\documentclass{amsart}

\usepackage{caption}

\usepackage{amsmath,amssymb}

\usepackage{graphicx}

\usepackage{epstopdf}

\usepackage{color}
\usepackage{subfigure}
\usepackage{ulem}
\usepackage{tikz}
\usepackage{amsthm}

\usetikzlibrary{decorations.pathmorphing, patterns,shapes}

\newtheorem{theorem}{Theorem}[section]

\theoremstyle{definition}
\newtheorem{definition}[theorem]{Definition}

\theoremstyle{remark}
\newtheorem{remark}[theorem]{Remark}

\numberwithin{equation}{section}

\begin{document}

\title[IST for nonlocal NLS with nonzero boundary conditions] {Inverse scattering transform for the nonlocal nonlinear Schr\"{o}dinger equation with nonzero boundary conditions}

\author{Mark J. Ablowitz}
\address{Department of Applied Mathematics, University of Colorado, Campus Box 526, Boulder, Colorado 80309-0526}
\email{mark.ablowitz@colorado.edu}
\thanks{}

\author{Xu-Dan Luo}
\address{Department of Applied Mathematics, University of Colorado at Boulder, Boulder, Colorado}
\email{lxdmathematics@gmail.com}

\author{Ziad H. Musslimani}
\address{Department of Mathematics, Florida State University, Tallahassee, FL 32306-4510}
\email{muslimani@math.fsu.edu}
\begin{abstract}
In 2013 a new nonlocal symmetry reduction of the well-known AKNS scattering problem was found; it was shown to give rise to a new nonlocal $PT$ symmetric and integrable Hamiltonian nonlinear Schr\"{o}dinger (NLS) equation. Subsequently, the inverse scattering transform was constructed for the case of rapidly decaying initial data and a family of spatially localized, time periodic one soliton solution were found.  In this paper, the inverse scattering transform for the nonlocal NLS equation with nonzero boundary conditions at infinity is presented in the four cases when the data at infinity have constant amplitudes. The direct and inverse scattering problems are analyzed. Specifically, the direct problem is formulated, the analytic properties of the eigenfunctions and scattering data and their symmetries are obtained. The inverse scattering problem is developed via a left-right Riemann-Hilbert problem in terms of a suitable uniformization variable and the time dependence of the scattering data is obtained. This leads to a method to linearize/solve the Cauchy problem. 
Pure soliton solutions are discussed and explicit 1-soliton solution and two 2-soliton solutions are provided for three of the four different cases corresponding to two different signs of nonlinearity and two different values of the phase difference between plus and minus infinity. In the one other case there are no solitons.
\end{abstract}
\maketitle
\section{Introduction}
Solitons are unique type of nonlinear wave that arise as a solution to integrable infinite dimensional Hamiltonian dynamical systems. They were first discovered by Zabusky and Kruskal while conducting numerical experiments on the Korteweg-deVries (KdV) equation. To their surprise, such solitons revealed an unusual particle-like behavior upon collisions despite the fact that they are inherently nonlinear ``objects". Their results sparked intense research interest on two parallel fronts. One is related to the physics and applications of solitons (or solitary waves) while the other is focused on the mathematical structure of integrable evolution equations.\\\\
From the physics point of view, solitons, or solitary waves, represent finite energy spatially localized structures that generally form as a balance between dispersion and
nonlinearity. They
have been theoretically predicted and observed in laboratory experiments in various settings in the
physical and optical sciences (see 
\cite{KIVSHAR}, \cite{Konotop}, \cite{Segev}   for extensive reviews).

 In fluid mechanics, they have been shown to appear as isolated humps in 
shallow water whereas in nonlinear optics they occur as a diffraction-free self guided nonlinear modes of a self-induced optical potential. Both disciplines provide exceptional situations where mathematical analysis, numerical simulations, mathematical modeling and laboratory experiments go hand in hand.

Mathematically speaking,
exactly solvable models play an essential role in the study of nonlinear wave propagation. 
There are many integrable equations that arise as important models in diverse physical phenomena. For example, the Korteweg-deVries (KdV) and the Kadomtsev-Petviashvili (KP) equations describe weakly nonlinear shallow water waves \cite{KDV,Ablowitz2} propagating in one and two dimensions respectively.  The cubic nonlinear Schr\"odinger (NLS) equation is also a physically important integrable model 
 \cite{ZS}. It describes the evolution of weakly nonlinear and quasi-monochromatic wave trains in media with cubic nonlinearities \cite{Ablowitz5}.
Solitons appear as a special class of solutions to these models  which are integrable in the sense that they admit an infinite number of conserved quantities.
The KdV, KP and NLS equations share the mathematical property that they all are exactly solvable evolution equations with many explicit solutions and linearizations known.

There are many other continuous and also discrete integrable evolution equations that are physically relevant. Applications are diverse and include problems in fluid mechanics, electromagnetics, gravitational waves, elasticity, fundamental physics and lattice dynamics, to name but a few \cite{Ablowitz1,Ablowitz2,Ablowitz3}.

Recently,  continuous and discrete integrable nonlinear nonlocal Schr\"odinger equations describing wave propagation in certain nonlinear $PT$ symmetric media were also found; remarkably, they have a very simple structure  \cite{AblowitzMusslimani,AblowitzMusslimani2}.

Generally speaking integrability is established once an infinite number of constants of motion or an infinite number of conservation laws are obtained. However considerably more information about the solution can be obtained if the inverse
scattering transform (IST) can be carried out \cite{AKNS}.

The inverse scattering transform (IST) to solve the initial-value problem with rapidly decaying data for the nonlocal NLS equation
\begin{equation}
\label{E:nonlocal NLS}
iq_{t}(x,t)=q_{xx}(x,t)- 2\sigma q^{2}(x,t)q^{*}(-x,t),
\end{equation}
where $q^*(x,t)$ denotes complex conjugate of $q(x,t), x\in\mathbb{R}, t\ge 0$ and $\sigma=\mp 1$,   has been developed in \cite{AblowitzMusslimani,AblowitzMusslimani3}. In this problem there are novel symmetry relations which relate analytic eigenfunctions as $x \rightarrow \infty$ to those as $x \rightarrow -\infty$. In turn it is useful to employ Riemann-Hilbert (RH) problems from both the left and right in order to  effectively develop the inverse scattering for both sets of eigenfunctions: i.e. those defined as
$x \rightarrow  \pm \infty$. We refer to this as  left-right Riemann-Hilbert problems. This is different from the classical NLS equation where the inverse problem is carried out using a RH problem using corresponding symmetries at either infinity \cite{Ablowitz3}. It is important to carry out the inverse scattering analysis not only to be able to solve the nonlinear equation but also because inverse scattering is important in its own right. The equation (\ref{E:nonlocal NLS}) was derived based upon on physical intuition. Recently this equation was derived in the physical context of magnetics \cite{ApplicNlclNLS16}.

It is well-known that the IST procedure for rapidly decaying potentials must be substantially modified when one is interested in potentials that do not decay as $|x|\rightarrow\infty$. This class of potentials is also relevant for
the nonlocal NLS equation, since it admits soliton solutions with nonzero boundary conditions (NZBCs).

For the classical NLS equation the first studies of NZBCs were done for the NLS equation \cite{ZS2}. The method to carry out the inverse problem employed  two Riemann surfaces associated with square root branch points in the eigenfunctions/scattering data. An improvement was made  with the introduction of a uniformization variable \cite{FT87}. This transforms the inverse problem to the more standard inverse problem in the upper lower/half planes in the new variable.  Subsequently a number of researchers have also studied NLS problems in this manner  cf. \cite{Prinari06,Demontis, Demontis2, BK14}.\\\\

In this paper, we consider the following nonzero boundary conditions
\begin{equation}
q(x,t)\rightarrow q_{\pm}(t)=q_{0}e^{i\theta_{\pm}(t)},\ \ as \ \ x\rightarrow\pm\infty,
\end{equation}
where $q_{0}>0$ is a constant , $0\leq \theta_{\pm}< 2\pi$ and  $\Delta \theta:=\theta_{+}-\theta_{-}$ is either $0$ or $\pi$.
If $\Delta \theta:=\theta_{+}-\theta_{-} \neq 0,\pi$ then the amplitude $q_{\pm}(t)$ is exponentially growing/decaying at one or the other infinity. We do not consider this situation.
We consider four different cases:  two different signs of $\sigma= \mp 1$ and two different values of $\Delta \theta=0,\pi $. First, we consider the case when $\sigma=-1, \Delta \theta = \pi. $
 We find the the following 1-soliton stationary in space, oscillating in time solution
\begin{equation}
q(x,t)=q_{0}\cdot e^{i(2q_{0}^{2}t+\theta_{+}- \pi)}\cdot \tanh \left[q_{0}x-i\theta_*
\right]
\end{equation}
where $\theta_*=\frac{1}{2}(\theta_{+}+\theta_{1}+\pi), \theta_1$ is a real constant related to the scattering data. This solution can be singular in the nongeneric case when $x=0$ and $\theta_{+}+\theta_{1}= 2n\pi, n \in \mathbb{Z}$. Apart from a complex phase shift, the above solution is similar to the well known black soliton solution in the standard integrable NLS equation. Second, we consider the case when $\sigma=-1, \Delta \theta=0.$ In this case a single eigenvalue is found to be in the continuous spectrum; there is no `proper exponentially decaying' pure one soliton solution. The simplest decaying pure reflectionless potential generates  a 2-soliton standing wave solution.
Third, we consider the case when $\sigma=1, \Delta \theta= \pi$. Here all solitons must arise from an even number of eigenvalues: $2N$. The simplest situation occurs when $N=1$ which leads to a 2-soliton traveling wave solution. Lastly, the case when $\sigma=1, \Delta \theta= 0$ is considered; in this case we show that there are no eigenvalues/solitons.

We use the uniformization methodology mentioned earlier for the nonlocal NLS problem with the above NZBCs. We first introduce a two sheeted Riemann surface and then introduce a suitable uniformization variable. There are a number of new features regarding the nonlocal NLS equation such as the introduction of important new symmetries which, when combined with a left-right RH problem allows us to construct the inverse scattering. There are situations when only an even number of solitons (eigenvalues) can be obtained and others in which there are no eigenvalues/solitons at all. In certain non-generic situations the solitons can be singular. We also study `box' like potentials and show that the eigenvalue spectrum is consistent with these results.\\\\
The outline of this paper is as follows. In section 2 some preliminaries are developed. The equation and compatible linear pair are given and the different nonzero boundary values at infinity that we will consider in this paper are presented. It is also shown that the only cases in which the amplitude at infinity are not exponentially growing/decaying are when $\sigma= \mp 1, \Delta \theta= 0,
\pi$. In section 3 the  direct scattering theory is analyzed and the analytic structure of the eigenfunctions and associated scattering data are found. From the symmetry of the potentials the corresponding symmetry of the eigenfunctions and scattering data are deduced. 
 A suitable uniformization variable is introduced and the inverse scattering from both the right and left is developed and pure reflectionless potentials and trace formulae are obtained.
From the time dependence of the scattering data the IST is constructed; pure soliton solutions are discussed and an explicit one soliton solution is given.

The analysis of the other three cases is carried out in a  the similar manner in subsequent sections. However the details, due to the underlying branch structure of the scattering data, are quite different in each case.
\section{Preliminaries}
The nonlocal nonlinear Schr\"{o}dinger (NLS) equation (\ref{E:nonlocal NLS})
 is associated with the following $2\times2$
compatible system  \cite{AblowitzMusslimani}:
\begin{equation}
v_{x}=Xv=
\left(\begin{array}{cc}
-ik& q(x,t)\\
\sigma q^{*}(-x,t)& ik
\end{array}\right)v,
\end{equation}
\begin{equation}
v_{t}=Tv=
\left(\begin{array}{cc}
2ik^{2}+i\sigma q(x,t)q^{*}(-x,t)& -2kq(x,t)-iq_{x}(x,t)\\
-2k\sigma q^{*}(-x,t)-i\sigma q^{*}_{x}(-x,t)& -2ik^{2}-i\sigma q(x,t)q^{*}(-x,t)
\end{array}\right)v,
\end{equation}
where $q(x,t)$ is a complex-valued function of the real variables $x$ and $t$.

Alternatively,the space part of the compatible system may be written in the form
\begin{equation}
\label{E:nonlocal ode}
v_{x}=(ikJ+Q)v, \ \ \ x\in \mathbb{R},
\end{equation}
where
\begin{equation}J=
\left(\begin{matrix}
-1& 0\\
0& 1
\end{matrix}\right),\ \ \
Q=
\left(\begin{matrix}
0& q(x,t)\\
\sigma q^{*}(-x,t)& 0
\end{matrix}\right),
\end{equation}
Here, $q(x,t)$ is called the potential and $k$ is a complex spectral parameter. In general as
$x \rightarrow \pm \infty, q \rightarrow q_{\pm}(t).$ Then equation (\ref{E:nonlocal NLS}) simplifies to
\begin{equation}
\label{NLS+}
iq_{+,t} =-2\sigma q_+^2q_-^*  \mbox{~~~as~~~~} ~~~~~x \rightarrow +\infty
\end{equation}
and

\begin{equation}
\label{NLS-}
iq_{-,t}=-2\sigma q_-^2q_+^*  \mbox{~~~as~~~~} ~~~~~x \rightarrow -\infty
\end{equation}
From the above equations we find the conserved quantity

\begin{equation}
\label{NLSinftycons}
q_+q_-^*=C_0,   ~~~C_0 ~~~\mbox{~~~constant~~~~}
\end{equation}
The solutions to equations (\ref{NLS+}-\ref{NLS-}) are then given by

\begin{equation}
\label{SolNLS_+}
q_{+} (t)= q_{0,+} e^{2i\sigma C_0t} ~~ \mbox{~~~as~~~~} ~~~~~x \rightarrow +\infty
\end{equation}
and
\begin{equation}
\label{SolNLS_-}
q_{-}(t) = q_{0,-} e^{2i\sigma C_0^*t} ~~ \mbox{~~~as~~~~} ~~~~~x \rightarrow -\infty
\end{equation}
where $q_{0,\pm}$ are constant. Further, if we take
\begin{equation}
\label{ConstInfty}
q_{\pm}(t) = q_{0,\pm}e^{i \theta_{\pm}}
\end{equation}
where $q_{0,\pm}>0$, then
\begin{equation}
\label{Constinfty2}
C_0= q_{0,+} q_{0,-}e^{i\Delta \theta}, ~\mbox{~~~where~~~~~~} \Delta \theta= \theta_+-\theta_-= \text{const}
\end{equation}
If
$$\Delta \theta =0 ~~\mbox{~~~or~~~} \Delta \theta =\pi $$
then $C_0$ is real. Otherwise it is complex and the background either grows or decays exponentially as $|t| \rightarrow \infty$. In this paper we shall only consider the cases $\Delta \theta =0 ~~\mbox{or~~} \pi$. For convenience we will take

$$q_{0,\pm}=q_0.$$
There is no material difference in the analysis if we  take $q_{0,+} \neq q_{0,-}$.

\bigskip
We also note that as $x\rightarrow\pm\infty$, the eigenfunctions of the scattering problem asymptotically satisfy
\begin{equation}
\label{E:large x}
\left(\begin{array}{cc}
v_{1}\\
v_{2}
\end{array}\right)_{x}
=
\left(\begin{array}{cc}
-ik& q_{0}e^{\mp2\sigma q_{0}^{2}t\sin \Delta\theta}\cdot e^{i(2\sigma q_{0}^{2}t\cos\Delta\theta+\theta_{\pm})}\\
\sigma q_{0}e^{\pm2\sigma q_{0}^{2}t\sin \Delta\theta}\cdot e^{-i(2\sigma q_{0}^{2}t\cos\Delta\theta+\theta_{\mp})}& ik
\end{array}\right)
\left(\begin{array}{cc}
v_{1}\\
v_{2}
\end{array}\right),
\end{equation}
i.e.,
\begin{equation}
v_{x}=(ikJ+Q_{\pm}(t))v,
\end{equation}
\begin{equation}
Q_{\pm}(t)=\left(\begin{array}{cc}
0& q_{0}e^{\mp2\sigma q_{0}^{2}t\sin \Delta\theta}\cdot e^{i(2\sigma q_{0}^{2}t\cos\Delta\theta+\theta_{\pm})}\\
\sigma q_{0}e^{\pm2\sigma q_{0}^{2}t\sin \Delta\theta}\cdot e^{-i(2\sigma q_{0}^{2}t\cos\Delta\theta+\theta_{\mp})}& 0
\end{array}\right),
\label{EVA}
\end{equation}
where
\begin{equation}
\label{NZBCA}
q(x,t)\rightarrow q_{\pm}(t)=q_{0}e^{\mp2\sigma q_{0}^{2}t\sin \Delta\theta}\cdot e^{i(2\sigma q_{0}^{2}t\cos\Delta\theta+\theta_{\pm})},\ \ as \ \ x\rightarrow\pm\infty.
\end{equation}
Here, $q_{0}>0, ~0\leq \theta_{\pm}<2\pi.$
\section{The case of $\sigma=-1$ with $\theta_{+}-\theta_{-}=\pi$}
\subsection{Direct Scattering}
In this section we consider the nonzero boundary conditions (NZBCs) given above in (\ref{NZBCA}) and $\sigma=-1$, $\Delta \theta:=\theta_{+}-\theta_{-}=\pi$.
With this condition, equation (\ref{EVA}) conveniently reduces to
\begin{equation}
\frac{\partial^{2} v_{j}}{\partial x^{2}}=-(k^{2}+q_{0}^{2}e^{i\Delta \theta})v_{j}=-(k^{2}-q_{0}^{2})v_{j}, \ \ \ j=1,2.
\end{equation}

Each of the two equations has two linearly independent solutions $e^{i\lambda x}$ and $e^{-i\lambda x}$ as $|x|\rightarrow\infty$, where
$\lambda=\sqrt{k^{2}-q_{0}^{2}}$.
The variable $k$ is then thought of as belonging to a Riemann surface $\mathbb{K}$ consisting of two sheets $\mathbb{C}_{1}$ and $\mathbb{C}_{2}$ 
with the complex plane cut along $(-\infty, -q_{0}]\cup [q_{0}, +\infty)$
with its edges glued in such a way that $\lambda(k)$ is continuous through the cut. We introduce the local polar coordinates
\begin{equation}
k-q_{0}=r_{1}e^{i\theta_{1}}, \ \ \ 0\leq\theta_{1}<2\pi,
\end{equation}
\begin{equation}
k+q_{0}=r_{2}e^{i\theta_{2}}, \ \ \ -\pi\leq\theta_{2}<\pi,
\end{equation}
where $r_{1}=|k-q_{0}|$ and $r_{2}=|k+q_{0}|$. Then the function $\lambda(k)$ becomes single-valued on $\mathbb{K}$, i.e.,
\begin{equation}
\label{E:lambda}
\lambda(k)=\left\{\begin{array}{ll}
\lambda_{1}(k)=(r_{1}r_{2})^{\frac{1}{2}}\cdot e^{i\frac{\theta_{1}+\theta_{2}}{2}},\ \ k\in\mathbb{C}_{1},\\
\lambda_{2}(k)=-(r_{1}r_{2})^{\frac{1}{2}}\cdot e^{i\frac{\theta_{1}+\theta_{2}}{2}}, \ \ k\in\mathbb{C}_{2}.\\
\end{array}\right.
\end{equation}
Moreover, if $k\in\mathbb{C}_{1}$, then $\Im \lambda\geq 0$; and if $k\in\mathbb{C}_{2}$, then $\Im \lambda\leq 0$.
Hence, the variable $\lambda$ is thought of as belonging to the complex plane consisting of the upper half plane $U_{+}$ :$\Im \lambda\geq 0$, and lower half plane $U_{-}$
:$\Im \lambda\leq 0$,  glued together along the real axis; the transition occurs at  $\Im \lambda= 0$. The transformation $k\rightarrow\lambda$ maps $\mathbb{C}_{1}$ onto $U_{+}$, $\mathbb{C}_{2}$ onto $U_{-}$, the cut $(-\infty, -q_{0}]\cup [q_{0}, +\infty)$ onto the real axis, and the points $\pm q_{0}$ to $0$; see also the figures 1,2 below.


\begin{figure}
\begin{tabular*}{\textwidth}{ccc}
\hspace*{-1cm}
\begin{minipage}{\dimexpr0.5\textwidth-2\tabcolsep}

\begin{tikzpicture}

\draw (0,-.1) -- (0.9,1.5) -- (5.8,1.5) -- (5,-.1);
\draw (-.1,-.4) -- (-1,-2) -- (4,-2) -- (4.9,-.4);

\draw (0,-.1) -- (1.8,-.1);
\draw (-.1,-.4) -- (1.8,-.4);
\draw [-stealth]  (1.8,-.1) -- (0.6,-.1);
\draw [-stealth] (-.1,-.4) -- (0.7,-.4);
\draw [-stealth]  (1.8,-.1) -- (1.5,-.1);
\draw [-stealth] (-.1,-.4) -- (1.6,-.4);

\draw (5,-.1) -- (3,-.1);
\draw (4.9,-.4) -- (3,-.4);
\draw  [-stealth] (3,-.1) -- (4.5,-.1) ;
\draw  [-stealth] (4.9,-.4) -- (4.4,-.4);
\draw  [-stealth] (3,-.1) -- (3.5,-.1) ;
\draw  [-stealth] (4.9,-.4) -- (3.4,-.4);

\draw  (1.8,-.1) arc (90:-90:0.15);
\draw  (3,-.4) arc (270:90:0.15);

\draw (0,-4.1) -- (0.9,-2.5) -- (5.8,-2.5) -- (5,-4.1);
\draw (-.1,-4.4) -- (-1,-6) -- (4,-6) -- (4.9,-4.4);

\draw (0,-4.1) -- (1.8,-4.1);
\draw (-.1,-4.4) -- (1.8,-4.4);
\draw [-stealth]  (1.8,-4.1) -- (0.6,-4.1);
\draw [-stealth] (-.1,-4.4) -- (0.7,-4.4);
\draw [-stealth]  (1.8,-4.1) -- (1.5,-4.1);
\draw [-stealth] (-.1,-4.4) -- (1.6,-4.4);

\draw (5,-4.1) -- (3,-4.1);
\draw (4.9,-4.4) -- (3,-4.4);
\draw  [-stealth] (3,-4.1) -- (4.5,-4.1) ;
\draw  [-stealth] (4.9,-4.4) -- (4.4,-4.4);
\draw  [-stealth] (3,-4.1) -- (3.5,-4.1) ;
\draw  [-stealth] (4.9,-4.4) -- (3.4,-4.4);

\draw  (1.8,-4.1) arc (90:-90:0.15);
\draw  (3,-4.4) arc (270:90:0.15);

\fill (2.4,-.25)  circle (1pt);
\fill (1.95,-.25)  circle (1pt);
\fill (2.85,-.25)  circle (1pt);

\fill (1.95,-4.25)  circle (1pt);
\fill (2.85,-4.25)  circle (1pt);

\fill (2.35,-4.25)  circle (1pt);

\fill ((1,-.1)  circle (1pt);
\fill (1,-.4)  circle (1pt);
\fill (4,-.1)  circle (1pt);
\fill (4,-.4)  circle (1pt);

\fill ((1,-4.1)  circle (1pt);
\fill (1,-4.4)  circle (1pt);
\fill (4,-4.1)  circle (1pt);
\fill (4,-4.4)  circle (1pt);

\draw (1.9,0) node () { \tiny $-q_0$};
\draw (2.85,0) node () { \tiny $q_0$};
\draw (2.4,-.5) node () { \tiny $0_{\rm I}$};

\draw (1,.1) node () { \tiny $b_{\rm I}$};
\draw (1,-.6) node () { \tiny $c_{\rm I}$};
\draw (4,.1) node () { \tiny $a_{\rm I}$};
\draw (4,-.6) node () { \tiny $d_{\rm I}$};

\draw (0.8,1.8) node () { \tiny Sheet I: Im $\lambda > 0$};

\draw (3,1) node () { \tiny Im $k > 0$};
\draw (2,-1.5) node () { \tiny  Im $k < 0$};

\draw (5.4,-.1) node () { \tiny $\infty_+$};
\draw (5.1,-.5) node () { \tiny $\infty_-$};
\draw (-.2,0) node () { \tiny $\infty_+$};
\draw (-.4,-.5) node () { \tiny $\infty_-$};


\draw (0.5,-2.3) node () { \tiny  Sheet II: Im $\lambda < 0$};

\draw (3,-3) node () { \tiny Im $k > 0$};
\draw (2,-5.5) node () { \tiny  Im $k < 0$};

\draw (1.9,-4) node () { \tiny $-q_0$};
\draw (2.85,-4) node () { \tiny $q_0$};
\draw (2.4,-4.5) node () { \tiny $0_{\rm II}$};

\draw (1,-3.9) node () { \tiny $b_{\rm II}$};
\draw (1,-4.6) node () { \tiny $c_{\rm II}$};
\draw (4,-3.9) node () { \tiny $a_{\rm II}$};
\draw (4,-4.6) node () { \tiny $d_{\rm II}$};

\draw (5.5,-4) node () { \tiny $\infty_-$};
\draw (5.1,-4.5) node () { \tiny $\infty_+$};
\draw (-.2,-4) node () { \tiny $\infty_-$};
\draw (-.4,-4.5) node () { \tiny $\infty_+$};

\end{tikzpicture}

\captionsetup{font=footnotesize}
\captionof{figure}{The two-sheeted Riemann surface $\mathbb{K}$. \label{fig1}}
\end{minipage}%

&
\hspace*{1cm}
\begin{minipage}{\dimexpr0.5\textwidth-2\tabcolsep}

\begin{tikzpicture}

 \draw  (0,1) arc (180:360:2 and 0.5);
 \draw  [dashed] (4,1) arc (0:180:2 and 0.5);
 \draw (0,1) arc (180:0:2 and 2);

\draw  [-stealth] (0,1) arc (180:210:2 and 0.5);
 \draw  [-stealth] (4,1) arc (0:-30:2 and 0.5);
 \draw  [-stealth] (0,1) arc (180:250:2 and 0.5);
 \draw  [-stealth] (4,1) arc (0:-70:2 and 0.5);
 \draw  [-stealth] [dashed] (0,1) arc (180:70:2 and 0.5);
 \draw  [-stealth] [dashed] (4,1) arc (0:110:2 and 0.5);
  \draw  [-stealth] [dashed] (0,1) arc (180:30:2 and 0.5);
 \draw  [-stealth] [dashed] (4,1) arc (0:150:2 and 0.5);

\fill (0,1)  circle (1pt);
\fill (4,1)  circle (1pt);
\fill (2,1.5)  circle (1pt);
\fill (2,0.5)  circle (1pt);

\fill (1,.57)  circle (1pt);
\fill (1,1.43)  circle (1pt);
\fill (3,.57)  circle (1pt);
\fill (3,1.43)  circle (1pt);

\fill (2,2.9)  circle (1pt);

\draw (-.25,1) node () { \tiny $-q_0$};
\draw (4.2,1) node () { \tiny $q_0$};
\draw (2,0.3) node () { \tiny $\infty_+$};
\draw (2,1.7) node () { \tiny $\infty_-$};

\draw (0.9,.4) node () { \tiny $b_{\rm I}$};
\draw (3.1,.4) node () { \tiny $a_{\rm I}$};
\draw (0.9,1.6) node () { \tiny $c_{\rm I}$};
\draw (3.1,1.6) node () { \tiny $d_{\rm I}$};

\draw (1.5,2.25) node () { \tiny Im $k>0$ (front)};
\draw (4,3) node () { \tiny Im $k<0$ (back)};
\draw (2,3.2) node () { \tiny $0_{ I}$};
\draw (0,3.5) node () { \tiny Sheet I: Im $\lambda > 0$};

 \filldraw[fill=gray!50!white] (0,-1) arc (180:360:2 and 0.5);
 \filldraw[fill=gray!50!white]  (0,-1) arc (180:0:2 and 0.5);
 \draw (0,-1) arc (-180:0:2 and 2);

 \draw  [-stealth] (0,-1) arc (180:210:2 and 0.5);
 \draw  [-stealth] (4,-1) arc (0:-30:2 and 0.5);
 \draw  [-stealth] (0,-1) arc (180:250:2 and 0.5);
 \draw  [-stealth] (4,-1) arc (0:-70:2 and 0.5);
 \draw  [-stealth] [dashed] (0,-1) arc (180:70:2 and 0.5);
 \draw  [-stealth] [dashed] (4,-1) arc (0:110:2 and 0.5);
  \draw  [-stealth] [dashed] (0,-1) arc (180:30:2 and 0.5);
 \draw  [-stealth] [dashed] (4,-1) arc (0:150:2 and 0.5);

\fill (0,-1)  circle (1pt);
\fill (4,-1)  circle (1pt);
\fill (2,-1.5)  circle (1pt);
\fill (2,-0.5)  circle (1pt);

\fill (1,-.57)  circle (1pt);
\fill (1,-1.43)  circle (1pt);
\fill (3,-.57)  circle (1pt);
\fill (3,-1.43)  circle (1pt);

\fill (2,-2.9)  circle (1pt);

\draw (-.25,-1) node () { \tiny $-q_0$};
\draw (4.2,-1) node () { \tiny $q_0$};
\draw (2,-0.3) node () { \tiny $\infty_-$};
\draw (2,-1.7) node () { \tiny $\infty_+$};

\draw (0.9,-.4) node () { \tiny $b_{\rm II}$};
\draw (3.1,-.4) node () { \tiny $a_{\rm II}$};
\draw (0.9,-1.6) node () { \tiny $c_{\rm II}$};
\draw (3.1,-1.6) node () { \tiny $d_{\rm II}$};

\draw (1.5,-2.25) node () { \tiny Im $k<0$ (front)};
\draw (2.75,-1) node () { \tiny Im $k>0$ (back)};
\draw (2,-3.2) node () { \tiny $0_{ II}$};
\draw (0,0) node () { \tiny Sheet II: Im $\lambda < 0$};

\end{tikzpicture}

\captionsetup{font=footnotesize}
\captionof{figure}{The genus $0$ surface is topologically equivalent to $\mathbb{K}$. \label{fig2}}
\end{minipage}%

\end{tabular*}%

\end{figure}

\subsection{Eigenfunctions}
It is natural to introduce the eigenfunctions defined by the following boundary conditions
\begin{equation}
\label{E:asymptotic 1}
\phi(x,k)\sim w e^{-i\lambda x}, \ \ \ \overline{\phi}(x,k)\sim \overline{w}e^{i\lambda x}
\end{equation}
as $x\rightarrow-\infty$,
\begin{equation}
\label{E:asymptotic 2}
\psi(x,k)\sim v e^{i\lambda x}, \ \ \ \overline{\psi}(x,k)\sim \overline{v}e^{-i\lambda x}
\end{equation}
as $x\rightarrow +\infty$. We substitute the above into (\ref{E:large x}), obtaining
\begin{equation}
\label{E:boundary conditions 1}
w=\left(\begin{array}{cc}
\lambda+k\\
-iq_{+}^{*}
\end{array}\right), \ \ \
\overline{w}=\left(\begin{array}{cc}
-iq_{-}\\
\lambda+k
\end{array}\right),
\end{equation}
\begin{equation}
\label{E:boundary conditions 2}
v=\left(\begin{array}{cc}
-iq_{+}\\
\lambda+k
\end{array}\right), \ \ \
\overline{v}=
\left(\begin{array}{cc}
\lambda+k\\
-iq_{-}^{*}
\end{array}\right)
\end{equation}
which satisfy the boundary conditions, but they are not unique.
In the following analysis, it is convenient to consider functions with constant boundary conditions. We define the bounded eigenfunctions as follows:
\begin{equation}
\label{E:definition 1}
M(x,k)=e^{i\lambda x}\phi(x,k), \ \ \ \overline{M}(x,k)=e^{-i\lambda x}\overline{\phi}(x,k),
\end{equation}
\begin{equation}
\label{E:definition 2}
N(x,k)=e^{-i\lambda x}\psi(x,k), \ \ \ \overline{N}(x,k)=e^{i\lambda x}\overline{\psi}(x,k).
\end{equation}
The eigenfunctions can be represented by means of the following integral equations
\begin{equation}
M(x,k)=
\left(\begin{array}{cc}
\lambda+k\\
-iq_{+}^{*}
\end{array}\right)
+\int_{-\infty}^{+\infty}G_{-}(x-x',k)((Q-Q_{-})M)(x',k)dx',
\end{equation}
\begin{equation}
\overline{M}(x,k)=
\left(\begin{array}{cc}
-iq_{-}\\
\lambda+k
\end{array}\right)
+\int_{-\infty}^{+\infty}\overline{G}_{-}(x-x',k)((Q-Q_{-})M)(x',k)dx',
\end{equation}
\begin{equation}
N(x,k)=
\left(\begin{array}{cc}
-iq_{+}\\
\lambda+k
\end{array}\right)
+\int_{-\infty}^{+\infty}G_{+}(x-x',k)((Q-Q_{+})M)(x',k)dx',
\end{equation}
\begin{equation}
\overline{N}(x,k)=\left(\begin{array}{cc}
\lambda+k\\
-iq_{-}^{*}
\end{array}\right)
+\int_{-\infty}^{+\infty}\overline{G}_{+}(x-x',k)((Q-Q_{+})M)(x',k)dx'.
\end{equation}
Using the Fourier transform method, we get
\begin{equation}
G_{-}(x,k)=\frac{\theta(x)}{2\lambda}[(1+e^{2i\lambda x})\lambda I-i(e^{2i\lambda x}-1)(ikJ+Q_{-})],
\end{equation}
\begin{equation}
\overline{G}_{-}(x,k)=\frac{\theta(x)}{2\lambda}[(1+e^{-2i\lambda x})\lambda I+i(e^{-2i\lambda x}-1)(ikJ+Q_{-})],
\end{equation}
\begin{equation}
G_{+}(x,k)=-\frac{\theta(-x)}{2\lambda}[(1+e^{-2i\lambda x})\lambda I+i(e^{-2i\lambda x}-1)(ikJ+Q_{+})],
\end{equation}
\begin{equation}
\overline{G}_{+}(x,k)=-\frac{\theta(-x)}{2\lambda}[(1+e^{2i\lambda x})\lambda I-i(e^{2i\lambda x}-1)(ikJ+Q_{+})],
\end{equation}
where $\theta(x)$ is the Heaviside function, i.e., $\theta(x)=1$ if $x>0$ and $\theta(x)=0$ if $x<0$.
\begin{definition}
We say $f\in L^{1}(\mathbb{R})$  if $\int_{-\infty}^{+\infty}|f(x)|dx<\infty$, and $f\in L^{1,2}(\mathbb{R})$ if $\int_{-\infty}^{+\infty}|f(x)|\cdot(1+|x|)^{2}dx<\infty$.
\end{definition}
Then we have the following result see also \cite{Demontis}).
\begin{theorem}
\label{T:1}
Suppose the entries of $Q-Q_{\pm}$ belong to $L^{1}(\mathbb{R})$, then for each $x\in\mathbb{R}$, the eigenfunctions $M(x,k)$ and $N(x,k)$ are continuous for $k\in \overline{\mathbb{C}}_{1}\setminus\{\pm q_{0}\}$ and analytic for $k\in \mathbb{C}_{1}$, $\overline{M}(x,k)$ and $\overline{N}(x,k)$ are continuous for $k\in \overline{\mathbb{C}}_{2}\setminus\{\pm q_{0}\}$ and analytic for $k\in \mathbb{C}_{2}$. In addition, if the entries of $Q-Q_{\pm}$ belong to $L^{1,2}(\mathbb{R})$, then for each $x\in\mathbb{R}$, the eigenfunctions $M(x,k)$ and $N(x,k)$ are continuous for $k\in \overline{\mathbb{C}}_{1}$ and analytic for $k\in \mathbb{C}_{1}$, $\overline{M}(x,k)$ and $\overline{N}(x,k)$ are continuous for $k\in \overline{\mathbb{C}}_{2}$ and analytic for $k\in \mathbb{C}_{2}$.
\end{theorem}
\begin{proof}
For $k\in(-\infty,-q_{0})\cup (q_{0},+\infty)$, the matrices $P_{-i\lambda}(k)^{\pm}$ and $P_{i\lambda}^{\pm}$ are defined as follows:
\begin{equation}
P_{-i\lambda}^{\pm}(k)=
\frac{1}{2\lambda}\left(\begin{array}{cc}
\lambda+k& iq_{\pm}\\
-iq_{\mp}^{*}& \lambda-k
\end{array}\right), \ \ \
P_{i\lambda}^{\pm}(k)
=\frac{1}{2\lambda}\left(\begin{array}{cc}
\lambda-k& -iq_{\pm}\\
iq_{\mp}^{*}& \lambda+k
\end{array}\right).
\end{equation}
We have $(P_{i\lambda}^{\pm})^{2}=P_{i\lambda}^{\pm}$, $(P_{-i\lambda}^{\pm})^{2}=P_{-i\lambda}^{\pm}$, $P_{i\lambda}^{\pm}+P_{-i\lambda}^{\pm}=I_{2}$, $P_{i\lambda}^{\pm}P_{-i\lambda}^{\pm}=P_{-i\lambda}^{\pm}P_{i\lambda}^{\pm}=0$.
Moreover,
\begin{equation}
(ikJ+Q_{\pm})P_{-i\lambda}^{\pm}(k)=-i\lambda P_{-i\lambda}^{\pm}(k)
\end{equation}
and
\begin{equation}
(ikJ+Q_{\pm})P_{i\lambda}^{\pm}(k)=i\lambda P_{i\lambda}^{\pm}(k).
\end{equation}
Then we can rewrite the Green's functions in terms of the projectors to find
\begin{equation}
M(x,k)=
\left(\begin{array}{cc}
\lambda+k\\
-iq_{+}^{*}
\end{array}\right)
+\int_{-\infty}^{x}[P_{-i\lambda}^{-}+e^{2i\lambda(x-x')}P_{i\lambda}^{-}]((Q-Q_{-})M)(x',k)dx',
\end{equation}
\begin{equation}
\overline{M}(x,k)=
\left(\begin{array}{cc}
-iq_{-}\\
\lambda+k
\end{array}\right)
+\int_{-\infty}^{x}[e^{2i\lambda(x'-x)}P_{-i\lambda}^{-}+P_{i\lambda}^{-}]((Q-Q_{-})M)(x',k)dx',
\end{equation}
\begin{equation}
N(x,k)=
\left(\begin{array}{cc}
-iq_{+}\\
\lambda+k
\end{array}\right)
+\int_{x}^{+\infty}[e^{2i\lambda(x'-x)}P_{-i\lambda}^{+}+P_{i\lambda}^{+}]((Q-Q_{+})M)(x',k)dx',
\end{equation}
\begin{equation}
\overline{N}(x,k)=\left(\begin{array}{cc}
\lambda+k\\
-iq_{-}^{*}
\end{array}\right)
+\int_{x}^{+\infty}[P_{-i\lambda}^{+}+e^{2i\lambda(x-x')}P_{i\lambda}^{+}]((Q-Q_{+})M)(x',k)dx'.
\end{equation}
The projections $P_{-i\lambda}^{\pm}(k)$ and $P_{i\lambda}^{\pm}(k)$ admit a natural continuation to $k\in\mathbb{K}\setminus\{\pm q_{0}\}$, i.e., $\lambda\in\mathbb{C}\setminus\{0\}$. Taking into account that these projections are singular matrices, we can show that 
their $l^{2}$-norm are given by
\begin{equation}
\parallel P_{-i\lambda}^{\pm}(k)\parallel_2=\parallel P_{i\lambda}^{\pm}(k)\parallel_2=
\sqrt{\frac{|\lambda-k|^{2}+|\lambda+k|^{2}+2q_{0}^{2}}{4|\lambda|^{2}}}=\frac{\sqrt{2|\lambda|^{2}+2|k|^{2}+2q_{0}^{2}}}{2|\lambda|}.
\end{equation}
In particular, if $k\in(-\infty, -q_{0}]\cup [q_{0}, +\infty)$, then
\begin{equation}
\parallel P_{-i\lambda}^{\pm}(k)\parallel_{2}=\frac{|k|}{|\lambda|}.
\end{equation}
We consider the Neumann series
\begin{equation}
M(x,k)=\sum_{n=0}^{\infty}M^{(n)}(x,k),
\end{equation}
where
\begin{equation}
M^{(0)}(x,k)=\left(\begin{array}{cc}
\lambda+k\\
-iq_{0}e^{-2iq_{0}^{2}t-i\theta_{+}}
\end{array}\right),
\end{equation}
\begin{equation}
M^{(n+1)}(x,k)=\int_{-\infty}^{x}[P_{-i\lambda}^{-}+e^{2i\lambda(x-x')}P_{i\lambda}^{-}]((Q-Q_{-})M^{(n)})(x',k)dx'.
\end{equation}
Then
\begin{equation}
\parallel M^{(n+1)}(x,k)\parallel\leq \sqrt{|\lambda+k|^{2}+q_{0}^{2}}+\frac{2|k|}{|\lambda|}\int_{-\infty}^{x}\parallel Q(x')-Q_{-}\parallel
\cdot\parallel M^{(n)}(x',k) \parallel dx'
\end{equation}
if $x'\leq x$ and $\lambda\in U_{+}\cup (-\infty,0)\cup (0,+\infty)$.

Since the entries of $Q-Q_{-}$ belong to $L^{1}(\mathbb{R})$, using the identities
\begin{equation}
\begin{split}
&\frac{1}{j!}\int_{-\infty}^{x}|f(\xi)|\left[\int_{-\infty}^{\xi}|f(\xi')|d\xi'\right]^{j}d\xi\\
&=\frac{1}{(j+1)!}\int_{-\infty}^{x}\frac{d}{d\xi}\left[\int_{-\infty}^{\xi}|f(\xi')|d\xi'\right]^{j+1}d\xi\\
&=\frac{1}{(j+1)!}\left[\int_{-\infty}^{x}|f(\xi)|d\xi\right]^{j+1},
\end{split}
\end{equation}
we can get the Neumann series itself is uniformly convergent for $k\in \overline{\mathbb{C}}_{1}\setminus\{\pm q_{0}\}$.
It follows that $M(x,k)$ is analytic for $k\in\mathbb{C}_{1}$ because a uniformly convergent series of analytic functions converges to an analytic function. Similarly, $M(x,k)$ is continuous for $k\in \overline{\mathbb{C}}_{1}\setminus\{\pm q_{0}\}$.

To extend the continuity properties at $k=\pm q_{0}$, we rewrite $P_{-i\lambda}^{-}+e^{2i\lambda(x-x')}P_{i\lambda}^{-}$ as
$I_{2}+[e^{2i\lambda(x-x')}-1]P_{i\lambda}^{-}$ and use the estimation
\begin{equation}
\begin{split}
\parallel I_{2}+[e^{2i\lambda(x-x')}-1]P_{i\lambda}^{-}\parallel
\leq 1+2|x-x'|\cdot|k|
&\leq \max\{1,2|k|\}(1+|x|)(1+|x'|)\\
&\leq \max\{1,2|k|\}(1+|x|)^{2}
\end{split}
\end{equation}
and the fact that $f\in L^{1,2}(\mathbb{R})\subseteq L^{1}(\mathbb{R})$, the proof then proceeds as before under the condition that the entries of $Q-Q_{-}$ belong to $L^{1,2}(\mathbb{R})$.
Similarly, the proof for $N(x,k)$, $\overline{M}(x,k)$ and $\overline{N}(x,k)$ can be done by using the above method.
\end{proof}
\begin{definition}
The Schwartz space or space of rapidly decreasing functions on $\mathbb{R}^{n}$ is the function space
\begin{equation}
S(\mathbb{R}^{n}):=\{f\in C^{\infty}(\mathbb{R}^{n}): \parallel f \parallel_{\alpha, \beta}<\infty, \forall \alpha, \beta\in \mathbb{Z}_{+}^{n}\},
\end{equation}
where $\alpha$, $\beta$ are multi-indices, $C^{\infty}(\mathbb{R}^{n})$ is the set of smooth functions from $\mathbb{R}^{n}$ to $\mathbb{C}$, and $\parallel f \parallel_{\alpha, \beta}=\sup_{x\in\mathbb{R}^{n}}|x^{\alpha}D^{\beta}f(x)|$. In particular, $x^{j}e^{-a|x|^{2}}\in S(\mathbb{R}^{n})$, where $j$ is a multi-index and $a$ is a positive real number.
\end{definition}
Then we have the following result.
\begin{theorem}
\label{T:Schwarz condition 1}
Suppose the entries of $Q-Q_{\pm}$ do not grow faster than $e^{-ax^{2}}$, where $a$ is a positive real number, then for each $x\in\mathbb{R}$, the eigenfunctions $M(x,k)$, $N(x,k)$, $\overline{M}(x,k)$ and $\overline{N}(x,k)$ are analytic in the Riemann surface $\mathbb{K}$.
\end{theorem}
\begin{proof}
The integral equation of $M(x,k)$ can be written in the alternative form
\begin{equation}
M_{1}(x,k)=\lambda+k+\int_{-\infty}^{x}\int_{-\infty}^{y}(q(y)-q_{-})(-q^{*}(-z)+q_{+}^{*})e^{2i\lambda(y-z)}M_{1}(z,k)dz dy
\end{equation}
and
\begin{equation}
M_{2}(x,k)=-iq_{+}^{*}+\int_{-\infty}^{x}e^{2i\lambda(x-y)}(-q^{*}(-y)+q_{+}^{*})M_{1}(y, k)dy,
\end{equation}
where $M(x,k)=\left(\begin{array}{cc}
M_{1}(x,k)\\
M_{2}(x,k)
\end{array}\right)$.
~\\
Since the entries of $Q-Q_{\pm}$ do not grow faster than $e^{-ax^{2}}$, it is easy to see that $M(x,k)$ is analytic in the Riemann surface $\mathbb{K}$. Similarly, we can prove $N(x,k)$, $\overline{M}(x,k)$ and $\overline{N}(x,k)$ are also analytic in $\mathbb{K}$.
\end{proof}
\subsection{Scattering data}
The two eigenfunctions with fixed boundary conditions as $x\rightarrow-\infty$ are linearly independent, as are the two eigenfunctions with fixed boundary conditions as $x\rightarrow+\infty$. Indeed, if $u(x,k)=(u_{1}(x,k), u_{2}(x,k))^{T}$ and $v(x,k)=(v_{1}(x,k), v_{2}(x,k))^{T}$ are any two solutions of (\ref{E:nonlocal ode}), we have
\begin{equation}
\frac{d}{dx}W(u,v)=0,
\end{equation}
where the Wronskian of $u$ and $v$, $W(u,v)$ is given by
\begin{equation}
W(u,v)=u_{1}v_{2}-u_{2}v_{1}.
\end{equation}
From the asymptotics (\ref{E:asymptotic 1}) and (\ref{E:asymptotic 2}), it follows that
\begin{equation}
\label{E:Wronskian 1}
W(\phi,\overline{\phi})=\lim_{x\rightarrow-\infty}W(\phi(x,k),\overline{\phi}(x,k))=2\lambda(\lambda+k)
\end{equation}
and
\begin{equation}
\label{E:Wronskian 2}
W(\psi,\overline{\psi})=\lim_{x\rightarrow+\infty}W(\psi(x,k),\overline{\psi}(x,k))=-2\lambda(\lambda+k),
\end{equation}
which proves that the functions $\phi(x,k)$ and $\overline{\phi}(x,k)$ are linearly independent, as are $\psi$ and $\overline{\psi}$,
with only the branch points $\pm q_{0}$ being excluded. Hence, we can write $\phi(x,k)$ and $\overline{\phi}(x,k)$ as linear combinations of
$\psi(x,k)$ and $\overline{\psi}(x,k)$, or vice versa. Thus, the relations
\begin{equation}
\label{E:linear combination 1}
\phi(x,k)=b(k)\psi(x,k)+a(k)\overline{\psi}(x,k)
\end{equation}
and
\begin{equation}
\label{E:linear combination 2}
\overline{\phi}(x,k)=\overline{a}(k)\psi(x,k)+\overline{b}(k)\overline{\psi}(x,k)
\end{equation}
hold for any $k$ such that all four eigenfunctions exist. Combining (\ref{E:Wronskian 1}) and (\ref{E:Wronskian 2}), we can deduce
that the scattering data satisfy the following characterization equation
\begin{equation}
a(k)\overline{a}(k)-b(k)\overline{b}(k)=1.
\end{equation}
The scattering data can be represented in terms of  Wronskians of the eigenfunctions, i.e.,
\begin{equation}
a(k)=\frac{W(\phi(x,k),\psi(x,k))}{W(\overline{\psi}(x,k),\psi(x,k))}=\frac{W(\phi(x,k),\psi(x,k))}{2\lambda(\lambda+k)},
\end{equation}
\begin{equation}
\overline{a}(k)=-\frac{W(\overline{\phi}(x,k),\overline{\psi}(x,k))}{W(\overline{\psi}(x,k),\psi(x,k))}
=-\frac{W(\overline{\phi}(x,k),\overline{\psi}(x,k))}{2\lambda(\lambda+k)},
\end{equation}
\begin{equation}
b(k)=-\frac{W(\phi(x,k),\overline{\psi}(x,k))}{W(\overline{\psi}(x,k),\psi(x,k))}=-\frac{W(\phi(x,k),\overline{\psi}(x,k))}{2\lambda(\lambda+k)},
\end{equation}
\begin{equation}
\overline{b}(k)=\frac{W(\overline{\phi}(x,k),\psi(x,k))}{W(\overline{\psi}(x,k),\psi(x,k))}
=\frac{W(\overline{\phi}(x,k),\psi(x,k))}{2\lambda(\lambda+k)}.
\end{equation}
Then from the analytic behavior of the eigenfunctions we have the following theorem (see also \cite{Demontis}).
\begin{theorem}
\label{T:2}
Suppose the entries of $Q-Q_{\pm}$ belong to $L^{1}(\mathbb{R})$, then $a(k)$ is continuous for $k\in \overline{\mathbb{C}}_{1}\setminus\{\pm q_{0}\}$ and analytic for $k\in \mathbb{C}_{1}$, and $\overline{a}(k)$ is continuous for $k\in \overline{\mathbb{C}}_{2}\setminus\{\pm q_{0}\}$ and analytic for $k\in \mathbb{C}_{2}$. Moreover, $b(k)$ and $\overline{b}(k)$ are continuous in $k\in(-\infty,-q_{0})\cup (q_{0},+\infty)$. In addition, if the entries of $Q-Q_{\pm}$ belong to $L^{1,2}(\mathbb{R})$,
then $a(k)\lambda(k)$ is continuous for $k\in \overline{\mathbb{C}}_{1}$ and analytic for $k\in \mathbb{C}_{1}$, and $\overline{a}(k)\lambda(k)$ is continuous for $k\in \overline{\mathbb{C}}_{2}$ and analytic for $k\in \mathbb{C}_{2}$. Moreover, $b(k)\lambda(k)$ and $\overline{b}(k)\lambda(k)$ are continuous for $k\in \mathbb{R}$. If the entries of $Q-Q_{\pm}$ do not grow faster than $e^{-ax^{2}}$, where $a$ is a positive real number, then $a(k)\lambda(k)$, $\overline{a}(k)\lambda(k)$, $b(k)\lambda(k)$ and $\overline{b}(k)\lambda(k)$ are analytic for $k\in\mathbb{K}$.
\end{theorem}

Note that (\ref{E:linear combination 1}) and (\ref{E:linear combination 2}) can be written as
\begin{equation}
\mu(x,k)=\rho(k)e^{2i\lambda x}N(x,k)+\overline{N}(x,k)
\end{equation}
and
\begin{equation}
\overline{\mu}(x,k)=N(x,k)+\overline{\rho}(k)e^{-2i\lambda x}\overline{N}(x,k),
\end{equation}
where $\mu(x,k)=M(x,k)a^{-1}(k)$, $\overline{\mu}(x,k)=\overline{M}(x,k)\overline{a}^{-1}(k)$, $\rho(k)=b(k)a^{-1}(k)$ and $\overline{\rho}(k)=\overline{b}(k)\overline{a}^{-1}(k)$. Introducing the $2\times2$ matrices
\begin{equation}
m_{+}(x,k)=(\mu(x,k), N(x,k)), \ \ \ m_{-}(x,k)=(\overline{N}(x,k), \overline{\mu}(x,k)),
\end{equation}
which are meromorphic in $\mathbb{C}_{1}$ and $\mathbb{C}_{2}$ respectively. Hence, we can write the Riemann-Hilbert problem or `jump' conditions in the $k$-plane as
\begin{equation}
m_{+}(x,k)-m_{-}(x,k)=m_{-}(x,k)\left(\begin{array}{cc}
-\rho(k)\overline{\rho}(k)& -\overline{\rho}(k)e^{-2i\lambda x}\\
\rho(k)e^{2i\lambda x}& 0
\end{array}\right)
\end{equation}
on the contour $\Sigma: k \in(-\infty,-q_{0}]\cup [q_{0},+\infty)$.


\begin{remark}
For the Riemann-Hilbert problem
\begin{equation}
F^{+}(\xi)-F^{-}(\xi)=F^{-}(\xi)g(\xi)
\end{equation}
on the contour $\Sigma$, where $g(\xi)$ is H\"{o}lder-continuously in $\Sigma$, we can consider the projection operators
\begin{equation}
(P_{j}(f))(k)=\frac{1}{2\pi i}\int_{\Sigma}\frac{\lambda(k)+\lambda(\xi)}{2\lambda(\xi)}\cdot\frac{f(\xi)}{\xi-k}d\xi, \ \ k\in \mathbb{C}_{j}, \ \ j=1,2,
\end{equation}
where $\frac{\lambda(k)+\lambda(\xi)}{2\lambda(\xi)}\cdot\frac{d\xi}{\xi-k}$ is the Weierstrass kernel and $\int_{\Sigma}$ denotes the integral along the oriented contour in sheet I of Figure \ref{fig1}. One can show that (\cite{Zverovich})
\begin{equation*}
\frac{\lambda(k)+\lambda(\xi)}{2\lambda(\xi)}\cdot\frac{d\xi}{\xi-k}=\frac{d\xi}{\xi-k}+ regular \ terms
\end{equation*}
for $(\xi, \lambda(\xi))\rightarrow (k, \lambda(k))$. If $f_{j}\ (j=1, 2)$, is sectionally analytic in $\mathbb{C}_{j}$ and rapidly decaying as $|k|\rightarrow\infty$ in the proper sheet, we have
\begin{equation}
(P_{j}(f_{j}))(k)=(-1)^{j-1}f_{j}(k), \ \ \ k\in \mathbb{C}_{j},
\end{equation}
\begin{equation}
(P_{j}(f_{l}))(k)=0, \ \ \ k\in \mathbb{C}_{j}.
\end{equation}
We obtain
\begin{equation}
F^{-}(k)=\frac{1}{2\pi i}\int_{\Sigma}\frac{\lambda(k)+\lambda(\xi)}{2\lambda(\xi)}\cdot\frac{F^{-}(\xi)g(\xi)}{\xi-k}d\xi,\ \ \ k\in\mathbb{C}_{2},
\end{equation}
\begin{equation}
F^{+}(k)=\frac{1}{2\pi i}\int_{\Sigma}\frac{\lambda(k)+\lambda(\xi)}{2\lambda(\xi)}\cdot\frac{F^{-}(\xi)g(\xi)}{\xi-k}d\xi,\ \ \ k\in\mathbb{C}_{1}.
\end{equation}
For $\xi\in\Sigma$, we can take limits from proper sheet; they are connected by the analogue of Plemelj-Sokhotskii formulas
\begin{equation}
F^{+}(\xi_{0})=\lim_{k\rightarrow\xi_{0}\in\Sigma, k\in\mathbb{C}_{1}}\frac{1}{2\pi i}\int_{\Sigma}\frac{\lambda(k)+\lambda(\xi)}{2\lambda(\xi)}\cdot\frac{F^{-}(\xi)g(\xi)}{\xi-k}d\xi,
\end{equation}
\begin{equation}
F^{-}(\xi_{0})=\lim_{k\rightarrow\xi_{0}\in\Sigma, k\in\mathbb{C}_{2}}\frac{1}{2\pi i}\int_{\Sigma}\frac{\lambda(k)+\lambda(\xi)}{2\lambda(\xi)}\cdot\frac{F^{-}(\xi)g(\xi)}{\xi-k}d\xi.
\end{equation}
\end{remark}

\subsection{Symmetry reductions}
The symmetry in the potential induces a symmetry between the eigenfunctions. Indeed, if $v(x,k)=(v_{1}(x,k), v_{2}(x,k))^{T}$ solves
(\ref{E:nonlocal ode}), then $(v_{2}^{*}(-x,-k^{*}), v_{1}^{*}(-x,-k^{*}))^{T}$ also solves (\ref{E:nonlocal ode}). Moreover, if $k\rightarrow -k^{*}$, according to (\ref{E:lambda}), we have $\theta_{1}\rightarrow \pi-\theta_{2}$ and $\theta_{2}\rightarrow \pi-\theta_{1}$. Hence, $\lambda_{j}^{*}(-k^{*})=-\lambda_{j}(k)$, where $j=1, 2$. Taking into account boundary conditions (\ref{E:boundary conditions 1}) and (\ref{E:boundary conditions 2}), we can obtain
\begin{equation}
\psi(x,k)=-\left(\begin{array}{cc}
0& 1\\
1& 0
\end{array}\right)\phi^{*}(-x,-k^{*})
\end{equation}
and
\begin{equation}
\overline{\psi}(x,k)=-\left(\begin{array}{cc}
0& 1\\
1& 0
\end{array}\right)\overline{\phi}^{*}(-x,-k^{*}).
\end{equation}
By (\ref{E:definition 1}) and (\ref{E:definition 2}), we can get the symmetry relations of the eigenfunctions, i.e.,
\begin{equation}
N(x,k)=-\left(\begin{array}{cc}
0& 1\\
1& 0
\end{array}\right)M^{*}(-x,-k^{*})
\end{equation}
and
\begin{equation}
\overline{N}(x,k)=-\left(\begin{array}{cc}
0& 1\\
1& 0
\end{array}\right)\overline{M}^{*}(-x,-k^{*}).
\end{equation}
From the Wronskian representations for the scattering data and the above symmetry relations, we have
\begin{equation}
a^{*}(-k^{*})=a(k),
\end{equation}
\begin{equation}
\overline{a}^{*}(-k^{*})=\overline{a}(k),
\end{equation}
\begin{equation}
b^{*}(-k^{*})=-\overline{b}(k).
\end{equation}

When using a particular single sheet for the Riemann surface of the function $\lambda^{2}=k^{2}-q_{0}^{2}$, the involution $(k, \lambda)\rightarrow (k, -\lambda)$ can only be considered across the cuts. The scattering data and eigenfunctions are defined by means of the corresponding values on the upper/lower edge of the cut; they are labeled with superscripts $\pm$ as clarified below.
Explicitly, one has
\begin{equation}
a^{\pm}(k)=\frac{W(\phi^{\pm}(x,k),\psi^{\pm}(x,k))}{2\lambda^{\pm}(\lambda^{\pm}+k)}, \ \ \ k\in (-\infty, -q_{0}]\cup [q_{0}, +\infty),
\end{equation}
\begin{equation}
\overline{a}^{\pm}(k)=-\frac{W(\overline{\phi}^{\pm}(x,k),\overline{\psi}^{\pm}(x,k))}{2\lambda^{\pm}(\lambda^{\pm}+k)}, \ \ \ k\in (-\infty, -q_{0}]\cup [q_{0}, +\infty),
\end{equation}
\begin{equation}
b^{\pm}(k)=-\frac{W(\phi^{\pm}(x,k),\overline{\psi}^{\pm}(x,k))}{2\lambda^{\pm}(\lambda^{\pm}+k)}, \ \ \ k\in (-\infty, -q_{0}]\cup [q_{0}, +\infty),
\end{equation}
\begin{equation}
\overline{b}^{\pm}(k)=\frac{W(\overline{\phi}^{\pm}(x,k),\psi^{\pm}(x,k))}{2\lambda^{\pm}(\lambda^{\pm}+k)},  \ \ \ k\in(-\infty, -q_{0}]\cup [q_{0}, +\infty).
\end{equation}
Using the notation $\lambda=\lambda^{+}=-\lambda^{-}$, we have the following symmetry:
\begin{equation}
\phi^{\mp}(x,k)=\frac{\lambda^{\mp}+k}{-iq_{-}}\overline{\phi}^{\pm}(x,k), \ \ \ \ \ \psi^{\mp}(x,k)=\frac{\lambda^{\mp}+k}{-iq_{-}^{*}}\overline{\psi}^{\pm}(x,k)
\end{equation}
for $k\in(-\infty, -q_{0}]\cup [q_{0}, +\infty)$.
Moreover,
\begin{equation}
a^{\pm}(k)=-\overline{a}^{\mp}(k), \ \ \ \ \ b^{\pm}(k)=-\frac{q_{0}^{2}}{q_{-}\cdot q_{+}}\cdot \overline{b}^{\mp}(k)
\end{equation}
for $k\in(-\infty, -q_{0}]\cup [q_{0}, +\infty)$.

%

\subsection{Uniformization coordinates}
Before discussing the properties of scattering data and solving the inverse problem, we introduce a uniformization variable $z$, defined by the conformal mapping:
\begin{equation}
z=z(k)=k+\lambda(k),
\end{equation}
where $\lambda= \sqrt{k^2-q_0^2}$ and the inverse mapping is given by
\begin{equation}
k=k(z)=\frac{1}{2}\left(z+\frac{q_{0}^{2}}{z}\right).
\end{equation}
Then
\begin{equation}
\lambda(z)=\frac{1}{2}\left(z-\frac{q_{0}^{2}}{z}\right).
\end{equation}
We observe that

(1) the upper sheet $\mathbb{C}_{1}$ and lower sheet $\mathbb{C}_{2}$ of the Riemann surface $\mathbb{K}$ are mapped onto the upper half plane $\mathbb{C}^{+}$ and lower half plane $\mathbb{C}^{-}$ of the complex $z-$ plane respectively;

(2) the cut $(-\infty, -q_{0}]\cup [q_{0}, +\infty)$ on the Riemann surface is mapped onto the real $z$ axis;

(3) the segments $[-q_{0}, q_{0}]$ on $\mathbb{C}_{1}$ and $\mathbb{C}_{2}$ are mapped onto the upper and lower semicircles of radius $q_{0}$ and centered at the origin of the complex $z-$ plane respectively.

From Theorem \ref{T:1}, we have the eigenfunctions $M(x,z)$ and $N(x,z)$ are analytic in the upper half $z-$ plane: i.e  $ z \in \mathbb{C}^{+}$, and $\overline{M}(x,z)$ and $\overline{N}(x,z)$ are analytic in the lower half $z-$ plane: i.e. $z \in \mathbb{C}^{-}$. Moreover, by Theorem \ref{T:2},
we find that $a(z)$ is analytic in the upper half $z-$ plane: $z \in \mathbb{C}^{+}$ and $\overline{a}(z)$ is analytic in the lower half plane: $z \in \mathbb{C}^{-}$.

\subsection{Symmetries via uniformization coordinates}
It is known that (1) when $z\rightarrow -z^{*}$, then $(k,\lambda)\rightarrow (-k^{*}, -\lambda^{*})$; (2) when $z\rightarrow \frac{q_{0}^{2}}{z}$, then $(k, \lambda)\rightarrow (k, -\lambda)$. Hence,
\begin{equation}
\psi(x,z)=-\left(\begin{array}{cc}
0& 1\\
1& 0
\end{array}\right)\phi^{*}(-x,-z^{*}),
\end{equation}
\begin{equation}
\overline{\psi}(x,z)=-\left(\begin{array}{cc}
0& 1\\
1& 0
\end{array}\right)\overline{\phi}^{*}(-x,-z^{*}),
\end{equation}
\begin{equation}
\phi\left(x,\frac{q_{0}^{2}}{z}\right)=\frac{\frac{q_{0}^{2}}{z}}{-iq_{-}}\cdot\overline{\phi}(x,z), \ \ \ \psi\left(x,\frac{q_{0}^{2}}{z}\right)=\frac{-iq_{+}}{z}\cdot\overline{\psi}(x,z), \ \ \ \Im z<0.
\end{equation}

Similarly, we can get
\begin{equation}
N(x,z)=-\left(\begin{array}{cc}
0& 1\\
1& 0
\end{array}\right)M^{*}(-x,-z^{*})
\end{equation}
and
\begin{equation}
\overline{N}(x,z)=-\left(\begin{array}{cc}
0& 1\\
1& 0
\end{array}\right)\overline{M}^{*}(-x,-z^{*}).
\end{equation}
Moreover,
\begin{equation}
a^{*}(-z^{*})=a(z), \ \ \ \Im z>0,
\end{equation}
\begin{equation}
\overline{a}^{*}(-z^{*})=\overline{a}(z), \ \ \ \Im z<0,
\end{equation}
\begin{equation}
b^{*}(-z^{*})=-\overline{b}(z),
\end{equation}
\begin{equation}
\label{E:scatering}
a\left(\frac{q_{0}^{2}}{z}\right)=-\overline{a}(z),\ \ \ \Im z<0; \ \ \ b\left(\frac{q_{0}^{2}}{z}\right)=-\frac{q_{0}^{2}}{q_{-}\cdot q_{+}}\cdot \overline{b}(z).
\end{equation}


We will assume that $a(z) ~\mbox{and}~ \bar{a}(z)$, have simple zero's in the upper/lower half $z-$ planes respectively. We assume that there are no multiple zero's and no zero's on $\Im z=0$.

\subsection{Asymptotic behavior of eigenfunctions and scattering data}
In order to solve the inverse problem, one has to determine the asymptotic behavior of eigenfunctions and scattering data both as $z\rightarrow\infty$ and as $z\rightarrow 0$. From the integral equations (in terms of Green's functions), we have
\begin{equation}
M(x,z)\sim\left\{\begin{array}{ll}
\left(\begin{array}{cc}
z\\
-iq^{*}(-x)
\end{array}\right), \ \ \ z\rightarrow\infty\\
\left(\begin{array}{cc}
-z\cdot\frac{q(x)}{q_{+}}\\
iq_{-}^{*}
\end{array}\right), \ \ \ z\rightarrow 0,\\
\end{array}\right.
\end{equation}

\begin{equation}
N(x,z)\sim\left\{\begin{array}{ll}
\left(\begin{array}{cc}
-iq(x)\\
z
\end{array}\right), \ \ \ z\rightarrow\infty\\
\left(\begin{array}{cc}
-iq_{+}\\
z\cdot \frac{q^{*}(-x)}{q_{-}^{*}}
\end{array}\right), \ \ \ z\rightarrow 0,\\
\end{array}\right.
\end{equation}

\begin{equation}
\overline{M}(x,z)\sim\left\{\begin{array}{ll}
\left(\begin{array}{cc}
-iq(x)\\
z
\end{array}\right), \ \ \ z\rightarrow\infty\\
\left(\begin{array}{cc}
iq_{+}\\
-z\cdot \frac{q^{*}(-x)}{q_{-}^{*}}
\end{array}\right), \ \ \ z\rightarrow 0,\\
\end{array}\right.
\end{equation}

\begin{equation}
\label{E:asymptotic 3}
\overline{N}(x,z)\sim\left\{\begin{array}{ll}
\left(\begin{array}{cc}
z\\
-iq^{*}(-x)
\end{array}\right), \ \ \ z\rightarrow\infty\\
\left(\begin{array}{cc}
z\cdot\frac{q(x)}{q_{+}}\\
-iq_{-}^{*}
\end{array}\right), \ \ \ z\rightarrow 0,\\
\end{array}\right.
\end{equation}

\begin{equation}
a(z)=
\left\{\begin{array}{ll}
1,\ \ \ z\rightarrow\infty,\\
-1, \ \ \ z\rightarrow 0,\\
\end{array}\right.
\end{equation}

\begin{equation}
\overline{a}(z)=
\left\{\begin{array}{ll}
1,\ \ \ z\rightarrow\infty,\\
-1, \ \ \ z\rightarrow 0,\\
\end{array}\right.
\end{equation}
\begin{equation}
\lim_{z\rightarrow\infty}zb(z)=0, \ \ \ \lim_{z\rightarrow0}\frac{b(z)}{z^{2}}=0.
\end{equation}

\subsection{Riemann-Hilbert problem via uniformization coordinates}
\subsubsection{Left scattering problem}
In order to take into account the behavior of the eigenfunctions, the `jump' conditions at the real $z-$ axis can be written from the left end as
\begin{equation}
\frac{M(x,z)}{za(z)}-\frac{\overline{N}(x,z)}{z}=\rho(z)e^{i\big(z-\frac{q_{0}^{2}}{z}\big)x}\cdot \frac{N(x,z)}{z}
\end{equation}
and
\begin{equation}
\frac{\overline{M}(x,z)}{z\overline{a}(z)}-\frac{N(x,z)}{z}=\overline{\rho}(z)e^{-i\big(z-\frac{q_{0}^{2}}{z}\big)x}\cdot \frac{\overline{N}(x,z)}{z},
\end{equation}
so that the functions will be bounded at infinity, though having an additional pole at $z=0$.
Note that $M(x,z)/a(z)$, as a function of
$z$, is defined in the upper half plane $\mathbb{C}^{+}$, where it has by assumption) simple poles $z_{j}$, i.e., $a(z_{j})=0$, and $\overline{M}(x,z)/\overline{a}(z)$,
is defined in the lower half plane $\mathbb{C}^{-}$, where it has simple poles $\overline{z}_{j}$, i.e., $\overline{a}(\overline{z}_{j})=0$. It follows that
\begin{equation}
\label{E:M}
M(x,z_{j})=b(z_{j})e^{i\big(z_{j}-\frac{q_{0}^{2}}{z_{j}}\big)x}\cdot N(x,z_{j})
\end{equation}
and
\begin{equation}
\overline{M}(x,\overline{z}_{j})=\overline{b}(\overline{z}_{j})e^{-i\big(\overline{z}_{j}-\frac{q_{0}^{2}}{\overline{z}_{j}}\big)x}\cdot \overline{N}(x,\overline{z}_{j}).
\end{equation}
Then subtracting the values at infinity, the induced pole at the origin and the poles, assumed simple in the upper/lower half planes respectively, at $a(z_j)=0, j=1,2...J$ and  $\bar{a}(\overline{z}_{j}), j=1,2...\bar{J} $ (later we will see that $J=\bar{J}$) gives
\begin{equation}
\label{E:jump 1}
\begin{split}
&\left[\frac{M(x,z)}{za(z)}-\left(\begin{array}{cc}
1\\
0
\end{array}\right)-
\frac{1}{z}\left(\begin{array}{cc}
0\\
-iq_{-}^{*}
\end{array}\right)
-\sum_{j=1}^{J}\frac{M(x,z_{j})}{(z-z_{j})z_{j}a'(z_{j})}\right]\\
&-\left[\frac{\overline{N}(x,z)}{z}-\left(\begin{array}{cc}
1\\
0
\end{array}\right)-
\frac{1}{z}\left(\begin{array}{cc}
0\\
-iq_{-}^{*}
\end{array}\right)
-\sum_{j=1}^{J}\frac{b(z_{j})e^{i\big(z_{j}-\frac{q_{0}^{2}}{z_{j}}\big)x}\cdot N(x,z_{j})}{(z-z_{j})z_{j}a'(z_{j})}\right]\\
&=\rho(z)e^{i\big(z-\frac{q_{0}^{2}}{z}\big)x}\cdot \frac{N(x,z)}{z}
\end{split}
\end{equation}
and
\begin{equation}
\label{E:jump 2}
\begin{split}
&\left[\frac{\overline{M}(x,z)}{z\overline{a}(z)}-\left(\begin{array}{cc}
0\\
1
\end{array}\right)-
\frac{1}{z}\left(\begin{array}{cc}
-iq_{+}\\
0
\end{array}\right)
-\sum_{j=1}^{\overline{J}}\frac{\overline{M}(x,\overline{z}_{j})}{(z-\overline{z}_{j})\overline{z}_{j}a'(\overline{z}_{j})}\right]\\
&-\left[\frac{N(x,z)}{z}-\left(\begin{array}{cc}
0\\
1
\end{array}\right)-
\frac{1}{z}\left(\begin{array}{cc}
-iq_{+}\\
0
\end{array}\right)
-\sum_{j=1}^{\overline{J}}\frac{\overline{b}(\overline{z}_{j})e^{-i\big(\overline{z}_{j}-\frac{q_{0}^{2}}{\overline{z}_{j}}\big)x}\cdot \overline{N}(x,\overline{z}_{j})}{(z-\overline{z}_{j})\overline{z}_{j}\overline{a}'(\overline{z}_{j})}\right]\\
&=\overline{\rho}(z)e^{-i\big(z-\frac{q_{0}^{2}}{z}\big)x}\cdot \frac{\overline{N}(x,z)}{z}.
\end{split}
\end{equation}
We now introduce the projection operators
\begin{equation}
P_{\pm}(f)(z)=\frac{1}{2\pi i}\int_{-\infty}^{+\infty}\frac{f(\xi)}{\xi-(z\pm i0)}d\xi,
\end{equation}
which are well-defined for any function $f(\xi)$ that is integrable on the real axis. If $f_{\pm}(\xi)$ is analytic in the upper/lower $z-$ plane and $f_{\pm}(\xi)$ is decaying at large $\xi$, then
\begin{equation}
P_{\pm}(f_{\pm})(z)=\pm f_{\pm}(z), \ \ \ P_{\mp}(f_{\pm})(z)=0.
\end{equation}
Applying $P_{-}$ to (\ref{E:jump 1}) and $P_{+}$ to (\ref{E:jump 2}), we can obtain
\begin{equation}
\begin{split}
\frac{\overline{N}(x,z)}{z}&=\left(\begin{array}{cc}
1\\
0
\end{array}\right)+
\frac{1}{z}\left(\begin{array}{cc}
0\\
-iq_{-}^{*}
\end{array}\right)
+\sum_{j=1}^{J}\frac{b(z_{j})e^{i\big(z_{j}-\frac{q_{0}^{2}}{z_{j}}\big)x}\cdot N(x,z_{j})}{(z-z_{j})z_{j}a'(z_{j})}\\
&+\frac{1}{2\pi i}\int_{-\infty}^{+\infty}\frac{\rho(\xi)}{\xi(\xi-z)}\cdot e^{i\big(\xi-\frac{q_{0}^{2}}{\xi}\big)x}\cdot N(x,\xi)d\xi
\end{split}
\end{equation}
and
\begin{equation}
\begin{split}
\frac{N(x,z)}{z}&=\left(\begin{array}{cc}
0\\
1
\end{array}\right)+
\frac{1}{z}\left(\begin{array}{cc}
-iq_{+}\\
0
\end{array}\right)
+\sum_{j=1}^{\overline{J}}\frac{\overline{b}(\overline{z}_{j})e^{-i\big(\overline{z}_{j}-\frac{q_{0}^{2}}{\overline{z}_{j}}\big)x}\cdot \overline{N}(x,\overline{z}_{j})}{(z-\overline{z}_{j})\overline{z}_{j}\overline{a}'(\overline{z}_{j})}\\
&-\frac{1}{2\pi i}\int_{-\infty}^{+\infty}\frac{\overline{\rho}(\xi)}{\xi(\xi-z)}\cdot e^{-i\big(\xi-\frac{q_{0}^{2}}{\xi}\big)x}\cdot \overline{N}(x,\xi)d\xi,
\end{split}
\end{equation}
i.e.,
\begin{equation}
\label{E:eigenfunction 1}
\begin{split}
\overline{N}(x,z)&=\left(\begin{array}{cc}
z\\
-iq_{-}^{*}
\end{array}\right)
+\sum_{j=1}^{J}\frac{z\cdot b(z_{j})e^{i\big(z_{j}-\frac{q_{0}^{2}}{z_{j}}\big)x}\cdot N(x,z_{j})}{(z-z_{j})z_{j}a'(z_{j})}\\
&+\frac{z}{2\pi i}\int_{-\infty}^{+\infty}\frac{\rho(\xi)}{\xi(\xi-z)}\cdot e^{i\big(\xi-\frac{q_{0}^{2}}{\xi}\big)x}\cdot N(x,\xi)d\xi
\end{split}
\end{equation}
and
\begin{equation}
\label{E:eigenfunction 2}
\begin{split}
N(x,z)&=\left(\begin{array}{cc}
-iq_{+}\\
z
\end{array}\right)
+\sum_{j=1}^{\overline{J}}\frac{z\cdot\overline{b}(\overline{z}_{j})e^{-i\big(\overline{z}_{j}-\frac{q_{0}^{2}}{\overline{z}_{j}}\big)x}\cdot \overline{N}(x,\overline{z}_{j})}{(z-\overline{z}_{j})\overline{z}_{j}\overline{a}'(\overline{z}_{j})}\\
&-\frac{z}{2\pi i}\int_{-\infty}^{+\infty}\frac{\overline{\rho}(\xi)}{\xi(\xi-z)}\cdot e^{-i\big(\xi-\frac{q_{0}^{2}}{\xi}\big)x}\cdot \overline{N}(x,\xi)d\xi.
\end{split}
\end{equation}

Since the symmetries are between eigenfunctions defined at both $\pm \infty$ we proceed to obtain the inverse scattering integral equations defined from the right end.

\subsubsection{Right scattering problem}
The right scattering problem can be written as
\begin{equation}
\psi(x,z)=\alpha(z)\overline{\phi}(x,z)+\beta(z)\phi(x,z)
\end{equation}
and
\begin{equation}
\overline{\psi}(x,z)=\overline{\alpha}(z)\phi(x,z)+\overline{\beta}(z)\overline{\phi}(x,z),
\end{equation}
where $\alpha(z)$, $\overline{\alpha}(z)$, $\beta(z)$ and $\overline{\beta}(z)$ are the right scattering data. Moreover, we can get the right scattering data and left scattering data satisfy the following relations
\begin{equation}
\overline{\alpha}(z)=\overline{a}(z), \ \ \ \alpha(z)=a(z), \ \ \ \overline{\beta}(z)=-b(z), \ \ \ \beta(z)=-\overline{b}(z).
\end{equation}
Thus,
\begin{equation}
\begin{split}
N(x,z)&=\alpha(z)\overline{M}(x,z)+\beta(z)M(x,z)e^{-i\big(z-\frac{q_{0}^{2}}{z}\big)x}\\
&=a(z)\overline{M}(x,z)-\overline{b}(z)M(x,z)e^{-i\big(z-\frac{q_{0}^{2}}{z}\big)x}
\end{split}
\end{equation}
and
\begin{equation}
\begin{split}
\overline{N}(x,z)&=\overline{\alpha}(z)M(x,z)+\overline{\beta}(z)\overline{M}(x,z)e^{i\big(z-\frac{q_{0}^{2}}{z}\big)x}\\
&=\overline{a}(z)M(x,z)-b(z)\overline{M}(x,z)e^{i\big(z-\frac{q_{0}^{2}}{z}\big)x}.
\end{split}
\end{equation}
The above two equations can be written as
\begin{equation}
\frac{N(x,z)}{za(z)}-\frac{\overline{M}(x,z)}{z}=-\frac{\overline{b}(z)}{a(z)}\cdot e^{-i\big(z-\frac{q_{0}^{2}}{z}\big)x}\cdot \frac{M(x,z)}{z}
\end{equation}
and
\begin{equation}
\frac{\overline{N}(x,z)}{z\overline{a}(z)}-\frac{M(x,z)}{z}=-\frac{b(z)}{\overline{a}(z)}\cdot e^{i\big(z-\frac{q_{0}^{2}}{z}\big)x}\cdot \frac{\overline{M}(x,z)}{z}.
\end{equation}
By the symmetry relations of scattering data, we have
\begin{equation}
\frac{N(x,z)}{za(z)}-\frac{\overline{M}(x,z)}{z}=\rho^{*}(-z^{*})\cdot e^{-i\big(z-\frac{q_{0}^{2}}{z}\big)x}\cdot \frac{M(x,z)}{z}
\end{equation}
and
\begin{equation}
\frac{\overline{N}(x,z)}{z\overline{a}(z)}-\frac{M(x,z)}{z}=\overline{\rho}^{*}(-z^{*})\cdot e^{i\big(z-\frac{q_{0}^{2}}{z}\big)x}\cdot \frac{\overline{M}(x,z)}{z},
\end{equation}
so that the functions will be bounded at infinity, though having an additional pole at $z=0$. Note that $N(x,z)/a(z)$, as a function of
$z$, is defined in the upper half plane $\mathbb{C}^{+}$, where it has simple poles $z_{j}$, i.e., $a(z_{j})=0$, and $\overline{N}(x,z)/\overline{a}(z)$,
is defined in the lower half plane $\mathbb{C}^{-}$, where it has simple poles $\overline{z}_{j}$, i.e., $\overline{a}(\overline{z}_{j})=0$. It follows that
\begin{equation}
N(x,z_{j})=-\overline{b}(z_{j})M(x,z_{j})e^{-i\big(z_{j}-\frac{q_{0}^{2}}{z_{j}}\big)x}
\end{equation}
and
\begin{equation}
\overline{N}(x,\overline{z}_{j})=-b(\overline{z}_{j})\overline{M}(x,\overline{z}_{j})e^{i\big(\overline{z}_{j}-\frac{q_{0}^{2}}{\overline{z}_{j}}\big)x}.
\end{equation}
Then, as before, subtracting the values at infinity, the induced pole at the origin and the poles, assumed simple in the upper/lower half planes respectively, at $a(z_j)=0, j=1,2...J$ and  $\bar{a}(\overline{z}_{j}), j=1,2...\bar{J} $, gives
\begin{equation}
\label{E:jump 3}
\begin{split}
&\left[\frac{N(x,z)}{za(z)}-\left(\begin{array}{cc}
0\\
1
\end{array}\right)
-\frac{1}{z}\left(\begin{array}{cc}
-iq_{-}\\
0
\end{array}\right)
-\sum_{j=1}^{J}\frac{N(x,z_{j})}{(z-z_{j})z_{j}a'(z_{j})}\right]\\
&-\left[\frac{\overline{M}(x,z)}{z}-\left(\begin{array}{cc}
0\\
1
\end{array}\right)
-\frac{1}{z}\left(\begin{array}{cc}
-iq_{-}\\
0
\end{array}\right)-\sum_{j=1}^{J}\frac{-\overline{b}(z_{j})M(x,z_{j})e^{-i\big(z_{j}-\frac{q_{0}^{2}}{z_{j}}\big)x}}{(z-z_{j})z_{j}a'(z_{j})}\right]\\
&=\rho^{*}(-z^{*})\cdot e^{-i\big(z-\frac{q_{0}^{2}}{z}\big)x}\cdot \frac{M(x,z)}{z}
\end{split}
\end{equation}
and
\begin{equation}
\label{E:jump 4}
\begin{split}
&\left[\frac{\overline{N}(x,z)}{z\overline{a}(z)}-\left(\begin{array}{cc}
1\\
0
\end{array}\right)
-\frac{1}{z}\left(\begin{array}{cc}
0\\
-iq_{+}^{*}
\end{array}\right)
-\sum_{j=1}^{\overline{J}}\frac{\overline{N}(x,\overline{z}_{j})}{(z-\overline{z}_{j})\overline{z}_{j}\overline{a}'(\overline{z}_{j})}\right]\\
&-\left[\frac{M(x,z)}{z}-\left(\begin{array}{cc}
1\\
0
\end{array}\right)
-\frac{1}{z}\left(\begin{array}{cc}
0\\
-iq_{+}^{*}
\end{array}\right)-\sum_{j=1}^{\overline{J}}\frac{-b(\overline{z}_{j})\overline{M}(x,\overline{z}_{j})e^{i\big(\overline{z}_{j}-\frac{q_{0}^{2}}{\overline{z}_{j}}\big)x}}
{(z-\overline{z}_{j})\overline{z}_{j}\overline{a}'(\overline{z}_{j})}\right]\\
&=\overline{\rho}^{*}(-z^{*})\cdot e^{i\big(z-\frac{q_{0}^{2}}{z}\big)x}\cdot \frac{\overline{M}(x,z)}{z}.
\end{split}
\end{equation}
Applying $P_{-}$ to (\ref{E:jump 3}) and $P_{+}$ to (\ref{E:jump 4}), we can obtain
\begin{equation}
\label{E:eigenfunction 3}
\begin{split}
\overline{M}(x,z)&=\left(\begin{array}{cc}
-iq_{-}\\
z
\end{array}\right)+
\sum_{j=1}^{J}\frac{-z\cdot\overline{b}(z_{j})M(x,z_{j})e^{-i\big(z_{j}-\frac{q_{0}^{2}}{z_{j}}\big)x}}{(z-z_{j})z_{j}a'(z_{j})}\\
&+\frac{z}{2\pi i}\int_{-\infty}^{+\infty}\frac{\rho^{*}(-\xi)}{\xi(\xi-z)}\cdot e^{-i\big(\xi-\frac{q_{0}^{2}}{\xi}\big)x}\cdot M(x,\xi)d\xi
\end{split}
\end{equation}
and
\begin{equation}
\label{E:eigenfunction 4}
\begin{split}
M(x,z)&=\left(\begin{array}{cc}
z\\
-iq_{+}^{*}
\end{array}\right)+
\sum_{j=1}^{\overline{J}}\frac{-z\cdot b(\overline{z}_{j})\overline{M}(x,\overline{z}_{j})e^{i\big(\overline{z}_{j}-\frac{q_{0}^{2}}{\overline{z}_{j}}\big)x}}
{(z-\overline{z}_{j})\overline{z}_{j}\overline{a}'(\overline{z}_{j})}\\
&-\frac{z}{2\pi i}\int_{-\infty}^{+\infty}\frac{\overline{\rho}^{*}(-\xi)}{\xi(\xi-z)}\cdot e^{i\big(\xi-\frac{q_{0}^{2}}{\xi}\big)x}\cdot \overline{M}(x,\xi)d\xi.
\end{split}
\end{equation}
\subsection{Recovery of the potentials}
In order to reconstruct the potential we use asymptotics . For example from equation (\ref{E:asymptotic 3}), we have
\begin{equation}
\frac{\overline{N}_{1}(x,z)}{z}\sim \frac{q(x)}{q_{+}}
\end{equation}
as $z\rightarrow 0$. By (\ref{E:eigenfunction 1}), we can get
\begin{equation}
\frac{\overline{N}_{1}(x,z)}{z}\sim 1+\sum_{j=1}^{J}\frac{b(z_{j})e^{i\big(z_{j}-\frac{q_{0}^{2}}{z_{j}}\big)x}}{-z_{j}^{2}a'(z_{j})}\cdot N_{1}(x,z_{j})+\frac{1}{2\pi i}\int_{-\infty}^{+\infty}\frac{\rho(\xi)}{\xi^{2}}\cdot e^{i\big(\xi-\frac{q_{0}^{2}}{\xi}\big)x}\cdot N_{1}(x,\xi)d\xi
\end{equation}
as $z\rightarrow 0$. Hence,
\begin{equation}
\label{asympN1c}
q(x)=q_{+}\cdot\left[1+\sum_{j=1}^{J}\frac{b(z_{j})e^{i\big(z_{j}-\frac{q_{0}^{2}}{z_{j}}\big)x}}{-z_{j}^{2}a'(z_{j})}\cdot N_{1}(x,z_{j})+\frac{1}{2\pi i}\int_{-\infty}^{+\infty}\frac{\rho(\xi)}{\xi^{2}}\cdot e^{i\big(\xi-\frac{q_{0}^{2}}{\xi}\big)x}\cdot N_{1}(x,\xi)d\xi\right].
\end{equation}
Note that the rapidly varying phase makes the integrals well defined at $\xi=0$.

\subsection{Closing the system}

We can find $J=\overline{J}$ from $a\left(\frac{q_{0}^{2}}{z}\right)=-\overline{a}(z)$. To close the system, by the symmetry relations between the eigenfunctions, we have
\begin{equation}
\label{E:system 1}
\begin{split}
\left(\begin{array}{cc}
-\overline{M}_{2}^{*}(-x,-z^{*})\\
-\overline{M}_{1}^{*}(-x,-z^{*})
\end{array}\right)&=\left(\begin{array}{cc}
z\\
-iq_{-}^{*}
\end{array}\right)
+\sum_{j=1}^{J}\frac{z\cdot b(z_{j})e^{i\big(z_{j}-\frac{q_{0}^{2}}{z_{j}}\big)x}}{(z-z_{j})z_{j}a'(z_{j})}\cdot \left(\begin{array}{cc}
N_{1}(x,z_{j})\\
N_{2}(x,z_{j})
\end{array}\right)\\
&+\frac{z}{2\pi i}\int_{-\infty}^{+\infty}\frac{\rho(\xi)}{\xi(\xi-z)}\cdot e^{i\big(\xi-\frac{q_{0}^{2}}{\xi}\big)x}\cdot \left(\begin{array}{cc}
N_{1}(x,\xi)\\
N_{2}(x,\xi)
\end{array}\right)d\xi,
\end{split}
\end{equation}
\begin{equation}
\label{E:system 2}
\begin{split}
\left(\begin{array}{cc}
N_{1}(x,z)\\
N_{2}(x,z)
\end{array}\right)&=\left(\begin{array}{cc}
-iq_{+}\\
z
\end{array}\right)
+\sum_{j=1}^{J}\frac{z\cdot\overline{b}(\overline{z}_{j})e^{-i\big(\overline{z}_{j}-\frac{q_{0}^{2}}{\overline{z}_{j}}\big)x} }{(z-\overline{z}_{j})\overline{z}_{j}\overline{a}'(\overline{z}_{j})}\cdot\left(\begin{array}{cc}
-\overline{M}_{2}^{*}(-x,-\overline{z}_{j}^{*})\\
-\overline{M}_{1}^{*}(-x,-\overline{z}_{j}^{*})
\end{array}\right)\\
&-\frac{z}{2\pi i}\int_{-\infty}^{+\infty}\frac{\overline{\rho}(\xi)}{\xi(\xi-z)}\cdot e^{-i\big(\xi-\frac{q_{0}^{2}}{\xi}\big)x}\cdot \left(\begin{array}{cc}
-\overline{M}_{2}^{*}(-x,-\xi)\\
-\overline{M}_{1}^{*}(-x,-\xi)
\end{array}\right)d\xi,
\end{split}
\end{equation}
\begin{equation}
\label{E:system 3}
\begin{split}
\left(\begin{array}{cc}
\overline{M}_{1}(x,z)\\
\overline{M}_{2}(x,z)
\end{array}\right)&=\left(\begin{array}{cc}
-iq_{-}\\
z
\end{array}\right)+
\sum_{j=1}^{J}\frac{-z\cdot\overline{b}(z_{j})e^{-i\big(z_{j}-\frac{q_{0}^{2}}{z_{j}}\big)x}}{(z-z_{j})z_{j}a'(z_{j})}\cdot \left(\begin{array}{cc}
-N_{2}^{*}(-x,-z_{j}^{*})\\
-N_{1}^{*}(-x,-z_{j}^{*})
\end{array}\right)\\
&+\frac{z}{2\pi i}\int_{-\infty}^{+\infty}\frac{\rho^{*}(-\xi)}{\xi(\xi-z)}\cdot e^{-i\big(\xi-\frac{q_{0}^{2}}{\xi}\big)x}\cdot \left(\begin{array}{cc}
-N_{2}^{*}(-x,-\xi)\\
-N_{1}^{*}(-x,-\xi)
\end{array}\right)d\xi,
\end{split}
\end{equation}
\begin{equation}
\label{E:system 4}
\begin{split}
\left(\begin{array}{cc}
-N_{2}^{*}(-x,-z^{*})\\
-N_{1}^{*}(-x,-z^{*})
\end{array}\right)&=\left(\begin{array}{cc}
z\\
-iq_{+}^{*}
\end{array}\right)+
\sum_{j=1}^{J}\frac{-z\cdot b(\overline{z}_{j})e^{i\big(\overline{z}_{j}-\frac{q_{0}^{2}}{\overline{z}_{j}}\big)x}}
{(z-\overline{z}_{j})\overline{z}_{j}\overline{a}'(\overline{z}_{j})}\cdot\left(\begin{array}{cc}
\overline{M}_{1}(x,\overline{z}_{j})\\
\overline{M}_{2}(x,\overline{z}_{j})
\end{array}\right)\\
&-\frac{z}{2\pi i}\int_{-\infty}^{+\infty}\frac{\overline{\rho}^{*}(-\xi)}{\xi(\xi-z)}\cdot e^{i\big(\xi-\frac{q_{0}^{2}}{\xi}\big)x}\cdot \left(\begin{array}{cc}
\overline{M}_{1}(x,\xi)\\
\overline{M}_{2}(x,\xi)
\end{array}\right)d\xi.
\end{split}
\end{equation}
Then
\begin{equation}
\label{E:closing system 1}
\begin{split}
&\left(\begin{array}{cc}
N_{1}(x,z)\\
N_{2}(x,z)
\end{array}\right)=\left(\begin{array}{cc}
-iq_{+}\\
z
\end{array}\right)
+\sum_{j=1}^{J}\frac{z\cdot\overline{b}(\overline{z}_{j})e^{-i\big(\overline{z}_{j}-\frac{q_{0}^{2}}{\overline{z}_{j}}\big)x} }{(z-\overline{z}_{j})\overline{z}_{j}\overline{a}'(\overline{z}_{j})}\cdot\\
&\left(\begin{array}{cc}
\overline{z}_{j}+\sum_{l=1}^{J}\frac{\overline{z}_{j}\cdot b(z_{l})e^{i\big(z_{l}-\frac{q_{0}^{2}}{z_{l}}\big)x}}{(\overline{z}_{j}-z_{l})z_{l}a'(z_{l})}\cdot N_{1}(x, z_{l})+\frac{\overline{z}_{j}}{2\pi i}\int_{-\infty}^{+\infty}\frac{\rho(\xi)}{\xi(\xi-\overline{z}_{j})}\cdot e^{i\big(\xi-\frac{q_{0}^{2}}{\xi}\big)x}\cdot N_{1}(x, \xi)d\xi\\
-iq_{-}^{*}+\sum_{l=1}^{J}\frac{\overline{z}_{j}\cdot b(z_{l})e^{i\big(z_{l}-\frac{q_{0}^{2}}{z_{l}}\big)x}}{(\overline{z}_{j}-z_{l})z_{l}a'(z_{l})}\cdot N_{2}(x, z_{l})+\frac{\overline{z}_{j}}{2\pi i}\int_{-\infty}^{+\infty}\frac{\rho(\xi)}{\xi(\xi-\overline{z}_{j})}\cdot e^{i\big(\xi-\frac{q_{0}^{2}}{\xi}\big)x}\cdot N_{2}(x, \xi)d\xi
\end{array}\right)\\
&-\frac{z}{2\pi i}\int_{-\infty}^{+\infty}\frac{\overline{\rho}(\xi)}{\xi(\xi-z)}\cdot e^{-i\big(\xi-\frac{q_{0}^{2}}{\xi}\big)x}\cdot\\
&\left(\begin{array}{cc}
\xi+\sum_{l=1}^{J}\frac{\xi\cdot b(z_{l})e^{i\big(z_{l}-\frac{q_{0}^{2}}{z_{l}}\big)x}}{(\xi-z_{l})z_{l}a'(z_{l})}\cdot N_{1}(x, z_{l})+\frac{\xi}{2\pi i}\int_{-\infty}^{+\infty}\frac{\rho(\eta)}{\eta(\eta-\xi)}\cdot e^{i\big(\eta-\frac{q_{0}^{2}}{\eta}\big)x}\cdot N_{1}(x, \eta)d\eta\\
-iq_{-}^{*}+\sum_{l=1}^{J}\frac{\xi\cdot b(z_{l})e^{i\big(z_{l}-\frac{q_{0}^{2}}{z_{l}}\big)x}}{(\xi-z_{l})z_{l}a'(z_{l})}\cdot N_{2}(x, z_{l})+\frac{\xi}{2\pi i}\int_{-\infty}^{+\infty}\frac{\rho(\eta)}{\eta(\eta-\xi)}\cdot e^{i\big(\eta-\frac{q_{0}^{2}}{\eta}\big)x}\cdot N_{2}(x, \eta)d\eta
\end{array}\right)d\xi,
\end{split}
\end{equation}
\begin{equation}
\label{E:closing system 2}
\begin{split}
&\left(\begin{array}{cc}
\overline{M}_{1}(x,z)\\
\overline{M}_{2}(x,z)
\end{array}\right)=\left(\begin{array}{cc}
-iq_{-}\\
z
\end{array}\right)+
\sum_{j=1}^{J}\frac{-z\cdot\overline{b}(z_{j})e^{-i\big(z_{j}-\frac{q_{0}^{2}}{z_{j}}\big)x}}{(z-z_{j})z_{j}a'(z_{j})}\cdot\\
& \left(\begin{array}{cc}
z_{j}+
\sum_{l=1}^{J}\frac{-z_{j}\cdot b(\overline{z}_{l})e^{i\big(\overline{z}_{l}-\frac{q_{0}^{2}}{\overline{z}_{l}}\big)x}}
{(z_{j}-\overline{z}_{l})\overline{z}_{l}\overline{a}'(\overline{z}_{l})}\cdot \overline{M}_{1}(x, \overline{z}_{l})-\frac{z_{j}}{2\pi i}\int_{-\infty}^{+\infty}\frac{\overline{\rho}^{*}(-\xi)}{\xi(\xi-z_{j})}\cdot e^{i\big(\xi-\frac{q_{0}^{2}}{\xi}\big)x}\cdot\overline{M}_{1}(x,\xi)d\xi\\
-iq_{+}^{*}+
\sum_{l=1}^{J}\frac{-z_{j}\cdot b(\overline{z}_{l})e^{i\big(\overline{z}_{l}-\frac{q_{0}^{2}}{\overline{z}_{l}}\big)x}}
{(z_{j}-\overline{z}_{l})\overline{z}_{l}\overline{a}'(\overline{z}_{l})}\cdot \overline{M}_{2}(x, \overline{z}_{l})-\frac{z_{j}}{2\pi i}\int_{-\infty}^{+\infty}\frac{\overline{\rho}^{*}(-\xi)}{\xi(\xi-z_{j})}\cdot e^{i\big(\xi-\frac{q_{0}^{2}}{\xi}\big)x}\cdot\overline{M}_{2}(x,\xi)d\xi
\end{array}\right)\\
&+\frac{z}{2\pi i}\int_{-\infty}^{+\infty}\frac{\rho^{*}(-\xi)}{\xi(\xi-z)}\cdot e^{-i\big(\xi-\frac{q_{0}^{2}}{\xi}\big)x}\cdot\\
&\left(\begin{array}{cc}
\xi+
\sum_{l=1}^{J}\frac{-\xi\cdot b(\overline{z}_{l})e^{i\big(\overline{z}_{l}-\frac{q_{0}^{2}}{\overline{z}_{l}}\big)x}}
{(\xi-\overline{z}_{l})\overline{z}_{l}\overline{a}'(\overline{z}_{l})}\cdot\overline{M}_{1}(x,\overline{z}_{l})-\frac{\xi}{2\pi i}\int_{-\infty}^{+\infty}\frac{\overline{\rho}^{*}(-\eta)}{\eta(\eta-\xi)}\cdot e^{i\big(\eta-\frac{q_{0}^{2}}{\eta}\big)x}\cdot \overline{M}_{1}(x,\eta)d\eta\\
-iq_{+}^{*}+
\sum_{l=1}^{J}\frac{-\xi\cdot b(\overline{z}_{l})e^{i\big(\overline{z}_{l}-\frac{q_{0}^{2}}{\overline{z}_{l}}\big)x}}
{(\xi-\overline{z}_{l})\overline{z}_{l}\overline{a}'(\overline{z}_{l})}\cdot\overline{M}_{2}(x,\overline{z}_{l})-\frac{\xi}{2\pi i}\int_{-\infty}^{+\infty}\frac{\overline{\rho}^{*}(-\eta)}{\eta(\eta-\xi)}\cdot e^{i\big(\eta-\frac{q_{0}^{2}}{\eta}\big)x}\cdot \overline{M}_{2}(x,\eta)d\eta
\end{array}\right)d\xi.
\end{split}
\end{equation}

We note that from eq (\ref{asympN1c})  $q(x)$  is given in  terms of the component $N_1$. We can use only the first component of eq (\ref{E:closing system 1}) to find this function and hence $q(x)$. Hence to complete the inverse scattering we reduce the problem to solving an integral equation in terms of the component $N_1$ only.

\subsection{Trace formula}
Since $a(z)$ and $\overline{a}(z)$ are analytic in the upper and lower $z-$ plane respectively. As mentioned above, we assume that $a(z)$ has simple zeros, which we call
$\widetilde{z}_{i}$ and $z_{j}$, where $\Re\widetilde{z}_{i}=0$ and $\Re z_{j}\neq 0$, then $-z_{j}^{*}$ is also a simple zero of $a(z)$ by $a^{*}(-z^{*})=a(z)$. By the symmetry relation $a\left(\frac{q_{0}^{2}}{z}\right)=-\overline{a}(z)$, we can deduce that $\overline{a}(z)$ has simple zeros $\frac{q_{0}^{2}}{z_{j}}$, $-\frac{q_{0}^{2}}{z_{j}^{*}}$ and $\frac{q_{0}^{2}}{\widetilde{z}_{i}}$.

We define
\begin{equation}
\gamma(z)=a(z)\cdot \prod_{j=1}^{J_{1}} \frac{z-\frac{q_{0}^{2}}{z_{j}}}{z-z_{j}}\cdot \frac{z+\frac{q_{0}^{2}}{z_{j}^{*}}}{z+z_{j}^{*}}
\cdot\prod_{i=1}^{J_{2}} \frac{z-\frac{q_{0}^{2}}{\widetilde{z}_{i}}}{z-\widetilde{z}_{i}},
\end{equation}
\begin{equation}
\overline{\gamma}(z)=\overline{a}(z)\cdot
\prod_{j=1}^{J_{1}}\frac{z-z_{j}}{z-\frac{q_{0}^{2}}{z_{j}}}\cdot\frac{z+z_{j}^{*}}{z+\frac{q_{0}^{2}}{z_{j}^{*}}}
\cdot\prod_{i=1}^{J_{2}}\frac{z-\widetilde{z}_{i}}{z-\frac{q_{0}^{2}}{\widetilde{z}_{i}}},
\end{equation}
where $2J_{1}+J_{2}=J$.

Then $\gamma(z)$ and $\overline{\gamma}(z)$ are analytic in the upper and lower $z-$ plane respectively and have no zeros in their respective half planes. We can get
\begin{equation}
\log \gamma(z)=\frac{1}{2\pi i}\int_{-\infty}^{+\infty}\frac{\log \gamma(\xi)}{\xi-z}d\xi, \ \ \
\frac{1}{2\pi i}\int_{-\infty}^{+\infty}\frac{\log \overline{\gamma}(\xi)}{\xi-z}d\xi=0, \ \ \ \Im z>0,
\end{equation}

\begin{equation}
\log \overline{\gamma}(z)=-\frac{1}{2\pi i}\int_{-\infty}^{+\infty}\frac{\log \overline{\gamma}(\xi)}{\xi-z}d\xi, \ \ \
\frac{1}{2\pi i}\int_{-\infty}^{+\infty}\frac{\log \gamma(\xi)}{\xi-z}d\xi=0, \ \ \ \Im z<0.
\end{equation}
Adding or subtracting the above equations in each half plane respectively, yields
\begin{equation}
\log \gamma(z)=\frac{1}{2\pi i}\int_{-\infty}^{+\infty}\frac{\log \gamma(\xi)\overline{\gamma}(\xi)}{\xi-z}d\xi, \ \ \ \Im z>0,
\end{equation}
\begin{equation}
\log \overline{\gamma}(z)=-\frac{1}{2\pi i}\int_{-\infty}^{+\infty}\frac{\log \gamma(\xi)\overline{\gamma}(\xi)}{\xi-z}d\xi, \ \ \ \Im z<0.
\end{equation}
Hence,
\begin{equation}
\log a(z)=\log\left(\prod_{j=1}^{J_{1}}\frac{z-z_{j}}{z-\frac{q_{0}^{2}}{z_{j}}}\cdot\frac{z+z_{j}^{*}}{z+\frac{q_{0}^{2}}{z_{j}^{*}}}
\cdot\prod_{i=1}^{J_{2}}\frac{z-\widetilde{z}_{i}}{z-\frac{q_{0}^{2}}{\widetilde{z}_{i}}}\right)+\frac{1}{2\pi i}\int_{-\infty}^{+\infty}\frac{\log \gamma(\xi)\overline{\gamma}(\xi)}{\xi-z}d\xi, \ \ \ \Im z>0,
\end{equation}

\begin{equation}
\log \overline{a}(z)=\log\left(\prod_{j=1}^{J_{1}} \frac{z-\frac{q_{0}^{2}}{z_{j}}}{z-z_{j}}\cdot \frac{z+\frac{q_{0}^{2}}{z_{j}^{*}}}{z+z_{j}^{*}}
\cdot\prod_{i=1}^{J_{2}} \frac{z-\frac{q_{0}^{2}}{\widetilde{z}_{i}}}{z-\widetilde{z}_{i}}\right)-\frac{1}{2\pi i}\int_{-\infty}^{+\infty}\frac{\log \gamma(\xi)\overline{\gamma}(\xi)}{\xi-z}d\xi, \ \ \ \Im z<0.
\end{equation}
Note that
\begin{equation}
\gamma(z)\overline{\gamma}(z)=a(z)\overline{a}(z),
\end{equation}
from the unitarity condition
\begin{equation}
a(z)\overline{a}(z)-b(z)\overline{b}(z)=1,
\end{equation}
then
\begin{equation}
\log a(z)=\log\left(\prod_{j=1}^{J_{1}}\frac{z-z_{j}}{z-\frac{q_{0}^{2}}{z_{j}}}\cdot\frac{z+z_{j}^{*}}{z+\frac{q_{0}^{2}}{z_{j}^{*}}}
\cdot\prod_{i=1}^{J_{2}}\frac{z-\widetilde{z}_{i}}{z-\frac{q_{0}^{2}}{\widetilde{z}_{i}}}\right)+\frac{1}{2\pi i}\int_{-\infty}^{+\infty}\frac{\log (1+b(\xi)\overline{b}(\xi))}{\xi-z}d\xi, \ \ \ \Im z>0,
\end{equation}

\begin{equation}
\log \overline{a}(z)=\log\left(\prod_{j=1}^{J_{1}} \frac{z-\frac{q_{0}^{2}}{z_{j}}}{z-z_{j}}\cdot \frac{z+\frac{q_{0}^{2}}{z_{j}^{*}}}{z+z_{j}^{*}}
\cdot\prod_{i=1}^{J_{2}} \frac{z-\frac{q_{0}^{2}}{\widetilde{z}_{i}}}{z-\widetilde{z}_{i}}\right)-\frac{1}{2\pi i}\int_{-\infty}^{+\infty}\frac{\log (1+b(\xi)\overline{b}(\xi))}{\xi-z}d\xi, \ \ \ \Im z<0.
\end{equation}
By the symmetry $b^{*}(-z^{*})=-\overline{b}(z)$, we can obtain
\begin{equation}
\log a(z)=\log\left(\prod_{j=1}^{J_{1}}\frac{z-z_{j}}{z-\frac{q_{0}^{2}}{z_{j}}}\cdot\frac{z+z_{j}^{*}}{z+\frac{q_{0}^{2}}{z_{j}^{*}}}
\cdot\prod_{i=1}^{J_{2}}\frac{z-\widetilde{z}_{i}}{z-\frac{q_{0}^{2}}{\widetilde{z}_{i}}}\right)+\frac{1}{2\pi i}\int_{-\infty}^{+\infty}\frac{\log (1-b(\xi)b^{*}(-\xi^{*}))}{\xi-z}d\xi, \ \ \ \Im z>0,
\end{equation}

\begin{equation}
\log \overline{a}(z)=\log\left(\prod_{j=1}^{J_{1}} \frac{z-\frac{q_{0}^{2}}{z_{j}}}{z-z_{j}}\cdot \frac{z+\frac{q_{0}^{2}}{z_{j}^{*}}}{z+z_{j}^{*}}
\cdot\prod_{i=1}^{J_{2}} \frac{z-\frac{q_{0}^{2}}{\widetilde{z}_{i}}}{z-\widetilde{z}_{i}}\right)-\frac{1}{2\pi i}\int_{-\infty}^{+\infty}\frac{\log (1-b(\xi)b^{*}(-\xi^{*}))}{\xi-z}d\xi, \ \ \ \Im z<0.
\end{equation}
Thus we can reconstruct $a(k), \bar{a}(k)$ in terms of the eigenvalues (zero's)  and only {\bf one function $b(k)$}.

\subsection{ Discrete scattering data and their symmetries}

In order to find reflectionless potentials/solitons we need to be able to calculate the relevant scattering data:

\[ b(z_j) ~\mbox{and} ~~~\bar{b}(\bar{z}_j) ,~~j=1,2,...J.\]
and

\[a'(z_j), ~~\bar{a}'(z_j) \]
The latter functions can be calculated via the trace formulae. So we concentrate on the former.

Since
\begin{equation}
N_{1}(x,z)=-M_{2}^{*}(-x,-z^{*}), \ \ \ N_{2}(x,z)=-M_{1}^{*}(-x,-z^{*}),
\end{equation}
\begin{equation}
M_{1}(x,z_{j})=b(z_{j})e^{i\big(z_{j}-\frac{q_{0}^{2}}{z_{j}}\big)x}\cdot N_{1}(x,z_{j})
\end{equation}
and
\begin{equation}
M_{2}(x,z_{j})=b(z_{j})e^{i\big(z_{j}-\frac{q_{0}^{2}}{z_{j}}\big)x}\cdot N_{2}(x,z_{j}),
\end{equation}
we have
\begin{equation}
\label{E:N1}
N_{1}(x,z_{j})=-b^{*}(-z_{j}^{*})\cdot e^{-i\left(z_{j}-\frac{q_{0}^{2}}{z_{j}}\right)x}\cdot N_{2}^{*}(-x,-z_{j}^{*}),
\end{equation}
\begin{equation}
\label{E:N2}
N_{2}(x,z_{j})=-b^{*}(-z_{j}^{*})\cdot e^{-\left(z_{j}-\frac{q_{0}^{2}}{z_{j}}\right)x}\cdot N_{1}^{*}(-x,-z_{j}^{*}).
\end{equation}
By rewriting (\ref{E:N2}), we obtain
\begin{equation}
N_{2}^{*}(-x,-z_{j}^{*})=-b(z_{j})\cdot e^{\left(z_{j}-\frac{q_{0}^{2}}{z_{j}}\right)x}\cdot N_{1}(x,z_{j}).
\end{equation}
Combining (\ref{E:N1}), we can deduce the following symmetry condition on the discrete data $b(z_j)$
\begin{equation}
\label{E:b}
b(z_{j})b^{*}(-z_{j}^{*})=1.
\end{equation}

Similar analysis shows that $\bar{b}(\bar{z}_j) $ satisfies an analogous equation

\[ \bar{b}(\bar{z}_j) \bar{b}^{*}(-\bar{z}_{j}^{*}) =1\]

Also, since $a(z) \sim -1$ as $z \rightarrow 0$, from the trace formula on the real axis, we have the general symmetry constraint
\begin{equation}
\label{E:product1}
\prod_{j=1}^{J_{1}}\frac{|z_{j}|^{4}}{q_{0}^{4}}\cdot\prod_{i=1}^{J_{2}}\frac{|\widetilde{z}_{i}|^{2}}{q_{0}^{2}}
e^{  {\frac{1}{2 \pi i}\int_{-\infty}^{\infty} \log(1+ b(\xi)b^*(-\xi))/\xi d\xi }}=1~~
\end{equation}
where $2J_{1}+J_{2}=J$.

\subsection{Reflectioness potentials and soliton solutions}
Reflectioness potentials and, soliton solutions when time dependence is added,  correspond to zero reflection coefficients, i.e., $\rho(\xi)=0$ and $\overline{\rho}(\xi)=0$ for all real $\xi$.
 We also note,   from the symmetry relation $\bar{b}(z)=-b^*(-z^*)$ it follows that the reflection coefficients $\rho(z)= b(z)/a(z), \bar{\rho}(z)= \bar{b}(z)/\bar{a}(z)$ will both vanish when $b(z)=0$ for $z$ on the real axis. By substituting $z=z_{l}$ in (\ref{E:closing system 1}) and $z=\overline{z}_{l}$ in (\ref{E:closing system 2}), the system (\ref{E:closing system 1}) and (\ref{E:closing system 2}) reduces to algebraic equations that determine the functional form of these special potentials. When time dependence is added the reflectionless potentials correspond to soliton solutions.
The reduced equations take the form


\begin{equation}
\begin{split}
\left(\begin{array}{cc}
N_{1}(x,z_{l})\\
N_{2}(x,z_{l})
\end{array}\right)&=\left(\begin{array}{cc}
-iq_{+}\\
z_{l}
\end{array}\right)
+\sum_{j=1}^{J}\frac{z_{l}\cdot\overline{b}(\overline{z}_{j})e^{-i\big(\overline{z}_{j}-\frac{q_{0}^{2}}{\overline{z}_{j}}\big)x} }{(z_{l}-\overline{z}_{j})\overline{z}_{j}\overline{a}'(\overline{z}_{j})}\cdot\\
&\left(\begin{array}{cc}
\overline{z}_{j}+\sum_{l=1}^{J}\frac{\overline{z}_{j}\cdot b(z_{l})e^{i\big(z_{l}-\frac{q_{0}^{2}}{z_{l}}\big)x}}{(\overline{z}_{j}-z_{l})z_{l}a'(z_{l})}\cdot N_{1}(x, z_{l})\\
-iq_{-}^{*}+\sum_{l=1}^{J}\frac{\overline{z}_{j}\cdot b(z_{l})e^{i\big(z_{l}-\frac{q_{0}^{2}}{z_{l}}\big)x}}{(\overline{z}_{j}-z_{l})z_{l}a'(z_{l})}\cdot N_{2}(x, z_{l})
\end{array}\right),
\end{split}
\end{equation}
\begin{equation}
\begin{split}
\left(\begin{array}{cc}
\overline{M}_{1}(x,\overline{z}_{l})\\
\overline{M}_{2}(x,\overline{z}_{l})
\end{array}\right)&=\left(\begin{array}{cc}
-iq_{-}\\
\overline{z}_{l}
\end{array}\right)+
\sum_{j=1}^{J}\frac{-\overline{z}_{l}\cdot\overline{b}(z_{j})e^{-i\big(z_{j}-\frac{q_{0}^{2}}{z_{j}}\big)x}}{(\overline{z}_{l}-z_{j})z_{j}a'(z_{j})}\cdot\\
&\left(\begin{array}{cc}
z_{j}+
\sum_{l=1}^{J}\frac{-z_{j}\cdot b(\overline{z}_{l})e^{i\big(\overline{z}_{l}-\frac{q_{0}^{2}}{\overline{z}_{l}}\big)x}}
{(z_{j}-\overline{z}_{l})\overline{z}_{l}\overline{a}'(\overline{z}_{l})}\cdot \overline{M}_{1}(x, \overline{z}_{l})\\
-iq_{+}^{*}+
\sum_{l=1}^{J}\frac{-z_{j}\cdot b(\overline{z}_{l})e^{i\big(\overline{z}_{l}-\frac{q_{0}^{2}}{\overline{z}_{l}}\big)x}}
{(z_{j}-\overline{z}_{l})\overline{z}_{l}\overline{a}'(\overline{z}_{l})}\cdot \overline{M}_{2}(x, \overline{z}_{l})
\end{array}\right).
\end{split}
\end{equation}

The above equations are an algebraic system to solve for either $N(x,z_{l})$ or $\overline{M}(x,\overline{z}_{l}^{*}),  l=1,2...J$. The potential are reconstructed from
equation (\ref{asympN1c}) with $\rho(\xi)=0,\bar{\rho}(\xi)=0$; i.e.,

\begin{equation}
\label{asympN1d1}
q(x)=q_{+}\cdot\left[1+\sum_{j=1}^{J}\frac{b(z_{j})e^{i\big(z_{j}-\frac{q_{0}^{2}}{z_{j}}\big)x}}{-z_{j}^{2}a'(z_{j})}\cdot N_{1}(x,z_{j})\right].
\end{equation}
As before, since $a(z)\sim -1$ as $z\rightarrow 0$, by the trace formula when $b(\xi)=0$ in the real axis, we have the following symmetry constraint for the reflectionless potentials
\begin{equation}
\label{E:product}
\prod_{j=1}^{J_{1}}\frac{|z_{j}|^{4}}{q_{0}^{4}}\cdot\prod_{i=1}^{J_{2}}\frac{|\widetilde{z}_{i}|^{2}}{q_{0}^{2}}=1,
\end{equation}
where $2J_{1}+J_{2}=J$.

\subsection{Reflectionless potential solution: 1-eigenvalue}
In this subsection, we show an explicit form for the  1 eigenvalue/1-soliton solution without time dependence, by setting $J=1$. Then $J_{1}=0$ and $J_{2}=1$. Now let $\widetilde{z}_{1}=iv_{1}$, we have $\overline{\widetilde{z}}_{1}=-i\frac{q_{0}^{2}}{v_{1}}:=-i\overline{v}_{1}$, where $v_{1}$ is real and positive. Hence, we have
\begin{equation}
N_{1}(x,iv_{1})=\frac{-iq_{+}+\frac{v_{1}\overline{b}(-i\overline{v}_{1})
e^{-\big(\overline{v}_{1}+\frac{q_{0}^{2}}{\overline{v}_{1}}\big)x}}{(v_{1}+\overline{v}_{1})\overline{a}'(-i\overline{v}_{1})}}
{1-\frac{\overline{b}(-i\overline{v}_{1})b(iv_{1})e^{-\big(v_{1}+\overline{v}_{1}+\frac{q_{0}^{2}}{v_{1}}
+\frac{q_{0}^{2}}{\overline{v}_{1}}\big)x}}{(v_{1}+\overline{v}_{1})^{2}\overline{a}'(-i\overline{v}_{1})a'(iv_{1})}},
\end{equation}

\begin{equation}
N_{2}(x,iv_{1})=\frac{iv_{1}+\frac{q_{-}^{*}\cdot v_{1}\overline{b}(-i\overline{v}_{1})
e^{-\big(\overline{v}_{1}+\frac{q_{0}^{2}}{\overline{v}_{1}}\big)x}}{(v_{1}+\overline{v}_{1})\overline{v}_{1}\overline{a}'(-i\overline{v}_{1})}}
{1-\frac{\overline{b}(-i\overline{v}_{1})b(iv_{1})e^{-\big(v_{1}+\overline{v}_{1}+\frac{q_{0}^{2}}{v_{1}}
+\frac{q_{0}^{2}}{\overline{v}_{1}}\big)x}}{(v_{1}+\overline{v}_{1})^{2}\overline{a}'(-i\overline{v}_{1})a'(iv_{1})}},
\end{equation}

\begin{equation}
\overline{M}_{1}(x,-i\overline{v}_{1})=\frac{-iq_{-}-\frac{\overline{v}_{1}\overline{b}(iv_{1})
e^{\big(v_{1}+\frac{q_{0}^{2}}{v_{1}}\big)x}}{(v_{1}+\overline{v}_{1})a'(iv_{1})}}
{1-\frac{\overline{b}(iv_{1})b(-i\overline{v}_{1})e^{\big(v_{1}+\overline{v}_{1}+\frac{q_{0}^{2}}{v_{1}}
+\frac{q_{0}^{2}}{\overline{v}_{1}}\big)x}}{(v_{1}+\overline{v}_{1})^{2}\overline{a}'(-i\overline{v}_{1})a'(iv_{1})}},
\end{equation}

\begin{equation}
\overline{M}_{2}(x,-i\overline{v}_{1})=\frac{-i\overline{v}_{1}+\frac{q_{+}^{*}\cdot\overline{v}_{1}\overline{b}(iv_{1})
e^{\big(v_{1}+\frac{q_{0}^{2}}{v_{1}}\big)x}}{(v_{1}+\overline{v}_{1})v_{1}a'(iv_{1})}}
{1-\frac{\overline{b}(iv_{1})b(-i\overline{v}_{1})e^{\big(v_{1}+\overline{v}_{1}+\frac{q_{0}^{2}}{v_{1}}
+\frac{q_{0}^{2}}{\overline{v}_{1}}\big)x}}{(v_{1}+\overline{v}_{1})^{2}\overline{a}'(-i\overline{v}_{1})a'(iv_{1})}}.
\end{equation}
Therefore,
\begin{equation}
\begin{split}
q(x)&=q_{+}\cdot\left[1+\frac{b(iv_{1})e^{-\big(v_{1}+\frac{q_{0}^{2}}{v_{1}}\big)x}}{v_{1}^{2}\cdot a'(iv_{1})}\cdot N_{1}(x,iv_{1})\right].
\end{split}
\end{equation}
By the trace formula when $b(\xi)=0$ in the real axis, we can get
\begin{equation}
a'(iv_{1})=\frac{1}{i(v_{1}+\overline{v}_{1})}, \ \ \ \overline{a}'(-i\overline{v}_{1})=\frac{i}{v_{1}+\overline{v}_{1}}.
\end{equation}
From (\ref{E:product}), we can deduce $v_{1}=\overline{v}_{1}=q_{0}$. Hence,
\begin{equation}
a'(iv_{1})=a'(iq_{0})=\frac{1}{2iq_{0}}, \ \ \ \overline{a}'(-i\overline{v}_{1})=\overline{a}'(-iq_{0})=\frac{i}{2q_{0}}.
\end{equation}

By choosing $\widetilde{z}_{1}=iq_{0}$ in (\ref{E:b}), we can get $|b(iq_{0})|^{2}=1$. Similarly, we have $|\overline{b}(-iq_{0})|^{2}=1$. Thus, we write $b(iq_{0})=e^{i\theta_{1}}$ and $\overline{b}(-iq_{0})=e^{i\overline{\theta}_{1}}$, where both $\theta_{1}$ and $\overline{\theta}_{1}$ are real. Moreover, we can obtain
$\overline{b}(iq_{0})=-e^{-i\theta_{1}}$ and $b(-iq_{0})=-e^{-i\overline{\theta}_{1}}$.  Moreover, from (\ref{E:scatering}), we can also get
\begin{equation}
e^{i \theta_{1}}=b(iq_{0})=b\left(\frac{q_{0}^{2}}{-iq_{0}}\right)=-\frac{\overline{b}(-iq_{0})}{e^{i(\theta_{+}+\theta_{-})}}
=\frac{e^{i\overline{\theta}_{1}}}{e^{2i\theta_{+}}},
\end{equation}
i.e., $\overline{\theta}_{1}=\theta_{1}+2\theta_{+}+2n\pi$, where $n\in\mathbb{Z}$. Thus we have the reflectionless potential corresponding to one eigenvalue:
\begin{equation}
\label{1refl}
q(x)=q_{0}e^{i\theta_{+}}\cdot
\left[1+ \frac{2e^{i\theta_{1}-2q_{0}x}\cdot(e^{i\theta_{+}}+e^{i(\theta_{1}+2\theta_{+})-2q_{0}x})}
{1-e^{2i(\theta_{1}+\theta_{+})-4q_{0}x}}\right].
\end{equation}
To find the corresponding soliton solution we need the time evolution of the data which we derive next.

\subsection{Time evolution}
Since
\begin{equation}
\label{E:time evolution}
v_{t}=
\left(\begin{array}{cc}
2ik^{2}-iq(x,t)q^{*}(-x,t)& -2kq(x,t)-iq_{x}(x,t)\\
2kq^{*}(-x,t)+iq^{*}_{x}(-x,t)& -2ik^{2}+iq(x,t)q^{*}(-x,t)
\end{array}\right)v,
\end{equation}
then
\begin{equation}
\frac{\partial v_{1}}{\partial t}=(2ik^{2}-iq(x,t)q^{*}(-x,t))v_{1}+(-2kq(x,t)-iq_{x}(x,t))v_{2}
\end{equation}
and
\begin{equation}
\frac{\partial v_{2}}{\partial t}=(2kq^{*}(-x,t)+iq^{*}_{x}(-x,t))v_{1}+(-2ik^{2}+iq(x,t)q^{*}(-x,t))v_{2}.
\end{equation}
Note that
\begin{equation}
q(x,t)\rightarrow q_{\pm}=q_{0}e^{2iq_{0}^{2}t+i\theta_{\pm}},\ \ as \ \ x\rightarrow\pm\infty,
\end{equation}
where $q_{0}>0$, $0\leq \theta_{\pm}<2\pi$, $\theta_{+}-\theta_{-}=\pi$. Thus,
\begin{equation}
\frac{\partial v_{1}}{\partial t}\sim i(2k^{2}+q_{0}^{2})\cdot v_{1}-2kq_{0}e^{2iq_{0}^{2}t+i\theta_{+}}\cdot v_{2}
\end{equation}
and
\begin{equation}
\frac{\partial v_{2}}{\partial t}\sim 2kq_{0}e^{-2iq_{0}^{2}t-i\theta_{-}}\cdot v_{1}-i(2k^{2}+q_{0}^{2})\cdot v_{2}
\end{equation}
as $x\rightarrow +\infty$;
\begin{equation}
\frac{\partial v_{1}}{\partial t}\sim i(2k^{2}+q_{0}^{2})\cdot v_{1}-2kq_{0}e^{2iq_{0}^{2}t+i\theta_{-}}\cdot v_{2}
\end{equation}
and
\begin{equation}
\frac{\partial v_{2}}{\partial t}\sim 2kq_{0}e^{-2iq_{0}^{2}t-i\theta_{+}}\cdot v_{1}-i(2k^{2}+q_{0}^{2})\cdot v_{2}
\end{equation}
as $x\rightarrow -\infty$.
As $x\rightarrow\pm\infty$, the eigenfunctions of the scattering problem asymptotically satisfy
\begin{equation}
\left(\begin{array}{cc}
v_{1}\\
v_{2}
\end{array}\right)_{x}
=
\left(\begin{array}{cc}
-ik& q_{0}e^{2iq_{0}^{2}t+i\theta_{\pm}}\\
-q_{0}e^{-2iq_{0}^{2}t-i\theta_{\mp}}& ik
\end{array}\right)
\left(\begin{array}{cc}
v_{1}\\
v_{2}
\end{array}\right),
\end{equation}
we can get
\begin{equation}
q_{0}e^{2iq_{0}^{2}t+i\theta_{\pm}} \cdot v_{2}\sim\frac{\partial v_{1}}{\partial x}+ik v_{1}
\end{equation}
and
\begin{equation}
q_{0}e^{-2iq_{0}^{2}t-i\theta_{\mp}}\cdot v_{1}\sim -\frac{\partial v_{2}}{\partial x}+ik v_{2}
\end{equation}
as $x\rightarrow\pm\infty$. Hence,
\begin{equation}
\frac{\partial v_{1}}{\partial t}\sim iq_{0}^{2}\cdot v_{1}-2k \frac{\partial v_{1}}{\partial x}
\end{equation}
and
\begin{equation}
\frac{\partial v_{2}}{\partial t}\sim -iq_{0}^{2}\cdot v_{2}-2k \frac{\partial v_{2}}{\partial x}
\end{equation}
as $x\rightarrow\pm\infty$.

Note that the eigenfunctions themselves, whose boundary values at space infinities, are not compatible with this time evolution. Therefore, one introduces time-dependent eigenfunctions
\begin{equation}
\Phi(x,t)=e^{i A_{\infty}t}\cdot \phi(x,t), \ \ \ \overline{\Phi}(x,t)=e^{i B_{\infty}t}\cdot \overline{\phi}(x,t),
\end{equation}
\begin{equation}
\Psi(x,t)=e^{i C_{\infty}t}\cdot \psi(x,t), \ \ \ \overline{\Psi}(x,t)=e^{i D_{\infty}t}\cdot \overline{\psi}(x,t)
\end{equation}
to be solutions of (\ref{E:time evolution}). We have
\begin{equation}
\frac{\partial \Phi_{1}(x,t)}{\partial t}=i A_{\infty} \Phi_{1}(x,t)+e^{iA_{\infty}t} \frac{\partial \phi_{1}(x,t)}{\partial t},
\ \ \ \frac{\partial \Phi_{2}(x,t)}{\partial t}=i A_{\infty} \Phi_{2}(x,t)+e^{iA_{\infty}t} \frac{\partial \phi_{2}(x,t)}{\partial t},
\end{equation}
\begin{equation}
\frac{\partial \overline{\Phi}_{1}(x,t)}{\partial t}=i B_{\infty} \overline{\Phi}_{1}(x,t)+e^{iB_{\infty}t} \frac{\partial \overline{\phi}_{1}(x,t)}{\partial t},
\ \ \ \frac{\partial \overline{\Phi}_{2}(x,t)}{\partial t}=i B_{\infty} \overline{\Phi}_{2}(x,t)+e^{iB_{\infty}t} \frac{\partial \overline{\phi}_{2}(x,t)}{\partial t},
\end{equation}
\begin{equation}
\frac{\partial \Psi_{1}(x,t)}{\partial t}=i C_{\infty} \Psi_{1}(x,t)+e^{iC_{\infty}t} \frac{\partial \psi_{1}(x,t)}{\partial t},
\ \ \ \frac{\partial \Psi_{2}(x,t)}{\partial t}=i C_{\infty} \Psi_{2}(x,t)+e^{iC_{\infty}t} \frac{\partial \psi_{2}(x,t)}{\partial t},
\end{equation}
\begin{equation}
\frac{\partial \overline{\Psi}_{1}(x,t)}{\partial t}=i D_{\infty} \overline{\Psi}_{1}(x,t)+e^{iD_{\infty}t} \frac{\partial \overline{\psi}_{1}(x,t)}{\partial t},
\ \ \ \frac{\partial \overline{\Psi}_{2}(x,t)}{\partial t}=i D_{\infty} \overline{\Psi}_{2}(x,t)+e^{iD_{\infty}t} \frac{\partial \overline{\psi}_{2}(x,t)}{\partial t}.
\end{equation}
Note that
\begin{equation}
\phi(x,t)\sim\left(\begin{array}{cc}
\lambda+k\\
-iq_{0}e^{-2iq_{0}^{2}t-i\theta_{+}}
\end{array}\right)e^{-i\lambda x}, \ \ \
\frac{\partial \phi(x,t)}{\partial t}\sim \left(\begin{array}{cc}
0\\
-2q_{0}^{3}e^{-2iq_{0}^{2}t-i\theta_{+}}
\end{array}\right)e^{-i\lambda x}
\end{equation}
as $x\rightarrow -\infty$. From
\begin{equation}
\frac{\partial \Phi_{1}(x,t)}{\partial t}\sim iq_{0}^{2}\cdot \Phi_{1}(x,t)-2k \frac{\partial \Phi_{1}(x,t)}{\partial x}
=i A_{\infty} \Phi_{1}(x,t)+e^{iA_{\infty}t} \frac{\partial \phi_{1}(x,t)}{\partial t},
\end{equation}
as $x\rightarrow -\infty$, we can deduce
\begin{equation}
A_{\infty}=q_{0}^{2}+2\lambda k.
\end{equation}
Similarly, we have
\begin{equation}
B_{\infty}=-A_{\infty}=-q_{0}^{2}-2\lambda k,
\end{equation}
\begin{equation}
C_{\infty}=-A_{\infty}=-q_{0}^{2}-2\lambda k,
\end{equation}
\begin{equation}
D_{\infty}=A_{\infty}=q_{0}^{2}+2\lambda k.
\end{equation}
Then
\begin{equation}
\frac{\partial \phi}{\partial t}=\left(\begin{array}{cc}
2ik^{2}-iq(x,t)q^{*}(-x,t)-i(q_{0}^{2}+2\lambda k)& -2kq(x,t)-iq_{x}(x,t)\\
2kq^{*}(-x,t)+iq^{*}_{x}(-x,t)& -2ik^{2}+iq(x,t)q^{*}(-x,t)-i(q_{0}^{2}+2\lambda k)
\end{array}\right)\phi,
\end{equation}
\begin{equation}
\frac{\partial \overline{\phi}}{\partial t}=\left(\begin{array}{cc}
2ik^{2}-iq(x,t)q^{*}(-x,t)+i(q_{0}^{2}+2\lambda k)& -2kq(x,t)-iq_{x}(x,t)\\
2kq^{*}(-x,t)+iq^{*}_{x}(-x,t)& -2ik^{2}+iq(x,t)q^{*}(-x,t)+i(q_{0}^{2}+2\lambda k)
\end{array}\right)\overline{\phi},
\end{equation}
\begin{equation}
\frac{\partial \psi}{\partial t}=\left(\begin{array}{cc}
2ik^{2}-iq(x,t)q^{*}(-x,t)+i(q_{0}^{2}+2\lambda k)& -2kq(x,t)-iq_{x}(x,t)\\
2kq^{*}(-x,t)+iq^{*}_{x}(-x,t)& -2ik^{2}+iq(x,t)q^{*}(-x,t)+i(q_{0}^{2}+2\lambda k)
\end{array}\right)\psi,
\end{equation}
\begin{equation}
\frac{\partial \overline{\psi}}{\partial t}=\left(\begin{array}{cc}
2ik^{2}-iq(x,t)q^{*}(-x,t)-i(q_{0}^{2}+2\lambda k)& -2kq(x,t)-iq_{x}(x,t)\\
2kq^{*}(-x,t)+iq^{*}_{x}(-x,t)& -2ik^{2}+iq(x,t)q^{*}(-x,t)-i(q_{0}^{2}+2\lambda k)
\end{array}\right)\overline{\psi}.
\end{equation}
Note that
\begin{equation}
\phi(x,t)=b(t)\psi(x,t)+a(t)\overline{\psi}(x,t)
\end{equation}
and
\begin{equation}
\overline{\phi}(x,t)=\overline{a}(t)\psi(x,t)+\overline{b}(t)\overline{\psi}(x,t),
\end{equation}
then
\begin{equation}
\begin{split}
&\left(\begin{array}{cc}
2ik^{2}-iq(x,t)q^{*}(-x,t)-i(q_{0}^{2}+2\lambda k)& -2kq(x,t)-iq_{x}(x,t)\\
2kq^{*}(-x,t)+iq^{*}_{x}(-x,t)& -2ik^{2}+iq(x,t)q^{*}(-x,t)-i(q_{0}^{2}+2\lambda k)
\end{array}\right)\\
&\cdot\left(\begin{array}{cc}
b(t)\psi_{1}(x,t)+a(t)\overline{\psi}_{1}(x,t)
\\
b(t)\psi_{2}(x,t)+a(t)\overline{\psi}_{2}(x,t)
\end{array}\right)=\frac{\partial b(t)}{\partial t}\cdot\left(\begin{array}{cc}
\psi_{1}(x,t)
\\
\psi_{2}(x,t)
\end{array}\right)\\
&+b(t)\left(\begin{array}{cc}
2ik^{2}-iq(x,t)q^{*}(-x,t)+i(q_{0}^{2}+2\lambda k)& -2kq(x,t)-iq_{x}(x,t)\\
2kq^{*}(-x,t)+iq^{*}_{x}(-x,t)& -2ik^{2}+iq(x,t)q^{*}(-x,t)+i(q_{0}^{2}+2\lambda k)
\end{array}\right)\left(\begin{array}{cc}
\psi_{1}(x,t)
\\
\psi_{2}(x,t)
\end{array}\right)\\
&+\frac{\partial a(t)}{\partial t}\cdot\left(\begin{array}{cc}
\overline{\psi}_{1}(x,t)
\\
\overline{\psi}_{2}(x,t)
\end{array}\right)+a(t)\cdot\\
&\left(\begin{array}{cc}
2ik^{2}-iq(x,t)q^{*}(-x,t)-i(q_{0}^{2}+2\lambda k)& -2kq(x,t)-iq_{x}(x,t)\\
2kq^{*}(-x,t)+iq^{*}_{x}(-x,t)& -2ik^{2}+iq(x,t)q^{*}(-x,t)-i(q_{0}^{2}+2\lambda k)
\end{array}\right)\left(\begin{array}{cc}
\overline{\psi}_{1}(x,t)
\\
\overline{\psi}_{2}(x,t)
\end{array}\right),
\end{split}
\end{equation}
i.e.,
\begin{equation}
\frac{\partial b(t)}{\partial t}\cdot\left(\begin{array}{cc}
\psi_{1}(x,t)
\\
\psi_{2}(x,t)
\end{array}\right)+b(t)\left(\begin{array}{cc}
2iA_{\infty}& 0\\
0& 2iA_{\infty}
\end{array}\right)\left(\begin{array}{cc}
\psi_{1}(x,t)
\\
\psi_{2}(x,t)
\end{array}\right)+\frac{\partial a(t)}{\partial t}\left(\begin{array}{cc}
\overline{\psi}_{1}(x,t)
\\
\overline{\psi}_{2}(x,t)
\end{array}\right)=0.
\end{equation}
Taking $x\rightarrow +\infty$, since $\psi(x,t)$ and $\overline{\psi}(x,t)$ are linearly independent, so
\begin{equation}
\frac{\partial a(t)}{\partial t}=0, \ \ \ \frac{\partial b(t)}{\partial t}=-2iA_{\infty}b(t).
\end{equation}
Similarly, we can get
\begin{equation}
\frac{\partial \overline{a}(t)}{\partial t}=0, \ \ \ \frac{\partial \overline{b}(t)}{\partial t}=2iA_{\infty}\overline{b}(t).
\end{equation}
Therefore, both $a(t)$ and $\overline{a}(t)$ are time independent, and
\begin{equation}
b(iq_{0}, t)=b(iq_{0}, 0)e^{-2i(q_{0}^{2}+2\lambda k)t}=e^{i\theta_{1}}\cdot e^{-2iq_{0}^{2}t}
\end{equation}
and
\begin{equation}
\overline{b}(-iq_{0}, t)=\overline{b}(-iq_{0}, 0)e^{2i(q_{0}^{2}+2\lambda k)t}=e^{i\overline{\theta}_{1}}\cdot e^{2iq_{0}^{2}t}
=e^{i(\theta_{1}+2\theta_{+})}\cdot e^{2iq_{0}^{2}t}.
\end{equation}
Putting all the above into the formula we had for the reflectionless potential (\ref{1refl}) we obtain the following one soliton solution
\begin{equation}
\begin{split}
q(x,t)&=q_{0}e^{2iq_{0}^{2}t+i\theta_{+}}\cdot
\left[1+ \frac{2e^{i\theta_{1}-2q_{0}x}\cdot(e^{i\theta_{+}}+e^{i(\theta_{1}+2\theta_{+})-2q_{0}x})}
{1-e^{2i(\theta_{1}+\theta_{+})-4q_{0}x}}\right]\\
&=q_{0}\cdot e^{i(2q_{0}^{2}t+\theta_{+})}\cdot \coth\left[q_{0}x-\frac{i}{2}(\theta_{+}+\theta_{1})\right].
\end{split}
\end{equation}
The above solution can be transformed to the following form
\begin{equation}
q(x,t)=q_{0}\cdot e^{i(2q_{0}^{2}t+\theta_{+}- \pi)}\cdot \tanh \left[q_{0}x- i\theta_* 
\right]
\end{equation}
where $\theta_*=\frac{1}{2}(\theta_{+}+\theta_{1}+\pi)$. This is similar to the well known black soliton of the standard integrable NLS equation with only a complex phase shift difference.
From this solution we see that there is a singularity only when $\theta_{+}+\theta_{1}=0, \pm 2\pi$ which puts a restriction on the otherwise arbitrary phases $\theta_{+},\theta_{1}$. The magnitude of the solution is stationary. In Fig. \ref{FigA} and Fig. \ref{FigA'} below we give typical one soliton solutions. We see in Fig. \ref{FigA} the magnitude rises from a constant background, whereas in Fig. \ref{FigA'} the magnitude dips from the constant background.

\begin{figure}[h]
\begin{tabular}{cc}
\includegraphics[width=0.5\textwidth]{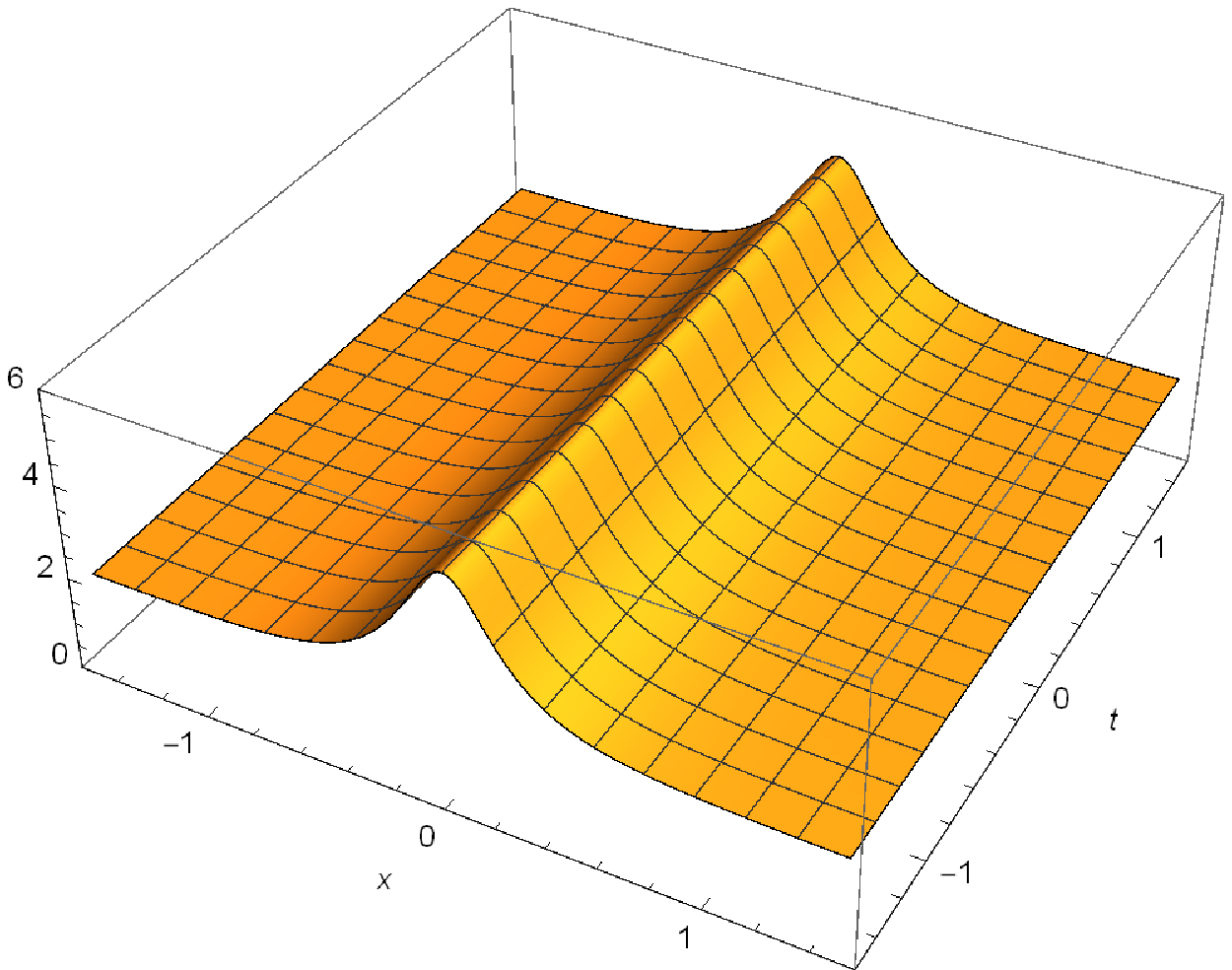}&
\includegraphics[width=0.5\textwidth]{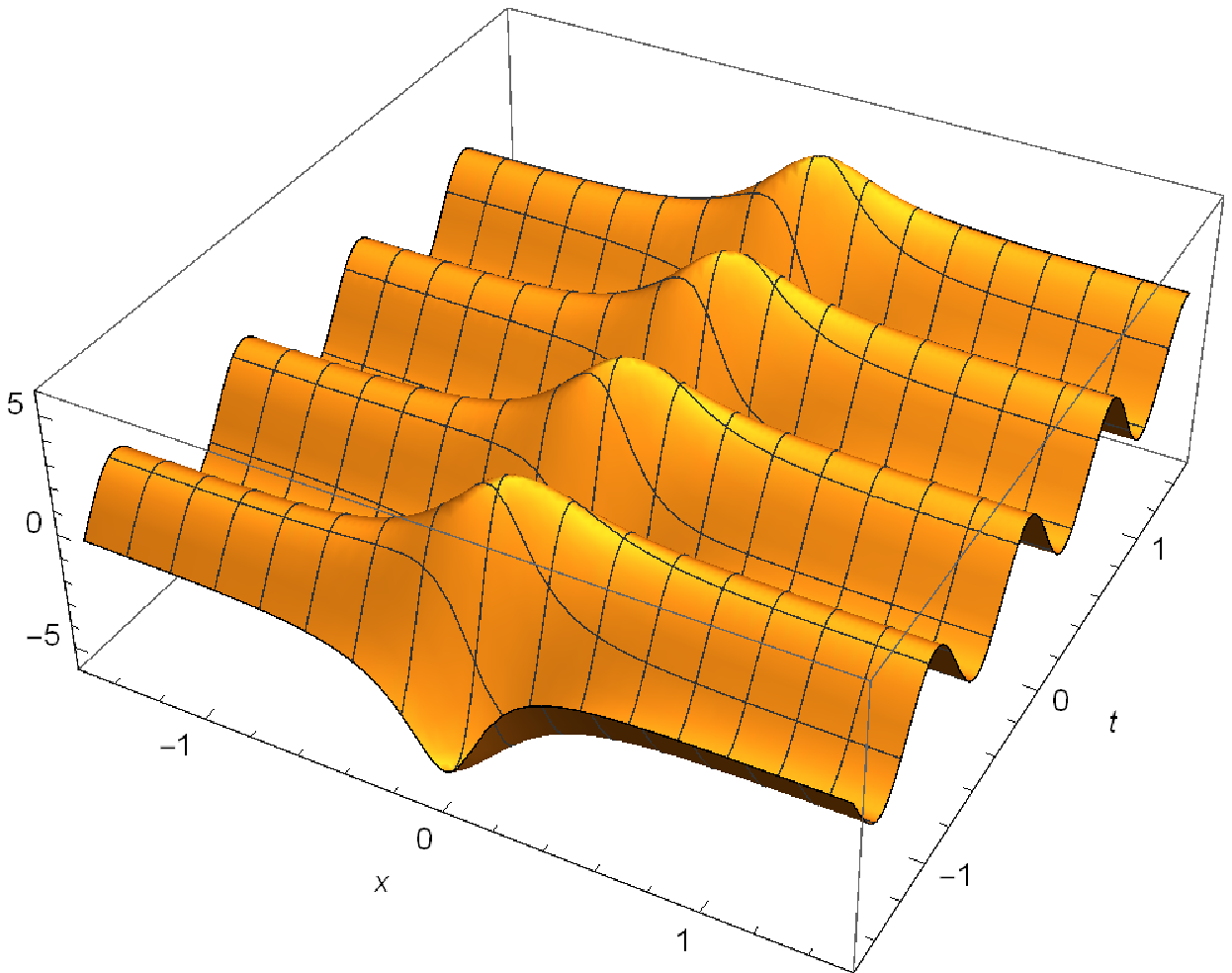}\\
(a) & (b)
\end{tabular}
\caption{(a) The amplitude of $q(x,t)$ with $\theta_{+}=\frac{\pi}{4}$, $\theta_{1}=0$ and $q_{0}=2$.  (b) The real part of $q(x,t)$ with $\theta_{+}=\frac{\pi}{4}$, $\theta_{1}=0$ and $q_{0}=2$.}
\label{FigA}
\end{figure}

\begin{figure}[h]
\begin{tabular}{cc}
\includegraphics[width=0.5\textwidth]{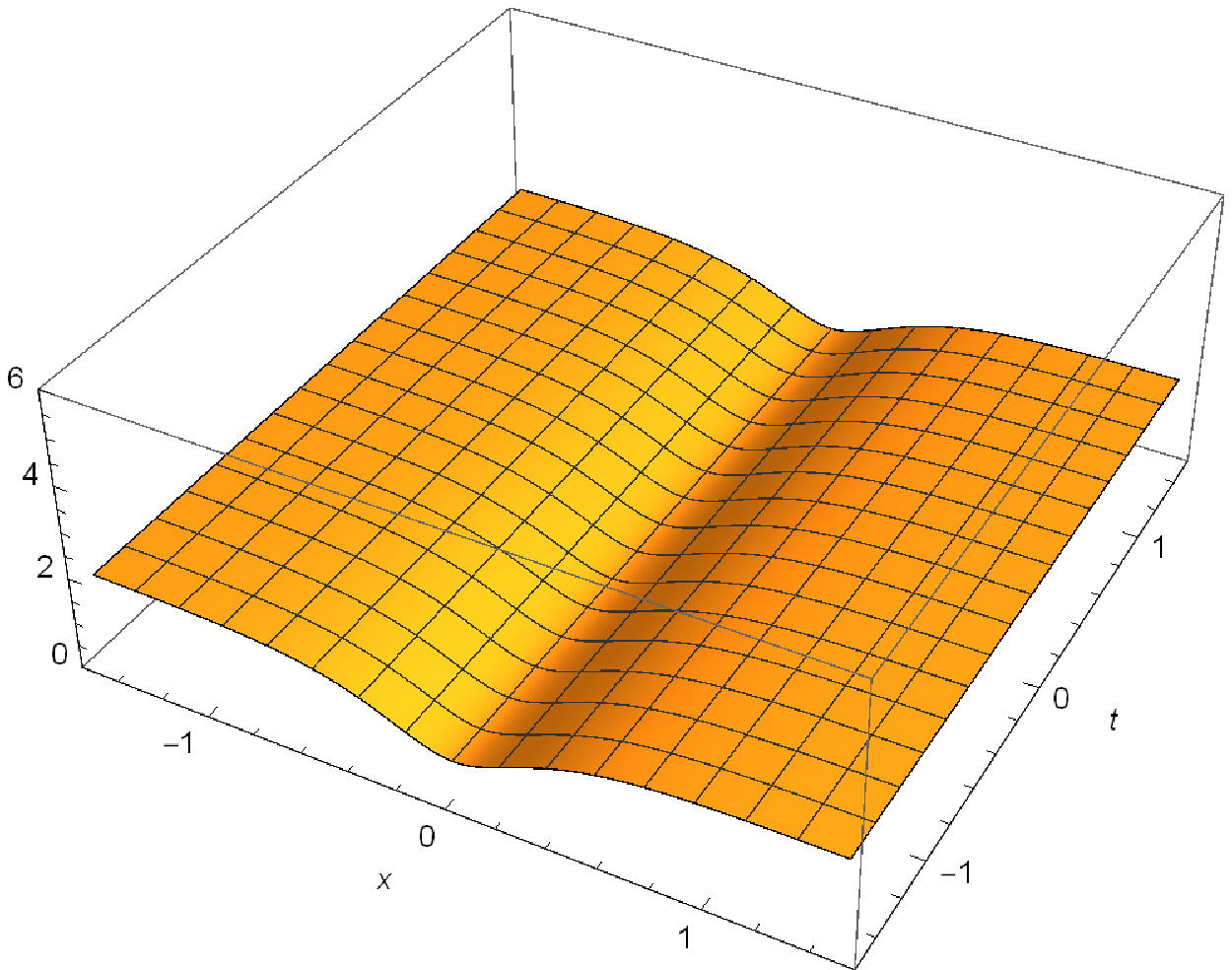}&
\includegraphics[width=0.5\textwidth]{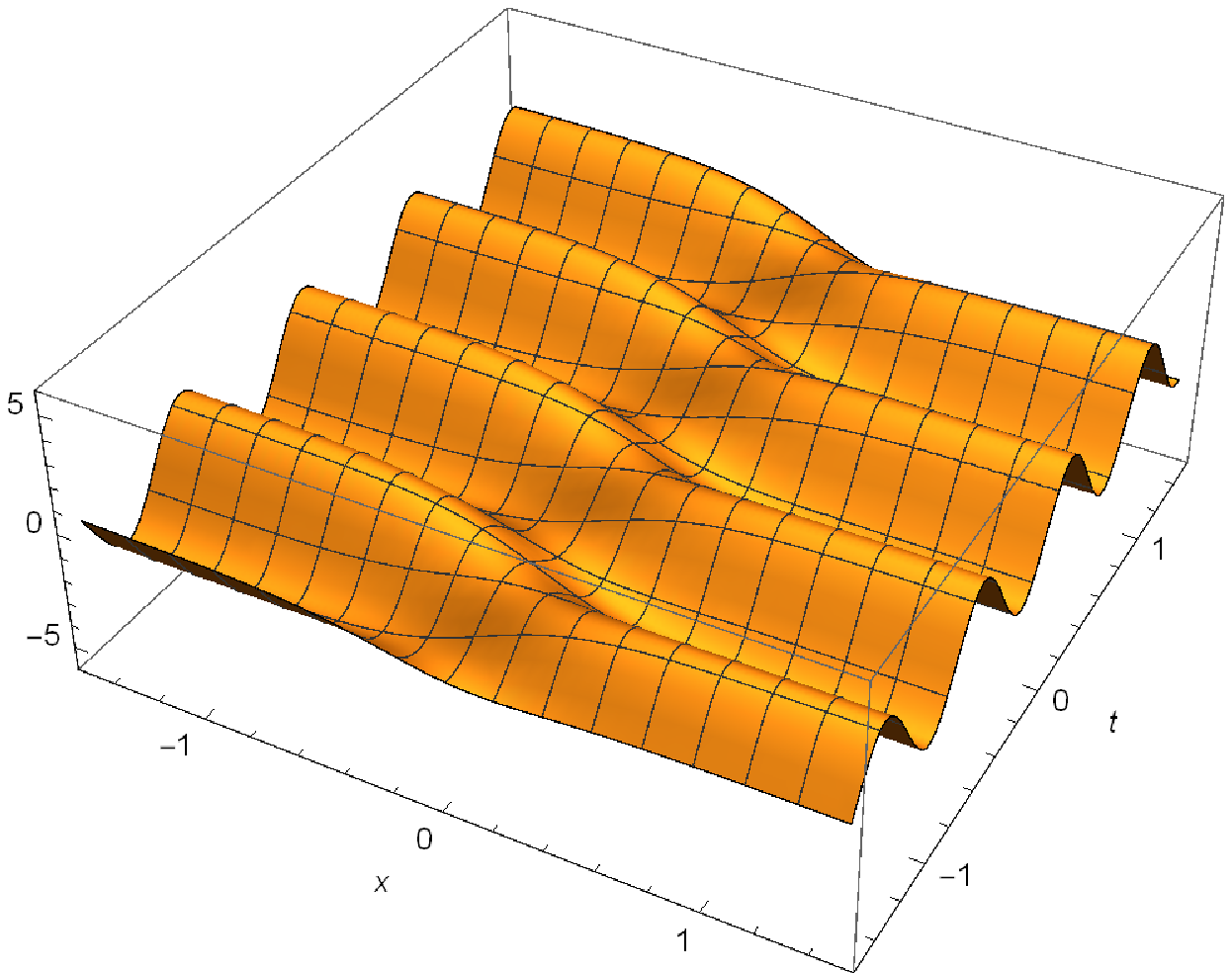}\\
(a) & (b)
\end{tabular}
\caption{(a) The amplitude of $q(x,t)$ with $\theta_{+}=\frac{5\pi}{4}$, $\theta_{1}=0$ and $q_{0}=2$.  (b) The real part of $q(x,t)$ with $\theta_{+}=\frac{5\pi}{4}$, $\theta_{1}=0$ and $q_{0}=2$.}
\label{FigA'}
\end{figure}

\section{The case of $\sigma=-1$ with $\theta_{+}-\theta_{-}=0$}
In this section we consider the nonzero boundary conditions (NZBCs) given in (\ref{NZBCA}) with$\sigma=-1$ , $\theta_{+}=\theta_{-}:=\theta$
\begin{equation}
q(x,t)\rightarrow q_{\pm}(t):=\widetilde{q}(t)=q_{0}e^{-2iq_{0}^{2}t+i\theta},\ \ as \ \ x\rightarrow\pm\infty,
\end{equation}
where $q_{0}>0$, $0\leq \theta<2\pi$.

\subsection{Direct scattering}
Equation
\begin{equation}
iq_{t}(x,t)=q_{xx}(x,t)+2q^{2}(x,t)q^{*}(-x,t)
\end{equation}
can be associated to the so-called ZS-AKNS scattering problem
\begin{equation}
v_{x}=(ikJ+Q)v, \ \ \ x\in \mathbb{R},
\end{equation}
where
\begin{equation}J=
\left(\begin{matrix}
-1& 0\\
0& 1
\end{matrix}\right),\ \ \
Q=
\left(\begin{matrix}
0& q(x,t)\\
-q^{*}(-x,t)& 0
\end{matrix}\right),
\end{equation}
$q(x,t)$ is the potential and $k$ is a complex spectral parameter.

As $x\rightarrow\pm\infty$, the eigenfunctions of the scattering problem asymptotically satisfy
\begin{equation}
\left(\begin{array}{cc}
v_{1}\\
v_{2}
\end{array}\right)_{x}
=
\left(\begin{array}{cc}
-ik& q_{0}e^{-2iq_{0}^{2}t+i\theta}\\
-q_{0}e^{2iq_{0}^{2}t-i\theta}& ik
\end{array}\right)
\left(\begin{array}{cc}
v_{1}\\
v_{2}
\end{array}\right),
\end{equation}
i.e.,
\begin{equation}
v_{x}=(ikJ+Q_{\pm}(t))v, \ \ \ Q_{\pm}(t)=\left(\begin{array}{cc}
0& q_{0}e^{-2iq_{0}^{2}t+i\theta}\\
-q_{0}e^{2iq_{0}^{2}t-i\theta}& 0
\end{array}\right),
\end{equation}
or
\begin{equation}
\frac{\partial^{2} v_{j}}{\partial x^{2}}=-(k^{2}+q_{0}^{2})v_{j}, \ \ \ j=1,2.
\end{equation}
Each of the two equations has two linearly independent solutions $e^{i\lambda x}$ and $e^{-i\lambda x}$ as $|x|\rightarrow\infty$, where
$\lambda=\sqrt{k^{2}+q_{0}^{2}}$. We introduce the local polar coordinates
\begin{equation}
k-iq_{0}=r_{1}e^{i\theta_{1}}, \ \ \ -\frac{\pi}{2}\leq\theta_{1}<\frac{3\pi}{2},
\end{equation}
\begin{equation}
k+iq_{0}=r_{2}e^{i\theta_{2}}, \ \ \ -\frac{\pi}{2}\leq\theta_{2}<\frac{3\pi}{2},
\end{equation}
where $r_{1}=|k-iq_{0}|$ and $r_{2}=|k+iq_{0}|$. One can write $\lambda(k)=(r_{1}r_{2})^{\frac{1}{2}}\cdot e^{i\cdot \frac{\theta_{1}+\theta_{2}}{2}+im\pi}$, and $m=0,1$ respectively on sheets I ($\mathbb{K}_{1}$) and II ($\mathbb{K}_{2}$). The variable $k$ is then thought of as belonging to a Riemann surface $\mathbb{K}$ consisting of sheets I and II with both coinciding with the complex plane cut along $\Sigma:=[-iq_{0}, iq_{0}]$
with its edges glued in such a way that $\lambda(k)$ is continuous through the cut. Along the real $k$ axis, we have $\lambda(k)=\pm sign(k) \sqrt{k^{2}+q_{0}^{2}}$, where the plus/minus signs apply, respectively, on sheet I and sheet II of the Riemann surface, and where the square root sign denotes the principal branch of the real-valued square root function. We denote $\mathbb{C}^{\pm}$ the open upper/lower complex half planes, and $\mathbb{K}^{\pm}$ the open upper/lower complex half planes cut along $\Sigma$. Then $\lambda$ provides one-to-one correspondences between the following sets:
1. $k\in\mathbb{K}^{+}=\mathbb{C}^{+}\setminus (0, iq_{0}]$ and $\lambda\in\mathbb{C}^{+}$,
~\\2. $k\in \partial\mathbb{K}^{+}=\mathbb{R}\cup\{is-0^{+}:0<s<q_{0}\}\cup\{iq_{0}\}\cup\{is+0^{+}:0<s<q_{0}\}$ and $\lambda\in\mathbb{R}$,
~\\3. $k\in\mathbb{K}^{-}=\mathbb{C}^{-}\setminus [-iq_{0}, 0)$ and $\lambda\in\mathbb{C}^{-}$,
~\\4. $k\in \partial\mathbb{K}^{-}=\mathbb{R}\cup\{is-0^{+}:-q_{0}<s<0\}\cup\{-iq_{0}\}\cup\{is+0^{+}:-q_{0}<s<0\}$ and $\lambda\in\mathbb{R}$.


Moreover, $\lambda^{\pm}(k)$ will denote the boundary values taken by $\lambda(k)$ for $k\in\Sigma$ from the right/left edge of the cut, with $\lambda^{\pm}(k)=\pm\sqrt{q_{0}^{2}-|k|^{2}}$, $k=is\pm 0^{+}$, $|s|<q_{0}$ on the right/left edge of the cut.

See also Figs. \ref{fig3}, \ref{fig4}. The contours along the real axis encircling $\pm iq_0$ are used in a Riemann-Hilbert formulation.


\begin{figure}
\begin{tabular*}{\textwidth}{@{}cc@{}}
\begin{minipage}{\dimexpr0.5\textwidth-2\tabcolsep}

\begin{tikzpicture}

\draw [fill=gray!50!white] (0,0) rectangle (-2.5,2.5);
\draw [fill=gray!50!white] (2.5,0) rectangle (-2.5,2.5);
\draw (2.5,-2.5) rectangle (-2.5,2.5);


\draw[green] (-2.5,.25) -- (-.4,.25);
\draw[green]   (-.4,.25) --  (-.4,1.5) ;
\draw[-stealth] (-.4,.99) --  (-.4,1) ;
\draw[green]  (.4,1.5) --  (.4,.25) ;
\draw[-stealth] (.4,.91) --  (.4,.9) ;
\draw[green]  (.4,.25) -- (2.5,.25);

\draw[green]  (-.4,1.5) arc (180:0:.4);	

\draw  [-stealth] (0,1)  (-1.51,.25) --(-1.5,.25) ;
\draw  [-stealth] (0,1)  (1.49,.25) -- (1.5,.25);

\draw[green]  (-2.5,-.25) -- (-.4,-.25);
\draw[green]  (-.4,-.25) --  (-.4,-1.5) ;
\draw [-stealth] (-.4,-.99) --  (-.4,-1) ;
\draw[green]  (.4,-1.5) --  (.4,-.25) ;
\draw [-stealth] (.4,-.91) --  (.4,-.9) ;
\draw[green]  (0.4,-.25) -- (2.5,-.25);

\draw[green]  (-.4,-1.5) arc (180:360:.4);	

\draw  [-stealth] (4,1) (-1.51,-.25) -- (-1.5,-.25);
\draw  [-stealth] (0,1)  (1.49,-.25) -- (1.5,-.25);

\draw[red, -stealth] (-2.5,0) -- (2.5,0);
\draw[-stealth] (0,-2.5) -- (0,2.5);
\draw[red] (0,-1.5) -- (0,1.5);

\draw (2.8,0) node[] {\tiny Re $k$};
\draw (0,2.8) node[] {\tiny Im $k$};
\draw (-.2,1.5) node[] {\tiny $iq_0$};
\draw (-.17,-1.6) node[] {\tiny -$iq_0$};
\draw (1.7,1.6) node[] {\tiny $k_n$};
\draw (1.7,-1.6) node[] {\tiny $k_n^*$};
\draw (.3,-.12) node[] {\tiny $0^+$};
\draw (-.25,-.14) node[] {\tiny $0^-$};

\draw (-2,3) node[] {\tiny Sheet I:};

\fill[blue] (0,1.5)  circle (1pt);
\fill[blue] (0,-1.5)  circle (1pt);
\fill[blue] (1.5,1.5)  circle (1pt);
\fill[blue] (1.5,-1.5)  circle (1pt);

\fill[blue] (.1,0)  circle (1pt);
\fill[blue] (-.1,0)  circle (1pt);

\end{tikzpicture}
\captionsetup{font=footnotesize}
\captionof{figure}{The first sheet of the Riemann surface, showing the branch cut (red), the contour (green) and the regions when $\Im \lambda>0$ (grey) and $\Im \lambda<0$ (white), where $\lambda(k)=(r_{1}r_{2})^{\frac{1}{2}}e^{i\frac{\theta_{1}+\theta_{2}}{2}}$.\label{fig3}}
\end{minipage}%
&
\begin{minipage}{\dimexpr0.5\textwidth-2\tabcolsep}

\begin{tikzpicture}

\draw [fill=gray!50!white] (0,0) rectangle (-2.5,-2.5);
\draw [fill=gray!50!white] (2.5,0) rectangle (-2.5,-2.5);
\draw (2.5,-2.5) rectangle (-2.5,2.5);

\draw[green] (-2.5,.25) -- (-.4,.25);
\draw[green]  (-.4,.25) --  (-.4,1.5) ;
\draw [-stealth]  (-.4,1.01) -- (-.4,1)  ;
\draw[green] (.4,1.5) --  (.4,.25) ;
\draw [-stealth] (.4,1.09)--  (.4,1.1)  ;
\draw[green] (.4,.25) -- (2.5,.25);

\draw[green] (-.4,1.5) arc (180:0:.4);	

\draw  [-stealth] (0,1)   (-1.49,.25)-- (-1.5,.25);
\draw  [-stealth] (0,1)   (1.5,.25)-- (1.49,.25);

\draw[green] (-2.5,-.25) -- (-.4,-.25);
\draw[green] (-.4,-.25) --  (-.4,-1.5) ;
\draw [-stealth] (-.4,-1) --  (-.4,-.99) ;
\draw[green] (.4,-1.5) --  (.4,-.25) ;
\draw [-stealth] (.4,-1.09) --  (.4,-1.1) ;
\draw[green] (0.4,-.25) -- (2.5,-.25);

\draw[green] (-.4,-1.5) arc (180:360:.4);	

\draw  [-stealth] (4,1) (-1.49,-.25) -- (-1.5,-.25);
\draw  [-stealth] (0,1)  (1.5,-.25) -- (1.49,-.25);

\draw[red,-stealth] (-2.5,0) -- (2.5,0);
\draw[-stealth] (0,-2.5) -- (0,2.5);
\draw[red] (0,-1.5) -- (0,1.5);

\draw (2.8,0) node[] {\tiny Re $k$};
\draw (0,2.8) node[] {\tiny Im $k$};
\draw (-.2,1.5) node[] {\tiny $iq_0$};
\draw (-.17,-1.6) node[] {\tiny -$iq_0$};
\draw (1.7,1.6) node[] {\tiny $k_n$};
\draw (1.7,-1.6) node[] {\tiny $k_n^*$};
\draw (.3,-.12) node[] {\tiny $0^+$};
\draw (-.25,-.14) node[] {\tiny $0^-$};

\draw (-2,3) node[] {\tiny Sheet II:};

\fill[blue] (0,1.5)  circle (1pt);
\fill[blue] (0,-1.5)  circle (1pt);
\fill[blue] (1.5,1.5)  circle (1pt);
\fill[blue] (1.5,-1.5)  circle (1pt);

\fill[blue] (.1,0)  circle (1pt);
\fill[blue] (-.1,0)  circle (1pt);

\end{tikzpicture}
\captionsetup{font=footnotesize}
\captionof{figure}{The second sheet of the Riemann surface, showing the branch cut (red), the contour (green) and the regions when $\Im \lambda<0$ (grey) and $\Im \lambda>0$ (white), where $\lambda(k)=-(r_{1}r_{2})^{\frac{1}{2}}e^{i\frac{\theta_{1}+\theta_{2}}{2}}$. \label{fig4}}
\end{minipage}

\end{tabular*}%

\captionof{figure}{The two sheets of the Riemann surface $\mathbb{K}$.}
\end{figure}

\begin{remark}
Topologically, the Riemann surface $\mathbb{K}$ is equivalent to a surface with genus $0$.
\end{remark}

%

The eigenfunctions are defined by the following boundary conditions
\begin{equation}
\phi(x,k)\sim w e^{-i\lambda x}, \ \ \ \overline{\phi}(x,k)\sim \overline{w}e^{i\lambda x}
\end{equation}
as $x\rightarrow-\infty$,
\begin{equation}
\psi(x,k)\sim v e^{i\lambda x}, \ \ \ \overline{\psi}(x,k)\sim \overline{v}e^{-i\lambda x}
\end{equation}
as $x\rightarrow +\infty$, where
\begin{equation}
\label{E:boundary conditions 1'}
w=\left(\begin{array}{cc}
\lambda+k\\
-i\widetilde{q}^{*}
\end{array}\right), \ \ \
\overline{w}=\left(\begin{array}{cc}
-i\widetilde{q}\\
\lambda+k
\end{array}\right),
\end{equation}
\begin{equation}
\label{E:boundary conditions 2'}
v=\left(\begin{array}{cc}
-i\widetilde{q}\\
\lambda+k
\end{array}\right), \ \ \
\overline{v}=
\left(\begin{array}{cc}
\lambda+k\\
-i\widetilde{q}^{*}
\end{array}\right)
\end{equation}
satisfy the boundary conditions, but they are not unique. Hence, we define the bounded eigenfunctions as follows:
\begin{equation}
\label{E:definition 1'}
M(x,k)=e^{i\lambda x}\phi(x,k), \ \ \ \overline{M}(x,k)=e^{-i\lambda x}\overline{\phi}(x,k),
\end{equation}
\begin{equation}
\label{E:definition 2'}
N(x,k)=e^{-i\lambda x}\psi(x,k), \ \ \ \overline{N}(x,k)=e^{i\lambda x}\overline{\psi}(x,k).
\end{equation}
The eigenfunctions can be represented by means of the following integral equations
\begin{equation}
M(x,k)=
\left(\begin{array}{cc}
\lambda+k\\
-i\widetilde{q}^{*}
\end{array}\right)
+\int_{-\infty}^{+\infty}G_{-}(x-x',k)((Q-Q_{-})M)(x',k)dx',
\end{equation}
\begin{equation}
\overline{M}(x,k)=
\left(\begin{array}{cc}
-i\widetilde{q}\\
\lambda+k
\end{array}\right)
+\int_{-\infty}^{+\infty}\overline{G}_{-}(x-x',k)((Q-Q_{-})M)(x',k)dx',
\end{equation}
\begin{equation}
N(x,k)=
\left(\begin{array}{cc}
-i\widetilde{q}\\
\lambda+k
\end{array}\right)
+\int_{-\infty}^{+\infty}G_{+}(x-x',k)((Q-Q_{+})M)(x',k)dx',
\end{equation}
\begin{equation}
\overline{N}(x,k)=\left(\begin{array}{cc}
\lambda+k\\
-i\widetilde{q}^{*}
\end{array}\right)
+\int_{-\infty}^{+\infty}\overline{G}_{+}(x-x',k)((Q-Q_{+})M)(x',k)dx',
\end{equation}
where
\begin{equation}
G_{-}(x,k)=\frac{\theta(x)}{2\lambda}[(1+e^{2i\lambda x})\lambda I-i(e^{2i\lambda x}-1)(ikJ+Q_{-})],
\end{equation}
\begin{equation}
\overline{G}_{-}(x,k)=\frac{\theta(x)}{2\lambda}[(1+e^{-2i\lambda x})\lambda I+i(e^{-2i\lambda x}-1)(ikJ+Q_{-})],
\end{equation}
\begin{equation}
G_{+}(x,k)=-\frac{\theta(-x)}{2\lambda}[(1+e^{-2i\lambda x})\lambda I+i(e^{-2i\lambda x}-1)(ikJ+Q_{+})],
\end{equation}
\begin{equation}
\overline{G}_{+}(x,k)=-\frac{\theta(-x)}{2\lambda}[(1+e^{2i\lambda x})\lambda I-i(e^{2i\lambda x}-1)(ikJ+Q_{+})].
\end{equation}
\begin{definition}
We say $f\in L^{1,1}(\mathbb{R})$ if $\int_{-\infty}^{+\infty}|f(x)|\cdot(1+|x|)dx<\infty$.
\end{definition}
Then, using similar methods as in the prior case ($\sigma=-1, \Delta \theta = \pi$) we find the following result (see also \cite{Demontis2})


\begin{theorem}
\label{Thm4.2}
Suppose the entries of $Q-Q_{\pm}$ belong to $L^{1,1}(\mathbb{R})$, then for each $x\in\mathbb{R}$, the eigenfunctions $M(x,k)$ and $N(x,k)$ are continuous for $k\in \overline{\mathbb{K}^{+}}\cup \partial \overline{\mathbb{K}^{-}}$ and analytic for $k\in \mathbb{K}^{+}$, $\overline{M}(x,k)$ and $\overline{N}(x,k)$ are continuous for $k\in \overline{\mathbb{K}^{-}}\cup \partial \overline{\mathbb{K}^{+}}$ and analytic for $k\in \mathbb{K}^{-}$.
\end{theorem}

If we assume that the entries of $Q-Q_{\pm}$ do not grow faster than $e^{-ax^{2}}$, where $a$ is a positive real number, by similar methods as in Case 1, we have the following result.

\begin{theorem}
\label{T:Schwarz condition 2}
Suppose the entries of $Q-Q_{\pm}$ do not grow faster than $e^{-ax^{2}}$, where $a$ is a positive real number, then for each $x\in\mathbb{R}$, the eigenfunctions $M(x,k)$, $N(x,k)$, $\overline{M}(x,k)$ and $\overline{N}(x,k)$ are analytic in the Riemann surface $\mathbb{K}$.
\end{theorem}

\subsubsection{Scattering data}
We have
\begin{equation}
\label{E:linear combination 1'}
\phi(x,k)=b(k)\psi(x,k)+a(k)\overline{\psi}(x,k)
\end{equation}
and
\begin{equation}
\label{E:linear combination 2'}
\overline{\phi}(x,k)=\overline{a}(k)\psi(x,k)+\overline{b}(k)\overline{\psi}(x,k)
\end{equation}
hold for any $k$ such that all four eigenfunctions exist. Moreover,
\begin{equation}
a(k)\overline{a}(k)-b(k)\overline{b}(k)=1,
\end{equation}
where
\begin{equation}
a(k)=\frac{W(\phi(x,k),\psi(x,k))}{2\lambda(\lambda+k)}, \ \ \ \overline{a}(k)=-\frac{W(\overline{\phi}(x,k),\overline{\psi}(x,k))}{2\lambda(\lambda+k)},
\end{equation}
\begin{equation}
b(k)=-\frac{W(\phi(x,k),\overline{\psi}(x,k))}{2\lambda(\lambda+k)}, \ \ \ \overline{b}(k)=\frac{W(\overline{\phi}(x,k),\psi(x,k))}{2\lambda(\lambda+k)}.
\end{equation}
When $k\in (-iq_{0}, iq_{0})$, the above scattering data and eigenfunctions are defined by means of the corresponding values on the right/left edge of the cut, are labeled with superscripts $\pm$ as clarified below. Explicitly, one has
\begin{equation}
a^{\pm}(k)=\frac{W(\phi^{\pm}(x,k),\psi^{\pm}(x,k))}{2\lambda^{\pm}(\lambda^{\pm}+k)}, \ \ \ k\in (-iq_{0}, iq_{0}),
\end{equation}
\begin{equation}
\overline{a}^{\pm}(k)=-\frac{W(\overline{\phi}^{\pm}(x,k),\overline{\psi}^{\pm}(x,k))}{2\lambda^{\pm}(\lambda^{\pm}+k)}, \ \ \ k\in (-iq_{0}, iq_{0}),
\end{equation}
\begin{equation}
b^{\pm}(k)=-\frac{W(\phi^{\pm}(x,k),\overline{\psi}^{\pm}(x,k))}{2\lambda^{\pm}(\lambda^{\pm}+k)}, \ \ \ k\in (-iq_{0}, iq_{0}),
\end{equation}
\begin{equation}
\overline{b}^{\pm}(k)=\frac{W(\overline{\phi}^{\pm}(x,k),\psi^{\pm}(x,k))}{2\lambda^{\pm}(\lambda^{\pm}+k)},  \ \ \ k\in (-iq_{0}, iq_{0}).
\end{equation}

Then from the analytic behavior of the eigenfunctions we have the following theorem.
\begin{theorem}
\label{T:2'}
Suppose the entries of $Q-Q_{\pm}$ belong to $L^{1,1}(\mathbb{R})$, then $a(k)$ is continuous for $k\in \overline{\mathbb{K}^{+}}\cup \partial \overline{\mathbb{K}^{-}}\setminus\{\pm iq_{0}\}$ and analytic for $k\in \mathbb{K}^{+}$, and $\overline{a}(k)$ is continuous for $k\in \overline{\mathbb{K}^{-}}\cup \partial \overline{\mathbb{K}^{+}}\setminus\{\pm iq_{0}\}$ and analytic for $k\in \mathbb{K}^{-}$. Moreover, $b(k)$ and $\overline{b}(k)$ are continuous in $k\in \mathbb{R}\cup (-iq_{0}, iq_{0})$. In addition, if the entries of $Q-Q_{\pm}$ do not grow faster than $e^{-ax^{2}}$, where $a$ is a positive real number, then $a(k)\lambda(k)$, $\overline{a}(k)\lambda(k)$, $b(k)\lambda(k)$ and $\overline{b}(k)\lambda(k)$ are analytic for $k\in\mathbb{K}$.
\end{theorem}

\subsection{Symmetry reductions}
The symmetry in the potential induces a symmetry between the eigenfunctions. Indeed, if $v(x,k)=(v_{1}(x,k), v_{2}(x,k))^{T}$ solves
(\ref{E:nonlocal ode}), then $(v_{2}^{*}(-x,-k^{*}), v_{1}^{*}(-x,-k^{*}))^{T}$ also solves (\ref{E:nonlocal ode}). Moreover, if $k\rightarrow -k^{*}$, then $\theta_{1}\rightarrow \pi-\theta_{1}$, $\theta_{2}\rightarrow \pi-\theta_{2}$. Hence, $\lambda_{j}^{*}(-k^{*})=-\lambda_{j}(k)$, where $j=1, 2$. Taking into account boundary conditions (\ref{E:boundary conditions 1'}) and (\ref{E:boundary conditions 2'}), we can obtain
\begin{equation}
\psi(x,k)=-\left(\begin{array}{cc}
0& 1\\
1& 0
\end{array}\right)\phi^{*}(-x,-k^{*})
\end{equation}
and
\begin{equation}
\overline{\psi}(x,k)=-\left(\begin{array}{cc}
0& 1\\
1& 0
\end{array}\right)\overline{\phi}^{*}(-x,-k^{*}).
\end{equation}
Similarly, we can get
\begin{equation}
N(x,k)=-\left(\begin{array}{cc}
0& 1\\
1& 0
\end{array}\right)M^{*}(-x,-k^{*})
\end{equation}
and
\begin{equation}
\overline{N}(x,k)=-\left(\begin{array}{cc}
0& 1\\
1& 0
\end{array}\right)\overline{M}^{*}(-x,-k^{*}).
\end{equation}
Moreover,
\begin{equation}
a^{*}(-k^{*})=a(k),
\end{equation}
\begin{equation}
\overline{a}^{*}(-k^{*})=\overline{a}(k),
\end{equation}
\begin{equation}
b^{*}(-k^{*})=-\overline{b}(k).
\end{equation}
When using a  particular single sheet for the Riemann surface of the function $\lambda^{2}=k^{2}+q_{0}^{2}$, the involution $(k,\lambda)\rightarrow (k,-\lambda)$ leads to relationships between eigenfunctions 
across the cut; namely it relates values of eigenfunctions and scattering data for the same value of $k$ from either side of the cut. One has
\begin{equation}
\phi^{\mp}(x,k)=\frac{-\lambda^{\pm}+k}{-i\widetilde{q}}\cdot\overline{\phi}^{\pm}(x,k), \ \ \ \psi^{\mp}(x,k)=\frac{-\lambda^{\pm}+k}{-i\widetilde{q}^{*}}\cdot\overline{\psi}^{\pm}(x,k)
\end{equation}
for $k\in [-iq_{0}, iq_{0}]$. Similarly,
\begin{equation}
M^{\mp}(x,k)=\frac{-\lambda^{\pm}+k}{-i\widetilde{q}}\cdot\overline{M}^{\pm}(x,k), \ \ \ N^{\mp}(x,k)=\frac{-\lambda^{\pm}+k}{-i\widetilde{q}^{*}}\cdot\overline{N}^{\pm}(x,k)
\end{equation}
for $k\in [-iq_{0}, iq_{0}]$.
Then
\begin{equation}
a^{\pm}(k)=\overline{a}^{\mp}(k), \ \ \ b^{\pm}(k)=\frac{\widetilde{q}^{*}}{\widetilde{q}}\cdot \overline{b}^{\mp}(k)
\end{equation}
for $k\in [-iq_{0}, iq_{0}]$.
\subsection{Riemann-Hilbert problem}
We will develop the Riemann-Hilbert problem across $\partial\mathbb{K}^{+}\cup \partial\mathbb{K}^{-}$.

(1) Riemann-Hilbert problem across the real axis (see Fig. \ref{fig6}.)
Note that (\ref{E:linear combination 1'}) and (\ref{E:linear combination 2'}) can be written as
\begin{equation}
\mu(x,k)=\rho(k)e^{2i\lambda x}N(x,k)+\overline{N}(x,k)
\end{equation}
and
\begin{equation}
\overline{\mu}(x,k)=N(x,k)+\overline{\rho}(k)e^{-2i\lambda x}\overline{N}(x,k),
\end{equation}
where $\mu(x,k)=M(x,k)a^{-1}(k)$, $\overline{\mu}(x,k)=\overline{M}(x,k)\overline{a}^{-1}(k)$, $\rho(k)=b(k)a^{-1}(k)$ and $\overline{\rho}(k)=\overline{b}(k)\overline{a}^{-1}(k)$. Introducing the $2\times2$ matrices
\begin{equation}
m_{+}(x,k)=(\mu(x,k), N(x,k)), \ \ \ m_{-}(x,k)=(\overline{N}(x,k), \overline{\mu}(x,k)).
\end{equation}
Hence, we can write the 'jump' conditions as
\begin{equation}
m_{+}(x,k)-m_{-}(x,k)=m_{-}(x,k)\left(\begin{array}{cc}
-\rho(k)\overline{\rho}(k)& -\overline{\rho}(k)e^{-2i\lambda x}\\
\rho(k)e^{2i\lambda x}& 0
\end{array}\right)
\end{equation}
for $k\in \mathbb{R}$.

(2) Riemann-Hilbert problem across $[0, iq_{0}]$ (see Fig. \ref{fig7}.). By (\ref{E:linear combination 1'}), (\ref{E:linear combination 2'}), the notation $\lambda=\lambda^{+}=-\lambda^{-}$ and the symmetry relations of Jost functions and scattering data, we have
\begin{equation}
\frac{M^{+}(x,k)}{a^{+}(k)}=\frac{b^{+}(k)}{a^{+}(k)}\cdot N^{+}(x,k)\cdot e^{2i\lambda x}+\frac{-i\widetilde{q}^{*}}{\lambda^{-}+k}\cdot N^{-}(x,k)
\end{equation}
and
\begin{equation}
\frac{-iq_{-}}{\lambda^{-}+k}\cdot\frac{M^{-}(x,k)}{a^{-}(k)}=N^{+}(x,k)+\frac{\widetilde{q}}{\widetilde{q}^{*}}\cdot\frac{b^{-}(k)}{a^{-}(k)}\cdot\frac{-i\widetilde{q}^{*}}{\lambda^{-}+k}\cdot N^{-}(x,k)\cdot e^{-2i\lambda x}.
\end{equation}
Introducing the $2\times2$ matrices
\begin{equation}
m^{+}(x,k)=(\mu^{+}(x,k), N^{+}(x,k)), \ \ \ m^{-}(x,k)=\left(\frac{-i\widetilde{q}^{*}}{\lambda^{-}+k}N^{-}(x,k), \frac{-i\widetilde{q}}{\lambda^{-}+k}\cdot\mu^{-}(x,k)\right),
\end{equation}
where $\mu^{+}(x,k)=M^{+}(x,k)a^{+}(k)^{-1}$, $\mu^{-}(x,k)=M^{-}(x,k)a^{-}(k)^{-1}$, $\rho^{+}(k)=b^{+}(k)a^{+}(k)^{-1}$ and $\rho^{-}(k)=b^{-}(k)a^{-}(k)^{-1}$. Then
\begin{equation}
m^{+}(x,k)-m^{-}(x,k)=m^{-}(x,k)\left(\begin{array}{cc}
-\frac{\widetilde{q}}{\widetilde{q}^{*}}\cdot\rho^{+}(k)\rho^{-}(k)& -\frac{\widetilde{q}}{\widetilde{q}^{*}}\cdot \rho^{-}(k)e^{-2i \lambda x}\\
\rho^{+}(k)e^{2i\lambda x}& 0
\end{array}\right)
\end{equation}
for $k\in \mathbb{C}^{+}$.

(3) Riemann-Hilbert problem across $[-iq_{0}, 0)$ (see Fig. \ref{fig8}). By (\ref{E:linear combination 1'}), (\ref{E:linear combination 2'}), the notation $\lambda=\lambda^{+}=-\lambda^{-}$ and the symmetry relations of Jost functions and scattering data, we can get
\begin{equation}
\frac{-\lambda^{-}+k}{-i\widetilde{q}}\cdot \frac{\overline{M}^{-}(x,k)}{\overline{a}^{-}(x,k)}=\frac{\widetilde{q}^{*}}{\widetilde{q}}\cdot \frac{\overline{b}^{-}(k)}{\overline{a}^{-}(k)}\cdot \frac{-\lambda^{-}+k}{-i\widetilde{q}^{*}}\cdot \overline{N}^{-}(x,k)\cdot e^{2i\lambda x}
+\overline{N}^{+}(x,k)
\end{equation}
and
\begin{equation}
\frac{\overline{M}^{+}(x,k)}{\overline{a}^{+}(k)}=\frac{-\lambda^{-}+k}{-i\widetilde{q}^{*}}\cdot \overline{N}^{-}(x,k)+\frac{\overline{b}^{+}(k)}{\overline{a}^{+}(k)}\cdot \overline{N}^{+}(x,k)\cdot e^{-2i \lambda x}.
\end{equation}
Introducing the $2\times2$ matrices
\begin{equation}
\overline{m}^{+}(x,k)=\left(\overline{\mu}^{+}(x,k), \overline{N}^{+}(x,k)\right), \ \ \ \overline{m}^{-}(x,k)=\left(\frac{-\lambda^{-}+k}{-i\widetilde{q}^{*}}\cdot \overline{N}^{-}(x,k), \frac{-\lambda^{-}+k}{-i\widetilde{q}}\cdot \overline{\mu}^{-}(k)\right),
\end{equation}
where $\overline{\mu}^{+}(x,k)=\overline{M}^{+}(x,k)\overline{a}^{+}(k)^{-1}$, $\overline{\mu}^{-}(x,k)=\overline{M}^{-}(x,k)\overline{a}^{-}(k)^{-1}$, $\overline{\rho}^{+}(k)=\overline{b}^{+}(k)\overline{a}^{+}(k)^{-1}$ and $\overline{\rho}^{-}(k)=\overline{b}^{-}(k)\overline{a}^{-}(k)^{-1}$. Then
\begin{equation}
\overline{m}^{+}(x,k)-\overline{m}^{-}(x,k)=\overline{m}^{-}(x,k)\left(\begin{array}{cc}
-\frac{\widetilde{q}^{*}}{\widetilde{q}}\cdot\overline{\rho}^{+}(k)\overline{\rho}^{-}(k)& -\frac{\widetilde{q}^{*}}{\widetilde{q}}\cdot \overline{\rho}^{-}(k)\cdot e^{2i\lambda x}\\
\overline{\rho}^{+}(k)e^{-2i\lambda x}& 0
\end{array}\right)
\end{equation}
for $k\in \mathbb{C}^{-}$.

\noindent
\begin{figure}
\begin{tabular*}{\textwidth}{ccc}
\hspace*{-2cm}
\begin{minipage}{\dimexpr0.5\textwidth-1.5\tabcolsep}

\begin{tikzpicture}

\draw[decoration = {zigzag,segment length = 2mm, amplitude = 1mm}, decorate,red] (-2.5,0) -- (2.4,0);
\draw[-stealth] (-2.5,0) -- (2.5,0);
\draw[-stealth] (0,-2.5) -- (0,2.5);

\draw (2.8,0) node[] {\tiny Re $k$};
\draw (0,2.8) node[] {\tiny Im $k$};
\draw (.3,-.17) node[] {\tiny $0^+$};
\draw (-.25,-.17) node[] {\tiny $0^-$};

\end{tikzpicture}

\captionsetup{font=footnotesize}
\captionof{figure}{Riemann-Hilbert problem across the real axis.}
\label{fig6}
\end{minipage}%
&
\hspace*{0cm}
\begin{minipage}{\dimexpr0.5\textwidth-2\tabcolsep}

\begin{tikzpicture}

\draw[-stealth] (-2.5,0) -- (2.5,0);
\draw[-stealth] (0,-2.5) -- (0,2.5);
\draw[decoration = {zigzag,segment length = 2mm, amplitude = 1mm}, decorate,red] (0,0) -- (0,1.8);

\draw (2.8,0) node[] {\tiny Re $k$};
\draw (0,2.8) node[] {\tiny Im $k$};
\draw (.3,-.14) node[] {\tiny $0^+$};
\draw (-.25,-.14) node[] {\tiny $0^-$};
\draw (0.3,1.8) node[] {\tiny $i q_0$};
\fill[blue] (0,1.8)  circle (1pt);

\end{tikzpicture}

\captionsetup{font=footnotesize}
\captionof{figure}{Riemann-Hilbert problem across $[0, iq_{0}]$.}
\label{fig7}
\end{minipage}%
&
\hspace*{0cm}
\begin{minipage}{\dimexpr0.5\textwidth-2\tabcolsep}

\begin{tikzpicture}

\draw[-stealth] (-2.5,0) -- (2.5,0);
\draw[-stealth] (0,-2.5) -- (0,2.5);
\draw[decoration = {zigzag,segment length = 2mm, amplitude = 1mm}, decorate,red] (0,0) -- (0,-1.8);

\draw (2.8,0) node[] {\tiny Re $k$};
\draw (0,2.8) node[] {\tiny Im $k$};
\draw (.3,-.14) node[] {\tiny $0^+$};
\draw (-.25,-.14) node[] {\tiny $0^-$};
\draw (-0.4,-1.8) node[] {\tiny $-i q_0$};
\fill[blue] (0,-1.8)  circle (1pt);

\end{tikzpicture}

\captionsetup{font=footnotesize}
\captionof{figure}{Riemann-Hilbert problem across $[-iq_{0}, 0)$.}
\label{fig8}
\end{minipage}%

\end{tabular*}%

\end{figure}

%
%
%
%

\begin{remark}
For the Riemann-Hilbert problem
\begin{equation}
F^{+}(\xi)-F^{-}(\xi)=F^{-}(\xi)g(\xi),
\end{equation}
on the contour $\widetilde{\Sigma}:=\Sigma_{1}\cup\Sigma_{2}\cup\Sigma_{3}$, where $\Sigma_{1}:=\mathbb{R}$, $\Sigma_{2}:=[0, iq_{0}]$, $\Sigma_{3}:=[-iq_{0}, 0)$, $g(\xi)$ is H\"{o}lder-continuously in $\widetilde{\Sigma}$ and $g(\xi)=g_{m}(\xi)$ is chosen depending on which piece of the contour is considered, where $m=1, 2, 3$. We can consider the projection operators
\begin{equation}
(P_{j}(f))(k)=\frac{1}{2\pi i}\int_{\widetilde{\Sigma}}\frac{\lambda(k)+\lambda(\xi)}{2\lambda(\xi)}\cdot\frac{f(\xi)}{\xi-k}d\xi, \ \ k\in \mathbb{C}_{j}, \ \ j=1,2,
\end{equation}
where $\frac{\lambda(k)+\lambda(\xi)}{2\lambda(\xi)}\cdot\frac{d\xi}{\xi-k}$ is the Weierstrass kernel, $\int_{\widetilde{\Sigma}}:=\int_{\Sigma_{1}\cup\Sigma_{2}}+\int_{\Sigma_{1}\cup\Sigma_{3}}$ denote the integrals along the oriented contours in Figure \ref{fig3}. One can show that (\cite{Zverovich})
\begin{equation*}
\frac{\lambda(k)+\lambda(\xi)}{2\lambda(\xi)}\cdot\frac{d\xi}{\xi-k}=\frac{d\xi}{\xi-k}+ regular \ terms
\end{equation*}
for $(\xi, \lambda(\xi))\rightarrow (k, \lambda(k))$. If $f_{j}\ (j=1, 2)$, is sectionally analytic in $\mathbb{C}_{j}$ and rapidly decaying as $|k|\rightarrow\infty$ in the proper sheet, we have
\begin{equation}
(P_{j}(f_{j}))(k)=(-1)^{j-1}f_{j}(k), \ \ \ k\in \mathbb{C}_{j},
\end{equation}
\begin{equation}
(P_{j}(f_{l}))(k)=0, \ \ \ k\in \mathbb{C}_{j}.
\end{equation}
We obtain
\begin{equation}
F^{-}(k)=\frac{1}{2\pi i}\int_{\widetilde{\Sigma}}\frac{\lambda(k)+\lambda(\xi)}{2\lambda(\xi)}\cdot\frac{F^{-}(\xi)g(\xi)}{\xi-k}d\xi,\ \ \ k\in\mathbb{C}_{2},
\end{equation}
\begin{equation}
F^{+}(k)=\frac{1}{2\pi i}\int_{\widetilde{\Sigma}}\frac{\lambda(k)+\lambda(\xi)}{2\lambda(\xi)}\cdot\frac{F^{-}(\xi)g(\xi)}{\xi-k}d\xi,\ \ \ k\in\mathbb{C}_{1}.
\end{equation}
For $\xi\in\widetilde{\Sigma}$, we can take limits from proper sheet and connected by the analogue of Plemelj-Sokhotskii formulas
\begin{equation}
F^{+}(\xi_{0})=\lim_{k\rightarrow\xi_{0}\in\widetilde{\Sigma}, k\in\mathbb{C}_{1}}\frac{1}{2\pi i}\int_{\widetilde{\Sigma}}\frac{\lambda(k)+\lambda(\xi)}{2\lambda(\xi)}\cdot\frac{F^{-}(\xi)g(\xi)}{\xi-k}d\xi,
\end{equation}
\begin{equation}
F^{-}(\xi_{0})=\lim_{k\rightarrow\xi_{0}\in\widetilde{\Sigma}, k\in\mathbb{C}_{2}}\frac{1}{2\pi i}\int_{\widetilde{\Sigma}}\frac{\lambda(k)+\lambda(\xi)}{2\lambda(\xi)}\cdot\frac{F^{-}(\xi)g(\xi)}{\xi-k}d\xi.
\end{equation}
\end{remark}

In principle, the above RH problem can be analyzed for the two sheeted problem. But a uniformizing coordinate makes the problem considerably more straight forward. This is discussed next.

\subsection{Uniformization coordinates}
Before discussing the properties of scattering data and solving the inverse problem, we introduce a uniformization variable $z$ see also \cite{BK14}), defined by the conformal mapping:
\begin{equation}
\label{zsect4}
z=z(k)=k+\lambda(k),
\end{equation}
where $ \lambda= \sqrt{k^2+q_0^2}$ and the inverse mapping is given by
\begin{equation}
k=k(z)=\frac{1}{2}\left(z-\frac{q_{0}^{2}}{z}\right).
\end{equation}
Then
\begin{equation}
\lambda(z)=\frac{1}{2}\left(z+\frac{q_{0}^{2}}{z}\right).
\end{equation}
We let $C_{0}$ be the circle of radius $q_{0}$ centered at the origin in $z-$ plane. We observe that

(1) The branch cut on either sheet is mapped onto $C_{0}$. In particular, $z(\pm iq_{0})=\pm iq_{0}$ from either sheet, $z(0_{I}^{\pm})=\pm q_{0}$ and $z(0_{II}^{\pm})=\mp q_{0}$.

(2) $\mathbb{K}_{1}$ is mapped onto the exterior of $C_{0}$, $\mathbb{K}_{2}$ is mapped onto the interior of $C_{0}$. In particular, $z(\infty_{I})=\infty$ and $z(\infty_{II})=0$; the first/second quadrants of $\mathbb{K}_{1}$ are mapped into the first/second quadrants outside $C_{0}$ respectively; the first/second quadrants of $\mathbb{K}_{2}$ are mapped into the second/first quadrants inside $C_{0}$ respectively; $z_{I}z_{II}=q_{0}^{2}$.

(3) The regions in $k-$ plane such that $\Im \lambda>0$ and $\Im \lambda<0$ are mapped onto $D^{+}=\{z\in\mathbb{C}:(|z|^{2}-q_{0}^{2})\cdot\Im z>0\}$ and $D^{-}=\{z\in\mathbb{C}:(|z|^{2}-q_{0}^{2})\cdot\Im z<0\}$ respectively.

From Theorem \ref{Thm4.2} and the definition of the uniformization variable $z$ by equation (\ref{zsect4}), we have that the eigenfunctions $M,N$ are analytic for $z \in D^{+}$ and
the eigenfunctions $\bar{M},\bar{N}$  are analytic for in $z \in D^{-}$.

See Figs \ref{fig5}, \ref{fig9}, \ref{fig10}.

\vspace{1 in}


\begin{figure}
\centering
\begin{tikzpicture}

\draw [fill=gray!50!white] (2.5,0) rectangle (-2.5,2.5);
\draw (2.5,-2.5) rectangle (-2.5,2.5);
 \draw  [red,fill=white]  (0,0) circle (1.5);

 \filldraw[red,fill=gray!50!white] (-1.5,0) arc (180:360:1.5);


\draw [-stealth] (-1.2,-.9) arc (224:225:1.5);
\draw [-stealth] (1.12,-1) arc (314:315:1.5);
\draw [-stealth] (1.12,1) arc (46:45:1.5);
\draw [-stealth] (-1.12,1) arc (136:135:1.5);

\draw[red] (-2.49,0) -- (2.5,0);
\draw[-stealth] (2.49,0) -- (2.5,0);
\draw[-stealth] (0,-2.5) -- (0,2.5);

\draw [-stealth]  (-1.76,0) -- (-1.75,0);
\draw [-stealth]  (-.99,0) -- (-1,0);
\draw [-stealth]  (.76,0) -- (.75,0);
\draw [-stealth]  (1.99,0) -- (2,0);

\draw (2.8,0) node[] {\tiny Re $z$};
\draw (0,2.8) node[] {\tiny Im $z$};
\draw (.25,1.7) node[] {\tiny $iq_0$};
\draw (.25,-1.7) node[] {\tiny -$iq_0$};
\draw (-.5,1) node[] {\tiny $q_0^2/z_n$};
\draw (-.5,-1) node[] {\tiny $-q_0^2/z_n^*$};
\draw (1.9,1.9) node[] {\tiny $z_n$};
\draw (1.9,-1.9) node[] {\tiny $z_n^*$};
\draw (2,-.2) node[] {\tiny $0_{I}^+ \Big( 0_{II}^- \Big)$};
\draw (-2,-.2) node[] {\tiny $0_{I}^-\Big( 0_{II}^+ \Big)$};
\draw (0.3,-.2) node[] {\tiny $\infty_{II}$};

\fill[blue] (0,1.5)  circle (1pt);
\fill[blue] (0,-1.5)  circle (1pt);
\fill[blue] (-.5,.75)  circle (1pt);
\fill[blue] (-.5,-.75)  circle (1pt);
\fill[blue] (0,0)  circle (1pt);
\fill[blue] (1.5,0)  circle (1pt);
\fill[blue] (-1.5,0)  circle (1pt);
\fill[blue] (1.75,1.75)  circle (1pt);
\fill[blue] (1.75,-1.75)  circle (1pt);

\end{tikzpicture}

\caption{The complex $z-$plane, showing the branch cut (red), the regions $D^{\pm}$ where $\Im \lambda>0$ (grey) and $\Im \lambda<0$ (white) respectively. In particular, the first sheet is mapped onto the region outside the circle, and the second sheet is mapped onto the region inside the circle.
Importantly, the eigenfunctions $M,N$ are analytic for  $z \in D^{+}$ and
the eigenfunctions $\bar{M},\bar{N}$  are analytic for $z \in D^{-}$.
\label{fig5}}

\end{figure}
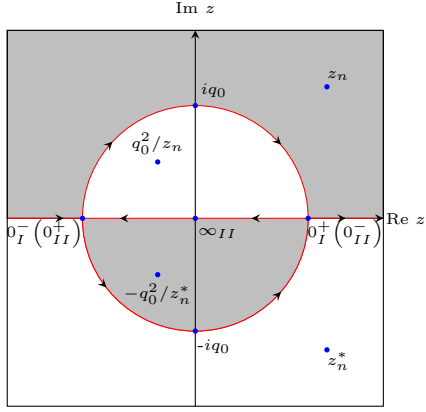

\vspace{1 in}
\noindent
\begin{figure}
\begin{tabular*}{\textwidth}{@{}cc@{}}
\begin{minipage}{\dimexpr0.5\textwidth-2\tabcolsep}

\begin{tikzpicture}

\draw[green]  (1.5,0) arc (180:0:1.5);
\draw[green] (4.5,-1) arc (360:180:1.5);
\draw [-stealth]  (4.12,1) arc (41:40:1.5);
\draw [-stealth] (1.88,-2) arc (219:220:1.5);
\draw [-stealth]  (1.88,1) arc (141:140:1.5);
\draw [-stealth] (4.12,-2) arc (319:320:1.5);

\draw[green] (0,0) -- (1.5,0) ;
\draw [-stealth] (.99,0) -- (1,0) ;
\draw[green] (0,-1) -- (1.5,-1) ;
\draw [-stealth] (.99,-1) -- (1,-1) ;

\draw[green] (4.5,0) -- (6,0) ;
\draw [-stealth] (5.49,0) -- (5.5,0) ;
\draw[green] (4.5,-1) -- (6,-1) ;
\draw [-stealth] (5.49,-1) -- (5.5,-1) ;

\fill[blue] (4.5,0)  circle (1pt);
\fill[blue] (4.5,-1)  circle (1pt);
\fill[blue] (1.5,0)  circle (1pt);
\fill[blue] (1.5,-1)  circle (1pt);
\fill[blue] (3,1.5)  circle (1pt);
\fill[blue] (3,-2.5)  circle (1pt);

\draw (3,1.8) node[] {\tiny $i q_0$};
\draw (3,-2.8) node[] {\tiny -$i q_0$};
\draw (4.3,0) node[] {\tiny $0_{I}^+$};
\draw (4.3,-1) node[] {\tiny $0_{I}^+$};
\draw (1.8,0) node[] {\tiny $0_{I}^-$};
\draw (1.8,-1) node[] {\tiny $0_{I}^-$};

\draw (1,2) node[] {\tiny Sheet I:};

\end{tikzpicture}


\captionsetup{font=footnotesize}
\captionof{figure}{The contour (green) in the first sheet is mapped into the above curve in $z-$ plane. \label{fig9}}
\end{minipage}%
&
\begin{minipage}{\dimexpr0.5\textwidth-2\tabcolsep}

\begin{tikzpicture}

\draw[green]  (1.5,0) arc (180:0:1.5);
\draw[green] (4.5,-1) arc (360:180:1.5);
\draw [-stealth]  (4.12,1) arc (41:40:1.5);
\draw [-stealth] (1.88,-2) arc (219:220:1.5);
\draw [-stealth]  (1.88,1) arc (141:140:1.5);
\draw [-stealth] (4.12,-2) arc (319:320:1.5);

\draw[green] (4.5,0) -- (1.5,0) ;
\draw [-stealth] (2.26,0) -- (2.25,0) ;
\draw [-stealth] (3.76,0) -- (3.75,0) ;

\draw[green] (4.5,-1) -- (1.5,-1) ;
\draw [-stealth] (2.26,-1) -- (2.25,-1) ;
\draw [-stealth] (3.76,-1)  -- (3.75,-1) ;

\fill[blue] (4.5,0)  circle (1pt);
\fill[blue] (4.5,-1)  circle (1pt);
\fill[blue] (1.5,0)  circle (1pt);
\fill[blue] (1.5,-1)  circle (1pt);
\fill[blue] (3,1.5)  circle (1pt);
\fill[blue] (3,-2.5)  circle (1pt);
\fill[blue] (3,0)  circle (1pt);
\fill[blue] (3,-1)  circle (1pt);

\draw (3,1.8) node[] {\tiny $i q_0$};
\draw (3,-2.8) node[] {\tiny -$i q_0$};
\draw (4.7,-.1) node[] {\tiny $0_{II}^-$};
\draw (4.7,-.9) node[] {\tiny $0_{II}^-$};
\draw (1.3,-.1) node[] {\tiny $0_{II}^+$};
\draw (1.3,-.9) node[] {\tiny $0_{II}^+$};
\draw (3,.2) node[] {\tiny $\infty_{ II}$};
\draw (3,-1.2) node[] {\tiny $\infty_{ II}$};

\draw (1,2) node[] {\tiny Sheet II:};

\end{tikzpicture}

\captionsetup{font=footnotesize}
\captionof{figure}{The contour (green) in the second sheet is mapped into the above curve in $z-$ plane. \label{fig10}}
\end{minipage}

\end{tabular*}%
\captionsetup{labelformat=empty}
\caption{The contours (green) in two sheets (Figure \ref{fig3} and Figure \ref{fig4}) are mapped into Figure \ref{fig9} and Figure \ref{fig10} respectively via the uniformization coordinate.}

\end{figure}

\subsection{Symmetries via uniformization coordinates}
It is found that (1) when $z\rightarrow -z^{*}$, then $(k,\lambda)\rightarrow (-k^{*}, -\lambda^{*})$; (2) when $z\rightarrow -\frac{q_{0}^{2}}{z}$, then $(k, \lambda)\rightarrow (k, -\lambda)$. Hence,
\begin{equation}
\psi(x,z)=-\left(\begin{array}{cc}
0& 1\\
1& 0
\end{array}\right)\phi^{*}(-x,-z^{*}),
\end{equation}
\begin{equation}
\overline{\psi}(x,z)=-\left(\begin{array}{cc}
0& 1\\
1& 0
\end{array}\right)\overline{\phi}^{*}(-x,-z^{*}),
\end{equation}
\begin{equation}
\phi\left(x,-\frac{q_{0}^{2}}{z}\right)=\frac{\frac{q_{0}^{2}}{z}}{i\widetilde{q}}\cdot\overline{\phi}(x,z), \ \ \ \psi\left(x,-\frac{q_{0}^{2}}{z}\right)=\frac{-i\widetilde{q}}{z}\cdot\overline{\psi}(x,z), \ \ \ z\in D^{-}.
\end{equation}

Similarly, we can get
\begin{equation}
\label{NM2}
N(x,z)=-\left(\begin{array}{cc}
0& 1\\
1& 0
\end{array}\right)M^{*}(-x,-z^{*})
\end{equation}
and
\begin{equation}
\label{NbMb2}
\overline{N}(x,z)=-\left(\begin{array}{cc}
0& 1\\
1& 0
\end{array}\right)\overline{M}^{*}(-x,-z^{*}).
\end{equation}
Moreover,
\begin{equation}
a^{*}(-z^{*})=a(z), \ \ \ \Im z>0,
\end{equation}
\begin{equation}
\overline{a}^{*}(-z^{*})=\overline{a}(z), \ \ \ \Im z<0,
\end{equation}
\begin{equation}
b^{*}(-z^{*})=-\overline{b}(z),
\end{equation}
\begin{equation}
\label{E:scatering'}
a\left(-\frac{q_{0}^{2}}{z}\right)=\overline{a}(z),\ \ \ z\in D^{-}; \ \ \ b\left(-\frac{q_{0}^{2}}{z}\right)=\frac{\widetilde{q}^{*}}{\widetilde{q}}\cdot \overline{b}(z).
\end{equation}

\subsection{Asymptotic behavior of eigenfunctions and scattering data}
In order to solve the inverse problem, one has to determine the asymptotic behavior of eigenfunctions and scattering data both as $z\rightarrow\infty$ in $\mathbb{K}_{1}$ and as $z\rightarrow 0$ in $\mathbb{K}_{2}$. We have
\begin{equation}
M(x,z)\sim\left\{\begin{array}{ll}
\left(\begin{array}{cc}
z\\
-iq^{*}(-x)
\end{array}\right), \ \ \ z\rightarrow\infty\\
\left(\begin{array}{cc}
z\cdot\frac{q(x)}{\widetilde{q}}\\
-i\widetilde{q}^{*}
\end{array}\right), \ \ \ z\rightarrow 0,\\
\end{array}\right.
\end{equation}

\begin{equation}
N(x,z)\sim\left\{\begin{array}{ll}
\left(\begin{array}{cc}
-iq(x)\\
z
\end{array}\right), \ \ \ z\rightarrow\infty\\
\left(\begin{array}{cc}
-i\widetilde{q}\\
z\cdot \frac{q^{*}(-x)}{\widetilde{q}^{*}}
\end{array}\right), \ \ \ z\rightarrow 0,\\
\end{array}\right.
\end{equation}

\begin{equation}
\overline{M}(x,z)\sim\left\{\begin{array}{ll}
\left(\begin{array}{cc}
-iq(x)\\
z
\end{array}\right), \ \ \ z\rightarrow\infty\\
\left(\begin{array}{cc}
-i\widetilde{q}\\
z\cdot \frac{q^{*}(-x)}{\widetilde{q}^{*}}
\end{array}\right), \ \ \ z\rightarrow 0,\\
\end{array}\right.
\end{equation}

\begin{equation}
\label{E:asymptotic 3'}
\overline{N}(x,z)\sim\left\{\begin{array}{ll}
\left(\begin{array}{cc}
z\\
-iq^{*}(-x)
\end{array}\right), \ \ \ z\rightarrow\infty\\
\left(\begin{array}{cc}
z\cdot\frac{q(x)}{\widetilde{q}}\\
-i\widetilde{q}^{*}
\end{array}\right), \ \ \ z\rightarrow 0,\\
\end{array}\right.
\end{equation}

\begin{equation}
a(z)=
\left\{\begin{array}{ll}
1,\ \ \ z\rightarrow\infty,\\
1, \ \ \ z\rightarrow 0,\\
\end{array}\right.
\end{equation}

\begin{equation}
\overline{a}(z)=
\left\{\begin{array}{ll}
1,\ \ \ z\rightarrow\infty,\\
1, \ \ \ z\rightarrow 0,\\
\end{array}\right.
\end{equation}
\begin{equation}
\lim_{z\rightarrow\infty}zb(z)=0, \ \ \ \lim_{z\rightarrow0}\frac{b(z)}{z^{2}}=0.
\end{equation}

\subsection{Riemann-Hilbert problem via uniformization coordinates}
\subsubsection{Left scattering problem}
In order to take into account the behavior of the eigenfunctions,  the `jump' conditions at $\Sigma$, where $\Sigma:=(-\infty,-q_{0})\cup (q_{0}, +\infty)\cup \overrightarrow{(q_{0}, -q_{0})}\cup \{q_{0}e^{i\theta}, \pi\leq \theta\leq 2\pi\}_{clockwise, upper \ circle}\cup \{q_{0}e^{i\theta}, -\pi\leq \theta\leq 0\}_{anticlockwise, lower \ circle}$, can be written as
\begin{equation}
\frac{M(x,z)}{za(z)}-\frac{\overline{N}(x,z)}{z}=\rho(z)e^{i\big(z+\frac{q_{0}^{2}}{z}\big)x}\cdot \frac{N(x,z)}{z}
\end{equation}
and
\begin{equation}
\frac{\overline{M}(x,z)}{z\overline{a}(z)}-\frac{N(x,z)}{z}=\overline{\rho}(z)e^{-i\big(z+\frac{q_{0}^{2}}{z}\big)x}\cdot \frac{\overline{N}(x,z)}{z},
\end{equation}
so that the functions will be bounded at infinity, though having an additional pole at $z=0$. Note that $M(x,z)/a(z)$, as a function of
$z$, is defined in $D^{+}$, where it has simple poles $z_{j}$, i.e., $a(z_{j})=0$, and $\overline{M}(x,z)/\overline{a}(z)$,
is defined in $D^{-}$, where it has simple poles $\overline{z}_{j}$, i.e., $\overline{a}(\overline{z}_{j})=0$. It follows that
\begin{equation}
M(x,z_{j})=b(z_{j})e^{i\big(z_{j}+\frac{q_{0}^{2}}{z_{j}}\big)x}\cdot N(x,z_{j})
\end{equation}
and
\begin{equation}
\overline{M}(x,\overline{z}_{j})=\overline{b}(\overline{z}_{j})e^{-i\big(\overline{z}_{j}+\frac{q_{0}^{2}}{\overline{z}_{j}}\big)x}\cdot \overline{N}(x,\overline{z}_{j}).
\end{equation}
Then subtracting the values at infinity, the induced pole at the origin and the poles, assumed simple in $D^{+}$/$D^{-}$ respectively, at $a(z_j)=0, j=1,2...J$ and  $\bar{a}(\overline{z}_{j}), j=1,2...\bar{J} $ gives
\begin{equation}
\label{E:jump 1'}
\begin{split}
&\left[\frac{M(x,z)}{za(z)}-\left(\begin{array}{cc}
1\\
0
\end{array}\right)-
\frac{1}{z}\left(\begin{array}{cc}
0\\
-i\widetilde{q}^{*}
\end{array}\right)
-\sum_{j=1}^{J}\frac{M(x,z_{j})}{(z-z_{j})z_{j}a'(z_{j})}\right]\\
&-\left[\frac{\overline{N}(x,z)}{z}-\left(\begin{array}{cc}
1\\
0
\end{array}\right)-
\frac{1}{z}\left(\begin{array}{cc}
0\\
-i\widetilde{q}^{*}
\end{array}\right)
-\sum_{j=1}^{J}\frac{b(z_{j})e^{i\big(z_{j}+\frac{q_{0}^{2}}{z_{j}}\big)x}\cdot N(x,z_{j})}{(z-z_{j})z_{j}a'(z_{j})}\right]\\
&=\rho(z)e^{i\big(z+\frac{q_{0}^{2}}{z}\big)x}\cdot \frac{N(x,z)}{z}
\end{split}
\end{equation}
and
\begin{equation}
\label{E:jump 2'}
\begin{split}
&\left[\frac{\overline{M}(x,z)}{z\overline{a}(z)}-\left(\begin{array}{cc}
0\\
1
\end{array}\right)-
\frac{1}{z}\left(\begin{array}{cc}
-i\widetilde{q}\\
0
\end{array}\right)
-\sum_{j=1}^{\overline{J}}\frac{\overline{M}(x,\overline{z}_{j})}{(z-\overline{z}_{j})\overline{z}_{j}a'(\overline{z}_{j})}\right]\\
&-\left[\frac{N(x,z)}{z}-\left(\begin{array}{cc}
0\\
1
\end{array}\right)-
\frac{1}{z}\left(\begin{array}{cc}
-i\widetilde{q}\\
0
\end{array}\right)
-\sum_{j=1}^{\overline{J}}\frac{\overline{b}(\overline{z}_{j})e^{-i\big(\overline{z}_{j}+\frac{q_{0}^{2}}{\overline{z}_{j}}\big)x}\cdot \overline{N}(x,\overline{z}_{j})}{(z-\overline{z}_{j})\overline{z}_{j}\overline{a}'(\overline{z}_{j})}\right]\\
&=\overline{\rho}(z)e^{-i\big(z+\frac{q_{0}^{2}}{z}\big)x}\cdot \frac{\overline{N}(x,z)}{z}.
\end{split}
\end{equation}
We now introduce the projection operators
\begin{equation}
P_{\pm}(f)(z)=\frac{1}{2\pi i}\int_{\Sigma}\frac{f(\xi)}{\xi-(z\pm i0)}d\xi,
\end{equation}
where $z$ lies in the $\pm$ regions and 
$\Sigma:=(-\infty,-q_{0})\cup (q_{0}, +\infty)\cup \overrightarrow{(q_{0}, -q_{0})}\cup \{q_{0}e^{i\theta}, \pi\leq \theta\leq 2\pi\}_{clockwise, upper \ circle}\cup \{q_{0}e^{i\theta}, -\pi\leq \theta\leq 0\}_{anticlockwise, lower \ circle}$ indicated by the curve in Figure \ref{fig5} (red) with (black) arrows.

%
%
%
%
%

If $f_{\pm}(\xi)$ is analytic in $D^{\pm}$ and $f_{\pm}(\xi)$ is decaying at large $\xi$, then
\begin{equation}
P_{\pm}(f_{\pm})(z)=\pm f_{\pm}(z), \ \ \ P_{\mp}(f_{\pm})(z)=0.
\end{equation}
Applying $P_{-}$ to (\ref{E:jump 1'}) and $P_{+}$ to (\ref{E:jump 2'}), we can obtain
\begin{equation}
\begin{split}
\frac{\overline{N}(x,z)}{z}&=\left(\begin{array}{cc}
1\\
0
\end{array}\right)+
\frac{1}{z}\left(\begin{array}{cc}
0\\
-i\widetilde{q}^{*}
\end{array}\right)
+\sum_{j=1}^{J}\frac{b(z_{j})e^{i\big(z_{j}+\frac{q_{0}^{2}}{z_{j}}\big)x}\cdot N(x,z_{j})}{(z-z_{j})z_{j}a'(z_{j})}\\
&+\frac{1}{2\pi i}\int_{\Sigma}\frac{\rho(\xi)}{\xi(\xi-z)}\cdot e^{i\big(\xi+\frac{q_{0}^{2}}{\xi}\big)x}\cdot N(x,\xi)d\xi
\end{split}
\end{equation}
and
\begin{equation}
\begin{split}
\frac{N(x,z)}{z}&=\left(\begin{array}{cc}
0\\
1
\end{array}\right)+
\frac{1}{z}\left(\begin{array}{cc}
-i\widetilde{q}\\
0
\end{array}\right)
+\sum_{j=1}^{\overline{J}}\frac{\overline{b}(\overline{z}_{j})e^{-i\big(\overline{z}_{j}+\frac{q_{0}^{2}}{\overline{z}_{j}}\big)x}\cdot \overline{N}(x,\overline{z}_{j})}{(z-\overline{z}_{j})\overline{z}_{j}\overline{a}'(\overline{z}_{j})}\\
&-\frac{1}{2\pi i}\int_{\Sigma}\frac{\overline{\rho}(\xi)}{\xi(\xi-z)}\cdot e^{-i\big(\xi+\frac{q_{0}^{2}}{\xi}\big)x}\cdot \overline{N}(x,\xi)d\xi,
\end{split}
\end{equation}
i.e.,
\begin{equation}
\label{E:eigenfunction 1'}
\begin{split}
\overline{N}(x,z)&=\left(\begin{array}{cc}
z\\
-i\widetilde{q}^{*}
\end{array}\right)
+\sum_{j=1}^{J}\frac{z\cdot b(z_{j})e^{i\big(z_{j}+\frac{q_{0}^{2}}{z_{j}}\big)x}\cdot N(x,z_{j})}{(z-z_{j})z_{j}a'(z_{j})}\\
&+\frac{z}{2\pi i}\int_{\Sigma}\frac{\rho(\xi)}{\xi(\xi-z)}\cdot e^{i\big(\xi+\frac{q_{0}^{2}}{\xi}\big)x}\cdot N(x,\xi)d\xi
\end{split}
\end{equation}
and
\begin{equation}
\label{E:eigenfunction 2'}
\begin{split}
N(x,z)&=\left(\begin{array}{cc}
-i\widetilde{q}\\
z
\end{array}\right)
+\sum_{j=1}^{\overline{J}}\frac{z\cdot\overline{b}(\overline{z}_{j})e^{-i\big(\overline{z}_{j}+\frac{q_{0}^{2}}{\overline{z}_{j}}\big)x}\cdot \overline{N}(x,\overline{z}_{j})}{(z-\overline{z}_{j})\overline{z}_{j}\overline{a}'(\overline{z}_{j})}\\
&-\frac{z}{2\pi i}\int_{\Sigma}\frac{\overline{\rho}(\xi)}{\xi(\xi-z)}\cdot e^{-i\big(\xi+\frac{q_{0}^{2}}{\xi}\big)x}\cdot \overline{N}(x,\xi)d\xi.
\end{split}
\end{equation}
\subsubsection{Right scattering problem}
The right scattering problem can be written as
\begin{equation}
\psi(x,z)=\alpha(z)\overline{\phi}(x,z)+\beta(z)\phi(x,z)
\end{equation}
and
\begin{equation}
\overline{\psi}(x,z)=\overline{\alpha}(z)\phi(x,z)+\overline{\beta}(z)\overline{\phi}(x,z),
\end{equation}
where $\alpha(z)$, $\overline{\alpha}(z)$, $\beta(z)$ and $\overline{\beta}(z)$ are the right scattering data. Moreover, we can get the right scattering data and left scattering data satisfy the following relations
\begin{equation}
\overline{\alpha}(z)=\overline{a}(z), \ \ \ \alpha(z)=a(z), \ \ \ \overline{\beta}(z)=-b(z), \ \ \ \beta(z)=-\overline{b}(z).
\end{equation}
Thus,
\begin{equation}
\begin{split}
N(x,z)&=\alpha(z)\overline{M}(x,z)+\beta(z)M(x,z)e^{-i\big(z+\frac{q_{0}^{2}}{z}\big)x}\\
&=a(z)\overline{M}(x,z)-\overline{b}(z)M(x,z)e^{-i\big(z+\frac{q_{0}^{2}}{z}\big)x}
\end{split}
\end{equation}
and
\begin{equation}
\begin{split}
\overline{N}(x,z)&=\overline{\alpha}(z)M(x,z)+\overline{\beta}(z)\overline{M}(x,z)e^{i\big(z+\frac{q_{0}^{2}}{z}\big)x}\\
&=\overline{a}(z)M(x,z)-b(z)\overline{M}(x,z)e^{i\big(z+\frac{q_{0}^{2}}{z}\big)x}.
\end{split}
\end{equation}
The above two equations can be written as
\begin{equation}
\frac{N(x,z)}{za(z)}-\frac{\overline{M}(x,z)}{z}=-\frac{\overline{b}(z)}{a(z)}\cdot e^{-i\big(z+\frac{q_{0}^{2}}{z}\big)x}\cdot \frac{M(x,z)}{z}
\end{equation}
and
\begin{equation}
\frac{\overline{N}(x,z)}{z\overline{a}(z)}-\frac{M(x,z)}{z}=-\frac{b(z)}{\overline{a}(z)}\cdot e^{i\big(z+\frac{q_{0}^{2}}{z}\big)x}\cdot \frac{\overline{M}(x,z)}{z}.
\end{equation}
By the symmetry relations of scattering data, we have
\begin{equation}
\frac{N(x,z)}{za(z)}-\frac{\overline{M}(x,z)}{z}=\rho^{*}(-z^{*})\cdot e^{-i\big(z+\frac{q_{0}^{2}}{z}\big)x}\cdot \frac{M(x,z)}{z}
\end{equation}
and
\begin{equation}
\frac{\overline{N}(x,z)}{z\overline{a}(z)}-\frac{M(x,z)}{z}=\overline{\rho}^{*}(-z^{*})\cdot e^{i\big(z+\frac{q_{0}^{2}}{z}\big)x}\cdot \frac{\overline{M}(x,z)}{z},
\end{equation}
so that the functions will be bounded at infinity, though having an additional pole at $z=0$. Note that $N(x,z)/a(z)$, as a function of
$z$, is defined in $D^{+}$, where it has simple poles $z_{j}$, i.e., $a(z_{j})=0$, and $\overline{N}(x,z)/\overline{a}(z)$,
is defined in $D^{-}$, where it has simple poles $\overline{z}_{j}$, i.e., $\overline{a}(\overline{z}_{j})=0$. It follows that
\begin{equation}
N(x,z_{j})=-\overline{b}(z_{j})M(x,z_{j})e^{-i\big(z_{j}+\frac{q_{0}^{2}}{z_{j}}\big)x}
\end{equation}
and
\begin{equation}
\overline{N}(x,\overline{z}_{j})=-b(\overline{z}_{j})\overline{M}(x,\overline{z}_{j})e^{i\big(\overline{z}_{j}+\frac{q_{0}^{2}}{\overline{z}_{j}}\big)x}.
\end{equation}
Then, as before, subtracting the values at infinity, the induced pole at the origin and the poles, assumed simple in $D^{+}$/ $D^{-}$ respectively, at $a(z_j)=0, j=1,2...J$ and  $\bar{a}(\overline{z}_{j}), j=1,2...\bar{J} $ gives
\begin{equation}
\label{E:jump 3'}
\begin{split}
&\left[\frac{N(x,z)}{za(z)}-\left(\begin{array}{cc}
0\\
1
\end{array}\right)
-\frac{1}{z}\left(\begin{array}{cc}
-i\widetilde{q}\\
0
\end{array}\right)
-\sum_{j=1}^{J}\frac{N(x,z_{j})}{(z-z_{j})z_{j}a'(z_{j})}\right]\\
&-\left[\frac{\overline{M}(x,z)}{z}-\left(\begin{array}{cc}
0\\
1
\end{array}\right)
-\frac{1}{z}\left(\begin{array}{cc}
-i\widetilde{q}\\
0
\end{array}\right)-\sum_{j=1}^{J}\frac{-\overline{b}(z_{j})M(x,z_{j})e^{-i\big(z_{j}+\frac{q_{0}^{2}}{z_{j}}\big)x}}{(z-z_{j})z_{j}a'(z_{j})}\right]\\
&=\rho^{*}(-z^{*})\cdot e^{-i\big(z+\frac{q_{0}^{2}}{z}\big)x}\cdot \frac{M(x,z)}{z}
\end{split}
\end{equation}
and
\begin{equation}
\label{E:jump 4'}
\begin{split}
&\left[\frac{\overline{N}(x,z)}{z\overline{a}(z)}-\left(\begin{array}{cc}
1\\
0
\end{array}\right)
-\frac{1}{z}\left(\begin{array}{cc}
0\\
-i\widetilde{q}^{*}
\end{array}\right)
-\sum_{j=1}^{\overline{J}}\frac{\overline{N}(x,\overline{z}_{j})}{(z-\overline{z}_{j})\overline{z}_{j}\overline{a}'(\overline{z}_{j})}\right]\\
&-\left[\frac{M(x,z)}{z}-\left(\begin{array}{cc}
1\\
0
\end{array}\right)
-\frac{1}{z}\left(\begin{array}{cc}
0\\
-i\widetilde{q}^{*}
\end{array}\right)-\sum_{j=1}^{\overline{J}}\frac{-b(\overline{z}_{j})\overline{M}(x,\overline{z}_{j})e^{i\big(\overline{z}_{j}+\frac{q_{0}^{2}}{\overline{z}_{j}}\big)x}}
{(z-\overline{z}_{j})\overline{z}_{j}\overline{a}'(\overline{z}_{j})}\right]\\
&=\overline{\rho}^{*}(-z^{*})\cdot e^{i\big(z+\frac{q_{0}^{2}}{z}\big)x}\cdot \frac{\overline{M}(x,z)}{z}.
\end{split}
\end{equation}
Applying $P_{-}$ to (\ref{E:jump 3'}) and $P_{+}$ to (\ref{E:jump 4'}), we can obtain
\begin{equation}
\label{E:eigenfunction 3'}
\begin{split}
\overline{M}(x,z)&=\left(\begin{array}{cc}
-i\widetilde{q}\\
z
\end{array}\right)+
\sum_{j=1}^{J}\frac{-z\cdot\overline{b}(z_{j})M(x,z_{j})e^{-i\big(z_{j}+\frac{q_{0}^{2}}{z_{j}}\big)x}}{(z-z_{j})z_{j}a'(z_{j})}\\
&+\frac{z}{2\pi i}\int_{\Sigma}\frac{\rho^{*}(-\xi^{*})}{\xi(\xi-z)}\cdot e^{-i\big(\xi+\frac{q_{0}^{2}}{\xi}\big)x}\cdot M(x,\xi)d\xi
\end{split}
\end{equation}
and
\begin{equation}
\label{E:eigenfunction 4'}
\begin{split}
M(x,z)&=\left(\begin{array}{cc}
z\\
-i\widetilde{q}^{*}
\end{array}\right)+
\sum_{j=1}^{\overline{J}}\frac{-z\cdot b(\overline{z}_{j})\overline{M}(x,\overline{z}_{j})e^{i\big(\overline{z}_{j}+\frac{q_{0}^{2}}{\overline{z}_{j}}\big)x}}
{(z-\overline{z}_{j})\overline{z}_{j}\overline{a}'(\overline{z}_{j})}\\
&-\frac{z}{2\pi i}\int_{\Sigma}\frac{\overline{\rho}^{*}(-\xi^{*})}{\xi(\xi-z)}\cdot e^{i\big(\xi+\frac{q_{0}^{2}}{\xi}\big)x}\cdot \overline{M}(x,\xi)d\xi.
\end{split}
\end{equation}
\subsection{Recovery of the potentials}
Note that $N_{1}(x,z)\sim -iq(x)$ as $z\rightarrow\infty$, and
\begin{equation}
N_{1}(x,z)\sim -i\widetilde{q}+\sum_{j=1}^{\overline{J}}\frac{\overline{b}(\overline{z}_{j})e^{-i\big(\overline{z}_{j}+\frac{q_{0}^{2}}{\overline{z}_{j}}\big)x}\cdot \overline{N}_{1}(x,\overline{z}_{j})}{\overline{z}_{j}\overline{a}'(\overline{z}_{j})}
+\frac{1}{2\pi i}\int_{\Sigma}\frac{\overline{\rho}(\xi)}{\xi}\cdot e^{-i\big(\xi+\frac{q_{0}^{2}}{\xi}\big)x}\cdot \overline{N}_{1}(x,\xi)d\xi,
\end{equation}
we have
\begin{equation}
\label{asympN1c'}
q(x)=\widetilde{q}+i\sum_{j=1}^{\overline{J}}\frac{\overline{b}(\overline{z}_{j})e^{-i\big(\overline{z}_{j}+\frac{q_{0}^{2}}{\overline{z}_{j}}\big)x}\cdot \overline{N}_{1}(x,\overline{z}_{j})}{\overline{z}_{j}\overline{a}'(\overline{z}_{j})}+\frac{1}{2\pi}\int_{\Sigma}\frac{\overline{\rho}(\xi)}{\xi}\cdot e^{-i\big(\xi+\frac{q_{0}^{2}}{\xi}\big)x}\cdot \overline{N}_{1}(x,\xi)d\xi.
\end{equation}

\subsection{Closing the system}

%

Similarly, we can get $J=\overline{J}$ from $a\left(-\frac{q_{0}^{2}}{z}\right)=\overline{a}(z)$. To close the system,
by the symmetry relations between the eigenfunctions, we have
\begin{equation}
\label{E:closing system 3}
\begin{split}
&\left(\begin{array}{cc}
N_{1}(x,z)\\
N_{2}(x,z)
\end{array}\right)=\left(\begin{array}{cc}
-i\widetilde{q}\\
z
\end{array}\right)
+\sum_{j=1}^{J}\frac{z\cdot\overline{b}(\overline{z}_{j})e^{-i\big(\overline{z}_{j}+\frac{q_{0}^{2}}{\overline{z}_{j}}\big)x} }{(z-\overline{z}_{j})\overline{z}_{j}\overline{a}'(\overline{z}_{j})}\cdot\\
&\left(\begin{array}{cc}
\overline{z}_{j}+\sum_{l=1}^{J}\frac{\overline{z}_{j}\cdot b(z_{l})e^{i\big(z_{l}+\frac{q_{0}^{2}}{z_{l}}\big)x}}{(\overline{z}_{j}-z_{l})z_{l}a'(z_{l})}\cdot N_{1}(x, z_{l})+\frac{\overline{z}_{j}}{2\pi i}\int_{\Sigma}\frac{\rho(\xi)}{\xi(\xi-\overline{z}_{j})}\cdot e^{i\big(\xi+\frac{q_{0}^{2}}{\xi}\big)x}\cdot N_{1}(x, \xi)d\xi\\
-i\widetilde{q}^{*}+\sum_{l=1}^{J}\frac{\overline{z}_{j}\cdot b(z_{l})e^{i\big(z_{l}+\frac{q_{0}^{2}}{z_{l}}\big)x}}{(\overline{z}_{j}-z_{l})z_{l}a'(z_{l})}\cdot N_{2}(x, z_{l})+\frac{\overline{z}_{j}}{2\pi i}\int_{\Sigma}\frac{\rho(\xi)}{\xi(\xi-\overline{z}_{j})}\cdot e^{i\big(\xi+\frac{q_{0}^{2}}{\xi}\big)x}\cdot N_{2}(x, \xi)d\xi
\end{array}\right)\\
&-\frac{z}{2\pi i}\int_{\Sigma}\frac{\overline{\rho}(\xi)}{\xi(\xi-z)}\cdot e^{-i\big(\xi+\frac{q_{0}^{2}}{\xi}\big)x}\cdot\\
&\left(\begin{array}{cc}
\xi+\sum_{l=1}^{J}\frac{\xi\cdot b(z_{l})e^{i\big(z_{l}-\frac{q_{0}^{2}}{z_{l}}\big)x}}{(\xi-z_{l})z_{l}a'(z_{l})}\cdot N_{1}(x, z_{l})+\frac{\xi}{2\pi i}\int_{\Sigma}\frac{\rho(\eta)}{\eta(\eta-\xi)}\cdot e^{i\big(\eta+\frac{q_{0}^{2}}{\eta}\big)x}\cdot N_{1}(x, \eta)d\eta\\
-i\widetilde{q}^{*}+\sum_{l=1}^{J}\frac{\xi\cdot b(z_{l})e^{i\big(z_{l}+\frac{q_{0}^{2}}{z_{l}}\big)x}}{(\xi-z_{l})z_{l}a'(z_{l})}\cdot N_{2}(x, z_{l})+\frac{\xi}{2\pi i}\int_{\Sigma}\frac{\rho(\eta)}{\eta(\eta-\xi)}\cdot e^{i\big(\eta+\frac{q_{0}^{2}}{\eta}\big)x}\cdot N_{2}(x, \eta)d\eta
\end{array}\right)d\xi,
\end{split}
\end{equation}
\begin{equation}
\label{E:closing system 4}
\begin{split}
&\left(\begin{array}{cc}
\overline{M}_{1}(x,z)\\
\overline{M}_{2}(x,z)
\end{array}\right)=\left(\begin{array}{cc}
-i\widetilde{q}\\
z
\end{array}\right)+
\sum_{j=1}^{J}\frac{-z\cdot\overline{b}(z_{j})e^{-i\big(z_{j}+\frac{q_{0}^{2}}{z_{j}}\big)x}}{(z-z_{j})z_{j}a'(z_{j})}\cdot\\
& \left(\begin{array}{cc}
z_{j}+
\sum_{l=1}^{J}\frac{-z_{j}\cdot b(\overline{z}_{l})e^{i\big(\overline{z}_{l}+\frac{q_{0}^{2}}{\overline{z}_{l}}\big)x}}
{(z_{j}-\overline{z}_{l})\overline{z}_{l}\overline{a}'(\overline{z}_{l})}\cdot \overline{M}_{1}(x, \overline{z}_{l})-\frac{z_{j}}{2\pi i}\int_{\Sigma}\frac{\overline{\rho}^{*}(-\xi^{*})}{\xi(\xi-z_{j})}\cdot e^{i\big(\xi+\frac{q_{0}^{2}}{\xi}\big)x}\cdot\overline{M}_{1}(x,\xi)d\xi\\
-i\widetilde{q}^{*}+
\sum_{l=1}^{J}\frac{-z_{j}\cdot b(\overline{z}_{l})e^{i\big(\overline{z}_{l}+\frac{q_{0}^{2}}{\overline{z}_{l}}\big)x}}
{(z_{j}-\overline{z}_{l})\overline{z}_{l}\overline{a}'(\overline{z}_{l})}\cdot \overline{M}_{2}(x, \overline{z}_{l})-\frac{z_{j}}{2\pi i}\int_{\Sigma}\frac{\overline{\rho}^{*}(-\xi^{*})}{\xi(\xi-z_{j})}\cdot e^{i\big(\xi+\frac{q_{0}^{2}}{\xi}\big)x}\cdot\overline{M}_{2}(x,\xi)d\xi
\end{array}\right)\\
&+\frac{z}{2\pi i}\int_{\Sigma}\frac{\rho^{*}(-\xi^{*})}{\xi(\xi-z)}\cdot e^{-i\big(\xi+\frac{q_{0}^{2}}{\xi}\big)x}\cdot\\
&\left(\begin{array}{cc}
\xi+
\sum_{l=1}^{J}\frac{-\xi\cdot b(\overline{z}_{l})e^{i\big(\overline{z}_{l}+\frac{q_{0}^{2}}{\overline{z}_{l}}\big)x}}
{(\xi-\overline{z}_{l})\overline{z}_{l}\overline{a}'(\overline{z}_{l})}\cdot\overline{M}_{1}(x,\overline{z}_{l})-\frac{\xi}{2\pi i}\int_{\Sigma}\frac{\overline{\rho}^{*}(-\eta^{*})}{\eta(\eta-\xi)}\cdot e^{i\big(\eta+\frac{q_{0}^{2}}{\eta}\big)x}\cdot \overline{M}_{1}(x,\eta)d\eta\\
-i\widetilde{q}^{*}+
\sum_{l=1}^{J}\frac{-\xi\cdot b(\overline{z}_{l})e^{i\big(\overline{z}_{l}+\frac{q_{0}^{2}}{\overline{z}_{l}}\big)x}}
{(\xi-\overline{z}_{l})\overline{z}_{l}\overline{a}'(\overline{z}_{l})}\cdot\overline{M}_{2}(x,\overline{z}_{l})-\frac{\xi}{2\pi i}\int_{\Sigma}\frac{\overline{\rho}^{*}(-\eta^{*})}{\eta(\eta-\xi)}\cdot e^{i\big(\eta+\frac{q_{0}^{2}}{\eta}\big)x}\cdot \overline{M}_{2}(x,\eta)d\eta
\end{array}\right)d\xi.
\end{split}
\end{equation}

We note that from eq (\ref{asympN1c'})  $q(x)$  is given in only terms of the component $\bar{N}_1$. We can use only the second component of eq (\ref{E:closing system 4}) to find $\bar{M}_2$ which via the symmetry relation (\ref{NbMb2})  yields this function and hence $q(x)$. Hence to complete the inverse scattering we reduce the problem to solving an integral equation in terms of the component $\bar{M}_2$ only.

\subsection{Trace formula}
In a similar manner to the prior section, we can get the Trace formula as follows.
\begin{equation}
\log a(z)=\log\left(\prod_{j=1}^{J_{1}}\frac{z-z_{j}}{z+\frac{q_{0}^{2}}{z_{j}}}\cdot\frac{z+z_{j}^{*}}{z-\frac{q_{0}^{2}}{z_{j}^{*}}}
\cdot\prod_{i=1}^{J_{2}}\frac{z-\widetilde{z}_{i}}{z+\frac{q_{0}^{2}}{\widetilde{z}_{i}}}\right)+\frac{1}{2\pi i}\int_{\Sigma}\frac{\log (1-b(\xi)b^{*}(-\xi^{*}))}{\xi-z}d\xi, \ \ \ z\in D^{+}
\end{equation}
and
\begin{equation}
\log \overline{a}(z)=\log\left(\prod_{j=1}^{J_{1}} \frac{z+\frac{q_{0}^{2}}{z_{j}}}{z-z_{j}}\cdot \frac{z-\frac{q_{0}^{2}}{z_{j}^{*}}}{z+z_{j}^{*}}
\cdot\prod_{i=1}^{J_{2}} \frac{z+\frac{q_{0}^{2}}{\widetilde{z}_{i}}}{z-\widetilde{z}_{i}}\right)-\frac{1}{2\pi i}\int_{\Sigma}\frac{\log (1-b(\xi)b^{*}(-\xi^{*}))}{\xi-z}d\xi, \ \ \ z\in D^{-},
\end{equation}
where $\Re\widetilde{z}_{i}=0$, $\Re z_{j}\neq 0$ and $2J_{1}+J_{2}=J$.

\subsection{Reflectioness potentials and soliton solutions}
Reflectioness potentials and soliton solutions correspond to zero reflection coefficients, i.e., $\rho(\xi)=0$ and $\overline{\rho}(\xi)=0$ for all 
$\xi\in \Sigma$.
 We also note,   from the symmetry relation $\bar{b}(z)=-b^*(-z^*)$ it follows that the reflection coefficients $\rho(z)= b(z)/a(z), \bar{\rho}(z)= \bar{b}(z)/\bar{a}(z)$ will both vanish when $b(z)=0$ for $z$ on the real axis. By substituting $z=z_{l}$ in (\ref{E:closing system 3}) and $z=\overline{z}_{l}$ in (\ref{E:closing system 4}), the system (\ref{E:closing system 3}) and (\ref{E:closing system 4}) reduces to an algebraic systems of equations that  determine the functional form of these special potentials. When time dependence is added the reflectionless potentials correspond to soliton solutions.
The reduced equations take the form
\begin{equation}
\begin{split}
\left(\begin{array}{cc}
N_{1}(x,z_{l})\\
N_{2}(x,z_{l})
\end{array}\right)&=\left(\begin{array}{cc}
-i\widetilde{q}\\
z_{l}
\end{array}\right)
+\sum_{j=1}^{J}\frac{z_{l}\cdot\overline{b}(\overline{z}_{j})e^{-i\big(\overline{z}_{j}+\frac{q_{0}^{2}}{\overline{z}_{j}}\big)x} }{(z_{l}-\overline{z}_{j})\overline{z}_{j}\overline{a}'(\overline{z}_{j})}\cdot\\
&\left(\begin{array}{cc}
\overline{z}_{j}+\sum_{l=1}^{J}\frac{\overline{z}_{j}\cdot b(z_{l})e^{i\big(z_{l}+\frac{q_{0}^{2}}{z_{l}}\big)x}}{(\overline{z}_{j}-z_{l})z_{l}a'(z_{l})}\cdot N_{1}(x, z_{l})\\
-i\widetilde{q}^{*}+\sum_{l=1}^{J}\frac{\overline{z}_{j}\cdot b(z_{l})e^{i\big(z_{l}+\frac{q_{0}^{2}}{z_{l}}\big)x}}{(\overline{z}_{j}-z_{l})z_{l}a'(z_{l})}\cdot N_{2}(x, z_{l})
\end{array}\right),
\end{split}
\end{equation}
\begin{equation}
\begin{split}
\left(\begin{array}{cc}
\overline{M}_{1}(x,\overline{z}_{l})\\
\overline{M}_{2}(x,\overline{z}_{l})
\end{array}\right)&=\left(\begin{array}{cc}
-i\widetilde{q}\\
\overline{z}_{l}
\end{array}\right)+
\sum_{j=1}^{J}\frac{-\overline{z}_{l}\cdot\overline{b}(z_{j})e^{-i\big(z_{j}+\frac{q_{0}^{2}}{z_{j}}\big)x}}{(\overline{z}_{l}-z_{j})z_{j}a'(z_{j})}\cdot\\
& \left(\begin{array}{cc}
z_{j}+
\sum_{l=1}^{J}\frac{-z_{j}\cdot b(\overline{z}_{l})e^{i\big(\overline{z}_{l}+\frac{q_{0}^{2}}{\overline{z}_{l}}\big)x}}
{(z_{j}-\overline{z}_{l})\overline{z}_{l}\overline{a}'(\overline{z}_{l})}\cdot \overline{M}_{1}(x, \overline{z}_{l})\\
-i\widetilde{q}^{*}+
\sum_{l=1}^{J}\frac{-z_{j}\cdot b(\overline{z}_{l})e^{i\big(\overline{z}_{l}+\frac{q_{0}^{2}}{\overline{z}_{l}}\big)x}}
{(z_{j}-\overline{z}_{l})\overline{z}_{l}\overline{a}'(\overline{z}_{l})}\cdot \overline{M}_{2}(x, \overline{z}_{l})
\end{array}\right).
\end{split}
\end{equation}



The above equations are an algebraic system to solve for either $N(x,z_{l})$ or $\overline{M}(x,\overline{z}_{l}), l=1,2...J$. The potential is reconstructed from
equation (\ref{asympN1c'}) with $\rho(\xi)=0,\bar{\rho}(\xi)=0$; i.e.,

\begin{equation}
\label{asympN1d1}
q(x)=\widetilde{q}+i\sum_{j=1}^{J}\frac{\overline{b}(\overline{z}_{j})e^{-i\big(\overline{z}_{j}+\frac{q_{0}^{2}}{\overline{z}_{j}}\big)x}\cdot \overline{N}_{1}(x,\overline{z}_{j})}{\overline{z}_{j}\overline{a}'(\overline{z}_{j})}.
\end{equation}
Since $a(z)\sim 1$ as $z\rightarrow 0$, from the trace formula when $b(\xi)=0$ in $\Sigma$, we have the following constraint for the reflectionless potentials
\begin{equation}
\label{E:product2}
\prod_{j=1}^{J_{1}}\frac{|z_{j}|^{4}}{q_{0}^{4}}\cdot\prod_{i=1}^{J_{2}}\left(\frac{|\widetilde{z}_{i}|^{2}}{q_{0}^{2}}\right)=1,
\end{equation}
where $2J_{1}+J_{2}=J$. We claim that $J\geq 2$. Otherwise, if $J=1$, then $J_{1}=0$ and $J_{2}=1$, which implies that $|\widetilde{z}_{1}|^{2}= q_{0}^{2}$; hence the eigenvalue lies on the circle, which in this case is the continuous spectrum. Such eigenvalues are not proper; they are not considered here.

\subsection{Discrete scattering data and their symmetries}
In order to find reflectionless potentials and solitons we need to be able to calculate the relevant discrete scattering data.

The coefficients
\[ b(z_j) ~\mbox{and} ~~~\bar{b}(\bar{z}_j) ,~~j=1,2,...J.\]
can be calculated in the same way as in Section 3.12


to find

\begin{equation}
b(z_{j})b^{*}(-z_{j}^{*})=1, \bar{b}(\bar{z}_j) \bar{b}^{*}(-\bar{z}_{j}^{*}) =1,
\end{equation}

\begin{equation}
-b(z_{j})\overline{b}(z_{j})=1, \ \ \ -b(\overline{z}_{j})\overline{b}(\overline{z}_{j})=1.
\end{equation}


The functions
\[a'(z_j), ~~\bar{a}'(z_j) \]
 can be calculated via the trace formulae.

\subsection{Reflectionless potential solution: 2-eigenvalue}

In this subsection, we construct  an explicit 2-eigenvalue reflectionless potential by setting $J=2$; (this solution turns into a soliton solution when time dependence is added. The simplest reflectionless potential case  obtains when $J_{1}=0$ and $J_{2}=2$; as mentioned above when $J_{1}=1, J_{2}=0$ leads to eigenvalues on the continuous spectrum.  Note that $|\widetilde{z}_{1}|\cdot|\widetilde{z}_{2}|=q_{0}^{2}$. Now let $\widetilde{z}_{1}=iv_{1}$, where $v_{1}>q_{0}$. Then $\widetilde{z}_{2}=-i\frac{q_{0}^{2}}{v_{1}}$, $a(iv_{1})=a(-i\frac{q_{0}^{2}}{v_{1}})=0$ and $\overline{a}(i \frac{q_{0}^{2}}{v_{1}})=\overline{a}(-i v_{1})=0$. Thus,
\begin{equation}
a(z)=\frac{z-iv_{1}}{z-i\frac{q_{0}^{2}}{v_{1}}}\cdot \frac{z+i\frac{q_{0}^{2}}{v_{1}}}{z+iv_{1}}, \ \ \
\overline{a}(z)=\frac{z-i\frac{q_{0}^{2}}{v_{1}}}{z-iv_{1}}\cdot \frac{z+iv_{1}}{z+i\frac{q_{0}^{2}}{v_{1}}},
\end{equation}
we can get
\begin{equation}
a'(iv_{1})=\frac{i\left(v_{1}+\frac{q_{0}^{2}}{v_{1}}\right)}{2(q_{0}^{2}-v_{1}^{2})}, \ \ \
a'\left(-i\frac{q_{0}^{2}}{v_{1}}\right)=\frac{i\left(\frac{q_{0}^{2}}{v_{1}}+v_{1}\right)}{2q_{0}^{2}\left(\frac{q_{0}^{2}}{v_{1}^{2}}-1\right)},
\end{equation}
\begin{equation}
\overline{a}'(-iv_{1})=\frac{i\left(v_{1}+\frac{q_{0}^{2}}{v_{1}}\right)}{2(v_{1}^{2}-q_{0}^{2})}, \ \ \
\overline{a}'\left(i \frac{q_{0}^{2}}{v_{1}}\right)=\frac{-i\left(v_{1}+\frac{q_{0}^{2}}{v_{1}}\right)}{2q_{0}^{2}\left(\frac{q_{0}^{2}}{v_{1}^{2}}-1\right)}.
\end{equation}
Similarly, from section 4.12,
we can get $|b(iv_{1})|=|\overline{b}(iv_{1})|=|b(-i\frac{q_{0}^{2}}{v_{1}})|=|\overline{b}(-i\frac{q_{0}^{2}}{v_{1}})|
=|b(-iv_{1})|=|\overline{b}(-iv_{1})|=|b(i\frac{q_{0}^{2}}{v_{1}})|=|\overline{b}(i\frac{q_{0}^{2}}{v_{1}})|=1$. So we can write
$b(iv_{1})=e^{i\varphi_{1}}$ and $b(-i\frac{q_{0}^{2}}{v_{1}})=e^{i\varphi_{2}}$.

By the symmetry relations $b^{*}(-z^{*})=-\overline{b}(z)$ and $b\left(-\frac{q_{0}^{2}}{z}\right)=\frac{\widetilde{q}^{*}}{\widetilde{q}} \overline{b}(z)$, we can obtain $\overline{b}(i \frac{q_{0}^{2}}{v_{1}})=e^{i(2\theta+\varphi_{1})}$, $\overline{b}(-iv_{1})=e^{i(2\theta+\varphi_{2})}$,
$\overline{b}(iv_{1})=-e^{-i\varphi_{1}}$, $\overline{b}(-i\frac{q_{0}^{2}}{v_{1}})=-e^{-i\varphi_{2}}$, $b(-iv_{1})=-e^{-i(2\theta+\varphi_{2})}$, $b(i\frac{q_{0}^{2}}{v_{1}})=-e^{-i(2\theta+\varphi_{1})}$.

Hence, we have
\begin{equation}
\overline{M}_{2}^{*}\left(-x,i \frac{q_{0}^{2}}{v_{1}}\right)=
\frac{iq_{0}e^{2v_{1}x}\cdot(q_{0}^{2}+v_{1}^{2})\cdot\left(q_{0}+v_{1}e^{\left(\frac{q_{0}^{2}}{v_{1}}-v_{1}\right)x}e^{i(\theta+\varphi_{2})}\right)}
{v_{1}\left[2q_{0}v_{1}\cdot e^{\left(\frac{q_{0}^{2}}{v_{1}}+v_{1}\right)x}\cdot e^{i\theta}(e^{i\varphi_{1}}-e^{i\varphi_{2}})
-(q_{0}^{2}+v_{1}^{2})\left(e^{2v_{1}x}-e^{\frac{2q_{0}^{2}}{v_{1}}x}\cdot e^{i(2\theta+\varphi_{1}+\varphi_{2})}\right)\right]},
\end{equation}
\begin{equation}
\overline{M}_{2}^{*}\left(-x,-iv_{1}\right)=
\frac{ie^{v_{1}x}\cdot(q_{0}^{2}+v_{1}^{2})\cdot\left(-q_{0}e^{\frac{q_{0}^{2}}{v_{1}}x}e^{i(\theta+\varphi_{1})}+v_{1}e^{v_{1}x}\right)}
{2q_{0}v_{1}\cdot e^{\left(\frac{q_{0}^{2}}{v_{1}}+v_{1}\right)x}\cdot e^{i\theta}(e^{i\varphi_{2}}-e^{i\varphi_{1}})
-(q_{0}^{2}+v_{1}^{2})\left(e^{2v_{1}x}-e^{\frac{2q_{0}^{2}}{v_{1}}x}\cdot e^{i(2\theta+\varphi_{1}+\varphi_{2})}\right)}.
\end{equation}

By (\ref{asympN1d1}), we have
\begin{equation}
\begin{split}
q(x)&= \widetilde{q}-i\frac{2\left(\frac{q_{0}^{2}}{v_{1}^{2}}-1\right)e^{i(2\theta+\varphi_{1})} e^{\left(\frac{q_{0}^{2}}{v_{1}}-v_{1}\right)x}}{1+\frac{q_{0}^{2}}{v_{1}^{2}}}\overline{M}_{2}^{*}\left(-x, i \frac{q_{0}^{2}}{v_{1}}\right)\\
&-i\frac{2\left(v_{1}-\frac{q_{0}^{2}}{v_{1}}\right) e^{i(2\theta+\varphi_{2})} e^{\left(\frac{q_{0}^{2}}{v_{1}}-v_{1}\right)x}}{\frac{q_{0}^{2}}{v_{1}}+v_{1}}\overline{M}_{2}^{*}\left(-x, -iv_{1}\right).
\end{split}
\end{equation}
\subsection{Time evolution}
As in the prior section we find,
\begin{equation}
\frac{\partial a(t)}{\partial t}=0, \ \ \ \frac{\partial \overline{a}(t)}{\partial t}=0, \ \ \ \frac{\partial b(t)}{\partial t}=2i(q_{0}^{2}-2\lambda k)b(t).
\end{equation}
Hence, $a(t)$ and $\overline{a}(t)$ are time independent. Moreover,
\begin{equation}
b(iv_{1},t)=e^{i\varphi_{1}}\cdot e^{2i\left[q_{0}^{2}+\frac{1}{2}\left(v_{1}^{2}-\frac{q_{0}^{4}}{v_{1}^{2}}\right)\right]t},
\end{equation}
\begin{equation}
b\left(-i \frac{q_{0}^{2}}{v_{1}}, t\right)=e^{i\varphi_{2}}\cdot e^{2i\left[q_{0}^{2}-\frac{1}{2}\left(v_{1}^{2}-\frac{q_{0}^{4}}{v_{1}^{2}}\right)\right]t}.
\end{equation}
Putting all the above into the formula we had for the reflectionless potential, we obtain the following 2-soliton solution, which is a oscillating in time or `breathing' soliton solution
\begin{equation}
\begin{split}
&q(x,t)=q_{0}e^{-2iq_{0}^{2}t+i\theta}-i\frac{2\left(\frac{q_{0}^{2}}{v_{1}^{2}}-1\right)\cdot e^{i(2\theta+\varphi_{1})} \cdot e^{-2iq_{0}^{2}t+i\left(v_{1}^{2}-\frac{q_{0}^{4}}{v_{1}^{2}}\right)t}\cdot e^{\left(\frac{q_{0}^{2}}{v_{1}}-v_{1}\right)x}}{1+\frac{q_{0}^{2}}{v_{1}^{2}}}\cdot\\
&\left(\frac{iq_{0}e^{2v_{1}x}(q_{0}^{2}+v_{1}^{2})\left(q_{0}+v_{1}e^{\left(\frac{q_{0}^{2}}{v_{1}}-v_{1}\right)x}e^{i(\theta+\varphi_{2})}e^{-i\left(v_{1}^{2}-\frac{q_{0}^{4}}{v_{1}^{2}}\right)t}\right)}
{v_{1}\left[2q_{0}v_{1} e^{\left(\frac{q_{0}^{2}}{v_{1}}+v_{1}\right)x} e^{i\theta}\left(e^{i\varphi_{1}}e^{i\left(v_{1}^{2}-\frac{q_{0}^{4}}{v_{1}^{2}}\right)t}-e^{i\varphi_{2}}e^{-i\left(v_{1}^{2}-\frac{q_{0}^{4}}{v_{1}^{2}}\right)t}\right)
-(q_{0}^{2}+v_{1}^{2})\left(e^{2v_{1}x}-e^{\frac{2q_{0}^{2}}{v_{1}}x}\cdot e^{i(2\theta+\varphi_{1}+\varphi_{2})}\right)\right]}\right)\\
&-i\frac{2\left(v_{1}-\frac{q_{0}^{2}}{v_{1}}\right) e^{i(2\theta+\varphi_{2})} \cdot e^{-2iq_{0}^{2}t-i\left(v_{1}^{2}-\frac{q_{0}^{4}}{v_{1}^{2}}\right)t}\cdot e^{\left(\frac{q_{0}^{2}}{v_{1}}-v_{1}\right)x}}{\frac{q_{0}^{2}}{v_{1}}+v_{1}}\cdot\\
&\left(\frac{ie^{v_{1}x}\cdot(q_{0}^{2}+v_{1}^{2})\left(-q_{0}e^{\frac{q_{0}^{2}}{v_{1}}x}e^{i(\theta+\varphi_{1})}e^{i\left(v_{1}^{2}-\frac{q_{0}^{4}}{v_{1}^{2}}\right)t}+v_{1}e^{v_{1}x}\right)}
{2q_{0}v_{1} e^{\left(\frac{q_{0}^{2}}{v_{1}}+v_{1}\right)x}\cdot e^{i\theta}\left(e^{i\varphi_{2}}\cdot e^{-i\left(v_{1}^{2}-\frac{q_{0}^{4}}{v_{1}^{2}}\right)t}-e^{i\varphi_{1}} e^{i\left(v_{1}^{2}-\frac{q_{0}^{4}}{v_{1}^{2}}\right)t}\right)
+(q_{0}^{2}+v_{1}^{2})\left(e^{2v_{1}x}-e^{\frac{2q_{0}^{2}}{v_{1}}x} e^{i(2\theta+\varphi_{1}+\varphi_{2})}\right)}\right)\\
&=\Bigg(e^{-\frac{i}{2}\left(\frac{2\left(q_{0}^{2}+v_{1}^{2}\right)\left(t\left(q_{0}^{2}+v_{1}^{2}\right)-iv_{1}x\right)}{v_{1}^{2}}+\varphi_{1}+\varphi_{2}\right)}
\cdot\Big(2q_{0}^{4}\cdot e^{2itv_{1}^{2}+\left(\frac{q_{0}^{2}}{v_{1}}+v_{1}\right)x+i(\theta+\varphi_{1})}-2v_{1}^{4}e^{\frac{2iq_{0}^{4}t}{v_{1}^{2}}+\left(\frac{q_{0}^{2}}{v_{1}}+v_{1}\right)x+i(\theta+\varphi_{2})}\\
&-q_{0}v_{1}(q_{0}^{2}+v_{1}^{2})\cdot e^{\frac{it\left(q_{0}^{4}+v_{1}^{4}\right)}{v_{1}^{2}}+2v_{1}x}+q_{0}v_{1}(q_{0}^{2}+v_{1}^{2})\cdot e^{i\cdot \frac{q_{0}^{4}t-2iq_{0}^{2}v_{1}x+v_{1}^{2}\left(tv_{1}^{2}+2\theta+\varphi_{1}+\varphi_{2}\right)}{v_{1}^{2}}} \Big)\Bigg)/\\
&\Bigg(2v_{1}\left(-2iq_{0}v_{1}\sin\left[\frac{1}{2}\left(2t\left(\frac{q_{0}^{4}}{v_{1}^{2}}-v_{1}^{2}\right)-\varphi_{1}+\varphi_{2}\right)\right]
\right)+(q_{0}^{2}+v_{1}^{2})\sinh\left[\frac{q_{0}^{2}x}{v_{1}}+\frac{-2v_{1}x+i(2\theta+\varphi_{1}+\varphi_{2})}{2}\right]\Bigg).
\end{split}
\end{equation}
A typical breather solution is plotted in figure \ref{breather} below.

From this solution, we can see that when $t=0$, the singularity occurs at
\begin{equation}
x=\frac{Ln \left(\frac{2q_{0}^{2}v_{1}e^{i\theta}(e^{i\varphi_{1}}-e^{i\varphi_{2}})\pm\sqrt{4q_{0}^{2}v_{1}^{2}
e^{2i\theta_{1}(e^{i\varphi_{1}}-e^{i\varphi_{2}})^{2}}+4e^{i(2\theta+\varphi_{1}+\varphi_{2})}(q_{0}^{2}+v_{1}^{2})^{2}}}{2(q_{0}^{2}+v_{1}^{2})}\right)}{v_{1}-\frac{q_{0}^{2}}{v_{1}}}.
\end{equation}

\begin{figure}[h]
\begin{tabular}{cc}
\includegraphics[width=0.5\textwidth]{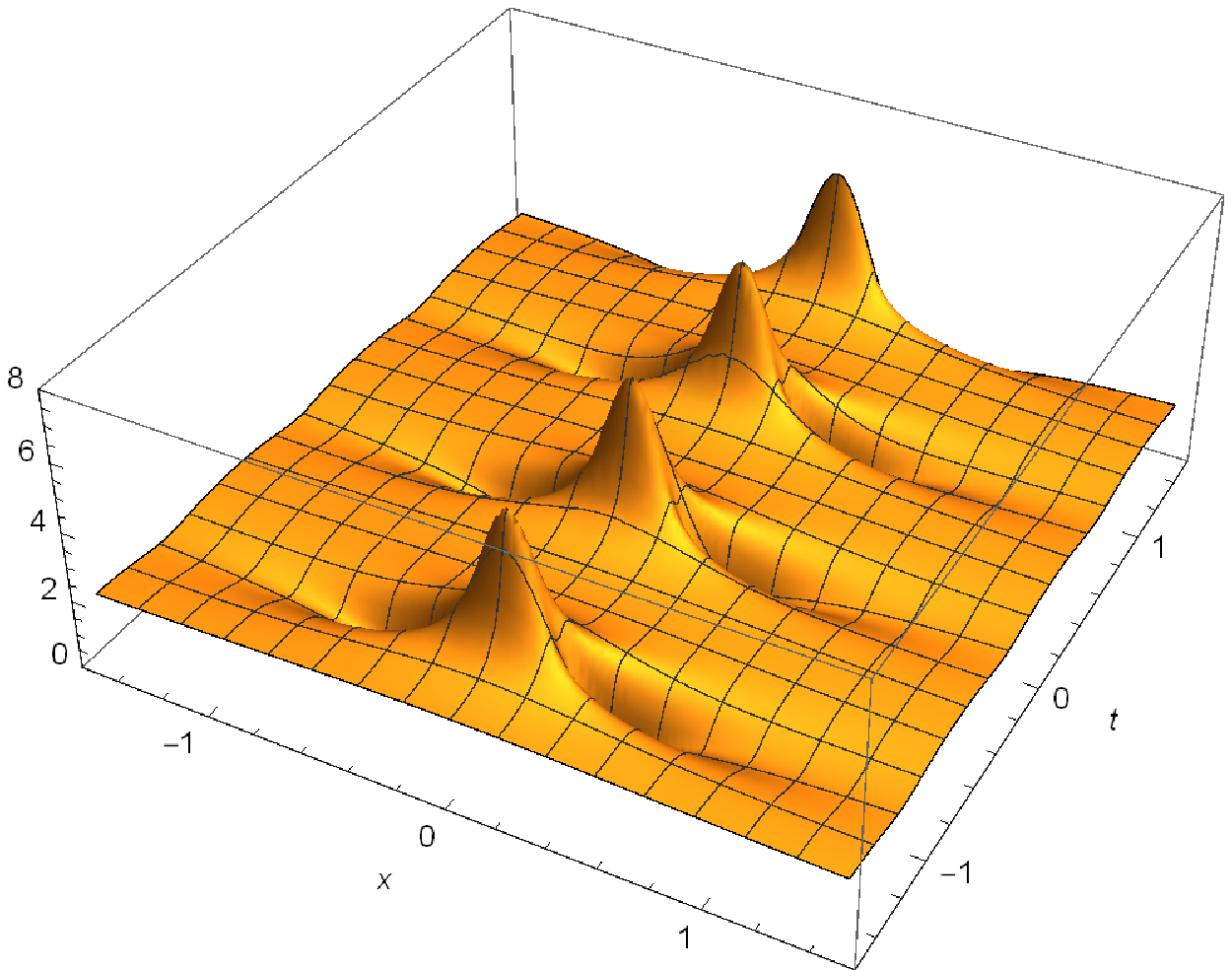}&
\includegraphics[width=0.5\textwidth]{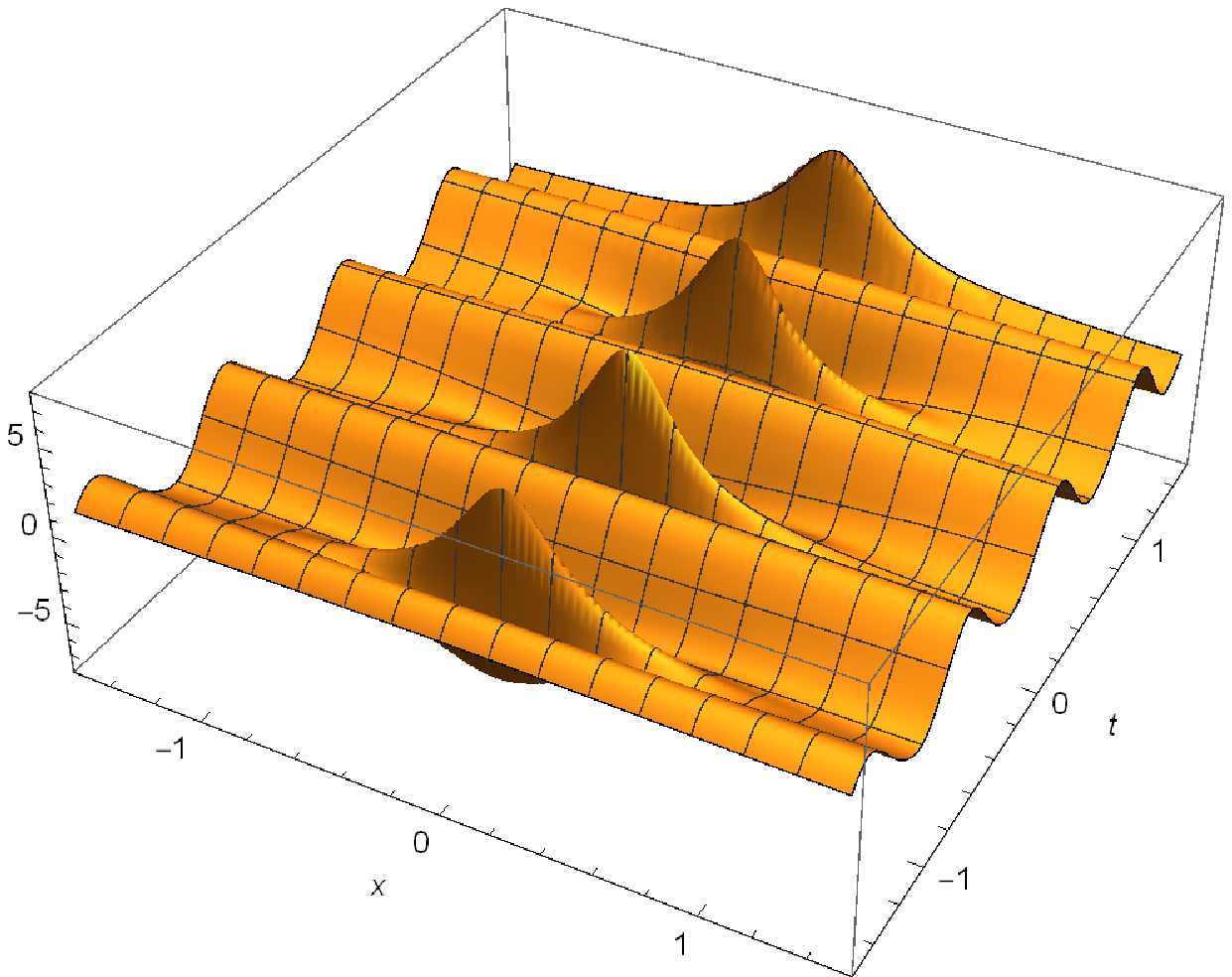}\\
(a) & (b)
\end{tabular}
\caption{(a) The amplitude of $q(x,t)$ with $\theta=\frac{\pi}{2}$, $\theta_{1}=0$, $\theta_{2}=0$, $v_{1}=3$ and $q_{0}=2$.  (b) The real part of $q(x,t)$ with $\theta=\frac{\pi}{2}$, $\theta_{1}=0$, $\theta_{2}=0$, $v_{1}=3$ and $q_{0}=2$.}
\label{breather}
\label{FigB}
\end{figure}


\section{The case of $\sigma=1$ with $\Delta \theta= \theta_{+}-\theta_{-}=0$}
Next we will study the nonlocal nonlinear Schr\"{o}dinger (NLS) equation (\ref{E:nonlocal NLS}) with $\sigma=1$
\begin{equation}
\label{E:nonlocal focusing NLS}
iq_{t}(x,t)=q_{xx}(x,t)-2q^{2}(x,t)q^{*}(-x,t)
\end{equation}
with nonzero boundary conditions (NZBCs)
\begin{equation}
q(x,t)\rightarrow \widehat{q}(t)=q_{0}e^{2iq_{0}^{2}t+i\theta},\ \ as \ \ x\rightarrow\pm\infty,
\end{equation}
where $q_{0}>0$, $0\leq \theta<2\pi$.

\subsection{Direct Scattering}
The nonlocal nonlinear Schr\"{o}dinger (NLS) equation (\ref{E:nonlocal NLS}) with $\sigma=1$ is associated with the following $2\times2$
compatible systems  \cite{AblowitzMusslimani}:
\begin{equation}
v_{x}=
\left(\begin{array}{cc}
-ik& q(x,t)\\
q^{*}(-x,t)& ik
\end{array}\right)v,
\end{equation}
\begin{equation}
v_{t}=
\left(\begin{array}{cc}
2ik^{2}+iq(x,t)q^{*}(-x,t)& -2kq(x,t)-iq_{x}(x,t)\\
-2kq^{*}(-x,t)-iq^{*}_{x}(-x,t)& -2ik^{2}-iq(x,t)q^{*}(-x,t)
\end{array}\right)v,
\end{equation}
where $q(x,t)$ is a complex-valued function of the real variables $x$ and $t$ with corresponding Lax pairs given by
\begin{equation}
X=
\left(\begin{matrix}
-ik& q(x,t)\\
q^{*}(-x,t)& ik
\end{matrix}\right),\ \
T=
\left(\begin{matrix}
2ik^{2}+iq(x,t)q^{*}(-x,t)& -2kq(x,t)-iq_{x}(x,t)\\
-2kq^{*}(-x,t)-iq^{*}_{x}(-x,t)& -2ik^{2}-iq(x,t)q^{*}(-x,t)
\end{matrix}\right).
\end{equation}
Alternatively the space part of the compatible system may be written in the form
\begin{equation}
v_{x}=(ikJ+Q)v, \ \ \ x\in \mathbb{R},
\end{equation}
where
\begin{equation}J=
\left(\begin{matrix}
-1& 0\\
0& 1
\end{matrix}\right),\ \ \
Q=
\left(\begin{matrix}
0& q(x,t)\\
q^{*}(-x,t)& 0
\end{matrix}\right),
\end{equation}
Here, $q(x,t)$ is called the potential and $k$ is a complex spectral parameter.

As $x\rightarrow\pm\infty$, the eigenfunctions of the scattering problem asymptotically satisfy
\begin{equation}
\left(\begin{array}{cc}
v_{1}\\
v_{2}
\end{array}\right)_{x}
=
\left(\begin{array}{cc}
-ik& q_{0}e^{2iq_{0}^{2}t+i\theta}\\
q_{0}e^{-2iq_{0}^{2}t-i\theta}& ik
\end{array}\right)
\left(\begin{array}{cc}
v_{1}\\
v_{2}
\end{array}\right),
\end{equation}
i.e.,
\begin{equation}
v_{x}=(ikJ+Q_{\pm}(t))v, \ \ \ Q_{\pm}(t)=\left(\begin{array}{cc}
0& q_{0}e^{2iq_{0}^{2}t+i\theta}\\
q_{0}e^{-2iq_{0}^{2}t-i\theta}& 0
\end{array}\right),
\end{equation}
which conveniently reduces to
\begin{equation}
\frac{\partial^{2} v_{j}}{\partial x^{2}}=-(k^{2}-q_{0}^{2})v_{j}, \ \ \ j=1,2.
\end{equation}
Each of the two equations has two linearly independent solutions $e^{i\lambda x}$ and $e^{-i\lambda x}$ as $|x|\rightarrow\infty$, where
$\lambda=\sqrt{k^{2}-q_{0}^{2}}$. The variable $k$ is then thought of as belonging to a Riemann surface $\mathbb{K}$ consisting of two sheets $\mathbb{C}_{1}$ and $\mathbb{C}_{2}$ with both coinciding with the complex plane cut along $(-\infty, -q_{0}]\cup [q_{0}, +\infty)$
with its edges glued in such a way that $\lambda(k)$ is continuous through the cut. The Riemann surface and the definition of $\lambda$ are the same as discussed above in section 3 ($\sigma=-1, \Delta \theta= \pi.$)
\subsection{Eigenfunctions}
Following the approach discussed in Section 3 we introduce the eigenfunctions defined by the following boundary conditions
\begin{equation}
\phi(x,k)\sim w e^{-i\lambda x}, \ \ \ \overline{\phi}(x,k)\sim \overline{w}e^{i\lambda x}
\end{equation}
as $x\rightarrow-\infty$,
\begin{equation}
\psi(x,k)\sim v e^{i\lambda x}, \ \ \ \overline{\psi}(x,k)\sim \overline{v}e^{-i\lambda x}
\end{equation}
as $x\rightarrow +\infty$, where
\begin{equation}
w=\left(\begin{array}{cc}
\lambda+k\\
i\widehat{q}^{*}
\end{array}\right), \ \ \
\overline{w}=\left(\begin{array}{cc}
-i\widehat{q}\\
\lambda+k
\end{array}\right),
\end{equation}
\begin{equation}
v=\left(\begin{array}{cc}
-i\widehat{q}\\
\lambda+k
\end{array}\right), \ \ \
\overline{v}=
\left(\begin{array}{cc}
\lambda+k\\
i\widehat{q}^{*}
\end{array}\right)
\end{equation}
satisfy the boundary conditions; they are not unique.
We consider functions with constant boundary conditions and define bounded eigenfunctions as follows:
\begin{equation}
M(x,k)=e^{i\lambda x}\phi(x,k), \ \ \ \overline{M}(x,k)=e^{-i\lambda x}\overline{\phi}(x,k),
\end{equation}
\begin{equation}
N(x,k)=e^{-i\lambda x}\psi(x,k), \ \ \ \overline{N}(x,k)=e^{i\lambda x}\overline{\psi}(x,k).
\end{equation}
The eigenfunctions  can be represented by means of the following integral equations
\begin{equation}
M(x,k)=
\left(\begin{array}{cc}
\lambda+k\\
i\widehat{q}^{*}
\end{array}\right)
+\int_{-\infty}^{+\infty}G_{-}(x-x',k)((Q-Q_{-})M)(x',k)dx',
\end{equation}
\begin{equation}
\overline{M}(x,k)=
\left(\begin{array}{cc}
-i\widehat{q}\\
\lambda+k
\end{array}\right)
+\int_{-\infty}^{+\infty}\overline{G}_{-}(x-x',k)((Q-Q_{-})M)(x',k)dx',
\end{equation}
\begin{equation}
N(x,k)=
\left(\begin{array}{cc}
-i\widehat{q}\\
\lambda+k
\end{array}\right)
+\int_{-\infty}^{+\infty}G_{+}(x-x',k)((Q-Q_{+})M)(x',k)dx',
\end{equation}
\begin{equation}
\overline{N}(x,k)=\left(\begin{array}{cc}
\lambda+k\\
i\widehat{q}^{*}
\end{array}\right)
+\int_{-\infty}^{+\infty}\overline{G}_{+}(x-x',k)((Q-Q_{+})M)(x',k)dx'.
\end{equation}
From the Fourier transform method, we get
\begin{equation}
G_{-}(x,k)=\frac{\theta(x)}{2\lambda}[(1+e^{2i\lambda x})\lambda I-i(e^{2i\lambda x}-1)(ikJ+Q_{-})],
\end{equation}
\begin{equation}
\overline{G}_{-}(x,k)=\frac{\theta(x)}{2\lambda}[(1+e^{-2i\lambda x})\lambda I+i(e^{-2i\lambda x}-1)(ikJ+Q_{-})],
\end{equation}
\begin{equation}
G_{+}(x,k)=-\frac{\theta(-x)}{2\lambda}[(1+e^{-2i\lambda x})\lambda I+i(e^{-2i\lambda x}-1)(ikJ+Q_{+})],
\end{equation}
\begin{equation}
\overline{G}_{+}(x,k)=-\frac{\theta(-x)}{2\lambda}[(1+e^{2i\lambda x})\lambda I-i(e^{2i\lambda x}-1)(ikJ+Q_{+})],
\end{equation}
where $\theta(x)$ is the Heaviside function, i.e., $\theta(x)=1$ if $x>0$ and $\theta(x)=0$ if $x<0$.

The analyticity of eigenfunctions and definitions of scattering data are the same as in Section 3.
\subsection{Symmetry reductions}
Taking into account boundary conditions, we can obtain
\begin{equation}
\psi(x,k)=\left(\begin{array}{cc}
0& 1\\
-1& 0
\end{array}\right)\phi^{*}(-x,-k^{*})
\end{equation}
and
\begin{equation}
\overline{\psi}(x,k)=\left(\begin{array}{cc}
0& -1\\
1& 0
\end{array}\right)\overline{\phi}^{*}(-x,-k^{*}).
\end{equation}
Similarly, we can get
\begin{equation}
N(x,k)=\left(\begin{array}{cc}
0& 1\\
-1& 0
\end{array}\right)M^{*}(-x,-k^{*})
\end{equation}
and
\begin{equation}
\overline{N}(x,k)=\left(\begin{array}{cc}
0& -1\\
1& 0
\end{array}\right)\overline{M}^{*}(-x,-k^{*}).
\end{equation}
Moreover,
\begin{equation}
a^{*}(-k^{*})=a(k),
\end{equation}
\begin{equation}
\overline{a}^{*}(-k^{*})=\overline{a}(k),
\end{equation}
\begin{equation}
b^{*}(-k^{*})=\overline{b}(k).
\end{equation}

When using a particular single sheet for the Riemann surface of the function $\lambda^{2}=k^{2}-q_{0}^{2}$, the involution $(k, \lambda)\rightarrow (k, -\lambda)$ can be considered across the cuts. The scattering data and eigenfunctions are defined by means of the corresponding values on the upper/lower edge of the cut, are labeled with superscripts $\pm$ as clarified below.
Explicitly, one has
\begin{equation}
a^{\pm}(k)=\frac{W(\phi^{\pm}(x,k),\psi^{\pm}(x,k))}{2\lambda^{\pm}(\lambda^{\pm}+k)}, \ \ \ k\in (-\infty, -q_{0}]\cup [q_{0}, +\infty),
\end{equation}
\begin{equation}
\overline{a}^{\pm}(k)=-\frac{W(\overline{\phi}^{\pm}(x,k),\overline{\psi}^{\pm}(x,k))}{2\lambda^{\pm}(\lambda^{\pm}+k)}, \ \ \ k\in (-\infty, -q_{0}]\cup [q_{0}, +\infty),
\end{equation}
\begin{equation}
b^{\pm}(k)=-\frac{W(\phi^{\pm}(x,k),\overline{\psi}^{\pm}(x,k))}{2\lambda^{\pm}(\lambda^{\pm}+k)}, \ \ \ k\in (-\infty, -q_{0}]\cup [q_{0}, +\infty),
\end{equation}
\begin{equation}
\overline{b}^{\pm}(k)=\frac{W(\overline{\phi}^{\pm}(x,k),\psi^{\pm}(x,k))}{2\lambda^{\pm}(\lambda^{\pm}+k)},  \ \ \ k\in(-\infty, -q_{0}]\cup [q_{0}, +\infty).
\end{equation}
Using the notation $\lambda=\lambda^{+}=-\lambda^{-}$, we have the following symmetry:
\begin{equation}
\phi^{\mp}(x,k)=\frac{\lambda^{\mp}+k}{-i\widehat{q}}\overline{\phi}^{\pm}(x,k), \ \ \ \ \ \psi^{\mp}(x,k)=\frac{\lambda^{\mp}+k}{i\widehat{q}^{*}}\overline{\psi}^{\pm}(x,k)
\end{equation}
for $k\in(-\infty, -q_{0}]\cup [q_{0}, +\infty)$.
Moreover,
\begin{equation}
a^{\pm}(k)=\overline{a}^{\mp}(k), \ \ \ \ \ b^{\pm}(k)=-\frac{\widehat{q}^{*}}{\widehat{q}}\cdot \overline{b}^{\mp}(k)
\end{equation}
for $k\in(-\infty, -q_{0}]\cup [q_{0}, +\infty)$.

\subsection{Uniformization coordinates}
As in Section 3 we introduce a uniformization variable $z$, defined by the conformal mapping:
\begin{equation}
z=z(k)=k+\lambda(k),
\end{equation}
and the inverse mapping is given by
\begin{equation}
k=k(z)=\frac{1}{2}\left(z+\frac{q_{0}^{2}}{z}\right).
\end{equation}
Then
\begin{equation}
\lambda(z)=\frac{1}{2}\left(z-\frac{q_{0}^{2}}{z}\right).
\end{equation}
\subsection{Symmetries via uniformization coordinates}
We have seen that (1) when $z\rightarrow -z^{*}$, then $(k,\lambda)\rightarrow (-k^{*}, -\lambda^{*})$; (2) when $z\rightarrow \frac{q_{0}^{2}}{z}$, then $(k, \lambda)\rightarrow (k, -\lambda)$. Hence,
\begin{equation}
\psi(x,z)=\left(\begin{array}{cc}
0& 1\\
-1& 0
\end{array}\right)\phi^{*}(-x,-z^{*}),
\end{equation}
\begin{equation}
\overline{\psi}(x,z)=\left(\begin{array}{cc}
0& -1\\
1& 0
\end{array}\right)\overline{\phi}^{*}(-x,-z^{*}),
\end{equation}
\begin{equation}
\phi\left(x, \frac{q_{0}^{2}}{z}\right)=\frac{\frac{q_{0}^{2}}{z}}{-i\widehat{q}}\overline{\phi}(x,z), \ \ \
\psi\left(x, \frac{q_{0}^{2}}{z}\right)=\frac{-i\widehat{q}}{z}\overline{\psi}(x,z), \ \ \ \Im z<0.
\end{equation}
Similarly, we can get
\begin{equation}
N(x,z)=\left(\begin{array}{cc}
0& 1\\
-1& 0
\end{array}\right)M^{*}(-x,-z^{*})
\end{equation}
and
\begin{equation}
\overline{N}(x,z)=\left(\begin{array}{cc}
0& -1\\
1& 0
\end{array}\right)\overline{M}^{*}(-x,-z^{*}).
\end{equation}

Moreover,
\begin{equation}
a^{*}(-z^{*})=a(z),
\end{equation}
\begin{equation}
\overline{a}^{*}(-z^{*})=\overline{a}(z),
\end{equation}
\begin{equation}
b^{*}(-z^{*})=\overline{b}(z),
\end{equation}
\begin{equation}
a\left(\frac{q_{0}^{2}}{z}\right)=\overline{a}(z), \ \ \ \Im z<0; \ \ \ b\left(\frac{q_{0}^{2}}{z}\right)=-\frac{\widehat{q}^{*}}{\widehat{q}}\cdot \overline{b}(z).
\end{equation}
\subsection{Asymptotic behavior of eigenfunctions and scattering data}
In order to solve the inverse problem, one has to determine the asymptotic behavior of eigenfunctions and scattering data both as $z\rightarrow\infty$ and as $z\rightarrow 0$. We have
\begin{equation}
M(x,z)\sim\left\{\begin{array}{ll}
\left(\begin{array}{cc}
z\\
iq^{*}(-x)
\end{array}\right), \ \ \ z\rightarrow\infty\\
\left(\begin{array}{cc}
z\cdot\frac{q(x)}{q_{+}}\\
i\widehat{q}^{*}
\end{array}\right), \ \ \ z\rightarrow 0,\\
\end{array}\right.
\end{equation}

\begin{equation}
N(x,z)\sim\left\{\begin{array}{ll}
\left(\begin{array}{cc}
-iq(x)\\
z
\end{array}\right), \ \ \ z\rightarrow\infty\\
\left(\begin{array}{cc}
-i\widehat{q}\\
z\cdot \frac{q^{*}(-x)}{\widehat{q}^{*}}
\end{array}\right), \ \ \ z\rightarrow 0,\\
\end{array}\right.
\end{equation}

\begin{equation}
\overline{M}(x,z)\sim\left\{\begin{array}{ll}
\left(\begin{array}{cc}
-iq(x)\\
z
\end{array}\right), \ \ \ z\rightarrow\infty\\
\left(\begin{array}{cc}
-i\widehat{q}\\
z\cdot \frac{q^{*}(-x)}{\widehat{q}^{*}}
\end{array}\right), \ \ \ z\rightarrow 0,\\
\end{array}\right.
\end{equation}

\begin{equation}
\label{E:asymptotic 3'''}
\overline{N}(x,z)\sim\left\{\begin{array}{ll}
\left(\begin{array}{cc}
z\\
iq^{*}(-x)
\end{array}\right), \ \ \ z\rightarrow\infty\\
\left(\begin{array}{cc}
z\cdot\frac{q(x)}{\widehat{q}}\\
i\widehat{q}^{*}
\end{array}\right), \ \ \ z\rightarrow 0,\\
\end{array}\right.
\end{equation}

\begin{equation}
a(z)=
\left\{\begin{array}{ll}
1,\ \ \ z\rightarrow\infty,\\
1, \ \ \ z\rightarrow 0,\\
\end{array}\right.
\end{equation}

\begin{equation}
\overline{a}(z)=
\left\{\begin{array}{ll}
1,\ \ \ z\rightarrow\infty,\\
1, \ \ \ z\rightarrow 0,\\
\end{array}\right.
\end{equation}
\begin{equation}
\lim_{z\rightarrow\infty}zb(z)=0, \ \ \ \lim_{z\rightarrow0}\frac{b(z)}{z^{2}}=0.
\end{equation}

\subsection{Left and right scattering problems}
Using similar the methods as in Section 3, we find
\begin{equation}
\label{E:eigenfunction 1''}
\begin{split}
\overline{N}(x,z)&=\left(\begin{array}{cc}
z\\
i\widehat{q}^{*}
\end{array}\right)
+\sum_{j=1}^{J}\frac{z\cdot b(z_{j})e^{i\big(z_{j}-\frac{q_{0}^{2}}{z_{j}}\big)x}\cdot N(x,z_{j})}{(z-z_{j})z_{j}a'(z_{j})}\\
&+\frac{z}{2\pi i}\int_{-\infty}^{+\infty}\frac{\rho(\xi)}{\xi(\xi-z)}\cdot e^{i\big(\xi-\frac{q_{0}^{2}}{\xi}\big)x}\cdot N(x,\xi)d\xi,
\end{split}
\end{equation}
\begin{equation}
\label{E:eigenfunction 2''}
\begin{split}
N(x,z)&=\left(\begin{array}{cc}
-i\widehat{q}\\
z
\end{array}\right)
+\sum_{j=1}^{\overline{J}}\frac{z\cdot\overline{b}(\overline{z}_{j})e^{-i\big(\overline{z}_{j}-\frac{q_{0}^{2}}{\overline{z}_{j}}\big)x}\cdot \overline{N}(x,\overline{z}_{j})}{(z-\overline{z}_{j})\overline{z}_{j}\overline{a}'(\overline{z}_{j})}\\
&-\frac{z}{2\pi i}\int_{-\infty}^{+\infty}\frac{\overline{\rho}(\xi)}{\xi(\xi-z)}\cdot e^{-i\big(\xi-\frac{q_{0}^{2}}{\xi}\big)x}\cdot \overline{N}(x,\xi)d\xi.
\end{split}
\end{equation}
\begin{equation}
\label{E:eigenfunction 3''}
\begin{split}
\overline{M}(x,z)&=\left(\begin{array}{cc}
-i\widehat{q}\\
z
\end{array}\right)+
\sum_{j=1}^{J}\frac{-z\cdot\overline{b}(z_{j})M(x,z_{j})e^{-i\big(z_{j}-\frac{q_{0}^{2}}{z_{j}}\big)x}}{(z-z_{j})z_{j}a'(z_{j})}\\
&+\frac{z}{2\pi i}\int_{-\infty}^{+\infty}\frac{\rho^{*}(-\xi)}{\xi(\xi-z)}\cdot e^{-i\big(\xi-\frac{q_{0}^{2}}{\xi}\big)x}\cdot M(x,\xi)d\xi
\end{split}
\end{equation}
and
\begin{equation}
\label{E:eigenfunction 4''}
\begin{split}
M(x,z)&=\left(\begin{array}{cc}
z\\
i\widehat{q}^{*}
\end{array}\right)+
\sum_{j=1}^{\overline{J}}\frac{-z\cdot b(\overline{z}_{j})\overline{M}(x,\overline{z}_{j})e^{i\big(\overline{z}_{j}-\frac{q_{0}^{2}}{\overline{z}_{j}}\big)x}}
{(z-\overline{z}_{j})\overline{z}_{j}\overline{a}'(\overline{z}_{j})}\\
&-\frac{z}{2\pi i}\int_{-\infty}^{+\infty}\frac{\overline{\rho}^{*}(-\xi)}{\xi(\xi-z)}\cdot e^{i\big(\xi-\frac{q_{0}^{2}}{\xi}\big)x}\cdot \overline{M}(x,\xi)d\xi.
\end{split}
\end{equation}
\subsection{Recovery of the potentials}
Note that
\begin{equation}
\frac{\overline{N}_{1}(x,z)}{z}\sim \frac{q(x)}{\widehat{q}}
\end{equation}
as $z\rightarrow 0$, and
\begin{equation}
\frac{\overline{N}_{1}(x,z)}{z}\rightarrow1+\sum_{j=1}^{J}\frac{b(z_{j})e^{i\big(z_{j}-\frac{q_{0}^{2}}{z_{j}}\big)x}}{-z_{j}^{2}a'(z_{j})}\cdot N_{1}(x,z_{j})+\frac{1}{2\pi i}\int_{-\infty}^{+\infty}\frac{\rho(\xi)}{\xi^{2}}\cdot e^{i\big(\xi-\frac{q_{0}^{2}}{\xi}\big)x}\cdot N_{1}(x,\xi)d\xi
\end{equation}
as $z\rightarrow 0$, therefore,
\begin{equation}
\label{E:asympN1c''}
q(z)=\widehat{q}\cdot \left[1+\sum_{j=1}^{J}\frac{b(z_{j})e^{i\big(z_{j}-\frac{q_{0}^{2}}{z_{j}}\big)x}}{-z_{j}^{2}a'(z_{j})}\cdot N_{1}(x,z_{j})+\frac{1}{2\pi i}\int_{-\infty}^{+\infty}\frac{\rho(\xi)}{\xi^{2}}\cdot e^{i\big(\xi-\frac{q_{0}^{2}}{\xi}\big)x}\cdot N_{1}(x,\xi)d\xi\right].
\end{equation}

\subsection{Closing the system}
We find $J=\overline{J}$ from $a\left(\frac{q_{0}^{2}}{z}\right)=\overline{a}(z)$. To close the system, 
by the symmetry relations between the eigenfunctions, we have
\begin{equation}
\label{E:closing system 5}
\begin{split}
&\left(\begin{array}{cc}
N_{1}(x,z)\\
N_{2}(x,z)
\end{array}\right)=\left(\begin{array}{cc}
-i\widehat{q}\\
z
\end{array}\right)
+\sum_{j=1}^{J}\frac{z\cdot\overline{b}(\overline{z}_{j})e^{-i\big(\overline{z}_{j}-\frac{q_{0}^{2}}{\overline{z}_{j}}\big)x} }{(z-\overline{z}_{j})\overline{z}_{j}\overline{a}'(\overline{z}_{j})}\cdot\\
&\left(\begin{array}{cc}
\overline{z}_{j}+\sum_{l=1}^{J}\frac{\overline{z}_{j}\cdot b(z_{l})e^{i\big(z_{l}-\frac{q_{0}^{2}}{z_{l}}\big)x}}{(\overline{z}_{j}-z_{l})z_{l}a'(z_{l})}\cdot N_{1}(x, z_{l})+\frac{\overline{z}_{j}}{2\pi i}\int_{-\infty}^{+\infty}\frac{\rho(\xi)}{\xi(\xi-\overline{z}_{j})}\cdot e^{i\big(\xi-\frac{q_{0}^{2}}{\xi}\big)x}\cdot N_{1}(x, \xi)d\xi\\
i\widehat{q}^{*}+\sum_{l=1}^{J}\frac{\overline{z}_{j}\cdot b(z_{l})e^{i\big(z_{l}-\frac{q_{0}^{2}}{z_{l}}\big)x}}{(\overline{z}_{j}-z_{l})z_{l}a'(z_{l})}\cdot N_{2}(x, z_{l})+\frac{\overline{z}_{j}}{2\pi i}\int_{-\infty}^{+\infty}\frac{\rho(\xi)}{\xi(\xi-\overline{z}_{j})}\cdot e^{i\big(\xi-\frac{q_{0}^{2}}{\xi}\big)x}\cdot N_{2}(x, \xi)d\xi
\end{array}\right)\\
&-\frac{z}{2\pi i}\int_{-\infty}^{+\infty}\frac{\overline{\rho}(\xi)}{\xi(\xi-z)}\cdot e^{-i\big(\xi-\frac{q_{0}^{2}}{\xi}\big)x}\cdot\\
&\left(\begin{array}{cc}
\xi+\sum_{l=1}^{J}\frac{\xi\cdot b(z_{l})e^{i\big(z_{l}-\frac{q_{0}^{2}}{z_{l}}\big)x}}{(\xi-z_{l})z_{l}a'(z_{l})}\cdot N_{1}(x, z_{l})+\frac{\xi}{2\pi i}\int_{-\infty}^{+\infty}\frac{\rho(\eta)}{\eta(\eta-\xi)}\cdot e^{i\big(\eta-\frac{q_{0}^{2}}{\eta}\big)x}\cdot N_{1}(x, \eta)d\eta\\
i\widehat{q}^{*}+\sum_{l=1}^{J}\frac{\xi\cdot b(z_{l})e^{i\big(z_{l}-\frac{q_{0}^{2}}{z_{l}}\big)x}}{(\xi-z_{l})z_{l}a'(z_{l})}\cdot N_{2}(x, z_{l})+\frac{\xi}{2\pi i}\int_{-\infty}^{+\infty}\frac{\rho(\eta)}{\eta(\eta-\xi)}\cdot e^{i\big(\eta-\frac{q_{0}^{2}}{\eta}\big)x}\cdot N_{2}(x, \eta)d\eta
\end{array}\right)d\xi,
\end{split}
\end{equation}
\begin{equation}
\label{E:closing system 6}
\begin{split}
&\left(\begin{array}{cc}
\overline{M}_{1}(x,z)\\
\overline{M}_{2}(x,z)
\end{array}\right)=\left(\begin{array}{cc}
-i\widehat{q}\\
z
\end{array}\right)+
\sum_{j=1}^{J}\frac{-z\cdot\overline{b}(z_{j})e^{-i\big(z_{j}-\frac{q_{0}^{2}}{z_{j}}\big)x}}{(z-z_{j})z_{j}a'(z_{j})}\cdot\\
& \left(\begin{array}{cc}
z_{j}+
\sum_{l=1}^{J}\frac{-z_{j}\cdot b(\overline{z}_{l})e^{i\big(\overline{z}_{l}-\frac{q_{0}^{2}}{\overline{z}_{l}}\big)x}}
{(z_{j}-\overline{z}_{l})\overline{z}_{l}\overline{a}'(\overline{z}_{l})}\cdot \overline{M}_{1}(x, \overline{z}_{l})-\frac{z_{j}}{2\pi i}\int_{-\infty}^{+\infty}\frac{\overline{\rho}^{*}(-\xi)}{\xi(\xi-z_{j})}\cdot e^{i\big(\xi-\frac{q_{0}^{2}}{\xi}\big)x}\cdot\overline{M}_{1}(x,\xi)d\xi\\
i\widehat{q}^{*}+
\sum_{l=1}^{J}\frac{-z_{j}\cdot b(\overline{z}_{l})e^{i\big(\overline{z}_{l}-\frac{q_{0}^{2}}{\overline{z}_{l}}\big)x}}
{(z_{j}-\overline{z}_{l})\overline{z}_{l}\overline{a}'(\overline{z}_{l})}\cdot \overline{M}_{2}(x, \overline{z}_{l})-\frac{z_{j}}{2\pi i}\int_{-\infty}^{+\infty}\frac{\overline{\rho}^{*}(-\xi)}{\xi(\xi-z_{j})}\cdot e^{i\big(\xi-\frac{q_{0}^{2}}{\xi}\big)x}\cdot\overline{M}_{2}(x,\xi)d\xi
\end{array}\right)\\
&+\frac{z}{2\pi i}\int_{-\infty}^{+\infty}\frac{\rho^{*}(-\xi)}{\xi(\xi-z)}\cdot e^{-i\big(\xi-\frac{q_{0}^{2}}{\xi}\big)x}\cdot\\
&\left(\begin{array}{cc}
\xi+
\sum_{l=1}^{J}\frac{-\xi\cdot b(\overline{z}_{l})e^{i\big(\overline{z}_{l}-\frac{q_{0}^{2}}{\overline{z}_{l}}\big)x}}
{(\xi-\overline{z}_{l})\overline{z}_{l}\overline{a}'(\overline{z}_{l})}\cdot\overline{M}_{1}(x,\overline{z}_{l})-\frac{\xi}{2\pi i}\int_{-\infty}^{+\infty}\frac{\overline{\rho}^{*}(-\eta)}{\eta(\eta-\xi)}\cdot e^{i\big(\eta-\frac{q_{0}^{2}}{\eta}\big)x}\cdot \overline{M}_{1}(x,\eta)d\eta\\
i\widehat{q}^{*}+
\sum_{l=1}^{J}\frac{-\xi\cdot b(\overline{z}_{l})e^{i\big(\overline{z}_{l}-\frac{q_{0}^{2}}{\overline{z}_{l}}\big)x}}
{(\xi-\overline{z}_{l})\overline{z}_{l}\overline{a}'(\overline{z}_{l})}\cdot\overline{M}_{2}(x,\overline{z}_{l})-\frac{\xi}{2\pi i}\int_{-\infty}^{+\infty}\frac{\overline{\rho}^{*}(-\eta)}{\eta(\eta-\xi)}\cdot e^{i\big(\eta-\frac{q_{0}^{2}}{\eta}\big)x}\cdot \overline{M}_{2}(x,\eta)d\eta
\end{array}\right)d\xi.
\end{split}
\end{equation}

We note that from eq (\ref{E:asympN1c''})  $q(x)$  is given in only terms of the component $N_1$. We can use only the first component of eq (\ref{E:closing system 5}) to find this function and hence $q(x)$. Hence to complete the inverse scattering we reduce the problem to solving an integral equation in terms of the component $N_1$ only.

\subsection{Trace formula}
Using $b^{*}(-z^{*})=\overline{b}(z)$ and following the analysis in Section 3, we can get
\begin{equation}
\log a(z)=\log\left(\prod_{j=1}^{J_{1}}\frac{z-z_{j}}{z-\frac{q_{0}^{2}}{z_{j}}}\cdot\frac{z+z_{j}^{*}}{z+\frac{q_{0}^{2}}{z_{j}^{*}}}
\cdot\prod_{i=1}^{J_{2}}\frac{z-\widetilde{z}_{i}}{z-\frac{q_{0}^{2}}{\widetilde{z}_{i}}}\right)+\frac{1}{2\pi i}\int_{-\infty}^{+\infty}\frac{\log (1+b(\xi)b^{*}(-\xi^{*}))}{\xi-z}d\xi, \ \ \ \Im z>0,
\end{equation}

\begin{equation}
\log \overline{a}(z)=\log\left(\prod_{j=1}^{J_{1}} \frac{z-\frac{q_{0}^{2}}{z_{j}}}{z-z_{j}}\cdot \frac{z+\frac{q_{0}^{2}}{z_{j}^{*}}}{z+z_{j}^{*}}
\cdot\prod_{i=1}^{J_{2}} \frac{z-\frac{q_{0}^{2}}{\widetilde{z}_{i}}}{z-\widetilde{z}_{i}}\right)-\frac{1}{2\pi i}\int_{-\infty}^{+\infty}\frac{\log (1+b(\xi)b^{*}(-\xi^{*}))}{\xi-z}d\xi, \ \ \ \Im z<0,
\end{equation}
where $\Re\widetilde{z}_{i}=0$, $\Re z_{j}\neq 0$ and $2J_{1}+J_{2}=J$.

\subsection{Discrete scattering data and their symmetries}
As discussed earlier, the data $a(z_j), \bar{a}(\bar{z}_j)$ are calculated =from the trace formulae. The symmetries on the data
\[ b(z_j) ~\mbox{and} ~~~\bar{b}(\bar{z}_j) ,~~j=1,2,...J.\]
can be calculated from their associated eigenfunctions.
Since
\begin{equation}
N_{1}(x,z)=M_{2}^{*}(-x,-z^{*}), \ \ \ N_{2}(x,z)=-M_{1}^{*}(-x,-z^{*}),
\end{equation}
\begin{equation}
M_{1}(x,z_{j})=b(z_{j})e^{i\big(z_{j}-\frac{q_{0}^{2}}{z_{j}}\big)x}\cdot N_{1}(x,z_{j})
\end{equation}
and
\begin{equation}
M_{2}(x,z_{j})=b(z_{j})e^{i\big(z_{j}-\frac{q_{0}^{2}}{z_{j}}\big)x}\cdot N_{2}(x,z_{j}),
\end{equation}
we have
\begin{equation}
\label{E:N1'}
N_{1}(x,z_{j})=b^{*}(-z_{j}^{*})\cdot e^{-i\left(z_{j}-\frac{q_{0}^{2}}{z_{j}}\right)x}\cdot N_{2}^{*}(-x,-z_{j}^{*}),
\end{equation}
\begin{equation}
\label{E:N2'}
N_{2}(x,z_{j})=-b^{*}(-z_{j}^{*})\cdot e^{-\left(z_{j}-\frac{q_{0}^{2}}{z_{j}}\right)x}\cdot N_{1}^{*}(-x,-z_{j}^{*}).
\end{equation}
By rewriting (\ref{E:N2'}), we obtain
\begin{equation}
N_{2}^{*}(-x,-z_{j}^{*})=-b(z_{j})\cdot e^{\left(z_{j}-\frac{q_{0}^{2}}{z_{j}}\right)x}\cdot N_{1}(x,z_{j}).
\end{equation}
Combining (\ref{E:N1'}), we can deduce
\begin{equation}
b(z_{j})b^{*}(-z_{j}^{*})=-1.
\end{equation}
If $\Re z_{j} = 0$ we are led to a contradiction. Thus we must have that $\Re z_{j}\neq 0$; consequently the eigenvalues come in complex conjugate pairs and $J$ must be even.

\subsection{Reflectioness potentials and soliton solutions}
Reflectioness potentials 
correspond to zero reflection coefficients, i.e., $\rho(\xi)=0$ and $\overline{\rho}(\xi)=0$ for all real $\xi$.
 We also note,   from the symmetry relation $\bar{b}(z)=-b^*(-z^*)$ it follows that the reflection coefficients $\rho(z)= b(z)/a(z), \bar{\rho}(z)= \bar{b}(z)/\bar{a}(z)$ will both vanish when $b(z)=0$ for $z$ on the real axis. By substituting $z=z_{l}$ in (\ref{E:closing system 5}) and $z=\overline{z}_{l}$ in (\ref{E:closing system 6}), the system (\ref{E:closing system 5}) and (\ref{E:closing system 6}) reduces to algebraic equations that determine the functional form of these special potentials. When time dependence is added the reflectionless potentials correspond to soliton solutions.
The reduced equations take the form
\begin{equation}
\begin{split}
\left(\begin{array}{cc}
N_{1}(x,z_{l})\\
N_{2}(x,z_{l})
\end{array}\right)&=\left(\begin{array}{cc}
-i\widehat{q}\\
z_{l}
\end{array}\right)
+\sum_{j=1}^{J}\frac{z_{l}\cdot\overline{b}(\overline{z}_{j})e^{-i\big(\overline{z}_{j}-\frac{q_{0}^{2}}{\overline{z}_{j}}\big)x} }{(z_{l}-\overline{z}_{j})\overline{z}_{j}\overline{a}'(\overline{z}_{j})}\cdot\\
&\left(\begin{array}{cc}
\overline{z}_{j}+\sum_{l=1}^{J}\frac{\overline{z}_{j}\cdot b(z_{l})e^{i\big(z_{l}-\frac{q_{0}^{2}}{z_{l}}\big)x}}{(\overline{z}_{j}-z_{l})z_{l}a'(z_{l})}\cdot N_{1}(x, z_{l})\\
i\widehat{q}^{*}+\sum_{l=1}^{J}\frac{\overline{z}_{j}\cdot b(z_{l})e^{i\big(z_{l}-\frac{q_{0}^{2}}{z_{l}}\big)x}}{(\overline{z}_{j}-z_{l})z_{l}a'(z_{l})}\cdot N_{2}(x, z_{l})
\end{array}\right),
\end{split}
\end{equation}
\begin{equation}
\begin{split}
\left(\begin{array}{cc}
\overline{M}_{1}(x,\overline{z}_{l})\\
\overline{M}_{2}(x,\overline{z}_{l})
\end{array}\right)&=\left(\begin{array}{cc}
-i\widehat{q}\\
\overline{z}_{l}
\end{array}\right)+
\sum_{j=1}^{J}\frac{-\overline{z}_{l}\cdot\overline{b}(z_{j})e^{-i\big(z_{j}-\frac{q_{0}^{2}}{z_{j}}\big)x}}{(\overline{z}_{l}-z_{j})z_{j}a'(z_{j})}\cdot\\
& \left(\begin{array}{cc}
z_{j}+
\sum_{l=1}^{J}\frac{-z_{j}\cdot b(\overline{z}_{l})e^{i\big(\overline{z}_{l}-\frac{q_{0}^{2}}{\overline{z}_{l}}\big)x}}
{(z_{j}-\overline{z}_{l})\overline{z}_{l}\overline{a}'(\overline{z}_{l})}\cdot \overline{M}_{1}(x, \overline{z}_{l})\\
i\widehat{q}^{*}+
\sum_{l=1}^{J}\frac{-z_{j}\cdot b(\overline{z}_{l})e^{i\big(\overline{z}_{l}-\frac{q_{0}^{2}}{\overline{z}_{l}}\big)x}}
{(z_{j}-\overline{z}_{l})\overline{z}_{l}\overline{a}'(\overline{z}_{l})}\cdot \overline{M}_{2}(x, \overline{z}_{l})
\end{array}\right).
\end{split}
\end{equation}



The above equations are an algebraic system to solve for $N(x,z_{l})$ or $\overline{M}(x, \overline{z}_{l})$, $l=1,2...J$. The potential are reconstructed from
equation (\ref{E:asympN1c''}) with $\rho(\xi)=0,\bar{\rho}(\xi)=0$; i.e.,

\begin{equation}
\label{asympN1d''}
q(z)=\widehat{q}\cdot \left[1+\sum_{j=1}^{J}\frac{b(z_{j})e^{i\big(z_{j}-\frac{q_{0}^{2}}{z_{j}}\big)x}}{-z_{j}^{2}a'(z_{j})}\cdot N_{1}(x,z_{j})\right].
\end{equation}
Since $a(z)\sim 1$ as $z\rightarrow 0$, by the trace formula when $b(\xi)=0$ in the real axis, we have
the constraint
\begin{equation}
\label{E:product3}
\prod_{j=1}^{J/2}\frac{|z_{j}|^{4}}{q_{0}^{4}}=1.
\end{equation}

\subsection{Reflectionless potential solution: 2-eigenvalue}
In this subsection, we construct an explicit form for the 2-reflectionless solution (this is the 2 soliton solution when time dependence is added) by setting $J=2$ and $z_{1}=\xi_{1}+i \eta_{1}$, where $\xi_{1}, \eta_{1}>0$ and $\xi_{1}^{2}+\eta_{1}^{2}=q_{0}^{2}$. Then $a(\xi_{1}+i \eta_{1})=a(-\xi_{1}+i \eta_{1})=0$ and $\overline{a}(\xi_{1}-i\eta_{1})=\overline{a}(-\xi_{1}-i\eta_{1})=0$. Thus,
\begin{equation}
a(z)=\frac{(z-(\xi_{1}+i\eta_{1}))(z-(-\xi_{1}+i\eta_{1}))}{(z-(\xi_{1}-i\eta_{1}))(z-(-\xi_{1}-i\eta_{1}))}, \ \ \
\overline{a}(z)=\frac{(z-(\xi_{1}-i\eta_{1}))(z-(-\xi_{1}-i\eta_{1}))}{(z-(\xi_{1}+i\eta_{1}))(z-(-\xi_{1}+i\eta_{1}))}
\end{equation}
and we then find
\begin{equation}
a'(\xi_{1}+i\eta_{1})=\frac{-i\xi_{1}(\xi_{1}-i\eta_{1})}{2q_{0}^{2}\eta_{1}}, \ \ \
a'(-\xi_{1}+i\eta_{1})=\frac{-i\xi_{1}(\xi_{1}+i\eta_{1})}{2q_{0}^{2}\eta_{1}},
\end{equation}
\begin{equation}
\overline{a}'(\xi_{1}-i\eta_{1})=\frac{i\xi_{1}(\xi_{1}+i\eta_{1})}{2q_{0}^{2}\eta_{1}}, \ \ \
\overline{a}'(-\xi_{1}-i\eta_{1})=\frac{i\xi_{1}(\xi_{1}-i\eta_{1})}{2q_{0}^{2}\eta_{1}}.
\end{equation}

Moreover, we write $b(\xi_{1}+i\eta_{1})=c_{1}e^{i\theta_{1}}$, by $b(z_{j})b^{*}(-z_{j}^{*})=-1$, we have $b(-\xi_{1}+i\eta_{1})=-\frac{1}{c_{1}}e^{i\theta_{1}}$.

From $b\left(\frac{q_{0}^{2}}{z}\right)=-\frac{\widehat{q}^{*}}{\widehat{q}}\cdot \overline{b}(z)$, we get
\begin{equation}
\overline{b}(\xi_{1}-i\eta_{1})=-e^{2i\theta}\cdot b(\xi_{1}+i\eta_{1})=-c_{1}e^{i(2\theta+\theta_{1})}
\end{equation}
and
\begin{equation}
\overline{b}(-\xi_{1}-i\eta_{1})=-e^{2i\theta}\cdot b(-\xi_{1}+i\eta_{1})=\frac{1}{c_{1}}e^{i(2\theta+\theta_{1})}.
\end{equation}

Hence, we have
\begin{equation}
\begin{split}
&N_{1}(x,\xi_{1}+i\eta_{1})=\Big(e^{2\eta_{1}x+i\theta}\cdot\xi_{1}\cdot \Big(e^{2\eta_{1}x+2i(\theta+\theta_{1})}q_{0}(\eta_{1}+c_{1}^{2}\eta_{1}-i\xi_{1})-ic_{1}^{2}q_{1}\xi_{1}e^{6\eta_{1}x}
+c_{1}\xi_{1}(\xi_{1}+i\eta_{1})e^{3(\theta+\theta_{1})i}\\
&+c_{1}(\eta_{1}-i\xi_{1})(\eta_{1}+c_{1}^{2}\eta_{1}+ic_{1}^{2}\xi_{1})e^{4\eta_{1}x+i(
\theta+\theta_{1})}\Big)\Big)/\Big(c_{1}^{2}\xi_{2}^{2}e^{8\eta_{1}x}+c_{1}^{2}\xi_{1}^{2}e^{4(
\theta+\theta_{1})i}+e^{4\eta_{1}x+2i(
\theta+\theta_{1})}\\
&\cdot\Big((1+c_{1}^{2})^{2}\eta_{1}^{2}+(1+c_{1}^{4})\xi_{1}^{2}\Big)\Big)
\end{split}
\end{equation}
and
\begin{equation}
\begin{split}
&N_{1}(x,-\xi_{1}+i\eta_{1})=\Big(c_{1}\xi_{1}e^{2\eta_{1}x+i\theta}\cdot\Big(e^{4\eta_{1}x+i(\theta+\theta_{1})}(\eta_{1}+c_{1}^{2}\eta_{1}-i\xi_{1})
(\eta_{1}+i\xi_{1})-ic_{1}q_{0}\xi_{1}e^{6\eta_{1}x}\\
&+c_{1}^{2}\xi_{1}(\xi_{1}-i\eta_{1})e^{3i(\theta+\theta_{1})}-c_{1}q_{0}(\eta_{1}+c_{1}^{2}\eta_{1}+ic_{1}^{2}\xi_{1})
e^{2\eta_{1}x+2i(\theta+\theta_{1})}\Big)\Big)/\Big(c_{1}^{2}\xi_{1}^{2}e^{8\eta_{1}x}+c_{1}^{2}\xi_{1}^{2}e^{4i(\theta+\theta_{1})}\\
&+e^{4\eta_{1}x+2i(\theta+\theta_{1})}\Big((1+c_{1}^{2})^{2}\eta_{1}^{2}+(1+c_{1}^{4})\xi_{1}^{2}\Big)\Big)
\end{split}
\end{equation}

By (\ref{E:asympN1c''}), we can obtain
\begin{equation}
\begin{split}
q(x)&=\widehat{q}\cdot\left[1+\frac{c_{1}e^{i\theta_{1}}e^{-2\eta_{1}x}}{\frac{i\xi_{1}(\xi_{1}+i\eta_{1})}{2\eta_{1}}}N_{1}(x, \xi_{1}+i\eta_{1})
+\frac{-\frac{1}{c_{1}}e^{i\theta_{1}}e^{-2\eta_{1}x}}{\frac{i\xi_{1}(\xi_{1}-i\eta_{1})}{2\eta_{1}}}N_{1}(x, -\xi_{1}+i\eta_{1})\right]
\end{split}
\end{equation}
\subsection{Time evolution}
As in previous sections we find that the scattering data satisfies 
\begin{equation}
\frac{\partial a(t)}{\partial t}=0, \ \ \ \frac{\partial \overline{a}(t)}{\partial t}=0, \ \ \ \frac{\partial b(t)}{\partial t}=-2i(q_{0}^{2}+2\lambda k)b(t).
\end{equation}
Hence, $a(t)$ and $\overline{a}(t)$ are time independent. Moreover,
\begin{equation}
b(\xi_{1}+i\eta_{1},t)=c_{1}e^{i\theta_{1}}\cdot e^{-2i(q_{0}^{2}+2\xi_{1}\eta_{1}i)t},
\end{equation}
\begin{equation}
b(-\xi_{1}+i\eta_{1},t)=-\frac{1}{c_{1}}e^{i\theta_{1}}\cdot e^{-2i(q_{0}^{2}-2\xi_{1}\eta_{1}i)t}.
\end{equation}
Putting all the above into the formula we had for the reflectionless potential, we obtain the following 2-soliton solution
\begin{equation}
\begin{split}
&q(x,t)=\Big(e^{i(2q_{0}^{2}t+\theta)}\Big(c_{1}^{2}q_{0}e^{8t\eta_{1}\sqrt{q_{0}^{2}-\eta_{1}^{2}}}\big(-2\eta_{1}^{2}e^{4\eta_{1}x
+2i(\theta+\theta_{1})}+(q_{0}^{2}-\eta_{1}^{2})(e^{8\eta_{1}x}+e^{4i(\theta+\theta_{1})})\\
&+c_{1}^{4}q_{0}e^{4\eta_{1}(x+4t\sqrt{q_{0}^{2}-\eta_{1}^{2}})+2i(\theta+\theta_{1})}\cdot\big(q_{0}^{2}-2\eta_{1}^{2}
-2i\eta_{1}\sqrt{q_{0}^{2}-\eta_{1}^{2}}\big)+2c_{1}\eta_{1}e^{2\eta_{1}x+4t\eta_{1}\sqrt{q_{0}^{2}-\eta_{1}^{2}}+i(\theta+\theta_{1})}\\
&\cdot(-e^{4\eta_{1}x}+e^{2i(\theta+\theta_{1})})(-q_{0}^{2}+\eta_{1}(\eta_{1}-i\sqrt{q_{0}^{2}-\eta_{1}^{2}}))-2c_{1}^{3}\eta_{1}
e^{2\eta_{1}(x+6t\sqrt{q_{0}^{2}-\eta_{1}^{2}})+i(\theta+\theta_{1})}\cdot\big(-e^{4\eta_{1}x}+e^{2i(\theta+\theta_{1})}\big)\\
&\cdot\big(-q_{0}^{2}+\eta_{1}(\eta_{1}-i\sqrt{q_{0}^{2}-\eta_{1}^{2}})\big)+q_{0}e^{4\eta_{1}x+2i(\theta+\theta_{1})}\cdot
\big(q_{0}^{2}+2i\eta_{1}(i\eta_{1}+\sqrt{q_{0}^{2}-\eta_{1}^{2}})\big)\Big)\Big)/\\
&\Big(q_{0}^{2}e^{4\eta_{1}x+2i(\theta+\theta_{1})}+c_{1}^{4}q_{0}^{2}e^{4\eta_{1}(x+4t\sqrt{q_{0}^{2}-\eta_{1}^{2}})+2i(\theta+\theta_{1})}
-c_{1}^{2}e^{8t\eta_{1}\sqrt{q_{0}^{2}-\eta_{1}^{2}}}\big(-2\eta_{1}^{2}e^{4\eta_{1}x+2i(\theta+\theta_{1})}\\
&+(\eta_{1}^{2}-q_{0}^{2})(e^{8\eta_{1}x+4i(\theta+\theta_{1})})\big)\Big).
\end{split}
\end{equation}
From this solution, we can see that the singularity occurs only when $\theta+\theta_{1}=\frac{\pi}{2}$ at
\begin{equation}
x=\frac{Ln \left(\frac{(1+c_{1}^{2})^{2}\eta_{1}^{2}+(1+c_{1}^{4})\xi_{1}^{2}\mp\sqrt{(1+c_{1}^{2})^{2}\eta_{1}^{2}+(1+c_{1}^{4})\xi_{1}^{2})^{2}-4c_{1}^{4}\xi_{1}^{4}}}{2c_{1}^{2}\xi_{1}^{2}}\right)}{4\eta_{1}}
\end{equation}
for $t=0$.

In particular, when $q_{0}=2$, $\theta=0$, $\theta_{1}=-\frac{3\pi}{4}$, $c_{1}=\sqrt{2}$, and $\eta_{1}=\sqrt{3}$, we can get the result in \cite{He}, i.e.,
\begin{equation}
q(x,t)=i e^{8it}\cdot \frac{(-\sqrt{3}+3i)\cosh(4\sqrt{3}t)+2i\cosh(2\sqrt{3}x)+(-3\sqrt{3}+i)\sinh(4\sqrt{3}t)-2\sinh(2\sqrt{3}x)}
{3\cosh(4\sqrt{3}t)-\cosh(2\sqrt{3}x)+\sinh(4\sqrt{3}t)-i\sinh(2\sqrt{3}x)}.
\end{equation}
Typical 2 soliton solutions, which travel and interact, including the one in  \cite{He}, are shown in Figures \ref{FigC}-\ref{FigC'}-\ref{FigC''} below. In terms of magnitude, Fig. \ref{FigC} shows the interaction of two waves of elevation,  Fig. \ref{FigC'} shows the interaction of two depressive waves (i.e. two `dips') and Fig. \ref{FigC''} illustrates the interaction of one elevated wave and one depressive wave.


If we choose $q_{0}=2$, $\eta_{1}=\sqrt{3}$, $\theta=0$, $c_{1}=\sqrt{\frac{p_{2}-p_{1}}{p_{2}+p_{1}}}$ and $e^{2i\theta}=\frac{-p_{3}-ip_{4}}{p_{3}-ip_{4}}$, we can get the general formula in \cite{He}.

\begin{figure}[h]
\begin{tabular}{cc}
\includegraphics[width=0.5\textwidth]{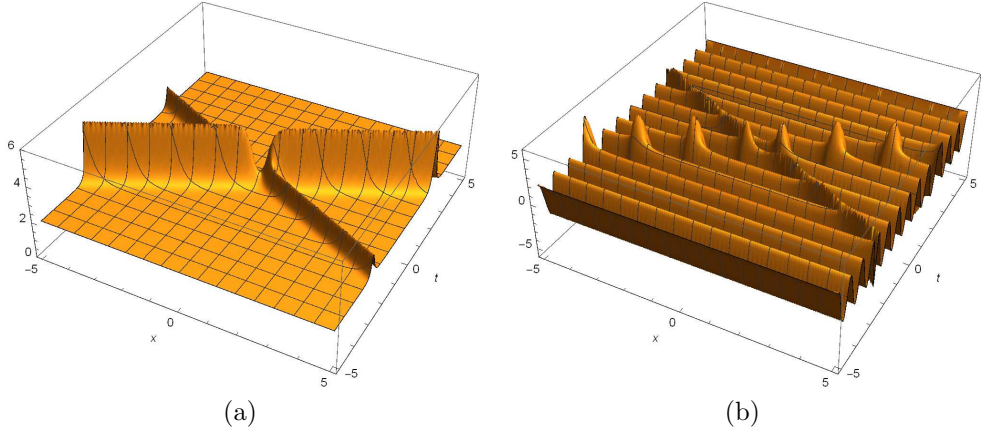}&
\includegraphics[width=0.5\textwidth]{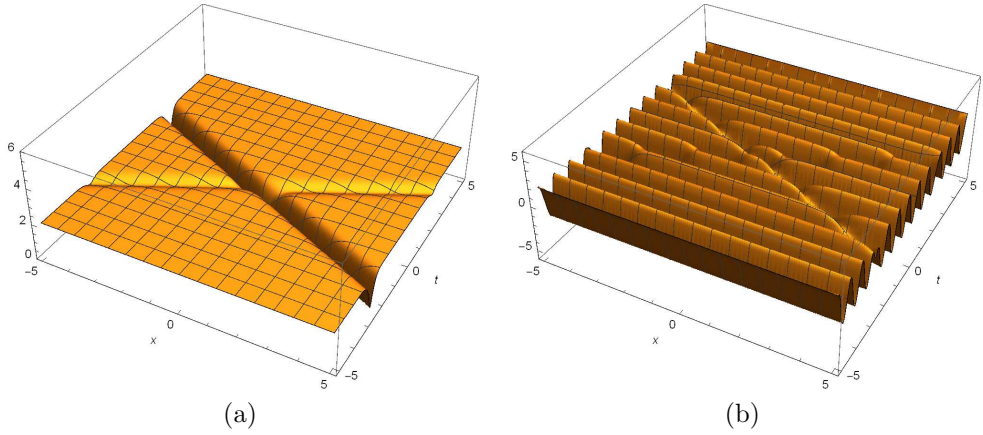}\\
(a) & (b)
\end{tabular}
\caption{(a) The amplitude of $q(x,t)$ with $\theta=0$, $\theta_{1}=-\frac{3\pi}{4}$, $c_{1}=\sqrt{2}$, $\eta_{1}=\sqrt{3}$ and $q_{0}=2$.  (b) The real part of $q(x,t)$ with $\theta=0$, $\theta_{1}=-\frac{3\pi}{4}$, $c_{1}=\sqrt{2}$, $\eta_{1}=\sqrt{3}$ and $q_{0}=2$. }
\label{FigC}
\end{figure}
\begin{figure}[h]
\begin{tabular}{cc}
\includegraphics[width=0.5\textwidth]{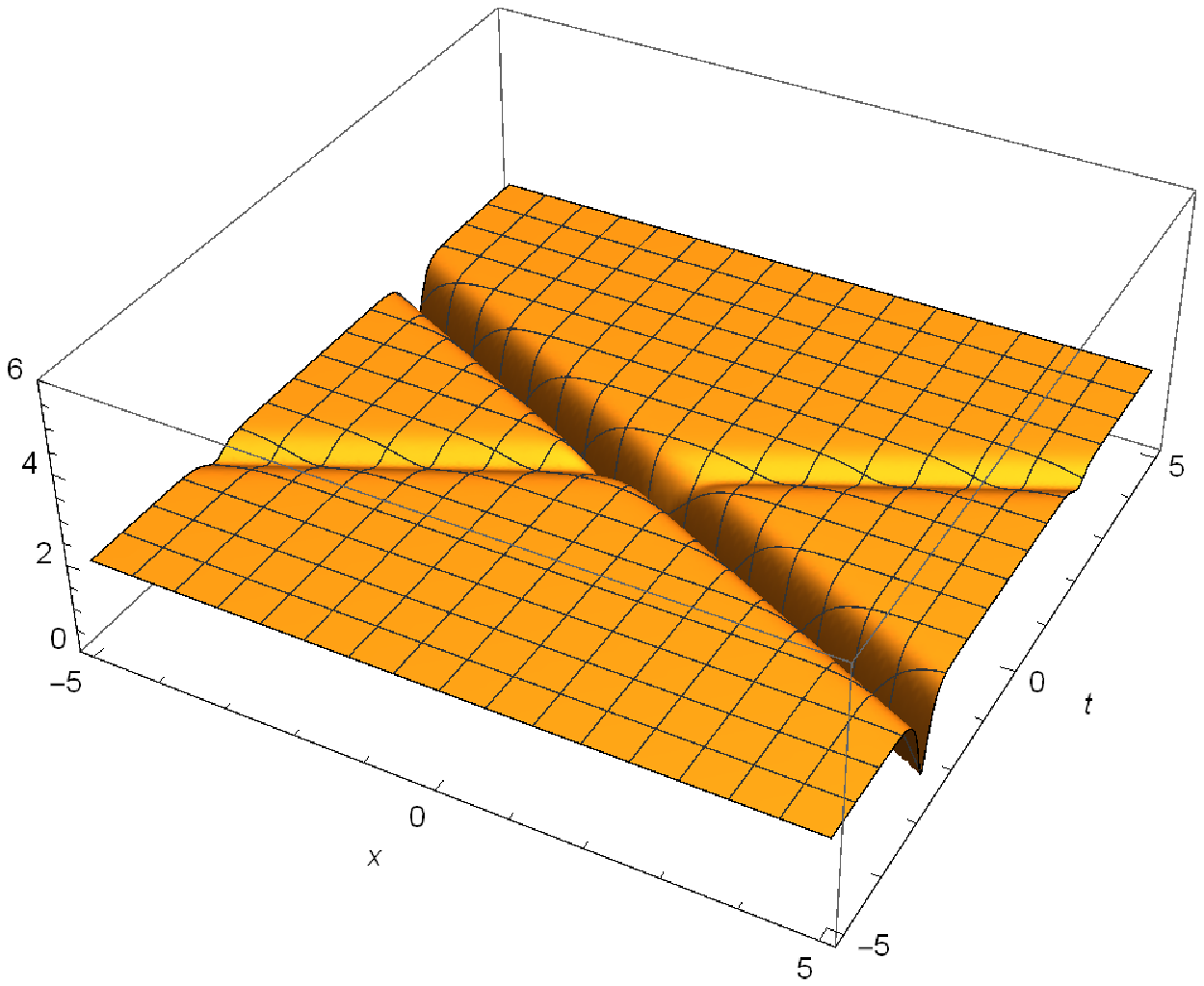}&
\includegraphics[width=0.5\textwidth]{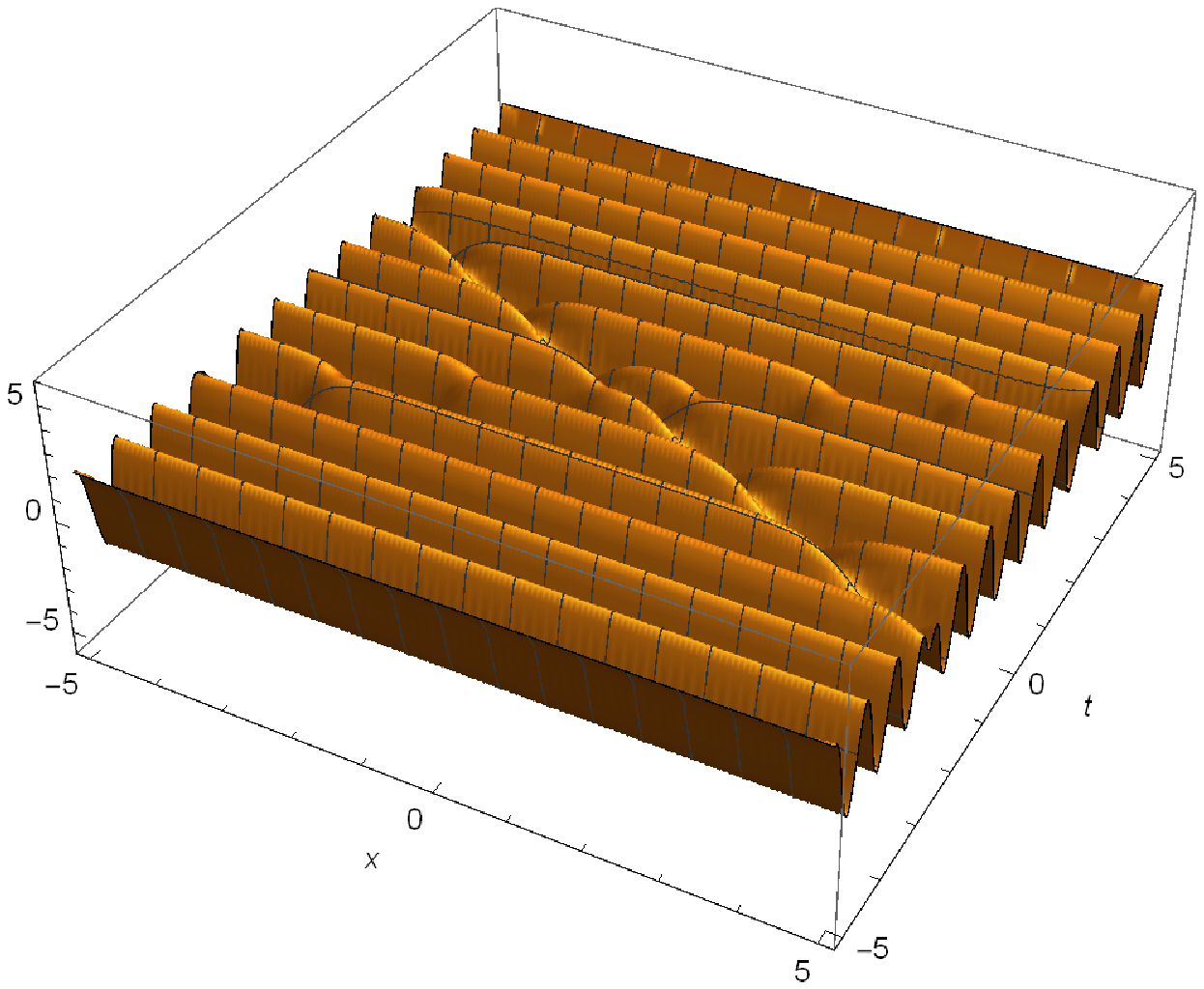}\\
(a) & (b)
\end{tabular}
\caption{(a) The amplitude of $q(x,t)$ with $\theta=0$, $\theta_{1}=\frac{3\pi}{4}$, $c_{1}=\sqrt{2}$, $\eta_{1}=\sqrt{3}$ and $q_{0}=2$.  (b) The real part of $q(x,t)$ with $\theta=0$, $\theta_{1}=\frac{3\pi}{4}$, $c_{1}=\sqrt{2}$, $\eta_{1}=\sqrt{3}$ and $q_{0}=2$. }
\label{FigC'}
\end{figure}
\begin{figure}[h]
\begin{tabular}{cc}
\includegraphics[width=0.5\textwidth]{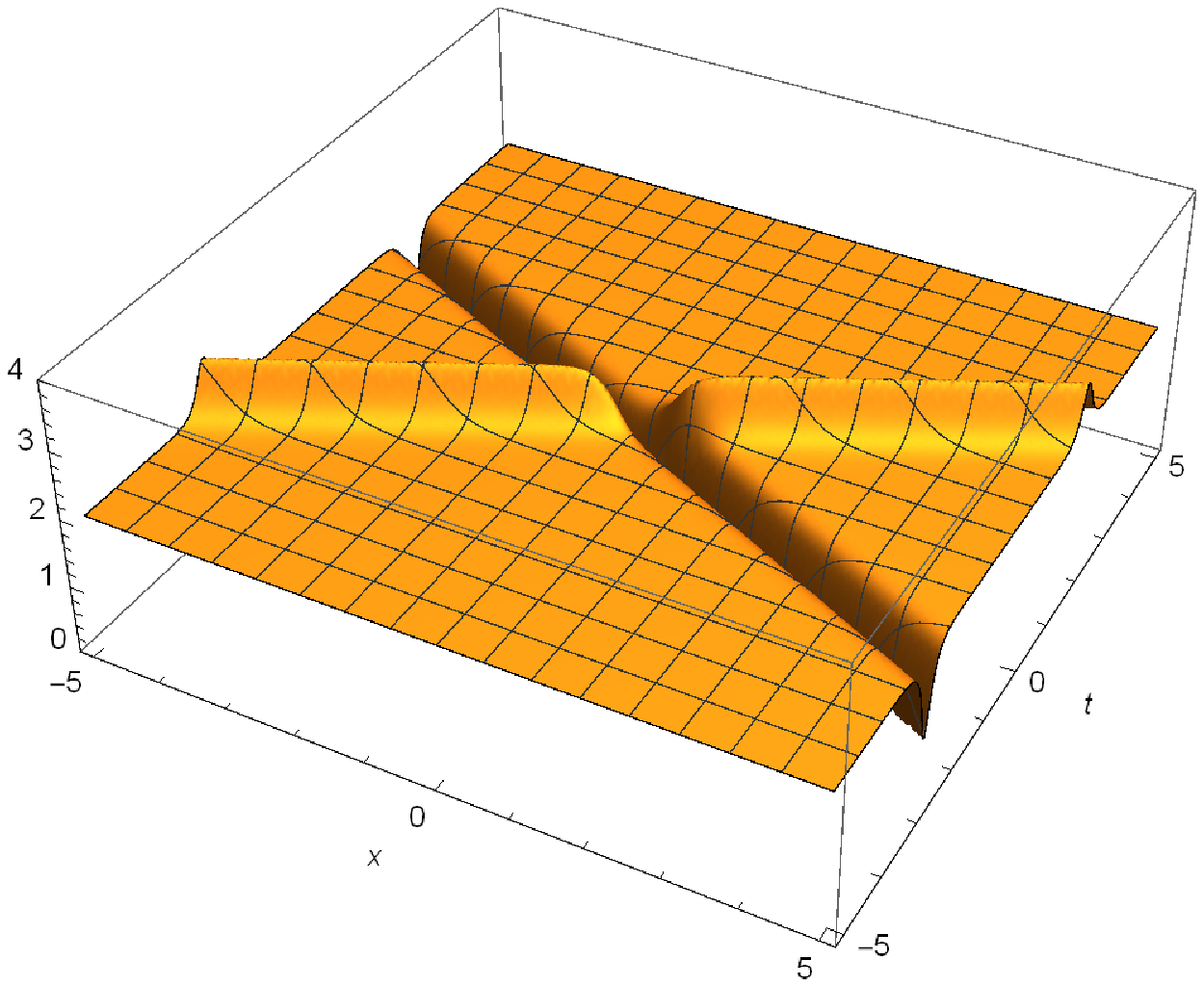}&
\includegraphics[width=0.5\textwidth]{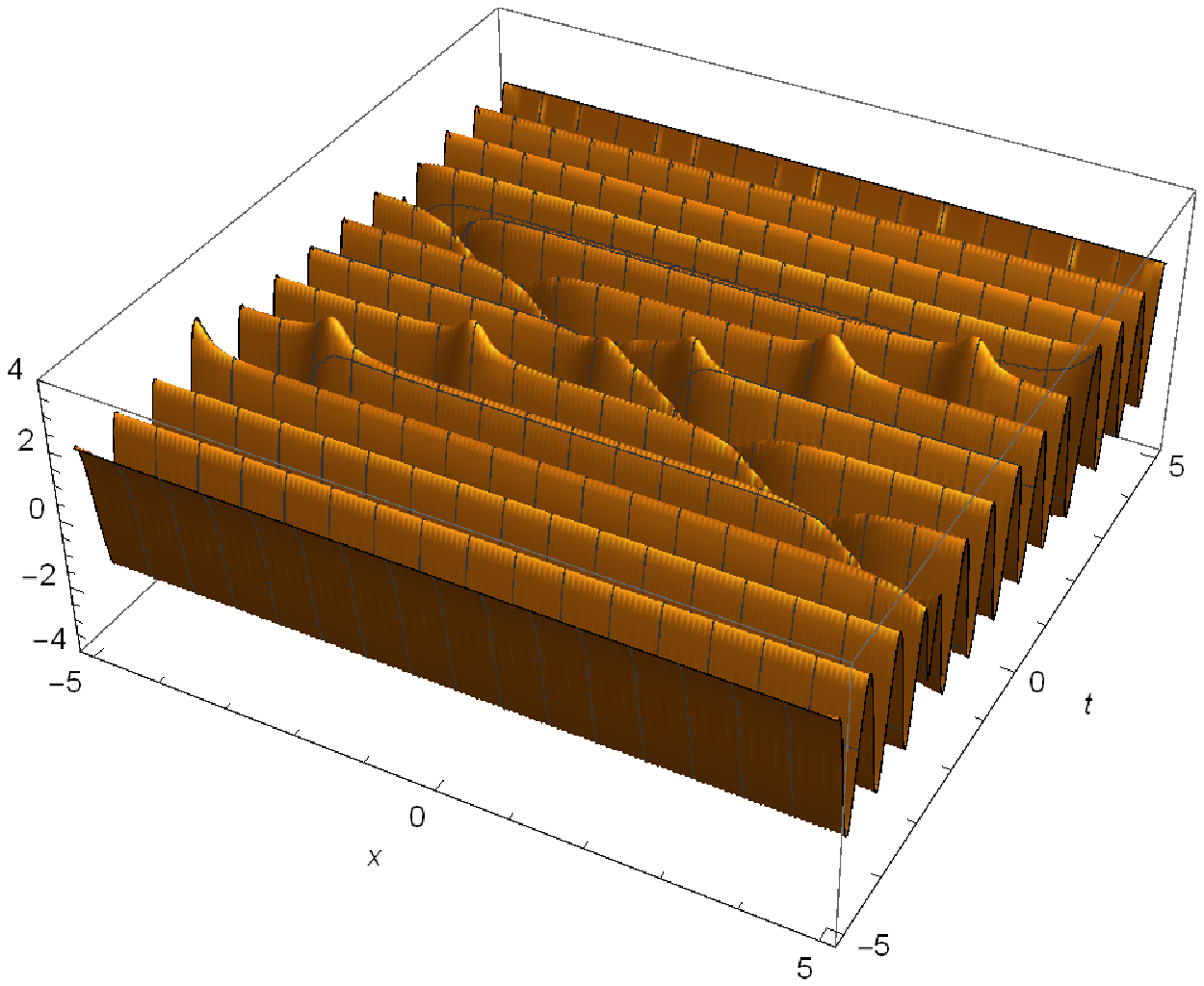}\\
(a) & (b)
\end{tabular}
\caption{(a) The amplitude of $q(x,t)$ with $\theta=0$, $\theta_{1}=\pi$, $c_{1}=\sqrt{2}$, $\eta_{1}=\sqrt{3}$ and $q_{0}=2$.  (b) The real part of $q(x,t)$ with $\theta=0$, $\theta_{1}=\pi$, $c_{1}=\sqrt{2}$, $\eta_{1}=\sqrt{3}$ and $q_{0}=2$. }
\label{FigC''}
\end{figure}

\section{The case of $\sigma=1$ with $\theta_{+}-\theta_{-}=\pi$}
The final case we will study is the nonlocal nonlinear Schr\"{o}dinger (NLS) equation (\ref{E:nonlocal NLS}) with $\sigma=1$
\begin{equation}
\label{E:nonlocal focusing NLS}
iq_{t}(x,t)=q_{xx}(x,t)-2q^{2}(x,t)q^{*}(-x,t)
\end{equation}
with nonzero boundary conditions (NZBCs)
\begin{equation}
q(x,t)\rightarrow q_{\pm}(t)=q_{0}e^{-2iq_{0}^{2}t+i\theta_{\pm}},\ \ as \ \ x\rightarrow\pm\infty,
\end{equation}
where $q_{0}>0$, $0\leq \theta_{\pm}<2\pi$, $\theta_{+}-\theta_{-}=\pi$.

The analysis of this case is similar to the NZBCs studied in section 4. As discussed in section 4, it is natural to introduce the eigenfunctions defined by the following boundary conditions
\begin{equation}
\phi(x,k)\sim w e^{-i\lambda x}, \ \ \ \overline{\phi}(x,k)\sim \overline{w}e^{i\lambda x}
\end{equation}
as $x\rightarrow-\infty$,
\begin{equation}
\psi(x,k)\sim v e^{i\lambda x}, \ \ \ \overline{\psi}(x,k)\sim \overline{v}e^{-i\lambda x}
\end{equation}
as $x\rightarrow +\infty$, where
\begin{equation}
w=\left(\begin{array}{cc}
\lambda+k\\
iq_{+}^{*}
\end{array}\right), \ \ \
\overline{w}=\left(\begin{array}{cc}
-iq_{-}\\
\lambda+k
\end{array}\right),
\end{equation}
\begin{equation}
v=\left(\begin{array}{cc}
-iq_{+}\\
\lambda+k
\end{array}\right), \ \ \
\overline{v}=
\left(\begin{array}{cc}
\lambda+k\\
iq_{-}^{*}
\end{array}\right)
\end{equation}
satisfy the boundary conditions. We define the bounded eigenfunctions as follows:
\begin{equation}
M(x,k)=e^{i\lambda x}\phi(x,k), \ \ \ \overline{M}(x,k)=e^{-i\lambda x}\overline{\phi}(x,k),
\end{equation}
\begin{equation}
N(x,k)=e^{-i\lambda x}\psi(x,k), \ \ \ \overline{N}(x,k)=e^{i\lambda x}\overline{\psi}(x,k).
\end{equation}
The Jost functions can be represented by means of the following integral equations
\begin{equation}
M(x,k)=
\left(\begin{array}{cc}
\lambda+k\\
iq_{+}^{*}
\end{array}\right)
+\int_{-\infty}^{+\infty}G_{-}(x-x',k)((Q-Q_{-})M)(x',k)dx',
\end{equation}
\begin{equation}
\overline{M}(x,k)=
\left(\begin{array}{cc}
-iq_{-}\\
\lambda+k
\end{array}\right)
+\int_{-\infty}^{+\infty}\overline{G}_{-}(x-x',k)((Q-Q_{-})M)(x',k)dx',
\end{equation}
\begin{equation}
N(x,k)=
\left(\begin{array}{cc}
-iq_{+}\\
\lambda+k
\end{array}\right)
+\int_{-\infty}^{+\infty}G_{+}(x-x',k)((Q-Q_{+})M)(x',k)dx',
\end{equation}
\begin{equation}
\overline{N}(x,k)=\left(\begin{array}{cc}
\lambda+k\\
iq_{-}^{*}
\end{array}\right)
+\int_{-\infty}^{+\infty}\overline{G}_{+}(x-x',k)((Q-Q_{+})M)(x',k)dx',
\end{equation}
where
\begin{equation}
Q_{\pm}(t)=\left(\begin{array}{cc}
0& q_{\pm}(t)\\
q^{*}_{\mp}(t)& 0
\end{array}\right).
\end{equation}
Using the Fourier transform method, we get
\begin{equation}
G_{-}(x,k)=\frac{\theta(x)}{2\lambda}[(1+e^{2i\lambda x})\lambda I-i(e^{2i\lambda x}-1)(ikJ+Q_{-})],
\end{equation}
\begin{equation}
\overline{G}_{-}(x,k)=\frac{\theta(x)}{2\lambda}[(1+e^{-2i\lambda x})\lambda I+i(e^{-2i\lambda x}-1)(ikJ+Q_{-})],
\end{equation}
\begin{equation}
G_{+}(x,k)=-\frac{\theta(-x)}{2\lambda}[(1+e^{-2i\lambda x})\lambda I+i(e^{-2i\lambda x}-1)(ikJ+Q_{+})],
\end{equation}
\begin{equation}
\overline{G}_{+}(x,k)=-\frac{\theta(-x)}{2\lambda}[(1+e^{2i\lambda x})\lambda I-i(e^{2i\lambda x}-1)(ikJ+Q_{+})],
\end{equation}
where $\theta(x)$ is the Heaviside function, i.e., $\theta(x)=1$ if $x>0$ and $\theta(x)=0$ if $x<0$.
The analyticity of eigenfunctions and definitions of scattering data are the same as the second case.
\subsection{Symmetry reductions}
Taking into account boundary conditions, we can obtain
\begin{equation}
\psi(x,k)=\left(\begin{array}{cc}
0& 1\\
-1& 0
\end{array}\right)\phi^{*}(-x,-k^{*})
\end{equation}
and
\begin{equation}
\overline{\psi}(x,k)=\left(\begin{array}{cc}
0& -1\\
1& 0
\end{array}\right)\overline{\phi}^{*}(-x,-k^{*}).
\end{equation}
Similarly, we can get
\begin{equation}
N(x,k)=\left(\begin{array}{cc}
0& 1\\
-1& 0
\end{array}\right)M^{*}(-x,-k^{*})
\end{equation}
and
\begin{equation}
\label{NbMb6}
\overline{N}(x,k)=\left(\begin{array}{cc}
0& -1\\
1& 0
\end{array}\right)\overline{M}^{*}(-x,-k^{*}).
\end{equation}
Moreover,
\begin{equation}
a^{*}(-k^{*})=a(k),
\end{equation}
\begin{equation}
\overline{a}^{*}(-k^{*})=\overline{a}(k),
\end{equation}
\begin{equation}
b^{*}(-k^{*})=\overline{b}(k).
\end{equation}

Similarly, we can define the eigenfunctions and scattering data on the left/right edge of the cut for $k\in (-iq_{0}, iq_{0})$. Then
\begin{equation}
\phi^{\mp}(x,k)=\frac{\lambda^{\mp}+k}{-iq_{-}}\cdot\overline{\phi}^{\pm}(x,k), \ \ \ \psi^{\mp}(x,k)=\frac{\lambda^{\mp}+k}{i\widetilde{q}_{-}^{*}}\cdot\overline{\psi}^{\pm}(x,k)
\end{equation}
for $k\in (-iq_{0}, iq_{0})$. Moreover, we have
\begin{equation}
a^{\pm}(k)=-\overline{a}^{\mp}(k), \ \ \ \ \ b^{\pm}(k)=\frac{q_{0}^{2}}{q_{+}\cdot q_{-}}\cdot \overline{b}^{\mp}(k)
\end{equation}
for $k\in (-iq_{0}, iq_{0})$.

\subsection{Uniformization coordinates}
Before discussing the properties of scattering data and solving the inverse problem, we introduce a uniformization variable $z$, defined by the conformal mapping:
\begin{equation}
z=z(k)=k+\lambda(k),
\end{equation}
and the inverse mapping is given by
\begin{equation}
k=k(z)=\frac{1}{2}\left(z-\frac{q_{0}^{2}}{z}\right).
\end{equation}
Then
\begin{equation}
\lambda(z)=\frac{1}{2}\left(z+\frac{q_{0}^{2}}{z}\right).
\end{equation}
\subsection{Symmetries via uniformization coordinates}
It is known that (1) when $z\rightarrow -z^{*}$, then $(k,\lambda)\rightarrow (-k^{*}, -\lambda^{*})$; (2) when $z\rightarrow -\frac{q_{0}^{2}}{z}$, then $(k, \lambda)\rightarrow (k, -\lambda)$. Hence,
\begin{equation}
\psi(x,z)=\left(\begin{array}{cc}
0& 1\\
-1& 0
\end{array}\right)\phi^{*}(-x,-z^{*}),
\end{equation}
\begin{equation}
\overline{\psi}(x,z)=\left(\begin{array}{cc}
0& -1\\
1& 0
\end{array}\right)\overline{\phi}^{*}(-x,-z^{*}),
\end{equation}
\begin{equation}
\phi\left(x,-\frac{q_{0}^{2}}{z}\right)=\frac{\frac{q_{0}^{2}}{z}}{iq_{-}}\overline{\phi}(x,z), \ \ \ \psi\left(x,-\frac{q_{0}^{2}}{z}\right)=\frac{-iq_{+}}{z}\overline{\psi}(x,z), \ \ \ z\in D^{-}.
\end{equation}
Similarly, we can get
\begin{equation}
N(x,z)=\left(\begin{array}{cc}
0& 1\\
-1& 0
\end{array}\right)M^{*}(-x,-z^{*})
\end{equation}
and
\begin{equation}
\overline{N}(x,z)=\left(\begin{array}{cc}
0& -1\\
1& 0
\end{array}\right)\overline{M}^{*}(-x,-z^{*}).
\end{equation}

Moreover,
\begin{equation}
a^{*}(-z^{*})=a(z),
\end{equation}
\begin{equation}
\overline{a}^{*}(-z^{*})=\overline{a}(z),
\end{equation}
\begin{equation}
b^{*}(-z^{*})=\overline{b}(z),
\end{equation}
\begin{equation}
a\left(-\frac{q_{0}^{2}}{z}\right)=-\overline{a}(z), \ \ \ z\in D^{-}, \ \ \ b\left(-\frac{q_{0}^{2}}{z}\right)=\frac{q_{0}^{2}}{q_{+}\cdot q_{-}}\cdot\overline{b}(z),
\end{equation}
where $D^{-}$ is the white regions of Figure \ref{fig5}.

\subsection{Asymptotic behavior of eigenfunctions and scattering data}
In order to solve the inverse problem, one has to determine the asymptotic behavior of eigenfunctions and scattering data both as $z\rightarrow\infty$ in $\mathbb{K}_{1}$ and as $z\rightarrow 0$ in $\mathbb{K}_{2}$. We have
\begin{equation}
M(x,z)\sim\left\{\begin{array}{ll}
\left(\begin{array}{cc}
z\\
iq^{*}(-x)
\end{array}\right), \ \ \ z\rightarrow\infty\\
\left(\begin{array}{cc}
-z\cdot\frac{q(x)}{q_{+}}\\
-iq_{-}^{*}
\end{array}\right), \ \ \ z\rightarrow 0,\\
\end{array}\right.
\end{equation}

\begin{equation}
N(x,z)\sim\left\{\begin{array}{ll}
\left(\begin{array}{cc}
-iq(x)\\
z
\end{array}\right), \ \ \ z\rightarrow\infty\\
\left(\begin{array}{cc}
-iq_{+}\\
z\cdot \frac{q^{*}(-x)}{q_{-}^{*}}
\end{array}\right), \ \ \ z\rightarrow 0,\\
\end{array}\right.
\end{equation}

\begin{equation}
\overline{M}(x,z)\sim\left\{\begin{array}{ll}
\left(\begin{array}{cc}
-iq(x)\\
z
\end{array}\right), \ \ \ z\rightarrow\infty\\
\left(\begin{array}{cc}
iq_{+}\\
-z\cdot \frac{q^{*}(-x)}{q_{-}^{*}}
\end{array}\right), \ \ \ z\rightarrow 0,\\
\end{array}\right.
\end{equation}

\begin{equation}
\overline{N}(x,z)\sim\left\{\begin{array}{ll}
\left(\begin{array}{cc}
z\\
iq^{*}(-x)
\end{array}\right), \ \ \ z\rightarrow\infty\\
\left(\begin{array}{cc}
z\cdot\frac{q(x)}{q_{+}}\\
iq_{-}^{*}
\end{array}\right), \ \ \ z\rightarrow 0,\\
\end{array}\right.
\end{equation}

\begin{equation}
a(z)=
\left\{\begin{array}{ll}
1,\ \ \ z\rightarrow\infty,\\
-1, \ \ \ z\rightarrow 0,\\
\end{array}\right.
\end{equation}

\begin{equation}
\overline{a}(z)=
\left\{\begin{array}{ll}
1,\ \ \ z\rightarrow\infty,\\
-1, \ \ \ z\rightarrow 0,\\
\end{array}\right.
\end{equation}
\begin{equation}
\lim_{z\rightarrow\infty}zb(z)=0, \ \ \ \lim_{z\rightarrow0}\frac{b(z)}{z^{2}}=0.
\end{equation}

\subsection{Left and right scattering problems}
Using the methods in the second case, we can get
\begin{equation}
\begin{split}
\overline{N}(x,z)&=\left(\begin{array}{cc}
z\\
iq_{-}^{*}
\end{array}\right)
+\sum_{j=1}^{J}\frac{z\cdot b(z_{j})e^{i\big(z_{j}+\frac{q_{0}^{2}}{z_{j}}\big)x}\cdot N(x,z_{j})}{(z-z_{j})z_{j}a'(z_{j})}\\
&+\frac{z}{2\pi i}\int_{\Sigma}\frac{\rho(\xi)}{\xi(\xi-z)}\cdot e^{i\big(\xi+\frac{q_{0}^{2}}{\xi}\big)x}\cdot N(x,\xi)d\xi,
\end{split}
\end{equation}
\begin{equation}
\begin{split}
N(x,z)&=\left(\begin{array}{cc}
-iq_{+}\\
z
\end{array}\right)
+\sum_{j=1}^{\overline{J}}\frac{z\cdot\overline{b}(\overline{z}_{j})e^{-i\big(\overline{z}_{j}+\frac{q_{0}^{2}}{\overline{z}_{j}}\big)x}\cdot \overline{N}(x,\overline{z}_{j})}{(z-\overline{z}_{j})\overline{z}_{j}\overline{a}'(\overline{z}_{j})}\\
&-\frac{z}{2\pi i}\int_{\Sigma}\frac{\overline{\rho}(\xi)}{\xi(\xi-z)}\cdot e^{-i\big(\xi+\frac{q_{0}^{2}}{\xi}\big)x}\cdot \overline{N}(x,\xi)d\xi,
\end{split}
\end{equation}
\begin{equation}
\begin{split}
\overline{M}(x,z)&=\left(\begin{array}{cc}
-iq_{-}\\
z
\end{array}\right)+
\sum_{j=1}^{J}\frac{-z\cdot\overline{b}(z_{j})M(x,z_{j})e^{-i\big(z_{j}+\frac{q_{0}^{2}}{z_{j}}\big)x}}{(z-z_{j})z_{j}a'(z_{j})}\\
&+\frac{z}{2\pi i}\int_{\Sigma}\frac{\rho^{*}(-\xi^{*})}{\xi(\xi-z)}\cdot e^{-i\big(\xi+\frac{q_{0}^{2}}{\xi}\big)x}\cdot M(x,\xi)d\xi
\end{split}
\end{equation}
and
\begin{equation}
\begin{split}
M(x,z)&=\left(\begin{array}{cc}
z\\
iq_{+}^{*}
\end{array}\right)+
\sum_{j=1}^{\overline{J}}\frac{-z\cdot b(\overline{z}_{j})\overline{M}(x,\overline{z}_{j})e^{i\big(\overline{z}_{j}+\frac{q_{0}^{2}}{\overline{z}_{j}}\big)x}}
{(z-\overline{z}_{j})\overline{z}_{j}\overline{a}'(\overline{z}_{j})}\\
&-\frac{z}{2\pi i}\int_{\Sigma}\frac{\overline{\rho}^{*}(-\xi^{*})}{\xi(\xi-z)}\cdot e^{i\big(\xi+\frac{q_{0}^{2}}{\xi}\big)x}\cdot \overline{M}(x,\xi)d\xi
\end{split}
\end{equation}
where we recall from the second case:  $\Sigma:=(-\infty,-q_{0})\cup (q_{0}, +\infty)\cup \overrightarrow{(q_{0}, -q_{0})}\cup \{q_{0}e^{i\theta}, \pi\leq \theta\leq 2\pi\}_{clockwise, upper \ circle}\cup \{q_{0}e^{i\theta}, -\pi\leq \theta\leq 0\}_{anticlockwise, lower \ circle}$.


\subsection{Recovery of the potentials}
Note that $N_{1}(x,z)\sim -iq(x)$ as $z\rightarrow\infty$, and
\begin{equation}
N_{1}(x,z)\sim -iq_{+}+\sum_{j=1}^{\overline{J}}\frac{\overline{b}(\overline{z}_{j})e^{-i\big(\overline{z}_{j}+\frac{q_{0}^{2}}{\overline{z}_{j}}\big)x}\cdot \overline{N}_{1}(x,\overline{z}_{j})}{\overline{z}_{j}\overline{a}'(\overline{z}_{j})}
+\frac{1}{2\pi i}\int_{\Sigma}\frac{\overline{\rho}(\xi)}{\xi}\cdot e^{-i\big(\xi+\frac{q_{0}^{2}}{\xi}\big)x}\cdot \overline{N}_{1}(x,\xi)d\xi,
\end{equation}
we have
\begin{equation}
\label{asympN1c6}
q(x)=q_{+}+i\sum_{j=1}^{\overline{J}}\frac{\overline{b}(\overline{z}_{j})e^{-i\big(\overline{z}_{j}+\frac{q_{0}^{2}}{\overline{z}_{j}}\big)x}\cdot \overline{N}_{1}(x,\overline{z}_{j})}{\overline{z}_{j}\overline{a}'(\overline{z}_{j})}+\frac{1}{2\pi}\int_{\Sigma}\frac{\overline{\rho}(\xi)}{\xi}\cdot e^{-i\big(\xi+\frac{q_{0}^{2}}{\xi}\big)x}\cdot \overline{N}_{1}(x,\xi)d\xi.
\end{equation}
\subsection{Trace formula}
As in the prior section we can show that
\begin{equation}
b(z_{j})b^{*}(-z_{j}^{*})=-1,
\end{equation}
which implies that $\Re z_{j}\neq 0$ and $J$ is even. Hence,
we can get the Trace formula as follows.
\begin{equation}
\log a(z)=\log\left(\prod_{j=1}^{J/2}\frac{z-z_{j}}{z+\frac{q_{0}^{2}}{z_{j}}}\cdot\frac{z+z_{j}^{*}}{z-\frac{q_{0}^{2}}{z_{j}^{*}}}
\right)+\frac{1}{2\pi i}\int_{\Sigma}\frac{\log (1+b(\xi)b^{*}(-\xi^{*}))}{\xi-z}d\xi, \ \ \ z\in D^{+}
\end{equation}
and
\begin{equation}
\log \overline{a}(z)=\log\left(\prod_{j=1}^{J/2} \frac{z+\frac{q_{0}^{2}}{z_{j}}}{z-z_{j}}\cdot \frac{z-\frac{q_{0}^{2}}{z_{j}^{*}}}{z+z_{j}^{*}}
\right)-\frac{1}{2\pi i}\int_{\Sigma}\frac{\log (1+b(\xi)b^{*}(-\xi^{*}))}{\xi-z}d\xi, \ \ \ z\in D^{-},
\end{equation}
where $\Re z_{j}\neq 0$.

We claim that this case does  not admit soliton solutions.] Indeed, by the asymptotic of $a(z)$, i.e.,  $a(z)\sim -1$ as $z\rightarrow 0$ and the fact $\Re z_{j}\neq 0$.
Moreover, the Trace formula yields
\begin{equation}
a(z)=\prod_{j=1}^{J/2}\frac{z-z_{j}}{z+\frac{q_{0}^{2}}{z_{j}}}\cdot \frac{z+z_{j}^{*}}{z-\frac{q_{0}^{2}}{z_{j}^{*}}} \mbox{exp}\left(\frac{1}{2\pi i}\int_{\Sigma}\frac{\log (1+b(\xi)b^{*}(-\xi^{*}))}{\xi-z}d\xi \right).
\end{equation}

Since the exponential term is real for all $\xi\in \Sigma$, it implies that

$$\prod_{j=1}^{J/2} \frac{|z_{j}|^{4}}{q_{0}^{4}}\mbox{exp}\left(\frac{1}{2\pi i}\int_{\Sigma}\frac{\log (1+b(\xi)b^{*}(-\xi^{*}))}{\xi}d\xi \right)=-1,$$
which is a contradiction. Thus, in this case there are no  solitons.

Indeed, we consider the integral $I:=\frac{1}{2\pi i}\int_{\Sigma}\frac{\log (1+b(\xi)b^{*}(-\xi^{*}))}{\xi}d\xi$, we can show that $I^{*}=I$. Set $I_{1}=\frac{1}{2\pi i}\int_{q_{0}}^{+\infty}\frac{\log(1+b(\xi)b^{*}(-\xi^{*}))}{\xi}d\xi$, $I_{2}=\frac{1}{2\pi i}\int_{-\infty}^{-q_{0}}\frac{\log(1+b(\xi)b^{*}(-\xi^{*}))}{\xi}d\xi$, $I_{3}=\frac{1}{2\pi i}\int_{q_{0}}^{-q_{0}}\frac{\log(1+b(\xi)b^{*}(-\xi^{*}))}{\xi}d\xi$, $I_{4}=\frac{1}{2\pi i}\int_{upper \ circle |\xi|=q_{0}, clockwise }\frac{\log(1+b(\xi)b^{*}(-\xi^{*}))}{\xi}d\xi$ and $I_{5}=\frac{1}{2\pi i}\int_{lower \ circle |\xi|=q_{0}, anticlockwise }\frac{\log(1+b(\xi)b^{*}(-\xi^{*}))}{\xi}d\xi$, then $I=I_{1}+I_{2}+I_{3}+I_{4}+I_{5}$. Moreover,
\begin{equation}
I_{1}=\frac{1}{2\pi i}\int_{-q_{0}}^{-\infty}\frac{\log(1+b(-\xi')b^{*}(\xi'))}{-\xi'}d(-\xi') \ \ \ (\xi=-\xi')
=\frac{1}{2\pi i}\int_{-q_{0}}^{-\infty}\frac{\log(1+b(-\xi')b^{*}(\xi'))}{\xi'}d\xi'.
\end{equation}
Thus,
\begin{equation}
I_{1}^{*}=-\frac{1}{2\pi i}\int_{-q_{0}}^{-\infty}\frac{\log(1+b^{*}(-\xi')b(\xi'))}{\xi'}d\xi'
=\frac{1}{2\pi i}\int_{-\infty}^{-q_{0}}\frac{\log(1+b(\xi')b^{*}(-\xi'))}{\xi'}d\xi'=I_{2}.
\end{equation}
Similarly, we have $I_{2}^{*}=I_{1}$ and $I_{3}^{*}=I_{3}$. Note that
\begin{equation}
I_{4}=-\frac{1}{2\pi i}\int_{0}^{\pi}\frac{\log(1+b(q_{0}e^{i\theta})b^{*}(-q_{0}e^{-i\theta}))}{q_{0}e^{i\theta}}d(q_{0}e^{i\theta})
=-\frac{1}{2\pi}\int_{0}^{\pi}\log(1+b(q_{0}e^{i\theta})b^{*}(-q_{0}e^{-i\theta}))d\theta
\end{equation}
and
\begin{equation}
\begin{split}
I_{4}^{*}&=-\frac{1}{2\pi}\int_{0}^{\pi}\log(1+b^{*}(q_{0}e^{i\theta})b(-q_{0}e^{-i\theta}))d\theta\\
&=-\frac{1}{2\pi}\int_{0}^{-\pi}\log(1+b^{*}(q_{0}e^{-i\theta'})b(-q_{0}e^{i\theta'}))d(-\theta') \ \ \ (\theta'=-\theta)\\
&=\frac{1}{2\pi}\int_{0}^{-\pi}\log(1+b^{*}(q_{0}e^{-i\theta'})b(-q_{0}e^{i\theta'}))d\theta'\\
&=\frac{1}{2\pi}\int_{\pi}^{0}\log(1+b^{*}(-q_{0}e^{-i\theta''})b(q_{0}e^{i\theta''}))d\theta'' \ \ \ (\theta'=\theta''+\pi)\\
&=-\frac{1}{2\pi}\int_{0}^{\pi}\log(1+b^{*}(-q_{0}e^{-i\theta''})b(q_{0}e^{i\theta''}))d\theta''\\
&=I_{4}.
\end{split}
\end{equation}
Similarly, we have $I_{5}^{*}=I_{5}$. Hence, we have proved $I^{*}=I$, i.e., $I$ is real.

\subsection{Closing the system}
This time the equation will not have a discrete spectrum contribution, to close the system, by the symmetry relations between the eigenfunctions, we have
\begin{equation}
\begin{split}
\left(\begin{array}{cc}
\overline{N}_{1}(x,z)\\
\overline{N}_{2}(x,z)
\end{array}\right)&=\left(\begin{array}{cc}
z\\
iq_{-}^{*}
\end{array}\right)+\frac{z}{2\pi i}\int_{\Sigma}\frac{\rho(\xi)}{\xi(\xi-z)}\cdot e^{i\big(\xi+\frac{q_{0}^{2}}{\xi}\big)x}\cdot \\
&\left(\begin{array}{cc}
-iq_{+}-\frac{\xi}{2\pi i}\int_{\Sigma}\frac{\overline{\rho}(\eta)}{\eta(\eta-\xi)}\cdot e^{-i\big(\eta+\frac{q_{0}^{2}}{\eta}\big)x}\cdot \overline{N}_{1}(x,\eta)d\eta\\
\xi-\frac{\xi}{2\pi i}\int_{\Sigma}\frac{\overline{\rho}(\eta)}{\eta(\eta-\xi)}\cdot e^{-i\big(\eta+\frac{q_{0}^{2}}{\eta}\big)x}\cdot \overline{N}_{2}(x,\eta)d\eta
\end{array}\right)d\xi,
\end{split}
\end{equation}
\begin{equation}
\begin{split}
\left(\begin{array}{cc}
M_{1}(x,z)\\
M_{2}(x,z)
\end{array}\right)
&=\left(\begin{array}{cc}
z\\
iq_{+}^{*}
\end{array}\right)
-\frac{z}{2\pi i}\int_{\Sigma}\frac{\overline{\rho}^{*}(-\xi^{*})}{\xi(\xi-z)}\cdot e^{i\big(\xi+\frac{q_{0}^{2}}{\xi}\big)x}\cdot \\ &\left(\begin{array}{cc}
-iq_{-}+\frac{\xi}{2\pi i}\int_{\Sigma}\frac{\rho^{*}(-\eta^{*})}{\eta(\eta-\xi)}\cdot e^{-i\big(\eta+\frac{q_{0}^{2}}{\eta}\big)x}\cdot M_{1}(x,\eta)d\eta\\
\xi+\frac{\xi}{2\pi i}\int_{\Sigma}\frac{\rho^{*}(-\eta^{*})}{\eta(\eta-\xi)}\cdot e^{-i\big(\eta+\frac{q_{0}^{2}}{\eta}\big)x}\cdot M(x,\eta)d\eta
\end{array}\right)d\xi.
\end{split}
\end{equation}

\subsection{Time evolution}
Similar to the case in section 4, we can get
\begin{equation}
\frac{\partial a(t)}{\partial t}=0, \ \ \ \frac{\partial \overline{a}(t)}{\partial t}=0, \ \ \ \frac{\partial b(t)}{\partial t}=2i(q_{0}^{2}-2\lambda k)b(t), \ \ \ \frac{\partial \overline{b}(t)}{\partial t}=-2i(q_{0}^{2}-2\lambda k)\overline{b}(t).
\end{equation}
Then
\begin{equation}
\rho(z,t)=\rho(z,0)e^{2i[q_{0}^{2}-\frac{1}{2}(\xi^{2}-\frac{q_{0}^{4}}{\xi^{2}})]t}, \ \ \ \overline{\rho}(z,t)=\overline{\rho}(z,0)e^{-2i[q_{0}^{2}-\frac{1}{2}(\xi^{2}-\frac{q_{0}^{4}}{\xi^{2}})]t}.
\end{equation}
It yields
\begin{equation}
q(x,t)=q_{0}e^{-2iq_{0}^{2}t+i\theta_{+}}+\frac{1}{2\pi}\int_{\Sigma}
\frac{\overline{\rho}(\xi,0)e^{-2i[q_{0}^{2}-\frac{1}{2}(\xi^{2}-\frac{q_{0}^{4}}{\xi^{2}})]t}}{\xi}e^{-i(\xi+\frac{q_{0}^{2}}{\xi})x}\cdot \overline{N}_{1}(x,\xi,t)d\xi.
\end{equation}
In principle we can formulate integral equations for $bar{M}_2$ and via the symmetry (\ref{NbMb6}) obtain $q(x)$ as done in section 4. 

\section{Box-type initial conditions for nonlocal NLS}
In this section, we consider the box-type initial conditions as follows:
\begin{equation}
q(x,0)=\left\{\begin{array}{ll}
\begin{array}{cc}
q_{0}e^{i\theta_{+}}, \ \ \ x>L,\\
q_{0}e^{i\theta_{-}}, \ \ \ x<-L,\\
q_{c}e^{i\theta_{c+}}, \ \ \ 0\leq x\leq L,\\
q_{c}e^{i\theta_{c-}}, \ \ \ -L\leq x<0,
\end{array}\\
\end{array}\right.
\end{equation}
where $q_{0}, q_{c}>0$, $0\leq \theta_{+}, \theta_{-}, \theta_{c+}, \theta_{c-}< 2\pi$ and $L>0$.
Then we will consider the four cases as we discuss above.

{\bf Case 1.} $\sigma=-1$ with $\theta_{+}-\theta_{-}=\pi$.
The box-type initial conditions can be specified as follows:
\begin{equation}
q(x,0)=\left\{\begin{array}{ll}
\begin{array}{cc}
q_{0}e^{i\theta}, \ \ \ x>L,\\
-q_{0}e^{i\theta}, \ \ \ x<-L,\\
q_{c}e^{i\theta}, \ \ \ -L\leq x\leq L,
\end{array}\\
\end{array}\right.
\end{equation}
where $q_{0}, q_{c}>0$ and $0\leq \theta< 2\pi$. As we show in Section 3, the two sets of eigenfunctions $\{\phi(x,k), \overline{\phi}(x,k)\}$ and $\{\psi(x,k), \overline{\psi}(x,k)\}$ satisfy the boundary conditions
\begin{equation}
\phi(x,k)\sim \left(\begin{array}{cc}
\lambda_{1}+k\\
-iq_{0}e^{-i\theta}
\end{array}\right)e^{-i\lambda_{1}x}, \ \ \
\overline{\phi}(x,k)\sim\left(\begin{array}{cc}
iq_{0}e^{i\theta}\\
\lambda_{1}+k
\end{array}\right)e^{i\lambda_{1}x}
\end{equation}
as $x\rightarrow -\infty$,
\begin{equation}
\psi(x,k)\sim \left(\begin{array}{cc}
-iq_{0}e^{i\theta}\\
\lambda_{1}+k
\end{array}\right)e^{i\lambda_{1}x}, \ \ \
\overline{\psi}(x,k)\sim\left(\begin{array}{cc}
\lambda_{1}+k\\
iq_{0}e^{-i\theta}
\end{array}\right)e^{-i\lambda_{1}x}
\end{equation}
as $x\rightarrow +\infty$, where $\lambda_{1}=\sqrt{k^{2}-q_{0}^{2}}$. Note that
\begin{equation}
\phi(x,k)=b(k)\psi(x,k)+a(k)\overline{\psi}(x,k),
\end{equation}
we have
\begin{equation}
\phi(x,k)\sim \left(\begin{array}{cc}
-iq_{0}e^{i\theta}b(k)e^{i\lambda_{1}x}+(\lambda_{1}+k)a(k)e^{-i\lambda_{1}x}\\
(\lambda_{1}+k)b(k)e^{i\lambda_{1}x}+iq_{0}e^{-i\theta}a(k)e^{-i\lambda_{1}x}
\end{array}\right)
\end{equation}
as $x\rightarrow +\infty$. We write the real axis as $(-\infty, +\infty)=(-\infty, -L)\cup[-L, L)\cup[L, +\infty)$.
In $[-L, L]$, we can solve
\begin{equation}
\left(\begin{array}{cc}
v_{1}\\
v_{2}
\end{array}\right)_{x}
=
\left(\begin{array}{cc}
-ik& q_{c}e^{i\theta}\\
-q_{c}e^{-i\theta}& ik
\end{array}\right)
\left(\begin{array}{cc}
v_{1}\\
v_{2}
\end{array}\right):=M_{-L}\left(\begin{array}{cc}
v_{1}\\
v_{2}
\end{array}\right)
\end{equation}
for $x$ as $v(x)=T(X)v(-L)$, $T(X)= \exp(X M_{-L})$, where $X:=x+L$, and the matrix $T(X)$ is explicitly written as
\begin{equation}
T(X)=\left(\begin{array}{cc}
\cos \lambda_{4}X-i\frac{k}{\lambda_{4}}\sin \lambda_{4}X& \frac{q_{c}e^{i\theta}}{\lambda_{4}}\sin \lambda_{4}X\\
-\frac{q_{c}e^{-i\theta}}{\lambda_{4}}\sin \lambda_{4}X& \cos \lambda_{4}X+i\frac{k}{\lambda_{4}}\sin \lambda_{4}X
\end{array}\right),
\end{equation}
where $\lambda_{4}=\sqrt{k^{2}+q_{c}^{2}}$.
Thus,
\begin{equation}
v(L)=T(2L)v(-L).
\end{equation}
The matrix $T$ is interpreted as a transfer matrix that connects two asymptotic forms in $x\rightarrow +\infty$ and $x\rightarrow -\infty$. We can get
\begin{equation}
\begin{split}
&\left(\begin{array}{cc}
-iq_{0}e^{i\theta}b(k)e^{i\lambda_{1}L}+(\lambda_{1}+k)a(k)e^{-i\lambda_{1}L}\\
(\lambda_{1}+k)b(k)e^{i\lambda_{1}L}+iq_{0}e^{-i\theta}a(k)e^{-i\lambda_{1}L}
\end{array}\right)\\
&=\phi(L,k)=T(2L)\phi(-L,k)=T(2L)\left(\begin{array}{cc}
\lambda_{1}+k\\
-iq_{0}e^{-i\theta}
\end{array}\right)e^{i\lambda_{1}L},
\end{split}
\end{equation}
i.e.,
\begin{equation}
\begin{split}
&\left(\begin{array}{cc}
-iq_{0}e^{i\theta}b(k)e^{i\lambda_{1}L}+(\lambda_{1}+k)a(k)e^{-i\lambda_{1}L}\\
(\lambda_{1}+k)b(k)e^{i\lambda_{1}L}+iq_{0}e^{-i\theta}a(k)e^{-i\lambda_{1}L}
\end{array}\right)\\
&=\left(\begin{array}{cc}
\cos 2\lambda_{4}L-i\frac{k}{\lambda_{4}}\sin 2\lambda_{4}L& \frac{q_{c}e^{i\theta}}{\lambda_{4}}\sin 2\lambda_{4}L\\
-\frac{q_{c}e^{-i\theta}}{\lambda_{4}}\sin 2\lambda_{4}L& \cos 2\lambda_{4}L+i\frac{k}{\lambda_{4}}\sin 2\lambda_{4}L
\end{array}\right)\cdot\left(\begin{array}{cc}
\lambda_{1}+k\\
-iq_{0}e^{-i\theta}
\end{array}\right)e^{i\lambda_{1}L}.
\end{split}
\end{equation}
Then
\begin{equation}
\begin{split}
&a(k)=(e^{2i\lambda_{1}L}((q_{0}^{2}+(k+\lambda_{1})^{2})\lambda_{4}\cos(2\lambda_{4}L)-i(k^{3}+2k^{2}\lambda_{1}+2q_{0}q_{c}\lambda_{1}\\
&+k(-q_{0}^{2}+2q_{0}q_{c}+\lambda_{1}^{2}))\sin(2\lambda_{4}L)))/((k-q_{0}+\lambda_{1})(k+q_{0}+\lambda_{1})\lambda_{4}).
\end{split}
\end{equation}

Set $k=i\xi$, where $\xi$ is a real number, we will distinguish two subcases.

{\bf Subcase 1.} $|\xi|\geq q_{c}$. We have $\lambda_{1}=i\sqrt{\xi^{2}+q_{0}^{2}}:=i\widetilde{\lambda}_{1}$ and $\lambda_{4}=i\sqrt{\xi^{2}-q_{c}^{2}}:=i\widetilde{\lambda}_{4}$, where both $\widetilde{\lambda}_{1}$ and $\widetilde{\lambda}_{4}$ are positive real numbers. Then $a(i\xi)=0$ if and only if $\xi$ is a root of the following equation:
\begin{equation}
\frac{e^{-2\widetilde{\lambda}_{4}L}-e^{2\widetilde{\lambda}_{4}L}}{e^{-2\widetilde{\lambda}_{4}L}+e^{2\widetilde{\lambda}_{4}L}}
=\frac{(q_{0}^{2}-(\xi+\widetilde{\lambda}_{1})^{2})\widetilde{\lambda}_{4}}{-\xi^{3}-2\xi^{2}\widetilde{\lambda}_{1}+2q_{0}q_{c}\widetilde{\lambda}_{1}
+\xi(-q_{0}^{2}+2q_{0}q_{c}-\widetilde{\lambda}_{1}^{2})}.
\end{equation}
Numerically, we do not find solutions to the above equation.


{\bf Subcase 2.} $|\xi|< q_{c}$. We have $\lambda_{1}=i\sqrt{\xi^{2}+q_{0}^{2}}:=i\widetilde{\lambda}_{1}$ and $\lambda_{4}=\sqrt{-\xi^{2}+q_{c}^{2}}$, where both $\widetilde{\lambda}_{1}$ and $\lambda_{4}$ are positive real numbers. Then $a(i\xi)=0$ if and only if $\xi$ is a root of the following equation:
\begin{equation}
\label{E:box 1}
\tan(2\lambda_{4}L)=\frac{(q_{0}^{2}-(\xi+\widetilde{\lambda}_{1})^{2})\lambda_{4}}{\xi^{3}+2\xi^{2}\widetilde{\lambda}_{1}-2q_{0}q_{c}\widetilde{\lambda}_{1}
-\xi(-q_{0}^{2}+2q_{0}q_{c}-\widetilde{\lambda}_{1}^{2})}.
\end{equation}
Numerically we find solutions to the above equation (\ref{E:box 1}). In Fig. \ref{fig11} below we show one intersection for the given parameters.
\begin{figure}[h]
\begin{tabular}{cc}
\includegraphics[width=0.8\textwidth]{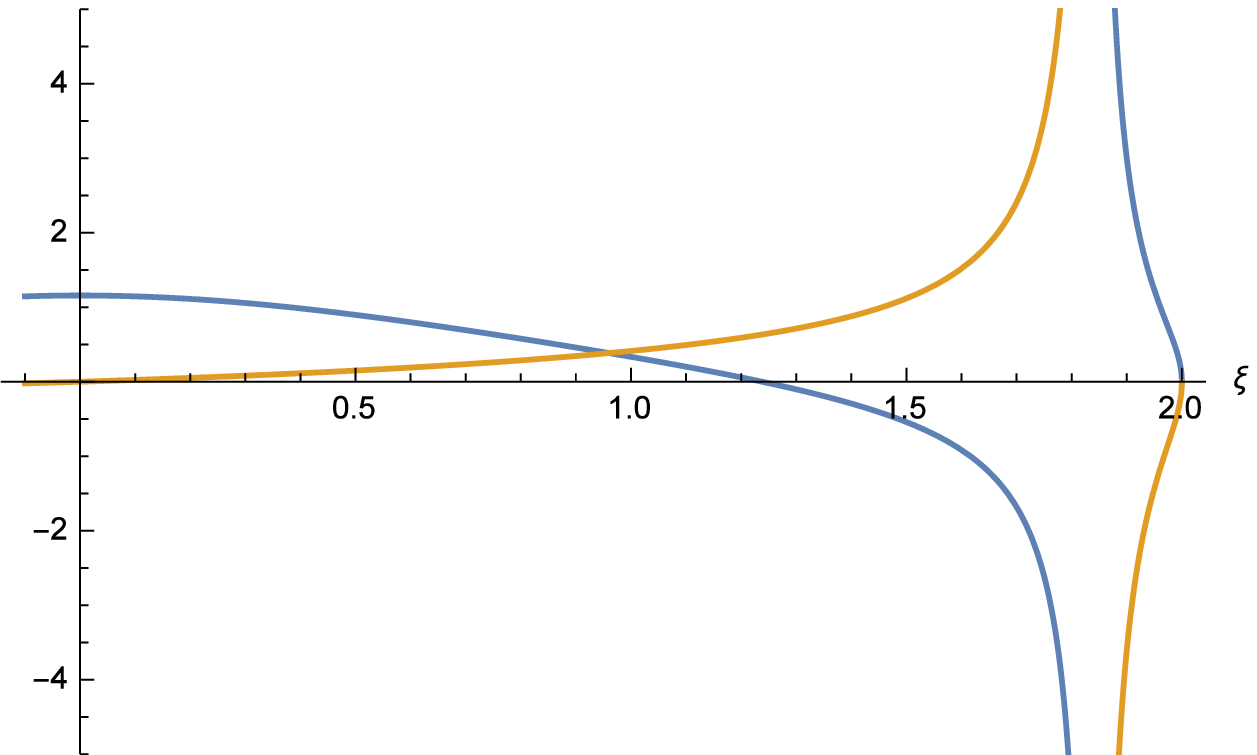}
\end{tabular}
\caption{The LHS (blue) and RHS (orange) of (\ref{E:box 1}) with $L=1$, $q_{0}=4.31$ and $q_{c}=2$, yielding one intersection between $0$ and $q_{c}$.}
\label{fig11}
\end{figure}

\begin{remark}
If we set a bigger $L$, then the LHS and RHS of (\ref{E:box 1}) can yield more intersections.
\end{remark}

{\bf Case 2.} $\sigma=-1$ with $\theta_{+}-\theta_{-}=0$.
The box-type initial conditions can be specified as follows:
\begin{equation}
q(x,0)=\left\{\begin{array}{ll}
\begin{array}{cc}
q_{0}e^{i\theta}, \ \ \ x\in (-\infty, -L)\cup (L, +\infty),\\
q_{c}e^{i\theta}, \ \ \ x\in [-L, L],
\end{array}\\
\end{array}\right.
\end{equation}
where $q_{0}, q_{c}>0$ and $0\leq \theta< 2\pi$. As we show in Section 4, the two sets of eigenfunctions $\{\phi(x,k), \overline{\phi}(x,k)\}$ and $\{\psi(x,k), \overline{\psi}(x,k)\}$ satisfy the boundary conditions
\begin{equation}
\phi(x,k)\sim \left(\begin{array}{cc}
\lambda_{3}+k\\
-iq_{0}e^{-i\theta}
\end{array}\right)e^{-i\lambda_{3}x}, \ \ \
\overline{\phi}(x,k)\sim\left(\begin{array}{cc}
-iq_{0}e^{i\theta}\\
\lambda_{3}+k
\end{array}\right)e^{i\lambda_{3}x}
\end{equation}
as $x\rightarrow -\infty$,
\begin{equation}
\psi(x,k)\sim \left(\begin{array}{cc}
-iq_{0}e^{i\theta}\\
\lambda_{3}+k
\end{array}\right)e^{i\lambda_{3}x}, \ \ \
\overline{\psi}(x,k)\sim\left(\begin{array}{cc}
\lambda_{3}+k\\
-iq_{0}e^{-i\theta}
\end{array}\right)e^{-i\lambda_{3}x}
\end{equation}
as $x\rightarrow +\infty$, where $\lambda_{3}=\sqrt{k^{2}+q_{0}^{2}}$.  Note that
\begin{equation}
\phi(x,k)=b(k)\psi(x,k)+a(k)\overline{\psi}(x,k),
\end{equation}
we have
\begin{equation}
\phi(x,k)\sim \left(\begin{array}{cc}
-iq_{0}e^{i\theta}b(k)e^{i\lambda_{3}x}+(\lambda_{3}+k)a(k)e^{-i\lambda_{3}x}\\
(\lambda_{3}+k)b(k)e^{i\lambda_{3}x}-iq_{0}e^{-i\theta}a(k)e^{-i\lambda_{3}x}
\end{array}\right)
\end{equation}
as $x\rightarrow +\infty$. We write the real axis as $(-\infty, +\infty)=(-\infty, -L)\cup[-L, 0)\cup[0, L)\cup[L, +\infty)$.
In $[-L, L]$, we can solve
\begin{equation}
\left(\begin{array}{cc}
v_{1}\\
v_{2}
\end{array}\right)_{x}
=
\left(\begin{array}{cc}
-ik& q_{c}e^{i\theta}\\
-q_{c}e^{-i\theta}& ik
\end{array}\right)
\left(\begin{array}{cc}
v_{1}\\
v_{2}
\end{array}\right):=M_{-L}\left(\begin{array}{cc}
v_{1}\\
v_{2}
\end{array}\right)
\end{equation}
for $x$ as $v(x)=T(X)v(-L)$, $T(X)= \exp(X M_{-L})$, where $X:=x+L$, and the matrix $T(X)$ is explicitly written as
\begin{equation}
T(X)=\left(\begin{array}{cc}
\cos \lambda_{4}X-i\frac{k}{\lambda_{4}}\sin \lambda_{4}X& \frac{q_{c}e^{i\theta}}{\lambda_{4}}\sin \lambda_{4}X\\
-\frac{q_{c}e^{-i\theta}}{\lambda_{4}}\sin \lambda_{4}X& \cos \lambda_{4}X+i\frac{k}{\lambda_{4}}\sin \lambda_{4}X
\end{array}\right),
\end{equation}
where $\lambda_{4}=\sqrt{k^{2}+q_{c}^{2}}$. The matrix $T$ is interpreted as a transfer matrix that connects two asymptotic forms in $x\rightarrow +\infty$ and $x\rightarrow -\infty$. We can get
\begin{equation}
\begin{split}
&\left(\begin{array}{cc}
-iq_{0}e^{i\theta}b(k)e^{i\lambda_{3}L}+(\lambda_{3}+k)a(k)e^{-i\lambda_{3}L}\\
(\lambda_{3}+k)b(k)e^{i\lambda_{3}L}-iq_{0}e^{-i\theta}a(k)e^{-i\lambda_{3}L}
\end{array}\right)\\
&=\phi(L,k)=T(2L)\phi(-L,k)=T(2L)\left(\begin{array}{cc}
\lambda_{3}+k\\
-iq_{0}e^{-i\theta}
\end{array}\right)e^{i\lambda_{3}L},
\end{split}
\end{equation}
i.e.,
\begin{equation}
\begin{split}
&\left(\begin{array}{cc}
-iq_{0}e^{i\theta}b(k)e^{i\lambda_{3}L}+(\lambda_{3}+k)a(k)e^{-i\lambda_{3}L}\\
(\lambda_{3}+k)b(k)e^{i\lambda_{3}L}-iq_{0}e^{-i\theta}a(k)e^{-i\lambda_{3}L}
\end{array}\right)\\
&=\left(\begin{array}{cc}
\cos 2\lambda_{4}L-i\frac{k}{\lambda_{4}}\sin 2\lambda_{4}L& \frac{q_{c}e^{i\theta}}{\lambda_{4}}\sin 2\lambda_{4}L\\
-\frac{q_{c}e^{-i\theta}}{\lambda_{4}}\sin 2\lambda_{4}L& \cos 2\lambda_{4}L+i\frac{k}{\lambda_{4}}\sin 2\lambda_{4}L
\end{array}\right)\left(\begin{array}{cc}
\lambda_{3}+k\\
-iq_{0}e^{-i\theta}
\end{array}\right)e^{i\lambda_{3}L}.
\end{split}
\end{equation}
Then
\begin{equation}
a(k)=\frac{e^{2i\lambda_{3}L}(\lambda_{4}(q_{0}^{2}+(k+\lambda_{3})^{2})\cos(2\lambda_{4}L)-i(k^{3}+2k^{2}\lambda_{3}+2q_{0}q_{c}\lambda_{3}
+k(-q_{0}^{2}+2q_{0}q_{c}+\lambda_{3}^{2}))\sin(2\lambda_{4}L))}
{\lambda_{4}(q_{0}^{2}+(k+\lambda_{3})^{2})}.
\end{equation}


Set $k=i\xi$, where $\xi$ is a real number, we will distinguish three subcases.

{\bf Subcase 1.} $|\xi|\leq\min\{q_{0}, q_{c}\}$. We obtain both $\lambda_{3}=\sqrt{-\xi^{2}+q_{0}^{2}}$ and $\lambda_{4}=\sqrt{-\xi^{2}+q_{c}^{2}}$ are positive real numbers. Then $a(i\xi)=0$ if and only if $\xi$ is a root of the following equation:
\begin{equation}
\tan(2\lambda_{4}L)=\frac{\lambda_{4}(q_{0}^{2}-\xi^{2}+2i\xi \lambda_{3}+\lambda_{3}^{2})}{i[-i\xi^{3}-2\xi^{2}\lambda_{3}+2q_{0}q_{c}\lambda_{3}+i\xi(-q_{0}^{2}+2q_{0}q_{c}+\lambda_{3}^{2})]},
\end{equation}
where the left hand side is real, however, the right hand side is complex, which is a contradiction.

{\bf Subcase 2.} $|\xi|\geq\max\{q_{0}, q_{c}\}$. We have $\lambda_{3}=i\sqrt{\xi^{2}-q_{0}^{2}}:=i\widetilde{\lambda}_{3}$
and $\lambda_{4}=i\sqrt{\xi^{2}-q_{c}^{2}}:=i\widetilde{\lambda}_{4}$, where both $\widetilde{\lambda}_{3}$ and $\widetilde{\lambda}_{4}$ are positive real numbers. Then $a(i\xi)=0$ if and only if $\xi$ is a root of the following equation:
\begin{equation}
\frac{e^{-2\widetilde{\lambda}_{4}L}-e^{2\widetilde{\lambda}_{4}L}}{e^{-2\widetilde{\lambda}_{4}L}+e^{2\widetilde{\lambda}_{4}L}}
=\frac{\widetilde{\lambda}_{4}[q_{0}^{2}-(\xi+\widetilde{\lambda}_{3})^{2}]}{-\xi^{3}-2\xi^{2}\widetilde{\lambda}_{3}+2q_{0}q_{c}\widetilde{\lambda}_{3}
+\xi(-q_{0}^{2}+2q_{0}q_{c}-\widetilde{\lambda}_{3}^{2})}.
\end{equation}
Numerically, we do not find solutions to the above equation.


{\bf Subcase 3.} $\min\{q_{0}, q_{c}\}<|\xi|<\max\{q_{0}, q_{c}\}$. If $q_{0}>q_{c}$, then $\lambda_{3}=\sqrt{-\xi^{2}+q_{0}^{2}}$ and $\lambda_{4}=i\sqrt{\xi^{2}-q_{c}^{2}}:=i\widetilde{\lambda}_{4}$, where both $\lambda_{3}$
and $\widetilde{\lambda}_{4}$ are positive real numbers. Then $a(i\xi)=0$ if and only if $\xi$ is a root of the following equation:
\begin{equation}
\frac{e^{-2\widetilde{\lambda}_{4}L}-e^{2\widetilde{\lambda}_{4}L}}{e^{-2\widetilde{\lambda}_{4}L}+e^{2\widetilde{\lambda}_{4}L}}
=\frac{i\widetilde{\lambda}_{4}[q_{0}^{2}+(i\xi+\lambda_{3})^{2}]}{-i\xi^{3}-2\xi^{2}\lambda_{3}+2q_{0}q_{c}\lambda_{3}+i\xi(-q_{0}^{2}+2q_{0}q_{c}+\lambda_{3}^{2})},
\end{equation}
where the left hand side is real, however, the right hand side is complex, which is a contradiction. If $q_{0}<q_{c}$, then $\lambda_{3}=i\sqrt{\xi^{2}-q_{0}^{2}}:=i\widetilde{\lambda}_{3}$ and $\lambda_{4}=\sqrt{-\xi^{2}+q_{c}^{2}}$, where both $\widetilde{\lambda}_{3}$ and $\lambda_{4}$ are positive real numbers. Then $a(i\xi)=0$ if and only if $\xi$ is a root of the following equation:
\begin{equation}
\label{E:box 2}
\tan(2\lambda_{4}L)=\frac{\lambda_{4}(q_{0}^{2}-(\xi+\widetilde{\lambda}_{3})^{2})}{\xi^{3}+2\xi^{2}\widetilde{\lambda}_{3}-2q_{0}q_{c}\widetilde{\lambda}_{3}
-\xi(-q_{0}^{2}+2q_{0}q_{c}-\widetilde{\lambda}_{3}^{2})}.
\end{equation}

Numerically we find solutions to the above equation (\ref{E:box 2}). In Fig. \ref{fig12} below we show two intersections for the given parameters.

\begin{figure}[h]
\begin{tabular}{cc}
\includegraphics[width=0.8\textwidth]{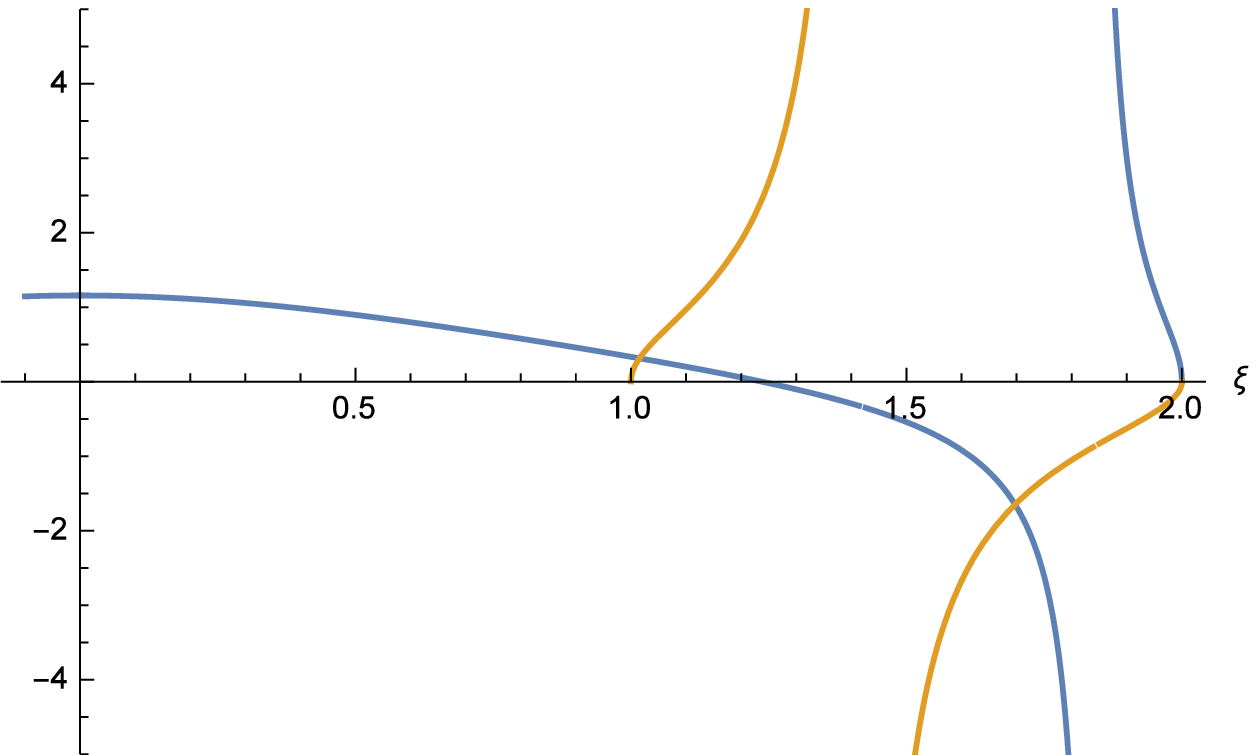}
\end{tabular}
\caption{The LHS (blue) and RHS (orange) of (\ref{E:box 2}) with $L=1$, $q_{0}=1$ and $q_{c}=2$, yielding two intersections between $q_{0}$ and $q_{c}$.}
\label{fig12}
\end{figure}

\begin{remark}
If we set a bigger $L$, then the LHS and RHS of (\ref{E:box 2}) can yield more intersections.
\end{remark}

{\bf Case 3.} $\sigma=1$ with $\theta_{+}-\theta_{-}=0$.
The box-type initial conditions can be specified as follows:
\begin{equation}
q(x,0)=\left\{\begin{array}{ll}
\begin{array}{cc}
q_{0}e^{i\theta}, \ \ \ x\in (-\infty, -L)\cup (L, +\infty),\\
q_{c}e^{i\theta}, \ \ \ x\in [-L, L],
\end{array}\\
\end{array}\right.
\end{equation}
where $q_{0}, q_{c}>0$ and $0\leq \theta< 2\pi$. As we show in Section 5, the two sets of eigenfunctions $\{\phi(x,k), \overline{\phi}(x,k)\}$ and $\{\psi(x,k), \overline{\psi}(x,k)\}$ satisfy the boundary conditions
\begin{equation}
\phi(x,k)\sim \left(\begin{array}{cc}
\lambda_{1}+k\\
iq_{0}e^{-i\theta}
\end{array}\right)e^{-i\lambda_{1}x}, \ \ \
\overline{\phi}(x,k)\sim\left(\begin{array}{cc}
-iq_{0}e^{i\theta}\\
\lambda_{1}+k
\end{array}\right)e^{i\lambda_{1}x}
\end{equation}
as $x\rightarrow -\infty$,
\begin{equation}
\psi(x,k)\sim \left(\begin{array}{cc}
-iq_{0}e^{i\theta}\\
\lambda_{1}+k
\end{array}\right)e^{i\lambda_{1}x}, \ \ \
\overline{\psi}(x,k)\sim\left(\begin{array}{cc}
\lambda_{1}+k\\
iq_{0}e^{-i\theta}
\end{array}\right)e^{-i\lambda_{1}x}
\end{equation}
as $x\rightarrow +\infty$. Note that
\begin{equation}
\phi(x,k)=b(k)\psi(x,k)+a(k)\overline{\psi}(x,k),
\end{equation}
we have
\begin{equation}
\phi(x,k)\sim \left(\begin{array}{cc}
-iq_{0}e^{i\theta}b(k)e^{i\lambda_{1}x}+(\lambda_{1}+k)a(k)e^{-i\lambda_{1}x}\\
(\lambda_{1}+k)b(k)e^{i\lambda_{1}x}+iq_{0}e^{-i\theta}a(k)e^{-i\lambda_{1}x}
\end{array}\right)
\end{equation}
as $x\rightarrow +\infty$.
We write the real axis as $(-\infty, +\infty)=(-\infty, -L)\cup[-L, 0)\cup[0, L)\cup[L, +\infty)$.
In $[-L, L]$, we can solve
\begin{equation}
\left(\begin{array}{cc}
v_{1}\\
v_{2}
\end{array}\right)_{x}
=
\left(\begin{array}{cc}
-ik& q_{c}e^{i\theta}\\
q_{c}e^{-i\theta}& ik
\end{array}\right)
\left(\begin{array}{cc}
v_{1}\\
v_{2}
\end{array}\right):=M_{-L}\left(\begin{array}{cc}
v_{1}\\
v_{2}
\end{array}\right)
\end{equation}
for $x$ as $v(x)=T(X)v(-L)$, $T(X)= \exp(X M_{-L})$, where $X:=x+L$, and the matrix $T(X)$ is explicitly written as
\begin{equation}
T(X)=\left(\begin{array}{cc}
\cos \lambda_{2}X-i\frac{k}{\lambda_{2}}\sin \lambda_{2}X& \frac{q_{c}e^{i\theta}}{\lambda_{2}}\sin \lambda_{2}X\\
\frac{q_{c}e^{-i\theta}}{\lambda_{2}}\sin \lambda_{2}X& \cos \lambda_{2}X+i\frac{k}{\lambda_{2}}\sin \lambda_{2}X
\end{array}\right),
\end{equation}
where $\lambda_{2}=\sqrt{k^{2}-q_{c}^{2}}$. The matrix $T$ is interpreted as a transfer matrix that connects two asymptotic forms in $x\rightarrow +\infty$ and $x\rightarrow -\infty$. We can get
\begin{equation}
\begin{split}
&\left(\begin{array}{cc}
-iq_{0}e^{i\theta}b(k)e^{i\lambda_{1}L}+(\lambda_{1}+k)a(k)e^{-i\lambda_{1}L}\\
(\lambda_{1}+k)b(k)e^{i\lambda_{1}L}+iq_{0}e^{-i\theta}a(k)e^{-i\lambda_{1}L}
\end{array}\right)\\
&=\phi(L,k)=T(2L)\phi(-L,k)=T(2L)\left(\begin{array}{cc}
\lambda_{1}+k\\
iq_{0}e^{-i\theta}
\end{array}\right)e^{i\lambda_{1}L},
\end{split}
\end{equation}
i.e.,
\begin{equation}
\begin{split}
&\left(\begin{array}{cc}
-iq_{0}e^{i\theta}b(k)e^{i\lambda_{1}L}+(\lambda_{1}+k)a(k)e^{-i\lambda_{1}L}\\
(\lambda_{1}+k)b(k)e^{i\lambda_{1}L}+iq_{0}e^{-i\theta}a(k)e^{-i\lambda_{1}L}
\end{array}\right)\\
&=\left(\begin{array}{cc}
\cos 2\lambda_{2}L-i\frac{k}{\lambda_{2}}\sin 2\lambda_{2}L& \frac{q_{c}e^{i\theta}}{\lambda_{2}}\sin 2\lambda_{2}L\\
\frac{q_{c}e^{-i\theta}}{\lambda_{2}}\sin 2\lambda_{2}L& \cos 2\lambda_{2}L+i\frac{k}{\lambda_{2}}\sin 2\lambda_{2}L
\end{array}\right)\left(\begin{array}{cc}
\lambda_{1}+k\\
iq_{0}e^{-i\theta}
\end{array}\right)e^{i\lambda_{1}L}.
\end{split}
\end{equation}
Then
\begin{equation}
a(k)=e^{2i\lambda_{1}L}\left(\cos(2\lambda_{2}L)-\frac{i(k^{3}+2k^{2}\lambda_{1}-2q_{0}q_{c}\lambda_{1}+k(q_{0}^{2}-2q_{0}q_{c}+\lambda_{1}^{2}))\sin(2\lambda_{2}L)}{\lambda_{2}(k-q_{0}+\lambda_{1})(k+q_{0}+\lambda_{1})}\right).
\end{equation}


Set $k=i\xi$, where $\xi$ is a real number, then $\lambda_{1}=i\sqrt{\xi^{2}+q_{0}^{2}}:=i\widetilde{\lambda}_{1}$ and $\lambda_{2}=i\sqrt{\xi^{2}+q_{c}^{2}}:=i\widetilde{\lambda}_{2}$, where both $\widetilde{\lambda}_{1}$ and $\widetilde{\lambda}_{2}$
are positive real numbers. We conclude that if $k=i\xi$ is a zero of $a(k)$, then $\xi$ is a root of the following equation:
\begin{equation}
\label{E:box 3}
\frac{e^{-2\widetilde{\lambda}_{2}L}-e^{2\widetilde{\lambda}_{2}L}}{e^{-2\widetilde{\lambda}_{2}L}+e^{2\widetilde{\lambda}_{2}L}}
=\frac{\widetilde{\lambda}_{2}[(\xi+\widetilde{\lambda}_{1})^{2}+q_{0}^{2}]}{2\xi^{3}+2\xi^{2}\widetilde{\lambda}_{1}+2q_{0}q_{c}\widetilde{\lambda}_{1}+2q_{0}q_{c}\xi}.
\end{equation}

\begin{figure}[h]
\begin{tabular}{cc}
\includegraphics[width=0.8\textwidth]{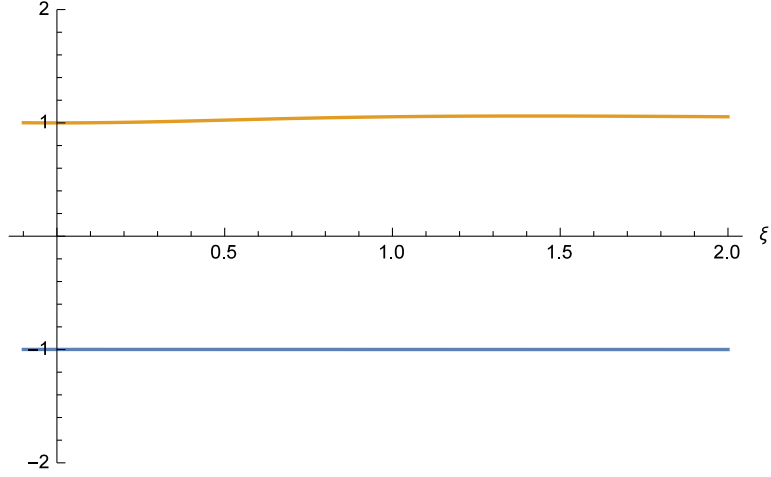}
\end{tabular}
\caption{The LHS (blue) and RHS (orange) of (\ref{E:box 3}) with $L=1$, $q_{0}=1$ and $q_{c}=2$, yielding no intersections for $\xi>0$.}
\label{fig13}
\end{figure}

From Fig. \ref{fig13}, we obtain that there are no eigenvalues lying in the imaginary axis. Next, we consider whether there are eigenvalues outside the imaginary axis.

Set $k=\xi+i\eta$, where both $\xi$ and $\eta$ are real numbers. We conclude that if $k=\xi+i\eta$ is a zero of $a(k)$, the $\xi$ and $\eta$ satisfy the following equation:
\begin{equation}
\label{E:box 4}
\begin{split}
&\tan(2L\sqrt{(\xi+i\eta)^{2}-q_{c}^{2}})\\
&=\frac{\sqrt{(\xi+i\eta)^{2}-q_{c}^{2}}(\xi+i\eta-q_{0}+\sqrt{(\xi+i\eta)^{2}-q_{0}^{2}})(\xi+i\eta+q_{0}+\sqrt{(\xi+i\eta)^{2}-q_{0}^{2}})}
{i[(\xi+i\eta)^{3}+2(\xi+i\eta)^{2}\sqrt{(\xi+i\eta)^{2}-q_{0}^{2}}-2q_{0}q_{c}\sqrt{(\xi+i\eta)^{2}-q_{0}^{2}}+(\xi+i\eta)((\xi+i\eta)^{2}-2q_{0}q_{c})]}.
\end{split}
\end{equation}

Numerically we find solutions to the above equation (\ref{E:box 4}). In Fig. \ref{fig14} below we show three paired intersections, i.e. six intersections total, for the given parameters.

\begin{figure}[h]
\begin{tabular}{cc}
\includegraphics[width=0.8\textwidth]{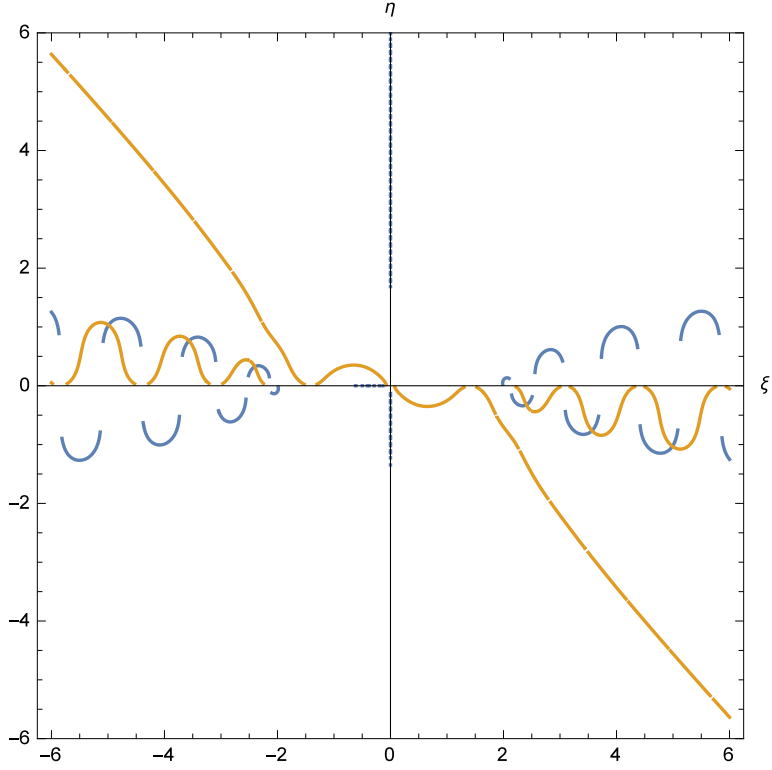}
\end{tabular}
\caption{The preimage of the intersections of real parts (blue) and imaginary parts (orange) of LHS and RHS of (\ref{E:box 4}) with $L=1$, $q_{0}=1$ and $q_{c}=2$.}
\label{fig14}
\end{figure}

\begin{remark}
When $L=1$, $q_{0}=1$ and $q_{c}=2$, $a(k)$ has zeros $k=2.40373-0.32349i$, $k=-2.40373+ 0.32349i$, $k=3.56577-0.764025i$, $k=-3.56577+0.764025i$, $k=4.98907-1.03002i$ and $k=-4.98907+1.03002i$,  i.e., the zeros of $a(k)$ appear in pairs.
\end{remark}

\begin{remark}
If we set a smaller $L$, then the LHS and RHS of (\ref{E:box 4}) can yield less intersections, such as $0$, $2$ and $4$. But they all appear in pairs. Similarly if $L$ is larger we can obtain more intersections. 
\end{remark}

{\bf Case 4.} $\sigma=1$ with $\theta_{+}-\theta_{-}=\pi$.
The box-type initial conditions can be specified as follows:
\begin{equation}
q(x,0)=\left\{\begin{array}{ll}
\begin{array}{cc}
q_{0}e^{i\theta}, \ \ \ x>L,\\
-q_{0}e^{i\theta}, \ \ \ x<-L,\\
q_{c}e^{i\theta}, \ \ \ -L\leq x\leq L,
\end{array}\\
\end{array}\right.
\end{equation}
where $q_{0}, q_{c}>0$ and $0\leq \theta< 2\pi$. As we show in Section 6, the two sets of eigenfunctions $\{\phi(x,k), \overline{\phi}(x,k)\}$ and $\{\psi(x,k), \overline{\psi}(x,k)\}$ satisfy the boundary conditions
\begin{equation}
\phi(x,k)\sim \left(\begin{array}{cc}
\lambda_{3}+k\\
iq_{0}e^{-i\theta}
\end{array}\right)e^{-i\lambda_{3}x}, \ \ \
\overline{\phi}(x,k)\sim\left(\begin{array}{cc}
iq_{0}e^{i\theta}\\
\lambda_{3}+k
\end{array}\right)e^{i\lambda_{3}x}
\end{equation}
as $x\rightarrow -\infty$,
\begin{equation}
\psi(x,k)\sim \left(\begin{array}{cc}
-iq_{0}e^{i\theta}\\
\lambda_{3}+k
\end{array}\right)e^{i\lambda_{3}x}, \ \ \
\overline{\psi}(x,k)\sim\left(\begin{array}{cc}
\lambda_{3}+k\\
-iq_{0}e^{-i\theta}
\end{array}\right)e^{-i\lambda_{3}x}
\end{equation}
as $x\rightarrow +\infty$.  Note that
\begin{equation}
\phi(x,k)=b(k)\psi(x,k)+a(k)\overline{\psi}(x,k),
\end{equation}
we have
\begin{equation}
\phi(x,k)\sim \left(\begin{array}{cc}
-iq_{0}e^{i\theta}b(k)e^{i\lambda_{3}x}+(\lambda_{3}+k)a(k)e^{-i\lambda_{3}x}\\
(\lambda_{3}+k)b(k)e^{i\lambda_{3}x}-iq_{0}e^{-i\theta}a(k)e^{-i\lambda_{3}x}
\end{array}\right)
\end{equation}
as $x\rightarrow +\infty$. We write the real axis as $(-\infty, +\infty)=(-\infty, -L)\cup[-L, L)\cup[L, +\infty)$.
In $[-L, L)$, we can solve
\begin{equation}
\left(\begin{array}{cc}
v_{1}\\
v_{2}
\end{array}\right)_{x}
=
\left(\begin{array}{cc}
-ik& q_{c}e^{i\theta}\\
q_{c}e^{-i\theta}& ik
\end{array}\right)
\left(\begin{array}{cc}
v_{1}\\
v_{2}
\end{array}\right):=M_{-L}\left(\begin{array}{cc}
v_{1}\\
v_{2}
\end{array}\right)
\end{equation}
for $x$ as $v(x)=T(X)v(-L)$, $T(X)= \exp(X M_{-L})$, where $X:=x+L$, and the matrix $T(X)$ is explicitly written as
\begin{equation}
T(X)=\left(\begin{array}{cc}
\cos \lambda_{2}X-i\frac{k}{\lambda_{2}}\sin \lambda_{2}X& \frac{q_{c}e^{i\theta}}{\lambda_{2}}\sin \lambda_{2}X\\
\frac{q_{c}e^{-i\theta}}{\lambda_{2}}\sin \lambda_{2}X& \cos \lambda_{2}X+i\frac{k}{\lambda_{2}}\sin \lambda_{2}X
\end{array}\right).
\end{equation}
Thus,
\begin{equation}
v(L)=T(2L)v(-L).
\end{equation}
The matrix $T$ is interpreted as a transfer matrix that connects two asymptotic forms in $x\rightarrow +\infty$ and $x\rightarrow -\infty$. We can get
\begin{equation}
\begin{split}
&\left(\begin{array}{cc}
-iq_{0}e^{i\theta}b(k)e^{i\lambda_{3}x}+(\lambda_{3}+k)a(k)e^{-i\lambda_{3}x}\\
(\lambda_{3}+k)b(k)e^{i\lambda_{3}x}-iq_{0}e^{-i\theta}a(k)e^{-i\lambda_{3}x}
\end{array}\right)\\
&=\phi(L,k)=T(2L)\phi(-L,k)=T(2L)\left(\begin{array}{cc}
\lambda_{3}+k\\
iq_{0}e^{-i\theta}
\end{array}\right)e^{i\lambda_{3}L},
\end{split}
\end{equation}
i.e.,
\begin{equation}
\begin{split}
&\left(\begin{array}{cc}
-iq_{0}e^{i\theta}b(k)e^{i\lambda_{3}x}+(\lambda_{3}+k)a(k)e^{-i\lambda_{3}x}\\
(\lambda_{3}+k)b(k)e^{i\lambda_{3}x}-iq_{0}e^{-i\theta}a(k)e^{-i\lambda_{3}x}
\end{array}\right)\\
&=\left(\begin{array}{cc}
\cos 2\lambda_{2}L-i\frac{k}{\lambda_{2}}\sin 2\lambda_{2}L& \frac{q_{c}e^{i\theta}}{\lambda_{2}}\sin 2\lambda_{2}L\\
\frac{q_{c}e^{-i\theta}}{\lambda_{2}}\sin 2\lambda_{2}L& \cos 2\lambda_{2}L+i\frac{k}{\lambda_{2}}\sin 2\lambda_{2}L
\end{array}\right)\cdot\left(\begin{array}{cc}
\lambda_{3}+k\\
iq_{0}e^{-i\theta}
\end{array}\right)e^{i\lambda_{3}L}.
\end{split}
\end{equation}
Then
\begin{equation}
\begin{split}
&a(k)=\frac{e^{2i\lambda_{3}L}}{\lambda_{2}(q_{0}^{2}+(k+\lambda_{3})^{2})}
(\lambda_{2}(k-q_{0}+\lambda_{3})(k+q_{0}+\lambda_{3})\cos(2\lambda_{2}L)\\
&-i(k^{3}+2k^{2}\lambda_{3}-2q_{0}q_{c}\lambda_{3}+k(q_{0}^{2}
-2q_{0}q_{c}+\lambda_{3}^{2}))\sin(2\lambda_{2}L)).
\end{split}
\end{equation}


Set $k=i\xi$, where $\xi$ is a real number, we will distinguish two subcases.

{\bf Subcase 1.} $|\xi|\leq q_{0}$. We have $\lambda_{3}=\sqrt{-\xi^{2}+q_{0}^{2}}$ and $\lambda_{2}=i\sqrt{\xi^{2}+q_{c}^{2}}:=i\widetilde{\lambda}_{2}$, where both $\lambda_{3}$ and $\widetilde{\lambda}_{2}$ are positive real numbers. Then $a(i\xi)=0$ if and only if $\xi$ is a root of the following equation:
\begin{equation}
\frac{e^{-2\widetilde{\lambda}_{2}L}-e^{2\widetilde{\lambda}_{2}L}}{e^{-2\widetilde{\lambda}_{2}L}+e^{2\widetilde{\lambda}_{2}L}}
=\frac{i\widetilde{\lambda}_{2}(i\xi-q_{0}+\lambda_{3})(i\xi+q_{0}+\lambda_{3})}{-i\xi^{3}-2\xi^{2}\lambda_{3}-2q_{0}q_{c}\lambda_{3}+i\xi(q_{0}^{2}-2q_{0}q_{c}+\lambda_{3}^{2})},
\end{equation}
where both the left hand side is real, however, both the right hand side is complex, which is a contradiction.

{\bf Subcase 2.} $|\xi|>q_{0}$. We have $\lambda_{3}=i\sqrt{\xi^{2}-q_{0}^{2}}:=i\widetilde{\lambda}_{3}$
and $\lambda_{2}=i\sqrt{\xi^{2}+q_{c}^{2}}:=i\widetilde{\lambda}_{2}$, where both $\widetilde{\lambda}_{3}$ and $\widetilde{\lambda}_{2}$ are positive real numbers. Then $a(i\xi)=0$ if and only if $\xi$ is a root of the following equation:
\begin{equation}
\label{E:box 5}
\frac{e^{-2\widetilde{\lambda}_{2}L}-e^{2\widetilde{\lambda}_{2}L}}{e^{-2\widetilde{\lambda}_{2}L}+e^{2\widetilde{\lambda}_{2}L}}
=\frac{\widetilde{\lambda}_{2}[(\xi+\widetilde{\lambda}_{3})^{2}+q_{0}^{2}]}{2\xi^{3}+2\xi^{2}\widetilde{\lambda}_{3}+2q_{0}q_{c}\widetilde{\lambda}_{3}-2q_{0}^{2}\xi+2q_{0}q_{c}\xi}.
\end{equation}

\begin{figure}[h]
\begin{tabular}{cc}
\includegraphics[width=0.8\textwidth]{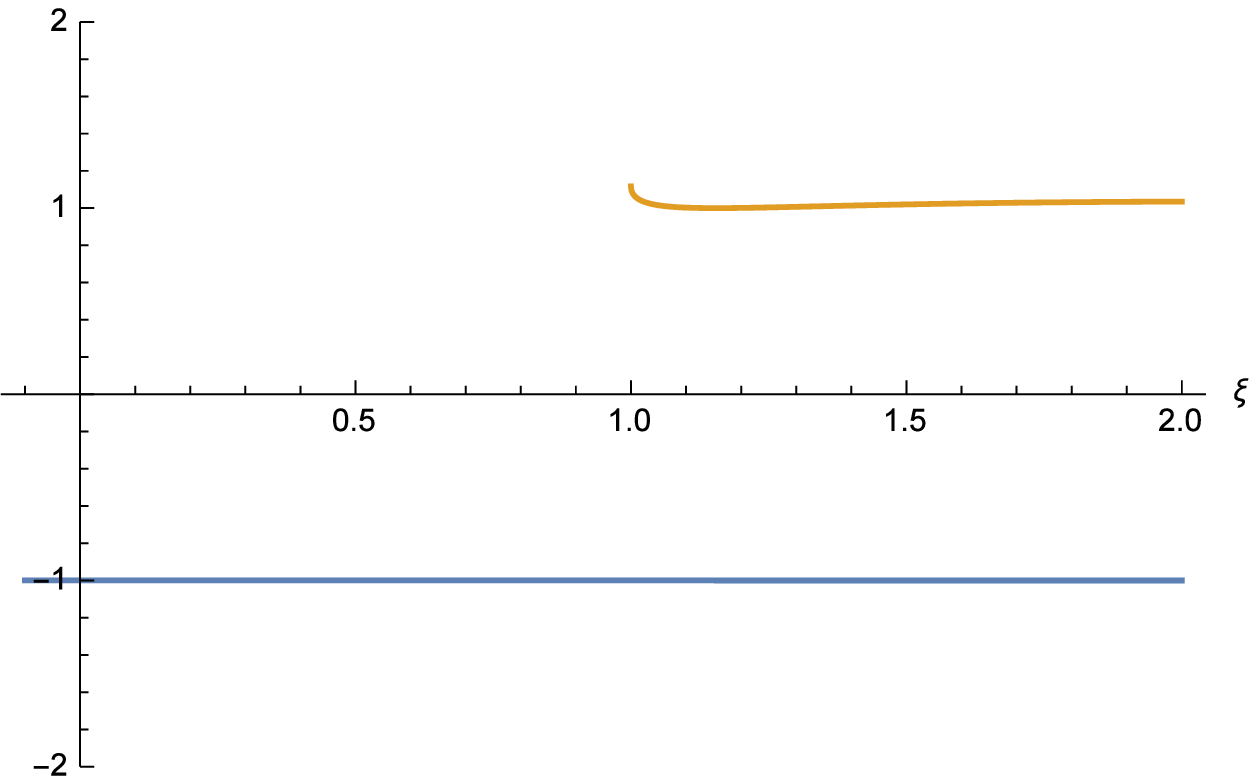}
\end{tabular}
\caption{The LHS (blue) and RHS (orange) of (\ref{E:box 5}) with $L=1$, $q_{0}=1$ and $q_{c}=2$, yielding no intersections for $\xi>q_{0}$.}
\label{fig15}
\end{figure}

From Fig \ref{fig15}, we obtain that there are no eigenvalues lying in the imaginary axis. Next, we consider whether there are eigenvalues outside the imaginary axis.

Set $k=\xi+i\eta$, where both $\xi$ and $\eta$ are real numbers. We conclude that if $k=\xi+i\eta$ is a zero of $a(k)$, the $\xi$ and $\eta$ satisfy the following equation:
\begin{equation}
\label{E:box 6}
\begin{split}
&\tan(2L\sqrt{(\xi+i\eta)^{2}-q_{c}^{2}})\\
&=\frac{\sqrt{(\xi+i\eta)^{2}-q_{c}^{2}}(\xi+i\eta-q_{0}+\sqrt{(\xi+i\eta)^{2}+q_{0}^{2}})(\xi+i\eta+q_{0}+\sqrt{(\xi+i\eta)^{2}+q_{0}^{2}})}
{i[(\xi+i\eta)^{3}+2(\xi+i\eta)^{2}\sqrt{(\xi+i\eta)^{2}+q_{0}^{2}}-2q_{0}q_{c}\sqrt{(\xi+i\eta)^{2}+q_{0}^{2}}+(\xi+i\eta)(-(\xi+i\eta)^{2}-2q_{0}q_{c})]}.
\end{split}
\end{equation}
Numerically, we do not find solutions to the above equation (\ref{E:box 6}). In Figure \ref{fig16} below we show no solutions for (\ref{E:box 6}).

\begin{figure}[h]
\begin{tabular}{cc}
\includegraphics[width=0.8\textwidth]{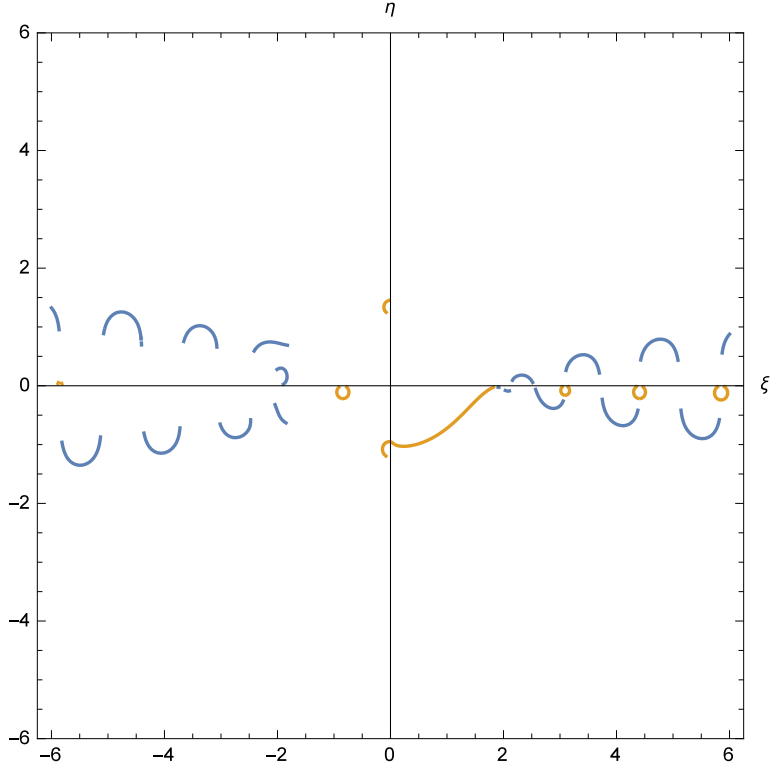}
\end{tabular}
\caption{The preimage of the intersections of real parts (blue) and imaginary parts (orange) of LHS and RHS of (\ref{E:box 6}) with $L=1$, $q_{0}=1$ and $q_{c}=2$.}
\label{fig16}
\end{figure}

\begin{remark}
When $L=1$, $q_{0}=1$ and $q_{c}=2$, $a(k)$ does not have zeros. We can obtain the same argument for other values of $L$, $q_{0}$ and $q_{c}$. In conclusion, it does not have soliton solution for Case 4.
\end{remark}
\newpage

\section{Modulational instability}
In this section we study modulation instability of a constant (in space) amplitude solution to the nonlocal NLS equation (\ref{E:nonlocal NLS}). More precisely, we study linear stability properties of the family of solutions defined by
\begin{equation}
\label{constant}
q(x,t) = A e^{i\omega t}\;,
\end{equation}
where $A$ is a constant complex amplitude and $\omega$ is a real frequency.
In the limit $x \rightarrow +\infty$ we define $A\equiv A_{+}$ and when $x \rightarrow -\infty$ we let $A\equiv A_{-}.$ If we let $A_{\pm} = q_0e^{i\theta_{\pm}}$ then in order for Eq.~(\ref{constant}) to be a solution to (\ref{E:nonlocal NLS}) we require that
$\theta_{+} - \theta_{-} = 0$ or $\pi.$ Thus, we distinguish here between four different cases:
$\theta_{+} - \theta_{-} = 0, \sigma = \pm 1$ and $\theta_{+} - \theta_{-} = \pi, \sigma = \pm 1.$
\subsection{Modulational instability I: $\theta_{+} - \theta_{-} = 0$}
In this case we have $ A_{+} =  A_{-}\equiv A$ and $\omega = 2\sigma |A|^2.$
To study modulational instability we write solutions to (\ref{E:nonlocal NLS}) in the form
\begin{equation}
\label{constant-pert}
q(x,t) = (A + \varepsilon \tilde{q}(x,t))e^{2 i\sigma |A|^2 t}\;.
\end{equation}
Substituting (\ref{constant-pert}) into (\ref{E:nonlocal NLS}) and retaining terms up to order $\varepsilon$ we obtain
\begin{equation}
\label{pert:nonlocal NLS}
i \tilde{q}_{t}(x,t) = \tilde{q}_{xx}(x,t) - 2\sigma |A|^2 \tilde{q}(x,t) - 2\sigma A^2 \tilde{q}^*(-x,t)\;.
\end{equation}
Equation (\ref{pert:nonlocal NLS}) governs the dynamical evolution of the perturbation
$\tilde{q}(x,t)$ subject to the boundary conditions $\tilde{q}(x,t)$ decays very fast at plus and minus infinity. We next decompose the perturbation in terms of its Fourier integral representation
\begin{equation}
\label{pert:four}
\tilde{q}(x,t) = \int_{-\infty}^{+\infty}
\left(
\tilde{q}_{1}(k) e^{i\lambda t} + \tilde{q}_{2}(k) e^{-i\lambda^* t}
\right) e^{ikx} dk
\;.
\end{equation}
It is evident that the constant solution is modulationally unstable if the eigenvalue $\lambda$ is complex. After some algebra we find
\begin{equation}
\label{eigs}
\lambda^2 = k^2(k^2 +4\sigma |A|^2)\;.
\end{equation}
Thus we conclude that  when $\sigma = +1$ the solution (\ref{constant}) is modulationaly stable whereas it is modulationaly unstable for $\sigma = -1.$

The unstable case corresponds to having two eigenvalues on the imaginary axis, or two solitons that are stationary and oscillate in time, as discussed in case 4.
\subsection{Modulational instability II: $\theta_{+} - \theta_{-} = \pi$}
In this second case we find $ A_{+} = -  A_{-}\equiv A$ thus
$\omega = -2\sigma |A|^2.$ Under this situation, the linear stability analysis proceeds as before by writing a perturbative solutions to (\ref{E:nonlocal NLS}) in the form (\ref{constant-pert}) with the exception that $q^*(-x,t) = (-A^* + \varepsilon \tilde{q}^*(-x,t))e^{-2 i\sigma |A|^2 t}.$
Substituting (\ref{constant-pert}) into (\ref{E:nonlocal NLS}) and retaining terms up to order
$\varepsilon$ we now obtain
\begin{equation}
\label{pert:nonlocal NLS-II}
i \tilde{q}_{t}(x,t) = \tilde{q}_{xx}(x,t) + 2\sigma |A|^2 \tilde{q}(x,t) - 2\sigma A^2 \tilde{q}^*(-x,t)\;.
\end{equation}
Decomposing the perturbation in terms of its Fourier integral representation (\ref{pert:four}) we then find
\begin{equation}
\label{eigsII}
\lambda^2 = (k^2 - 2\sigma |A|^2)^2 + 4|A|^2\;.
\end{equation}
It is thus evident that the constant amplitude solution is modulationally stable for both signs of
 $\sigma.$

\section{Conclusion}
In this paper the nonlocal NLS equation  (\ref{E:nonlocal NLS}) with nonzero boundary values at infinity is considered. Depending on the signs of nonlinearity: $\sigma= \pm 1$  and phase difference of the boundary values from plus infinity to minus infinity: $\Delta \theta= \theta_+-\theta_-$, there are four distinct cases to consider when the amplitude at infinity is constant.

The direct scattering problem and corresponding analytic properties and symmetries of the eigenfunctions and scattering data are obtained. The inverse problem is constructed via a Riemann-Hilbert problem formulated  from both the right and left in terms of a convenient uniformization variable.
The simplest pure soliton solutions are obtained.  When $\sigma=-1, \Delta \theta= \pi$ a pure one soliton solution is found; when  $\sigma=-1, \Delta \theta= 0$ a pure stationary two soliton solution is obtained;
and with $\sigma=1, \Delta \theta= 0$ a traveling bidirectional interacting two soliton solution is found.
In the final case $\sigma=1, \Delta \theta= \pi$, there are no soliton solutions allowed. A number of box potentials are analyzed; their associated eigenvalues are found and in all case shown to be consistent with the prior analytical results. Modulational instability of a plane wave is also discussed.

\section{Acknowledgements}
The authors are pleased to acknowledge Justin T. Cole and Yi-Ping Ma for considerable assistance in the figures of this paper and in Mathematica respectively. MJA was partially supported by NSF under Grant No. DMS-1310200.

\bibliographystyle{amsplain}

\end{document}